\documentclass[letterpaper,12pt,titlepage,oneside,final]{book}
 
\newcommand{\href}[1]{#1} 

\usepackage{ifthen}
\newboolean{PrintVersion}
\setboolean{PrintVersion}{false} 

\usepackage{nomencl} 
\usepackage{amsmath,amssymb,amstext,amsthm}
\usepackage{array}
\usepackage{braket}
\usepackage{bbold}
\usepackage{bbm}
\usepackage{commath}
\usepackage{enumerate}
\usepackage{fancyhdr}
\usepackage{float}
\usepackage[pdftex]{graphicx} 
\usepackage{hhline}
\usepackage{mathrsfs}
\usepackage{mathtools} 
\usepackage{multirow}
\usepackage{stmaryrd}
\usepackage{subcaption}
\usepackage{tcolorbox}
\usepackage{tikz,pgfplots}
\usepackage{thmtools}
\usepackage{xcolor}

\usepackage[T2A]{fontenc}
\usepackage[utf8]{inputenc}
\usepackage[russian,main=english]{babel}

\usetikzlibrary{arrows.meta}

\usepackage{csquotes}

\usepackage[backend=bibtex,isbn=false,url=false,sorting=nyt,style=numeric-comp]{biblatex}

\appto{\bibsetup}{\raggedright}
\usepackage[pdftex,pagebackref=false]{hyperref} 
\hypersetup{
    plainpages=false,       
    unicode=false,          
    pdftoolbar=true,        
    pdfmenubar=true,        
    pdffitwindow=false,     
    pdfstartview={FitH},    
    pdftitle={Locality and Exceptional Points in Pseudo-Hermitian Physics},  
    pdfauthor={Jacob L. Barnett},    
    pdfnewwindow=true,      
    colorlinks=true,        
    linkcolor=blue,         
    citecolor=green,        
    filecolor=magenta,      
    urlcolor=cyan           
}
\ifthenelse{\boolean{PrintVersion}}{   
\hypersetup{	
    citecolor=black,%
    filecolor=black,%
    linkcolor=black,%
    urlcolor=black}
}{} 

\usepackage{cleveref}

\usepackage[automake,toc,abbreviations,nonumberlist
]{glossaries-extra} 
\usepackage{glossary-longextra}
\setglossarystyle{long-name-desc}

\newcommand{\kket}[1]{\vert #1 \rangle\rangle}
\newcommand{\bbra}[1]{\langle\langle #1 \vert}

\newcommand{\Mod}[1]{\ (\mathrm{mod}\ #1)}
\newcommand\mtiny[1]{\mbox{\tiny\ensuremath{#1}}}

\newcommand{\dbrac}[1]{\llbracket #1 \rrbracket}

\newcommand\xqed[1]{%
  \leavevmode\unskip\penalty9999 \hbox{}\nobreak\hfill
  \quad\hbox{#1}}
\newcommand\demo{\xqed{$\triangle$}}

\let\origdoublepage\cleardoublepage
\newcommand{\clearemptydoublepage}{%
  \clearpage{
  
  \pagestyle{empty}
\renewcommand{\headrulewidth}{0pt}

  \origdoublepage}}
\let\cleardoublepage\clearemptydoublepage

\tcbuselibrary{theorems}

\newtheorem{theorem}{Theorem}[section]
\newtheorem{thm}[theorem]{Theorem}
\theoremstyle{definition}
\newtheorem{defn}[theorem]{Definition}
\theoremstyle{definition}
\newtheorem{definition}[theorem]{Definition}
\theoremstyle{definition}
\newtheorem{proposition}[theorem]{Proposition}
\newtheorem{lemma}[theorem]{Lemma}
\newtheorem{conjecture}[theorem]{Conjecture}
\newtheorem{corollary}{Corollary}[theorem]

\newtheorem{ex}{Example}[chapter]

\renewcommand*{\arraystretch}{1.2}
\newglossary*{symbols}{List of Symbols}

\newglossaryentry{AbsoluteValue}
{
name={$|z|$},
sort={label},
type=symbols,
description={The \textbf{absolute value} of $z \in \mathbb{C}$.}
}

\newglossaryentry{Adjoint}
{
name={$A^\dag$},
sort={label},
type=symbols,
description={The \textbf{adjoint} of an operator $A:\text{Dom}(A) \to \mathcal{H}$, see \cref{adjoint}.}
}

\newglossaryentry{B(H)}
{
name={\ensuremath{\mathcal{B}(\mathcal{H})}},
sort={label},
type=symbols,
description={The set of \textbf{bounded operators} on $\mathcal{H}$, see \cref{boundedOperators}.}
}

\newglossaryentry{bra}
{
name={$\bra{\psi}$},
sort={label},
type=symbols,
description={A \textbf{bra}, the linear functional on a Hilbert space, $\bra{\psi}:\mathcal{H} \to \mathbb{C}$, defined by $\bra{\psi}(\phi) = \braket{\psi|\phi}$.}
}

\newglossaryentry{ei}
{
name={\ensuremath{e_i}},
sort={label},
type=symbols,
description={An element of the \textbf{canonical basis} of $\mathbb{C}^n$, where $(e_i)_j = \delta^i_j$.}
}

\newglossaryentry{card}
{
name={\ensuremath{\text{card}}},
sort={label},
type=symbols,
description={The \textbf{cardinality} of a set.}
}

\newglossaryentry{ChebyshevFirst}
{
name={$T_n$},
sort={label},
type=symbols,
description={The $n$-th \textbf{Chebyshev polynomial of the first kind}. See \cref{ChebyshevTable}.}
}

\newglossaryentry{ChebyshevSecond}
{
name={$U_n$},
sort={label},
type=symbols,
description={The $n$-th \textbf{Chebyshev polynomial of the second kind}. See \cref{ChebyshevTable}.}
}

\newglossaryentry{cl}
{
name={\ensuremath{\text{cl}}},
sort={label},
type=symbols,
description={The \textbf{closure}. Given a topological space, $(X, \tau)$, and a subset $S\in X$, $\text{cl}(S)$ is the intersection of every closed set in $X$ containing $S$.}
}

\newglossaryentry{Complement}
{
name={$A^c$},
sort={label},
type=symbols,
description={The \textbf{complement} of a set.}
}

\newglossaryentry{ComplexVectors}
{
name = {$\mathbb{C}^n$},
sort = {label},
type = symbols,
description = {The $n$-dimensional \textbf{complex coordinate space}, which is the set of ordered $n$-tuples of complex numbers. $\mathbb{C}^n$ is a Hilbert space with the inner product $\braket{z|w} = \sum_{i = 1}^n z^*_i  w_i$, and a $C^*$-algebra where multiplication of $w, z \in \mathbb{C}^n$ is given by the dot product, $(w \cdot z)_i := w_i z_i$, and the involution is $(z^*)_i= \overline{z_i}$.}
}

\newglossaryentry{commutator}
{
name={$[A,B]_\pm$},
sort={label},
type=symbols,
description={The \textbf{commutator} (-) or \textbf{anti-commutator} (+) of two elements of an associative algebra, $[A,B]_\pm := AB \pm BA$.}
}

\newglossaryentry{C*algebra}
{
name={$\mathfrak{A}$},
sort={label},
type=symbols,
description={A $C^*$-\textbf{algebra}, see \cref{c*Defn}.}
}

\newglossaryentry{delta}
{
type=symbols,
name= $\delta^{i}_{j}$,
description={The \textbf{Kronecker Delta}, defined by $\delta^{i}_{j} := \begin{cases}1 & \text{ if } i = j\\ 0 & \text{ otherwise }. \end{cases}$}
}

\newglossaryentry{Derivative}
{
name={$\dot{f}$},
sort={label},
type=symbols,
description={The \textbf{derivative} of a function, $f$, where either $\text{Dom}(f) = \mathbb{R}$ or $\text{Dom}(f) = \mathbb{C}$.}
}

\newglossaryentry{det}
{
name={\ensuremath{\text{det}}},
sort={label},
type=symbols,
description={The \textbf{determinant}.}
}

\newglossaryentry{dimension}
{
name={$\text{dim}(V)$},
sort={label},
type=symbols,
description={The \textbf{dimension} of the vector space $V$.}
}

\newglossaryentry{DirectProduct}
{
name={$\times$},
sort={label},
type=symbols,
description={The \textbf{direct product}.}
}

\newglossaryentry{DirectSum}
{
name={$\oplus$},
sort={label},
type=symbols,
description={The \textbf{direct sum}.}
}

\newglossaryentry{Domain}
{
name={$\text{Dom}(f)$},
sort={label},
type=symbols,
description={The \textbf{domain} of a function, $f$.}
}

\newglossaryentry{GLnA}
{
name = {\ensuremath{\text{GL}_n(\mathfrak{A})}},
sort = {label},
type = symbols,
description = {The \textbf{general linear group} over $\gls{MnA}$, see \cref{inverseDefn} and \cref{GnAdefn}. I use shorthand $\text{GL}(\mathfrak{A}) = \text{GL}_1(\mathfrak{A})$.}
}

\newglossaryentry{HilbertSpace}
{
name={$\mathcal{H}$},
sort={label},
type=symbols,
description={A separable \textbf{Hilbert space}, see \cref{HilbertSpaceDefn}.}
}

\newglossaryentry{id}
{
name = {\ensuremath{\mathbb{1}}},
sort = {label},
type = symbols,
description = {The \textbf{identity map} on a vector space. $\mathbb{1}$ also refers to the \textbf{multiplicative identity} in a unital associative algebra.}
}

\newglossaryentry{Im}
{
name = {\ensuremath{\text{Im}(z)}},
sort = {label},
type = symbols,
description = {The \textbf{imaginary part} of the complex number $z \in \mathbb{C}$.}
}

\newglossaryentry{imagUnit}
{
type=symbols,
name= $\mathfrak{i}$,
description={The \textbf{imaginary unit}.}
}

\newglossaryentry{InnerProduct}
{
name={$\braket{\psi|\phi}$},
sort={label},
type=symbols,
description={An \textbf{inner product}, which is linear in $\phi$ and anti-linear in $\psi$. See \cref{innerProduct}.}
}

\newglossaryentry{ComplexConjugate}
{
name = {$z^*$},
sort = {label},
type = symbols,
description = {The \textbf{involution} applied to an element of a $C^*$-algebra. If $z \in \mathbb{C}$, then $z^*$ denotes its \textbf{complex conjugate}.}
}

\newglossaryentry{ker}
{
name={\ensuremath{\text{ker}(A)}},
sort={label},
type=symbols,
description={The \textbf{kernel} or \textbf{nullspace} of the operator $A$.}
}

\newglossaryentry{ket}
{
name={$\ket{\psi}$},
sort={label},
type=symbols,
description={A \textbf{ket}, an element of a Hilbert space.}
}

\newglossaryentry{MnA}
{
name = {\ensuremath{\mathfrak{M}_n(\mathfrak{A})}},
sort = {label},
type = symbols,
description = {The \textbf{algebra of} $n \times n$  \textbf{matrices} whose entries are elements of the associative algebra $\mathfrak{A}$, see \cref{matrixAlgebraExample}.}
}

\newglossaryentry{N}
{
name={\ensuremath{\mathbb{N}}},
sort={label},
type=symbols,
description={The set of \textbf{non-negative integers}, which includes zero.}
}

\newglossaryentry{norm}
{
name={$||\cdot||$},
sort={label},
type=symbols,
description={A \textbf{norm}, see \cref{normedSpace}.}
}

\newglossaryentry{ortho}
{
name={$S^\perp$},
sort={label},
type=symbols,
description={The \textbf{orthogonal complement} of a subset of an inner product space, see \cref{OrthoComplement-Defn}.}
}

\newglossaryentry{PositiveElements}
{
name={$\mathfrak{A}^+$},
sort={label},
type=symbols,
description={The set of \textbf{positive elements} of a $C^*$ algebra, $\mathfrak{A}$. See \cref{positivityDef}.}
}

\newglossaryentry{pwr}
{
name={\ensuremath{\mathbb{P}(S)}},
sort={label},
type=symbols,
description={The \textbf{power set} of a set $S$, also referred to as the \textbf{discrete topology}.}
}

\newglossaryentry{Product}
{
name={$\prod\limits_{i \in S} A_i$},
sort={label},
type=symbols,
description={The \textbf{product} over a finite, well-ordered set, $S$, of elements $A_i$ of an associative algebra. If $S = \emptyset$, we define $\prod_{i \in S} A_i = 1$. Since countable sets are totally ordered, we will define $\prod_{i \in S} A_i = A_l \prod_{i \in S \setminus \{l\}} A_i$, where $l \in S$ is the least element.}
}

\newglossaryentry{proj}
{
name={\ensuremath{\mathscr{P}_V}},
sort={label},
type=symbols,
description={The \textbf{orthogonal projection} onto a closed linear subspace of a Hilbert space, $V$, see \cref{orthoProjection}.}
}

\newglossaryentry{ran}
{
name = {\ensuremath{\mathcal{R}}},
sort = {label},
type = symbols,
description = {The \textbf{range}, or \textbf{image}, of a function.}
}

\newglossaryentry{Re}
{
name = {\ensuremath{\text{Re}(z)}},
sort = {label},
type = symbols,
description = {The \textbf{real part} of the complex number $z \in \mathbb{C}$.}
}

\newglossaryentry{rel-comp}
{
name={\ensuremath{\setminus}},
sort={label},
type=symbols,
description={The \textbf{relative complement} or \textbf{set difference}.}
}

\newglossaryentry{span}
{
name={\ensuremath{\text{span}}},
sort={label},
type=symbols,
description={The \textbf{linear span}.}
}

\newglossaryentry{spectrum}
{
name={$\sigma(A)$},
sort={label},
type=symbols,
description={The \textbf{spectrum} of an operator or an element of a unital associative algebra, see \cref{spectrumDefn}.}
}

\newglossaryentry{tensorProduct}
{
name={$\otimes$},
sort={label},
type=symbols,
description={The \textbf{tensor product}.}
}

\newglossaryentry{tr}
{
name={\ensuremath{\text{Tr}}},
sort={label},
type=symbols,
description={The \textbf{trace}, which is the sum of diagonal elements of a square matrix.}
}

\newglossaryentry{Z}
{
name={\ensuremath{\mathbb{Z}}},
sort={label},
type=symbols,
description={The set of \textbf{integers}.}
}

\newglossaryentry{Z+}
{
name={\ensuremath{\mathbb{Z}_+}},
sort={label},
type=symbols,
description={The set of \textbf{positive integers}, which excludes zero.}
}

\newglossaryentry{stMaryRd}
{
name={$\dbrac{m,n}$},
sort={label},
type=symbols,
description={Given $m, n \in \mathbb{Z}$, $\dbrac{m,n}$ is shorthand for $\{m, \dots, n\}$. I also use the shorthand $\dbrac{n} := \dbrac{1,n}$.}
}

\makeglossaries

\DeclarePairedDelimiter{\floor}{\lfloor}{\rfloor}

\begin{document}

\pagestyle{empty}
\pagenumbering{roman}

\begin{titlepage}
        \begin{center}
        \vspace*{1.0cm}

        \Huge
        {\bf Locality and Exceptional Points in Pseudo-Hermitian Physics}

        \vspace*{1.0cm}

        \normalsize
        by \\

        \vspace*{1.0cm}

        \Large
        Jacob L. Barnett \\

        \vspace*{3.0cm}

        \normalsize
        A thesis \\
        presented to the University of Waterloo \\ 
        in fulfilment of the \\
        thesis requirement for the degree of \\
        Doctor of Philosophy \\
        in \\
        Physics \\

        \vspace*{2.0cm}

        Waterloo, Ontario, Canada, 2023 \\

        \vspace*{1.0cm}

        \copyright\ Jacob L. Barnett 2023 \\
        \end{center}
\end{titlepage}

\pagestyle{plain}
\setcounter{page}{2}

\cleardoublepage 

\begin{center}\textbf{Examining Committee Membership}\end{center}
  \noindent
The following served on the Examining Committee for this thesis. The decision of the Examining Committee is by majority vote.
  \bigskip
  
  \noindent
\begin{tabbing}
Internal-External Member: \=  \kill 
External Examiner: \>  Fabio Bagarello, PhD \\ 
\> Professor \\ 
\> University of Palermo \\
\end{tabbing} 
  \bigskip
  
  \noindent
\begin{tabbing}
Internal-External Member: \=  \kill 
Supervisors: \> James A. Forrest, PhD \\
\> Professor \\ 
\> University of Waterloo \\
\\
\> Yogesh N. Joglekar, PhD  \\
\> Professor \\ 
\> Indiana University-Purdue University Indianapolis \\
\end{tabbing}
  \bigskip
  
  \noindent
  \begin{tabbing}
Internal-External Member: \=  \kill 
Internal Member: \> Roger Melko, PhD \\
\> Professor \\
\> University of Waterloo \\
\end{tabbing}
  \bigskip
  
  \noindent
\begin{tabbing}
Internal-External Member: \=  \kill 
Internal-External Member: \> Marek Stastna, PhD \\
\> Professor\\ 
\> University of Waterloo \\
\end{tabbing}
  \bigskip
  
  \noindent
\begin{tabbing}
Internal-External Member: \=  \kill 
Other Member: \> Sung-Sik Lee, PhD \\
\> Professor \\ 
\> McMaster University \\
\end{tabbing}

\cleardoublepage

  \noindent
  
\begin{center}\textbf{Author's Declaration}\end{center}  
  
This thesis consists of material all of which I authored or co-authored: see the Statement of Contributions included in the thesis. This is a true copy of the thesis, including any required final revisions, as accepted by my examiners. 

  \bigskip
  
  \noindent
I understand that my thesis may be made electronically available to the public.

\cleardoublepage

\cleardoublepage

  \noindent
  
  
\begin{center}\textbf{Statement of Contributions}\end{center}  
  
Material from this thesis is presented in two manuscripts. Section~\ref{TPMSection}, \cref{FermionObservableAlgebraSection}, and \cref{FermionsIntro} are based on \cite{Barnett_2021}, which I am the sole author of. Chapter~\ref{toyModelsChapter} and \cref{epSurfaceSingularities} is based on \cite{Barnett2023}, co-authored with Yogesh N. Joglekar. In addition, this thesis contains unpublished original research in \cref{Tactics} and \cref{Local-Equivalence-Section}.

\cleardoublepage


\begin{center}\textbf{Abstract}\end{center}

Pseudo-Hermitian operators generalize the concept of Hermiticity. Included in this class of operators are the quasi-Hermitian operators, which define a generalization of quantum theory with real-valued measurement outcomes and unitary time evolution. This thesis is devoted to the study of locality in quasi-Hermitian theory, the symmetries and conserved quantities associated with non-Hermitian operators, and the perturbative features of pseudo-Hermitian matrices.

An implicit assumption of the tensor product model of locality is that the inner product factorizes with the tensor product. Quasi-Hermitian quantum theory generalizes the tensor product model by modifying the Born rule via a metric operator with nontrivial Schmidt rank. Local observable algebras and expectation values are examined in \cref{LocalityChapter}. Observable algebras of two one-dimensional fermionic quasi-Hermitian chains are explicitly constructed. Notably, there can be spatial subsystems with no nontrivial observables. 
Despite devising a new framework for local quantum theory, I show that expectation values of local quasi-Hermitian observables can be equivalently computed as expectation values of local Hermitian observables. 
Thus, quasi-Hermitian theories do not increase the values of nonlocal games set by Hermitian theories. Furthermore, Bell's inequality violations in quasi-Hermitian theories never exceed the Tsirelson bound of Hermitian quantum theory.

A perturbative feature present in pseudo-Hermitian curves which has no Hermitian counterpart is the \textit{exceptional point}, a branch point in the set of eigenvalues. An original finding presented in \cref{epSurfaceSingularities} is a correspondence between cusp singularities of algebraic curves and higher-order exceptional points.

Eigensystems of one-dimensional lattice models admit closed-form expressions that can be used to explore the new features of non-Hermitian physics.
One-dimensional lattice models with a pair of non-Hermitian defect potentials with balanced gain and loss, $\Delta \pm \mathfrak{i} \gamma$, are investigated in \cref{toyModelsChapter}. Conserved quantities and positive-definite metric operators are examined. When the defects are nearest neighbour, the entire spectrum simultaneously becomes complex when $\gamma$ increases beyond a second-order exceptional point. When the defects are at the edges of the chain and the hopping amplitudes are 2-periodic, as in the Su-Schrieffer-Heeger chain, the $\mathcal{PT}$-phase transition is dictated by the topological phase of the system. In the thermodynamic limit, $\mathcal{PT}$-symmetry spontaneously breaks in the topologically non-trivial phase due to the presence of edge states.

Chiral symmetry and representation theory are utilized in \cref{Tactics} to derive large classes of pseudo-Hermitian operators with closed-form intertwining operators. These intertwining operators include positive-definite metric operators in the quasi-Hermitian case. The $\mathcal{PT}$-phase transition is explicitly determined in a special case.

\cleardoublepage


\begin{center}\textbf{Acknowledgements}\end{center}

I appreciate the support of all the people who made this thesis possible.

I would like to thank \textbf{Yogesh Joglekar} for being there for me and giving guidance at every step along the way.

\textbf{Kasia Rejzner} listened and responded to my ramblings on $C^*$-algebras, for which I am very appreciative.

I would like to thank \textbf{James Forrest} for keeping me on track and for his patience. 

The extensive writing advice given by \textbf{Ma\"\i t\'e Dupuis} and \textbf{Alisa Klinger} has increased the clarity of this thesis; both the reader and I owe them our thanks. 

My sincere gratitude extends to \textbf{Rada} and \textbf{Yasha Neiman} for giving me courage and supporting my personal and intellectual development through a research internship program at the Okinawa Institute for Science and Technology.

\textbf{Rob Myers} and \textbf{Lee Smolin} have both provided financial support for my PhD studies, and I want to thank both of them for making my work possible.

I would also like to thank \textbf{Neil Turok} for giving me my first research problem at Perimeter Institute and for guidance through my early graduate career.

I would also like to thank the numerous people whose interactions with me led me to either new perspectives or ideas for calculations for this thesis. An incomplete list of people not already mentioned includes \textbf{Kaustubh Agarwal}, \textbf{Richard Cleve}, \textbf{Daniel Gottesman}, \textbf{Laurent Freidel}, \textbf{Tobias Fritz}, \textbf{Lauren Hayward}, \textbf{Marvin Kemple}, \textbf{Chia-Yi Ju}, \textbf{Sung-Sik Lee}, \textbf{Jacob Muldoon}, \textbf{Alexandru Nica}, \textbf{Daniel Ranard}, \textbf{Nic Shannon}, \textbf{Barbara \v{S}oda}, \textbf{Robert Spekkens}, \textbf{Marek Stastna}, and \textbf{Richard Ward}.


\textbf{Mario Ariganello}, \textbf{Tung Bui}, \textbf{Fred Gower}, \textbf{Dave Graore}, \textbf{Jeffrey Schneider}, and \textbf{Steve Zeoke} believed in me, and my foosball skills, when nobody else did, so I cannot thank them enough.

I also thank my inner circle of friends for supporting me when I needed it most. Friends who have extended incredible amounts of generosity and who have not been mentioned above include 
\textbf{Chris Adamantidis}, \textbf{Rosemarie Archer}, \textbf{Alexander Atanasov}, \textbf{Clyde Baker}, \textbf{Dan Belcher}, \textbf{Melanie Brunet}, \textbf{Juan Cayuso}, \textbf{Dave Defebaugh}, \textbf{Daniel Guariento}, \textbf{Dave Hoffman}, \textbf{Kayla Head}, \textbf{Keni Kim}, \textbf{Jolanta Komornicka}, \textbf{Adri{\'a}n L{\'o}pez},  \textbf{Campbell MacDonald}, \textbf{Fiona McCarthy}, \textbf{John McDermott}, \textbf{Andrew}, \textbf{Bronwen}, and \textbf{Lily Meikle}, \textbf{Barak Shoshany}, and \textbf{Ryan Snell}.

I am especially grateful to my immediate family \textbf{Ethan}, \textbf{Michael}, and \textbf{Wesley}.

I would like to thank \textbf{Olga Schneider} for accepting me for who I am.


%
%
%

\renewcommand\contentsname{Table of Contents}
\tableofcontents
\cleardoublepage
\phantomsection    

\addcontentsline{toc}{chapter}{List of Figures}
\listoffigures
\cleardoublepage
\phantomsection		

\addcontentsline{toc}{chapter}{List of Tables}
\listoftables
\cleardoublepage
\phantomsection		

\printunsrtglossary[type=symbols]

\cleardoublepage
\phantomsection		

\pagenumbering{arabic}







\chapter{Outline of Thesis}

This chapter provides brief summaries of later chapters, with the exception of the concluding chapter~\ref{conclusionChapter}. A companion talk to this thesis can be found at \cite{GradSeminar}.

\section{Summary of Chapter~\ref{introChapter}}

The goal of this chapter is to introduce the concepts and contexts upon which this thesis is built. Precise definitions of key terms in non-Hermitian physics along with their fundamental relationships and properties are given. Detailed discussions of the topic of non-Hermiticity in physics can be found in textbooks \cite{benderBook,Bagarello2015,Moiseyev2009,christodoulides2018parity} and review articles \cite{MakingSenseNonHerm,MostaReview,Ashida2020}. The mathematics of non-Hermitian operators and their natural setting, the indefinite inner product space, is examined in textbooks such as \cite{azizov1989linear,Gohberg2005,Bognr1974,Gohberg1969}. In some cases, I have written proofs of central results that I find simpler than what can easily be found in the literature. In the interest of generality, I have adopted definitions that address the possibility of infinite-dimensional Hilbert spaces. Simplifications occur in the finite-dimensional setting, which is the context of most of the examples.

Section~\ref{hermQuantumSection} discusses the importance of quantum theory and summarizes a mathematical framework for quantum theory, the Dirac-von Neumann axioms.

Section~\ref{leaveHerm} discusses physical contexts in which non-Hermitian operators play a role and how non-Hermitian operators can address shortcomings in the Dirac-von Neumann axioms.

The next objective of \cref{introChapter} is to introduce properties of families of non-Hermitian operators. Operators with antilinear symmetry, e.g., those with time-reversal symmetry, are discussed in \ref{PTSymmetryIntroSection}. Pseudo-Hermitian and quasi-Hermitian operators are introduced in sections~\ref{PseudoHermSection} and \ref{quasiHermSection} respectively. Quasi-Hermiticity is the necessary and sufficient criteria for the existence of real eigenvalues in finite-dimensional Hilbert spaces \cite{Drazin1962} and, consequently, defines a fundamental extension of quantum theory via a modified Born rule. 

Perturbation theory is the setting of \ref{ExceptionalPointsIntro}. Non-Hermitian operators exhibit perturbative features unparalleled by Hermitian operators. For example, non-Hermitian operators typically contain \textit{exceptional points}, a branch point in the curves generated by eigenvalues. A novel correspondence between higher-order exceptional points and cusp singularities of algebraic curves is presented in \cref{epSurfaceSingularities}. 

All of the above properties are examined in the context of a qubit in \cref{PTqubit}. 

The $C^*$-algebra picture of both quasi-Hermitian and Hermitian quantum theory is introduced in \cref{C*AlgebrasIntro}. 

I have opted to order the results of the introduction chapter in one particular way. A modular approach can be taken to reading the introduction chapter, whereby readers can jump between sections they find most interesting. Nonlinear readers will find \cref{fig:Nonlinear-Reader} helpful in deciding how they wish to approach \cref{introChapter}.

\begin{figure}[!ht]
\centering
\includegraphics[width=\textwidth]{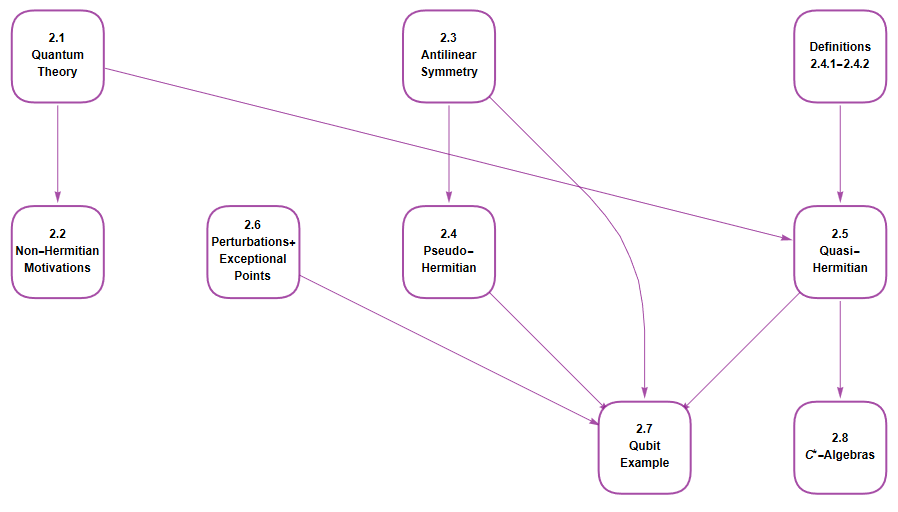}
\caption{This figure displays dependencies between the sections of the introduction chapter. Sections depicted at the end of an arrow are recommended to be read after material from the tail of all arrows pointing towards it.}
\label{fig:Nonlinear-Reader}
\end{figure}

\section{Summary of Chapter~\ref{toyModelsChapter}} \label{Toy Models Summary}

A majority of the original results presented in this chapter have been submitted to J. Phys. A and can be found on the arXiv \cite{Barnett2023}.

My primary objective when performing the calculations of this chapter was to build a basis of toy models that could be used to examine the novelties of non-Hermitian operators. The class of tridiagonal matrices with perturbed corners proved to be particularly tractable. These matrices physically correspond to systems with one spatial dimension consisting of a finite number of sites. The eigenvalue problem associated with these matrices is discussed in \cref{tridiagEsysSection}. Three special cases with a pair of non-Hermitian defect potentials are addressed, with a strong emphasis on the $\mathcal{PT}$-symmetric case.


The first case, studied in \cref{nearestNeighbour}, models a system with gain and loss at nearest neighbour sites at the center of a one-dimensional chain with an even number of sites. 
A one-dimensional family of intertwining operators, which physically correspond to conserved quantities, is determined in \cref{pseudoHermOpenChainThm}. Included in this family is a subset of positive-definite metric operators in the $\mathcal{PT}$-unbroken domain. In the $\mathcal{PT}$-unbroken domain, a similar Hermitian Hamiltonian and a $\mathcal{C}$-symmetry are constructed explicitly in \cref{Equivalent Hamiltonian Section}. Every eigenvalue gains an imaginary part when the impurity strength is increased above a threshold corresponding to a second-order exceptional point. This threshold is determined solely by the central hopping amplitude's magnitude, $|t_m|$.

The second case, studied in \cref{SSHSection}, is a non-Hermitian perturbation of a Su-Schrieffer-Heeger chain with gain and loss at the edges. Owing to the connection between the topological phase and the presence of edge states, in the thermodynamic limit, the $\mathcal{PT}$-unbroken domain coincides with the topologically trivial phase. A three-dimensional surface of exceptional points is plotted in \cref{EP Surface Section}. This surface, an example of an algebraic variety, has ridges of singular points which cannot exist in manifolds. The ridges correspond to third-order exceptional points which include cusp points corresponding to fourth-order exceptional points.

{}{}
The third case, studied in \cref{uniformSection}, is a chain with open boundary conditions and uniform hopping. Cases where either a subset or the entirety of the spectrum can be computed are summarized in sections~\ref{ExactEvalues} and \ref{Constant Evalues}. An analytical, numerical, and perturbative analysis of the exceptional points is undertaken in sections ~\ref{EPContours Uniform Chain} and \ref{Perturbative Section}.
One set of results not reported in \cite{Barnett2023} that appears in this chapter is a discussion of the positivity of an intertwining operator introduced in \cite{farImpurityMetric,Ruzicka2015}. 

\section{Summary of Chapter~\ref{Tactics}}

This chapter is devoted to finding pseudo-Hermitian operators, $H$, where intertwining operators, $\eta$, can be given in closed-form. I will refer to such a pair $(\eta, H)$ as a pseudo-Hermitian pair.

Section~\ref{anticommutPencilSection} discusses the case of a pencil with anticommuting invertible Hermitian generators. When both generators are invertible, and one is an involution, the spectrum is determined exactly. A closed-form intertwining operator is determined for every element of this pencil, which is a positive-definite metric if and only if the spectrum of its corresponding pencil element is real.

Section~\ref{section:PseudoHermFromRep} utilizes $*$-homomorphisms to generate pseudo-Hermitian pairs from a known vector space of $\eta$-self-adjoint operators. Curiously, the pairs generated in this way can be in a higher dimension than the initial space of known pairs. Section~\ref{Commutative-Section} characterizes the complexity of the spectrum in the case where the known vector space is commutative. In particular, this case exhibits an unusual antilinear-symmetry breaking phase transition, where the entire spectrum of a pseudo-Hermitian operator gains a nonzero imaginary part as a parameter is tuned past an exceptional point.

\section{Summary of Chapter~\ref{LocalityChapter}}
Sections~\ref{TPMSection} and \ref{FermionObservableAlgebraSection} are based on, and expand upon, material published in \cite{Barnett_2021}.

This chapter examines local observable algebras, expectation values of local observables, and nonlocal games in the context of quasi-Hermitian quantum theory.

Section~\ref{localityIntroSection} introduces locality in quantum theory. 

In the tensor product model, when the Schmidt rank of the metric operator is greater than one, the dimensions of local quasi-Hermitian observable algebras are smaller than their Hermitian counterpart. For some choices of metric operators, there are spatial subsystems with no nontrivial local quasi-Hermitian observables. 

A quasi-Hermitian model of particle-conserving free fermions is examined. A second-quantized metric operators corresponding to a first-quantized relative is elaborated on, and the number operator is shown to be an observable for this second-quantized metric. The local observable algebras are characterized in general, and then computed explicitly for two one-dimensional $\mathcal{PT}$-symmetric chains. The structure of the local observable algebras depends strongly on the Hamiltonian.

Expectation values of local observables in a quasi-Hermitian theory are shown to be computable by expectation values of local observables in a Hermitian theory. Consequently, a quasi-Hermitian strategy for a nonlocal game cannot outperform a Hermitian strategy. Equivalently, quasi-Hermitian Bell inequality violations cannot exceed their Hermitian counterparts.

\section{Summary of Appendix~\ref{functionalAnalysisAppendix}}

The majority of this appendix is devoted to reviewing mathematical tools from the fields of functional analysis and operator algebras. One could view these fields as abstractions of linear algebra which are suited to handling infinite-dimensional vector spaces.  

Some formulas relating to Chebyshev polynomials, a sequence of orthogonal polynomials which naturally appear in tridiagonal matrix problems, are summarized in \cref{Chebyshev Appendix}.

\chapter{Introduction} \label{introChapter}

\section{Hermitian Quantum Theory} \label{hermQuantumSection}


Quantum theory is the most successful framework for modern physics. Originally invented to describe blackbody radiation \cite{planck1900}, the photoelectric effect  
\cite{HertzPhotoelectric,einstein1905erzeugung,Millikan1914}, and the structure of the atom \cite{BohrAtom,Gerlach1922}, physicists have used quantum theory to describe a diverse set of phenomena, including neutron stars \cite{Potekhin2010}, semiconductors \cite{shockley1959electrons}, superconductivity \cite{Bardeen1957,Leggett2006,onnes1911discovery}, and the standard model of particle physics \cite{langacker2017standard}. The design of ubiquitous devices, such as lasers \cite{Schawlow1958,siegman1986lasers}, transistors \cite{shockley1959electrons}, magnetic resonance image scanners  \cite{LAUTERBUR1973}, and atomic clocks \cite{ESSEN1955} was only possible due to the predictive power of quantum theory. Innovation in technology continues to be driven by quantum theory, exemplified by current research into quantum computing \cite{operatorSchmidt,QuantumComputing2018,Nielsen2012}. 
Phenomena such as Bell's inequality violations, recently demonstrated experimentally in \cite{Hensen2015,Giustina2015,Shalm2015}, are incapable of being modelled by classical physics \cite{Bell1964}, further necessitating the use of quantum theory. Furthermore, all predictions of Newtonian physics can be obtained in the semiclassical limit of quantum theory.




A commonly applied mathematical formulation of quantum theory, attributable to Dirac \cite{dirac1981principles} and von Neumann \cite{vonNeumannBook,vonNeumannBookEnglish}, is given in \cref{DiracVonNeumann}. In the interest of generality, the Dirac-von Neumann axioms allow for the possibility of infinite-dimensional Hilbert spaces\footnote{A physical motivation for introducing infinite-dimensional Hilbert spaces is to model a quantum particle which obeys the Heisenberg uncertainty principle \cite{Kennard1927}. 
One method to realize the Heisenberg uncertainty principle is to find a representation of the canonical commutation relations 
and apply the Robertson \cite{Robertson1929} or Schr{\"o}dinger \cite{schrodinger1930sitzungsberichte,schrodinger1930translated} uncertainty inequalities. Representations of the canonical commutation relations only exist in infinite-dimensional Hilbert spaces \cite{Weyl1927,WintnerUnbounded}.}. Various tools from measure theory and functional analysis are required to rigorously handle this possibility, which are reviewed in \cref{functionalAnalysisAppendix}. As illustrated in \cref{finiteDimQuantum}, assuming the Hilbert space dimension is finite simplifies the mathematics of quantum theory. In this case, the basis of quantum theory is linear algebra on complex vector spaces. 

\begin{defn}[Dirac-von Neumann axioms] \label{DiracVonNeumann}
\leavevmode 
\begin{enumerate}
\item The \textit{state} of a quantum system, $\ket{\psi}$, is an element of a separable Hilbert space, $\mathcal{H}$, with unit norm, $\braket{\psi|\psi} = 1$\footnote{I use a convention typical for physicists that inner products are antilinear in the first argument and linear in the second argument.}. \\
\item \textit{Observables}, $A$, are self-adjoint linear operators on $\mathcal{H}$ satisfying $A = A^\dag$. By the spectral \cref{SpectralTheorem} for self-adjoint linear operators, there exists a unique projection-valued measure, $P_A:\mathfrak{B}_{\sigma(A)} \to \mathcal{B}(\mathcal{H})$, such that
\begin{align}
A = \int_{\sigma(A)} \lambda \, d P_A(\lambda),
\end{align} 
where $\mathfrak{B}_{\sigma(A)}$ denotes the Borel $\sigma$-algebra generated by the spectrum $\sigma(A) \subseteq \mathbb{R}$ of $A$, defined in \cref{Borel}, and \gls{B(H)} is the set of bounded operators on the Hilbert space $\mathcal{H}$, defined in \cref{boundedOperators}. 
\\
\item \textit{Measurements} of observables, $A$, yield real-valued outcomes from the spectrum of $A$, $\sigma(A) \subseteq \mathbb{R}$. If a system is in the state $\ket{\psi}$, the probability of measuring an outcome from the Borel set $B \in \mathfrak{B}_{\sigma(A)}$, $
\textbf{Pr}(B)$, is determined by the \textit{Born rule} \cite{Born1926},
\begin{align}
\textbf{Pr}(B) = \braket{\psi|P_A(B)|\psi}. \label{measurementOutcomeProbability}
\end{align}
$\textbf{Pr}:\mathfrak{B}_{\sigma(A)} \to [0,1]$ is a probability measure. 
After a measurement, if an observer models the outcome as uniformly sampled from $B \in \mathfrak{B}_{\sigma(A)}$, then the state is updated according to the \textit{L\"{u}ders rule} \cite{Luders1950,Lders2006},
\begin{align}
\ket{\psi} \rightarrow \frac{P_A(B)\ket{\psi}}{\sqrt{\braket{\psi|P_A(B)|\psi}}}.
\end{align} 
The \textit{expectation value} of the observable $A$, $\braket{A}$, is computed by integrating \cref{measurementOutcomeProbability}, 
\begin{align}
\braket{A} &= \braket{\psi| A | \psi}.
\end{align}
\\
\item \textit{Time evolution} is determined via a strongly continuous unitary representation of the additive group $(\mathbb{R}, +)$ on $\mathcal{H}$, which I denote as $U:\mathbb{R} \to \mathcal{B}(\mathcal{H})$. More explicitly, the map $U$ must satisfy $U(t+s) = U(t)U(s)$ for all $s, t \in \mathbb{R}$, $U(t)$ must be unitary for all $t \in \mathbb{R}$, and the map $t \to U(t) \ket{\psi}$ must be continuous in the norm topology, defined in \cref{normedSpace}, for all $\ket{\psi} \in \mathcal{H}$. If the state at the time $t_0 \in \mathbb{R}$ is $\ket{\psi(t_0)}$, then the state at time\footnote{This thesis considers the \textit{Schr{\"o}dinger picture} of time evolution, where the state evolves in time and operators are held constant. Observable quantities can equivalently be computed by evolving operators instead of states, which is the so-called \textit{Heisenberg picture}, or by evolving both operators and states, as in the \textit{interaction picture}. The Heisenberg picture of non-Hermitian time evolution is rather nontrivial \cite{Bagarello2022}.} $t \in \mathbb{R}$ is


\begin{align}
\ket{\psi(t)} = U(t-t_0) \ket{\psi(t_0)}.
\end{align}
By Stone's theorem \cite{Stone1930,Stone1932}, there exists a (possibly unbounded) self-adjoint operator, $H$, referred to as the \textit{Hamiltonian}, which generates time evolution via 
\begin{align}
U(t) = e^{-\mathfrak{i} t H/\hbar},
\end{align}
where $\hbar$ denotes the \textit{reduced Planck's constant}. If $\ket{\psi} \in \cap_{k \in \gls{N}} \text{Dom}(H^k)$, then time evolution is governed by the \textit{Schr{\"o}dinger equation} \cite{Schrodinger1926},
\begin{align}
\mathfrak{i} \hbar \frac{d}{d t} \ket{\psi} = H \ket{\psi}.
\end{align}
\end{enumerate}
\end{defn}

\begin{ex}[Finite-Dimensional Quantum Theory] \label{finiteDimQuantum}
Finite-dimensional Hilbert spaces are isomorphic to $\mathbb{C}^{\text{dim} \mathcal{H}}$. Self-adjoint linear operators are represented by Hermitian\footnote{Hermitian matrices derive their name from the work of Charles Hermite, who proved that Hermitian matrices have a real spectrum in \cite{Hermite1855}.} matrices\footnote{In this thesis, the term "matrix" refers to finite-dimensional matrices.}, $A \in \mathfrak{M}_{\dim \mathcal{H}}(\mathbb{C})$, whose elements satisfy $A_{ij} = A^*_{ji}$. The spectrum of $A$, $\sigma(A)$, is the set of eigenvalues of $A$, and $\mathfrak{B}_{\sigma(A)}$ is the power set of $\sigma(A)$, namely $\mathfrak{B}_{\sigma(A)} = \mathbb{P}(\sigma(A))$. The projection-valued measure associated to a Hermitian matrix, $A$, is 
\begin{align}
P_A(B) = \sum_{\lambda \in B} P_A(\{\lambda\}),
\end{align}
where $P_A(\{\lambda\})$ is the orthogonal projection onto the eigenspace associated to $\lambda$.
\demo
\end{ex}

\section{Motivations for Abandoning Hermiticity} \label{leaveHerm}
\subsection{Why Do We Use Hermitian Operators?}

Hermiticity plays two crucial roles in the Dirac-von Neumann axioms for quantum theory.

The first important feature of Hermitian linear operators is that they have a spectral decomposition. Consequently, Hermitian observables have real measurement outcomes associated to their spectrum. Furthermore, associated to every observable and every state, there exists a consistent assignment of probabilities to subsets of measurement outcomes. More formally, a state defines a probability measure, given in \cref{measurementOutcomeProbability}, on the Borel $\sigma$-algebra generated by the measurement outcomes of every observable.

Secondly, time evolution generated by a Hermitian Hamiltonian is \textit{unitary}, implying that the normalization of the state vector, $\braket{\psi|\psi} = 1$, is a constant in time. The normalization condition is necessary to generate a probability measure from the state vector, since it ensures the probability that a measurement of an observable yields a result from its spectrum equals one\footnote{Equivalently, the measure $\textbf{Pr}$ satisfies the axiom of unit measure, which is one of the Kolmogorov axioms for probability measures \cite[p. 2]{kolmogorov2018foundations}.}.

\subsection{Non-Hermiticity in Quantum Theory}
The axioms of quantum theory presented in \cref{DiracVonNeumann} have several weaknesses. In this section, I discuss four of these shortcomings: $1.$ the modelling of open or effective quantum systems, $2.$ the lack of a quantum description of spacetime, $3.$ the physical and mathematical origins of these axioms, and $4.$ the measurement problem. 
In each case, I discuss what the introduction of non-Hermitian operators can do to alleviate these problems.

\begin{enumerate}
\item 
First, I claim the goal of physics is to describe the outcomes of experiments. A typical system has many more microscopic degrees of freedom than observable qualities. Thus, it can be impractical to make predictions using the entire system's Hilbert space. 

For example, in the context of \textit{open quantum systems}, there is a distinction between a \textit{subsystem} and an \textit{environment}, and one is often only interested in properties of the subsystem. An extremal 
example is that of a qubit coupled to a thermal bath; if a scientist is interested in the properties of the qubit, ideally the scientist would make predictions using only the mathematical setting naturally associated with qubits, $\mathbb{C}^2$. 
Non-Hermitian Hamiltonians are typically required to achieve this goal. For example,
under certain assumptions about the initial conditions of the quantum state, subsystem dynamics can be modelled with the \textit{Lindblad equation} \cite{Lindblad1976,Gorini1976}, 
which is a Schr{\"o}dinger equation generated by a non-Hermitian Hamiltonian. This Hamiltonian is referred to as the \textit{Liouvillian superoperator}, and it acts on the Hilbert space of Hilbert-Schmidt operators local to the subsystem. A simpler Hamiltonian, which acts directly on the Hilbert space of states local to the subsystem, can be obtained by neglecting the effects of quantum jumps in the Liouvillian. 

Elaborating on this point further, non-Hermitian operators appear frequently as effective or phenomenological Hamiltonians. For instance, they appear in the context of renormalization \cite{LeeModelPT,LeeModel,quantumRG}, in the Feshbach projection formalism \cite{Feshbach1958,Feshbach1962}, and in the complex scaling technique used to describe resonance phenomena \cite{Moiseyev1998,Moiseyev2009}. My intuition for why this happens is that non-Hermiticity can encapsulate flows of probability into and out of a system. 

\item Despite much theoretical work \cite{polchinski2005string,rovelli2004quantum}, at the time of writing this thesis, there is no experimentally tested quantum description of gravity. Modern theories of gravity, such as general relativity \cite{einstein1915general,wald2010general}, 
alter the primitive notions of space and time in Newtonian physics. It is generally accepted that a quantum theory of gravity must challenge our understanding of space, time, and locality. In \cref{LocalityChapter}, I argue that non-Hermitian quantum theory generates new notions of locality. Consequently, I dream that quasi-Hermitian quantum theories could serve as a bridge between gravity and quantum theory. Evidence for this dream includes a recent result that non-Hermitian models can be the dual to Hermitian models in curved spacetime \cite{Lv2022} and the introduction of non-Hermiticity in the context of noncommutative spacetime \cite{Fring2010}. Furthermore, the mathematics used to define an inner product in quasi-Hermitian quantum theory rhymes with the vielbein formalism in general relativity \cite{Ju2022}.

\item The axiom of Hermiticity of the observables is not directly derived from physical considerations. A physically motivated constraint on the observables is that they possess an \textit{antilinear symmetry}. For example, the antilinear $CPT$-symmetry is exhibited in every experiment performed to date; an exhaustive search for $CPT$-violations has thus far found none \cite{CPTReview}. Every matrix with an antilinear symmetry\footnote{In this section, "antilinear symmetry" is synonymous with "bijective antilinear symmetry". Every operator has the trivial antilinear symmetry $\Theta = 0$, assuming bijectivity avoids this sort of pathology.} 
is \textit{pseudo-Hermitian} \cite{Solombrino2002,mostafazadeh2002pseudo3,Siegl2009,siegl2008quasi}, rather than Hermitian; this result is presented in \cref{pseudoEquivThm}. Pseudo-Hermiticity is a generalization of Hermiticity which includes Hermiticity as a special case. Applying the additional constraint of \textit{unbroken} antilinear symmetry to a matrix, so that the eigenspaces are invariant under the antilinear symmetry, is equivalent to imposing the reality of the matrix's eigenvalues, as discussed in \cref{unbrokenRealSpectrum}. The stricter constraint that a matrix is \textit{quasi-Hermitian} is equivalent to the statement that said matrix is diagonalizable with a real spectrum \cite{Drazin1962}, as discussed in \cref{quasiHermSection}. Based on this observation, quasi-Hermitian operators define a fundamental description of quantum theory via a modified inner product
\cite{QuasiHerm92}. The key properties of these operators are summarized in \cref{operatorClassTable} and their relationships in the finite-dimensional case are explained in \cref{matrixClassFig}.
\end{enumerate}

\begin{table*}[!htp]
\centering
\begingroup
\setlength{\tabcolsep}{6pt} 
\renewcommand{\arraystretch}{2} 
\begin{tabular}{|m{2.7cm}|m{2.2cm}|m{4.5cm}|m{4.8cm}|}
\hline
\renewcommand{\arraystretch}{2}
\textbf{Operator Type} & \textbf{Spectrum} & \textbf{Sesquilinear Form} & \textbf{Physical Consequences} \\
\hhline{|=|=|=|=|}
\parbox{2.7cm}{Antilinear\\Symmetric} 
& c.c. pairs 
& \parbox{4.5cm}{Indefinite inner product\\(when $\dim \mathcal{H} < \infty$)} 
& \parbox{4.5cm}{Time-reversal symmetry, \\ Real partition function} \\[8pt]
\hline
\parbox{2.7cm}{Pseudo-\\Hermitian}
& c.c. pairs 
& Indefinite inner product 
& \parbox{4.5cm}{Conserved quantities, \\ Real partition function} \\[8pt]
\hline
\mbox{Unbroken} \mbox{Antilinear} \mbox{Symmetric} 
& Real 
& Indefinite inner product (when $\dim \mathcal{H} < \infty$)
& \\[8pt]
\hline
\parbox{2.7cm}{Quasi-\\Hermitian}
& \parbox{2.2cm}{Semi-simple\\Real} 
& Inner product 
& \mbox{Real-valued measurements,} \mbox{States generate}  \mbox{probability measures}
\\[8pt]
\hline
Hermitian 
& \parbox{2.2cm}{Semi-simple\\Real} 
& Hilbert space 
& \mbox{Real-valued measurements,} \mbox{States generate}  \mbox{probability measures} \\[8pt]
\hline
\end{tabular}
\endgroup
\caption{Comparison of non-Hermitian operator classes. Complex conjugate is abbreviated to c.c.
Physical consequences refer to contexts where an operator is interpreted as either an observable or as the generator of time evolution. The sesquilinear form defines the meaning of unitarity of (non-)Hermitian Hamiltonians.
}
\label{operatorClassTable}
\end{table*} 

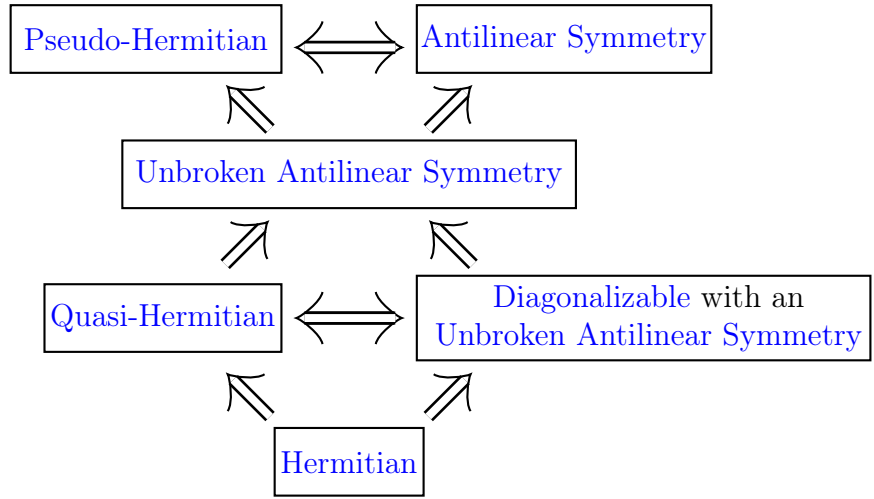
\begin{figure}[!ht]
\begin{center}
\begin{tikzpicture}[thick, scale=0.9]
\draw (-1.1,-.5) rectangle (1.1,.5) node[pos=.5] {\hyperlink{link:Hermitian}{Hermitian}};
\draw (-4.5,1.625) rectangle (-1,2.625) node[pos=.5] {\hyperlink{link:QuasiHerm}{Quasi-Hermitian}};
\draw (1,1.5) rectangle (7.7,2.75) node[pos=.5] {
\renewcommand{\arraystretch}{1}
$\begin{array}{c}
\hyperlink{link:Diagonable}{\text{Diagonalizable}}\text{ with an} \\
\hyperlink{link:Unbroken}{\text{Unbroken Antilinear Symmetry}}
\end{array}$
};
\draw (-3.35,3.75) rectangle (3.35,4.75) node[pos=.5] {\hyperlink{link:Unbroken}{Unbroken Antilinear Symmetry}
};
\draw (-5,5.75) rectangle (-1,6.75) node[pos=.5] {\hyperlink{link:PseudoHerm}{Pseudo-Hermitian}};
\draw (1,5.75) rectangle (5.35,6.75) node[pos=.5] {
\hyperlink{link:Antilinear}{Antilinear Symmetry} 
};

\draw node (Diag3) at (0,7) {};

\draw node (HR) at (1,.5) {};
\draw node (HL) at (-1,.5) {};
\draw node (QH4) at (-2,1.5) {};
\draw node (Diag3) at (2,1.5) {};

\draw node (QHR) at (-1,2.125) {};
\draw node (DiagL) at (1,2.125) {};
\draw node (QHR2) at (-.8,2.125) {};
\draw node (DiagL2) at (.85,2.125) {};

\draw node (QHU) at (-2,2.75) {};
\draw node (DiagU) at (2,2.75) {};
\draw node (UASL) at (-1,3.75) {};
\draw node (UASR) at (1,3.75) {};

\draw node (UAS2) at (1,4.75) {};
\draw node (UAS1) at (-1,4.75) {};
\draw node (PH4) at (-2,5.75) {};
\draw node (AL3) at (2,5.75) {};

\draw node (PHR) at (-1,6.125) {};
\draw node (ALL) at (1,6.125) {};
\draw node (PHR2) at (-.8,6.125) {};
\draw node (ALL2) at (.85,6.125) {};

\draw[line width=1pt, double distance=3pt,-{Classical TikZ Rightarrow[sharp]}] (HL) -- (QH4);
\draw[line width=1pt, double distance=3pt,-{Classical TikZ Rightarrow[sharp]}] (HR) -- (Diag3);
\draw[line width=1pt, double distance=3pt,-{Classical TikZ Rightarrow[sharp]}] (QHR2) -- (DiagL);
\draw[line width=1pt, double distance=3pt,-{Classical TikZ Rightarrow[sharp]}] (DiagL2) -- (QHR);
\draw[line width=1pt, double distance=3pt,-{Classical TikZ Rightarrow[sharp]}] (QHU) -- (UASL);
\draw[line width=1pt, double distance=3pt,-{Classical TikZ Rightarrow[sharp]}] (DiagU) -- (UASR);
\draw[line width=1pt, double distance=3pt,-{Classical TikZ Rightarrow[sharp]}] (UAS2) -- (AL3);
\draw[line width=1pt, double distance=3pt,-{Classical TikZ Rightarrow[sharp]}] (UAS1) -- (PH4);
\draw[line width=1pt, double distance=3pt,-{Classical TikZ Rightarrow[sharp]}] (PHR2) -- (ALL);
\draw[line width=1pt, double distance=3pt,-{Classical TikZ Rightarrow[sharp]}] (ALL2) -- (PHR);
\end{tikzpicture}
\end{center}
\caption{This figure summarizes the relationship between different types of non-Hermitian matrices. If an arrow points from one term to another, then every instance of a term at the tail of this arrow is also an instance of a term at the point of this arrow. Relevant definitions are hyperlinked to this figure. Some authors require matrices with unbroken antilinear symmetry to be diagonalizable, in which case the constraint of unbroken antilinear symmetry is equivalent to quasi-Hermiticity. Equivalence between pseudo-Hermiticity and antilinear symmetry only holds in finite-dimensional Hilbert spaces \cite{Siegl2009,siegl2008quasi}.} \label{matrixClassFig}
\end{figure}

\begin{enumerate}[4.]
\item In my view, the measurement problem in quantum theory is that there is no axiomatic procedure for determining which processes correspond to unitary time evolution and which processes induce measurements. A related issue is that observers are modelled as deterministic entities. The distinction between observer and system is sometimes referred to as a \textit{Heisenberg cut}, and there is no first principles procedure which determines when or where a Heisenberg cut takes place. Due to the presence of a Heisenberg cut, no quantum theory can be a universal theory. While I doubt that the measurement problem can be solved by simply making things non-Hermitian, it is interesting to point out that a certain limit of continuously performed measurements can be modelled via a non-Hermitian Hamiltonian \cite{Dubey2021}. 
\end{enumerate}


I conclude this section with a perhaps counter-intuitive claim: analysis of the non-Hermitian regime yields insight into the Hermitian domain. As an example, consider a Rayleigh-Schr{\"o}dinger perturbation of a Hermitian operator, which is a linear perturbation whose strength is controlled by a single real parameter. The eigenvalues and eigenstates of this operator admit a Taylor series expansion which has a possibly finite radius of convergence \cite{Kato1995}. Exceptional points, a property of non-Hermitian operators, dictate this radius; as one generalizes the perturbation to include a complex parameter, branch points in the spectrum can occur if the parameter gains a sufficiently large imaginary part, and a perturbative expansion cannot encapsulate behaviour beyond this point \cite{Klaiman2008}.


%


\subsection{Non-Hermiticity in Classical Physics}


A surprising feature of mathematical equations is their ability to model a variety of distinct physical systems. This section examines how the Schr{\"o}dinger equation is no exception and how it appears in multiple contexts in classical physics. When applied to settings other than isolated quantum systems, the state which solves the Schr{\"o}dinger equation need not be physically interpreted as a probability distribution. Consequently, the dynamics of such a state need not be unitary, and the Hamiltonian can be non-Hermitian. Introduction of non-Hermitian Hamiltonians allows the Schr{\"o}dinger equation to model systems exhibiting \textit{loss} or \textit{gain}.

In optical systems, such as coupled waveguides, the propagation of electromagnetic waves is modelled with a Schr{\"o}dinger-like equation referred to as the \textit{paraxial equation} \cite[\S 7.3]{siegman1986lasers}. Optical systems with loss (attenuation) or gain (as in a lasing medium) can be modelled via a refractive index with a negative or positive imaginary part respectively, which leads to the non-Hermiticity of the Hamiltonian in the paraxial equation. Optical systems with \textit{balanced} gain and loss can be modelled by a $\mathcal{PT}$-symmetric Hamiltonian \cite{Ruschhaupt2005,ElGanainy2007}. $\mathcal{PT}$-symmetric optical systems were experimentally realized in \cite{Guo2009,Rter2010}. A review article which goes into depth on this topic is \cite{ElGanainy2018}.

In the context of linear electronic circuits, gain is implemented with operational amplifiers, and loss is implemented with resistors. In such a circuit with coupled gain and loss, the time evolution of electrical charges and currents is governed by a non-Hermitian "Hamiltonian" matrix \cite{Schindler2011}. In cases where the gain and loss is balanced, the resulting Hamiltonian is $\mathcal{PT}$-symmetric.

Even motion in simple mechanical systems, such as coupled driven pendulums, can be modelled with a non-Hermitian but $\mathcal{PT}$-symmetric Hamiltonian \cite{Bender2013}.




\section{Antilinear Symmetry} \label{PTSymmetryIntroSection}
%

A major guiding principle in physics is \textit{symmetry}. Symmetry considerations grant enhanced insight into systems, often greatly simplifying their analysis. In theories governed by an action principle, Noether's theorem establishes a correspondence between differentiable symmetries and conservation laws \cite{Noether1918} and consequently derives angular momentum and energy-momentum conservation laws from rotational and space-time translation invariance respectively. A non-exhaustive list of discrete, non-differentiable symmetries includes charge conjugation; chiral symmetry; lattice reflections, rotations, and translations; parity; particle-hole symmetry; and time-reversal. Applications of discrete symmetry considerations are in classifying topological quantum states of matter \cite{Chiu2016,Kawabata2019}; deriving selection rules in systems such as atomic nuclei, atoms, and molecules; separating the Lorentz group into connected components, and reducing the computational effort required to derive eigenspaces. 
This section is devoted to understanding the mathematical consequences of imposing an antilinear symmetry, such as time-reversal, on an operator such as the Hamiltonian.


Given a state evolving in time according to some laws, these laws are said to exhibit \textit{time-reversal symmetry} if there exists a corresponding time-reversed state which evolves according to the time reverse of these laws. In the Newtonian mechanics of a single particle in $d$ spatial dimensions, whose state space is represented by a pair of real numbers, $(q,p) \in \mathbb{R}^d \times \mathbb{R}^d$, corresponding to the \textit{position} and \textit{momentum} of a particle respectively, the effect of time reversal is to map $(q,p) \mapsto (q,-p)$. As discovered by Eugene Wigner \cite{Wigner1932,WignerCollectedWorks}, time-reversal symmetry of the Schr{\"o}dinger equation is implemented by forcing the Hamiltonian to have an \textit{antilinear symmetry}, as defined below. 
\begin{defn} \label{symmetryDefn} 
An (anti)linear operator\footnote{Definition~\ref{operatorDefn} defines (anti)linear operators.}, $\Theta$, is called an \hypertarget{link:Antilinear}{\textit{(anti)linear symmetry}} of an operator, $A$, if $A$ and $\Theta$ commute and\footnote{The constraint that $\Theta$ preserves the domains associated to powers of $A$ is only relevant in the case where $A$ acts on an infinite-dimensional space.} $\Theta(\text{Dom}(A^k)) \subseteq \text{Dom}(A^k)$ for all $k \in \mathbb{N}$. 
\end{defn}
If $\psi(t)$ satisfies the Schr{\"o}dinger equation generated by a Hamiltonian which has the antilinear symmetry $\Theta$, then $\Theta \psi(-t)$ also satisfies the Schr{\"o}dinger equation. If an induced Hilbert space norm of physical states is required to be a constant in time, as is the case for quantum theory, then $\Theta$ must be a norm-preserving map. Equivalently, $\Theta$ must be \textit{antiunitary}, which means it satisfies\footnote{The adjoint of an antilinear map is defined in \cref{adjoint}.} $\Theta^\dag \Theta = \mathbb{1}$.

\begin{ex} \label{ex:HermitianQuantumParticle}
A quantum particle in one-dimension is modelled via the Hilbert space $L^2(\mathbb{R})$. Two examples of densely-defined operators on $L^2$ are the \textit{multiplication operators}
\begin{align}
\hat{f} \ket{\psi(x)} = \ket{f(x) \psi(x)}
\end{align}
and the \textit{differentiation operator}
\begin{align}
\hat{D} \ket{\psi(x)} = \ket{\frac{d \psi(x)}{dx}}.
\end{align}
The Hamiltonian, $H$, is traditionally assumed to be a \textit{Schr{\"o}dinger operator}, which means $H$ can be expressed in the form
\begin{align}
H = -\frac{\hbar^2}{2m} \hat{D}^2 + \hat{V},
\end{align}
where $m > 0$ is the \textit{mass} of the particle, and $V(x)$ is the \textit{potential}. If $V(x) = V^*(x)$ is real-valued and continuous, then $H$ is self-adjoint\footnote{Technically, without further information on its domain, $H$ is only \textit{essentially} self-adjoint; this means $H$ is symmetric and its closure is self-adjoint. For more details on this point, see \cite[Chp. 9]{hall2013quantum}.}, and the antilinear \textit{complex conjugation} operator
\begin{align}
\mathcal{T} \ket{\psi(x)} = \ket{\psi^*(x)}
\end{align}
is an antilinear symmetry of $H$. Thus, $\mathcal{T}$ generates time-reversal. The \textit{position} and \textit{momentum} observables are 
\begin{align}
Q &= \hat{x} \label{positionOperator} \\
P &= -\mathfrak{i} \hbar \hat{D}. \label{momentumOperator}
\end{align}
The action of $\mathcal{T}$ on position and momentum is 
\begin{align}
\mathcal{T} Q \mathcal{T} &= Q\\
\mathcal{T} P \mathcal{T} &= -P,
\end{align}
as one would expect from analogy with Newtonian mechanics.
Intuitively, time-reversal preserves the position of a particle, but it inverts its direction of motion.
\demo
\end{ex}

Reality of the potential of a Schr{\"o}dinger operator is not a necessary condition for antilinear symmetry. One example of a non-Hermitian Hamiltonian with an antilinear symmetry is the Schr{\"o}dinger operator with the complex-valued potential $V(x) = x^2 (\mathfrak{i} x)^\epsilon$, where $\epsilon \in \mathbb{R}$. The corresponding antilinear symmetry is $\mathcal{PT}$-\text{symmetry} in the sense defined in the next example. This potential was considered by \cite{bender1998real} and is commonly regarded as having piqued the physics community's interest in antilinear symmetry.

\begin{ex} \label{benderPTExample}
Observe that the position and momentum operators of \cref{positionOperator,momentumOperator} satisfy the \textit{canonical commutation relations},
\begin{align}
[Q,P]_- = \mathfrak{i} \hbar \mathbb{1}.
\end{align}
Max Born observed that performing the substitution $Q \to P, P \to -Q$ results in an alternative representation of the canonical commutation relations, a concept often referred to as \textit{Born reciprocity} \cite{Born1938}. Naturally, Born reciprocity raises the question of whether the position of a particle can be represented with a differentiation operator instead of a multiplication operator. If we make this change, since the classical limit requires the position observable to be time-reversal invariant, $\mathcal{T}$ can no longer represent the antilinear symmetry of the Hamiltonian. A trick to generating an antilinear symmetry is to use the combined operation of \textit{Parity}, $\mathcal{P}$, with $\mathcal{T}$, where
\begin{align}
\mathcal{P} \ket{\psi(x)} = \ket{\psi(-x)}.
\end{align}
For a Schr{\"o}dinger operator to be $\mathcal{PT}$-symmetric, we must have $V(x) = V^*(-x)$. An example of such a potential is the class of $\mathcal{PT}$-symmetric Hamiltonians with potentials 
\begin{align}
V(x) = x^2 (\mathfrak{i} x)^\epsilon
\end{align}
with $\epsilon \in \mathbb{R}$ \cite{bender1998real}.
\demo
\end{ex}

Analysis of finite-dimensional Hilbert spaces is simpler than analysis of their infinite-dimensional counterparts. Thus, a discrete version of the $\mathcal{PT}$ operator from the previous example is desirable, and is given in the following example. 
 
\begin{ex} \label{typicalPTex}
Let the Hilbert space be the complex coordinate space, $\mathbb{C}^n$. An orthonormal basis of $\mathbb{C}^n$ is the \textit{canonical basis}, defined by $(e_i)_j = \delta_{ij}$ for $i \in \{1, \dots, n\}$. A typical choice for parity is the exchange matrix\footnote{Some sources refer to $\mathcal{P}$ as the \textit{standard involutionary permutation matrix}, or \textit{sip matrix} for short \cite{Gohberg2005}.}, 
\begin{align}
\mathcal{P} e_i = e_{n - i + 1}, \label{parityFinite}
\end{align}
and a typical choice for time reversal is the complex conjugation operator in the basis $\{e_i\,|\,i\in \{1, \dots, n\}\}$. Explicitly, given a general vector of the form $\sum_{i = 1}^n \psi_i e_i \in \mathbb{C}^n$ with $\psi_i \in \mathbb{C}$, the action of $\mathcal{T}$ is
\begin{align}
\mathcal{T} \sum_{i = 1}^n \psi_i \gls{ei} = \sum_{i = 1}^n \psi_i^* e_i. \label{timeReverseFinite}
\end{align} 
\demo
\end{ex} 
\begin{defn} \label{centrohermitianDefn}
If the elements of a matrix, $A \in \mathfrak{M}_n(\mathbb{C})$, satisfy $A_{i,j} = A^*_{n-i+1,n-j+1}$, then $A$ is a \textit{centrohermitian} matrix \cite{Lee1980}. $A$ is centrohermitian if and only if $A$ is $\mathcal{PT}$-symmetric with $\mathcal{P}$ and $\mathcal{T}$ defined in \cref{typicalPTex}.
\end{defn}

The following theorem characterizes the spectrum of an operator with an antilinear symmetry. In infinite-dimensional Hilbert spaces, the spectrum of an operator can be larger than the set of its eigenvalues. The set of eigenvalues is also referred to as the \textit{point spectrum}, is denoted by $\sigma_p$, and is analysed in \cref{Section:AntilinearSymmetryBreaking}.
\begin{thm} \label{AntilinearSpectrumSymmetry}
Let $A$ denote an operator with a bounded, bijective antilinear symmetry. Let $\sigma_{p, c, r}(A)$ denote the decomposition of the spectrum of $A$ into point, continuum, and residual parts\footnote{Spectral theory is reviewed in \cref{Section:Spectral-Theory}, which includes a discussion of the decomposition of the spectrum into point, continuun, and residual parts in definitions~\ref{pointSpectrum} and \ref{continuousResidual}.}. Then 
{\normalfont \cite{Siegl2009,siegl2008quasi}}
\begin{align}
\sigma_{p,c,r}(A) = \sigma_{p,c,r}(A)^*.
\end{align}
\end{thm}


Somewhat misleadingly, when the antilinear symmetry is a product of a linear operator, $\mathcal{P}$, with a complex conjugation operator, such as $\mathcal{T}$, the conjugation operation is often referred to as \textit{time-reversal}, even though it does not generate time-reversal for a general $\mathcal{PT}$-symmetric Hamiltonian. To attempt to explain this terminology, consider a Hamiltonian with continuous dependence on a parameter. A simple choice of such a Hamiltonian is one where every operator has the same antilinear symmetry, $\mathcal{PT}$. If there exists a parameter such that $\mathcal{P}$ is a linear symmetry of the Hamiltonian, then $\mathcal{T}$ is a time-reversal operator for this choice of parameter. In the example of a Schr{\"o}dinger operator with complex-valued potential $V(x) = x^2 (\mathfrak{i} x)^\epsilon$ which was mentioned in \cref{benderPTExample}, the Hermitian limit $\epsilon \to 0$ corresponds to the $\mathcal{T}$-symmetric choice.

\subsection{Antilinear Symmetry Breaking} \label{Section:AntilinearSymmetryBreaking}
\textit{Spontaneous symmetry breaking} is a ubiquitous concept in physics. In condensed matter physics, spontaneous symmetry breaking is used to classify phases of matter \cite{kittel2018introduction}.
In the standard model of particle physics, spontaneous symmetry breaking is used to generate the masses of fundamental particles via the Higgs mechanism \cite[\S 4.3]{langacker2017standard}. 

This section discusses the relationship between antilinear symmetry breaking and the reality of the spectrum of an antilinear symmetric operator.
\begin{defn} \label{unbrokenDefn}
Given an antilinear operator on a Hilbert space, $\Theta$, a linear operator, $A: \text{Dom}(A) \to \mathcal{H}$, is called $\Theta$-\hypertarget{link:Unbroken}{\textit{unbroken}} if and only if $\Theta$ is an antilinear symmetry of $A$ and the eigenspaces of $A$ are invariant subspaces\footnote{See \cref{invariantSubspaceOperatorDefn} for the definition of an invariant subspace.} of $\Theta$,
\begin{align}
\Theta \ker (\lambda \mathbb{1} - A) \subseteq \ker (\lambda \mathbb{1} - A)  &\quad& \forall \lambda \in \sigma_p(A).
\end{align} 
If $A$ has the antilinear symmetry $\Theta$ and $A$ is not $\Theta$-unbroken, then $A$ is called $\Theta$-\textit{broken}.
\end{defn}

The following result relates the reality of an eigenvalue of an operator with an antilinear symmetry to the symmetry properties of its eigenspace.
\begin{theorem} \label{PT-Unbroken-Lemma}
Consider a linear operator on a Hilbert space, $A: \text{Dom}(A) \to \mathcal{H}$, with an injective antilinear symmetry, $\Theta$, and consider an eigenvalue $\lambda \in \sigma_p(A)$.
\begin{enumerate}
\item $\Theta \ker(\lambda \mathbb{1} - A) \subseteq \ker(\lambda^* \mathbb{1} - A)$. Thus, for every eigenvalue $\lambda$, its complex conjugate is an eigenvalue of $A$, so $\lambda^* \in \sigma_p(A)$.\\
\item The eigenspace of $A$ associated to $\lambda$ is an invariant subspace of $\Theta$ if and only if $\lambda \in \mathbb{R}$. \label{PT-Unbroken-Lemma-2}\\
\item If the eigenspace associated to $\lambda$ is an invariant subspace of $\Theta$, then the rank-$m$ generalized eigenspaces associated to $\lambda$ are also invariant subspaces of $\Theta$, so
\begin{align}
\Theta \ker (\lambda \mathbb{1} - A)^m \subseteq \ker (\lambda \mathbb{1} - A)^m  &\quad& \forall (m, \lambda) \in \mathbb{Z}_+ \times \sigma_p(A).
\end{align} 
\end{enumerate}
\end{theorem}
\begin{proof}
This proof is adapted from \cite{bender1999pt}, which addressed the case of simple eigenvalues. 
\begin{enumerate}
\item A brief computation relying on the antilinear symmetry of $A$ yields
\begin{align}
A \Theta v &= \Theta A v \\
&= \Theta \lambda v \\
&= \lambda^* \Theta v. \label{lemma:tempEq}
\end{align}
\item I'll start by displaying the "only if" direction of the assertion of \cref{{PT-Unbroken-Lemma-2}}. Consider an eigenvector corresponding to $\lambda$, $v \in \ker(\lambda \mathbb{1} - A) \setminus \{0\}$. Since $\Theta$ is injective, $\Theta v$ is a nonzero eigenvector of $A$ with the eigenvalue $\lambda^*$. By supposition, the vector $\Theta v$ is in the eigenspace of $A$ associated to the eigenvalue $\lambda$. Thus, $\lambda = \lambda^* \in \mathbb{R}$.

The remaining direction of the assertion of \cref{{PT-Unbroken-Lemma-2}} is proven by contrapositive. Assume $\ker(\lambda \mathbb{1} - A)$ is not an invariant subspace of $\Theta$. Then, there exists an eigenvector, $v \in \text{Dom}(A) \setminus \{0\}$, such that $\Theta v$ is not in the eigenspace associated to $\lambda$. Equation~\eqref{lemma:tempEq} demonstrated that $\Theta v$ is an eigenvector with the eigenvalue $\lambda^*$, thus, $\lambda \neq \lambda^*$.

\item Since $\lambda = \lambda^*$ has been established,
\begin{align}
(\lambda \mathbb{1} - H)^m \Theta = \Theta (\lambda \mathbb{1} - H)^m
\end{align} 
holds for all $m \in \mathbb{Z}^+$. Thus, if $v \in \ker(\lambda \mathbb{1} - H)^m$, then $\Theta v \in \ker (\lambda \mathbb{1} - H)^m$.
\end{enumerate}
\end{proof}

\begin{ex}[continues=ex:HermitianQuantumParticle]
Schr{\"o}dinger operators with real-valued potentials, which are the Hamiltonians associated to the quantum theory of non-interacting particles in $n$-dimensional Euclidean space, are $\mathcal{T}$-unbroken. This stems from two observations: firstly, that the spectrum is real, and secondly, if $\psi(x)$ is an eigenfunction, then so is $\psi^*(x)$. 
\demo
\end{ex} %

\begin{ex}[continues=benderPTExample]
If the potential is complex, then the standard results from Sturm-Liouville theory regarding reality of the spectrum and eigenfunctions no longer apply. Continuing with the analysis of $V(x) = x^2 (\mathfrak{i} x)^\epsilon$, unbroken $\mathcal{PT}$-symmetry protects reality of the spectrum for $\epsilon \geq 0$, while broken $\mathcal{PT}$-symmetry manifests for $\epsilon < 0$, where only a finite subset of the spectrum is real \cite{bender1998real,Dorey2001}. Even in the (non-Hermitian) $\mathcal{PT}$-unbroken case, the eigenfunctions are no longer real-valued. Consequently, the node theorem, which characterizes the interlacing of roots of eigenfunctions, no longer holds. An interesting generalization of the interlacing property to $\mathcal{PT}$-unbroken Schr{\"o}dinger operators regards the \textit{winding numbers} of eigenfunctions. Numerical results suggest winding numbers of an eigenfunction of a $\mathcal{PT}$-symmetric Schr{\"o}dinger operator match their corresponding quantum numbers \cite{Schindler2017}. 
\demo
\end{ex}
The next theorem asserts the equivalence of the existence of an unbroken antilinear symmetry with the reality of the spectrum for operators on finite-dimensional Hilbert spaces.


Note in the case where the spectrum of $A$ consists only of simple eigenvalues, $A$ is $\Theta$-unbroken if and only if every eigenvector of $A$ is an eigenvector of $\Theta$.


\begin{thm} \label{unbrokenRealSpectrum}
The spectrum of a matrix, $A \in \mathfrak{M}_n(\mathbb{C})$, is a subset of the reals, $\sigma(A) \subset \mathbb{R}$, if and only if there exists a antilinear symmetry, $\Theta$, satisfying $\Theta^2 = \gls{id}$, such that $A$ is $\Theta$-unbroken. If $A$ is $\Theta$-unbroken for one invertible antilinear symmetry, then $A$ is $\Theta$-unbroken for all invertible antilinear symmetries.
\end{thm}

%

\begin{proof}
See \cref{UnbrokenTheoremProof}.
\end{proof}

I would find a generalization of \cref{unbrokenRealSpectrum} which applies to operators on infinite-dimensional Hilbert spaces quite interesting. As far as I'm aware, no such generalization is presently known. For such a generalization to exist, the notion of unbroken antilinear symmetry must be different from the one provided in \cref{unbrokenDefn}. Notably, in the setting of a one-dimensional particle given in \cref{benderPTExample}, the operator $\mathfrak{i} \hat{x}^3$ has no eigenspaces, so it is trivially $\mathcal{PT}$-unbroken; yet, $\mathfrak{i} \hat{x}^3$ has an imaginary spectrum.




%

%

\section{Pseudo-Hermiticity} \label{PseudoHermSection}

A class of operators which is closely related to the set of operators with antilinear symmetry is the set \textit{pseudo-Hermitian operators}. This class of operators defines an extension of Hermitian operators, and includes Hermitian operators as a limiting case. Theorem~\ref{pseudoEquivThm} characterizes the structure of pseudo-Hermitian matrices and could be considered the pinnacle of this section. In particular, this theorem includes the assertion that an operator acting on a finite-dimensional space is pseudo-Hermitian if and only if it has an invertible antilinear symmetry.

Pseudo-Hermiticity is defined using an \textit{intertwining operator}. Physically, the intertwining operator corresponds to a real-valued conserved quantity \cite{bian2019time}; this quantity is conserved in time when evolution is dictated by a pseudo-Hermitian Hamiltonian in the Schr{\"o}dinger equation. 

The natural setting for pseudo-Hermitian operators is an \textit{indefinite inner product space}, which is defined in generality in \cref{indefiniteInnerProduct}. This thesis only considers a special type of an indefinite inner product space; the $\eta$-\textit{space} defined in \cite[p. 39]{azizov1989linear} and below.
\begin{defn} \label{etaSpace}
A bounded injective\footnote{An operator is injective if and only if its kernel is the set containing only the zero vector.} linear operator on a Hilbert space, $\eta \in \mathcal{B}(\mathcal{H})$, defines an $\eta$-\textit{space}, an indefinite inner product space on $\mathcal{H}$ with the indefinite inner product 
\begin{align}
\braket{\psi|\phi}_\eta := \braket{\psi|\eta \phi}. \label{etaInnerProduct}
\end{align}
The special case of a traditional inner product space corresponds to the choice $\eta = \mathbb{1}$.
\end{defn}

\begin{defn} \label{pseudoHermDefn}
An operator, $A: \text{Dom}(A) \to \mathcal{H}$, is \textit{weakly pseudo-Hermitian} if and only if there exists a bounded, bijective, linear \textit{intertwining operator}\footnote{$\eta$ is sometimes referred to as a \textit{Gram operator} \cite[p. 89]{Bognr1974}.}, $\eta \in \text{GL}(\mathcal{B}(\mathcal{H}))$, such that
\begin{align}
A= \eta^{-1} A^\dag  \eta.
\end{align}
If the intertwining operator, $\eta$, is also self-adjoint, then $A$ is called \hypertarget{link:PseudoHerm}{\textit{pseudo-Hermitian}} or $\eta$-\textit{self-adjoint}\footnote{The terminology of $\eta$-self-adjointness can be found in many textbooks, such as \cite[p. 111]{MethodsOfMatrixAlgebra}, \cite[p. 104]{azizov1989linear}, and \cite[p. 48]{Gohberg2005}.}. 
If instead the intertwining operator is bounded, bijective, Hermitian, and antilinear, $A$ is \textit{anti-pseudo-Hermitian}.
\end{defn}

Special cases of pseudo-Hermitian matrices include the centrohermitian matrices of \cref{typicalPTex}, the perhermitian matrices, as defined in \cite{Hill1990}, and the $\kappa$-real and $\kappa$-Hermitian matrices of \cite{Hill1992}.
%

A closely related concept is that of \textit{pseudo-unitarity}, as defined in \cite[p. 134]{azizov1989linear}.
\begin{defn} \label{pseudoUnitary}
An operator, $U \in \mathcal{B}(\mathcal{H})$ is \textit{pseudo-unitary} or $\eta$\textit{-unitary} if and only if $U$ is surjective\footnote{Pseudo-unitary operators are necessarily surjective in the finite-dimensional case.} and there exists a bounded, bijective linear operator $\eta \in \text{GL}(\mathcal{B}(\mathcal{H}))$ such that
\begin{align}
U^\dag \eta U = \eta.
\end{align}
\end{defn}

Some simple properties of self-adjoint and unitary operators can be generalized to their $\eta$-self-adjoint and $\eta$-unitary counterparts. In analogy to \cref{adjointDefn}, if $A$ is $\eta$-self-adjoint, then
\begin{align}
\braket{\psi|A \phi}_\eta = \braket{A \psi|\phi}_\eta &\quad& \forall \psi, \phi \in \text{Dom}(A). \label{pseudoAdjoint}
\end{align}
If $U$ is $\eta$-unitary, then it preserves the indefinite inner product,
\begin{align}
\braket{U \psi|U \phi}_\eta = \braket{\psi|\phi}_\eta &\quad& \forall \psi, \phi \in \mathcal{H}.
\end{align}
If $H \in \mathcal{B}(\mathcal{H})$ is $\eta$-self-adjoint, then $U := e^{-\mathfrak{i} t H}$ is $\eta$-unitary when $t \in \mathbb{R}$. The following definition introduces pseudo-unitary representations, which can be used to define time evolution. 
\begin{defn}
Let $(G, \cdot)$ denote a group with the binary operation $\cdot$ and identity $e$. A \textit{representation} of $(G,\cdot)$ on a Banach space\footnote{A Banach space is defined in \cref{defn:Banach}.}, $(X,||\cdot||)$, is a map $\Phi:G \to \text{GL}(\mathcal{B}(X))$ such that for all $g, g' \in G$,
\begin{align}
\Phi(e) &= \mathbb{1} \\
\Phi(g) \Phi(g') &= \Phi(g g') \\
\Phi(g^{-1}) &= (\Phi(g))^{-1}.
\end{align}
A representation of a group is \textit{strongly continuous} if it is continuous in the strong operator topology. Given a map $\eta:G \to \text{GL}(\mathcal{B}(\mathcal{H}))$ such that $\eta(g) = \eta^\dag(g)$, a representation is $\eta$\textit{-unitary} if for every $g \in G$,
\begin{align}
\Phi^\dag(g) \eta(g) \Phi(g) = \eta(g).
\end{align}
\end{defn}

If time evolution of the state $\psi \in \mathcal{H}$ is governed by a strongly-continuous $\eta$-unitary representation of the reals, so that $\psi(t) = U(t-t_0) \psi(t_0)$, then the intertwining operator $\eta$ corresponds to a real-valued conserved quantity, \cite{bian2019time}
\begin{align}
\frac{d}{dt} \braket{\psi(t)|\psi(t)}_\eta = 0.
\end{align} 

A pseudo-Hermitian operator has many intertwining operators associated to it. 
In particular, consider an $\eta$-self-adjoint operator, $H$, which commutes with a bijective operator, $O$, so that $H O = O H$. If $O$ is $\eta$-self-adjoint, then $H$ is $\eta O$-self-adjoint as well. Examples for $O$ include the powers of $H$. Thus, the operators 
\begin{align}
\eta_k := \eta H^k \label{generative}
\end{align}
are also intertwiners associated to $H$ \cite{bian2019time}. If the spectrum of a matrix $H$ consists of only simple eigenvalues, then the span of the $\eta_k$ is the entire space of intertwining operators.

To conclude this section, I note that in the spirit of \cref{pseudoAdjoint}, adjoints can be defined on $\eta$-spaces \cite[\S 4.1]{Gohberg2005} \cite[\S 2.1]{azizov1989linear}. In the case where $A$ is bounded,
\begin{align}
A^{\dag_\eta} := \eta^{-1} A^\dag \eta.
\end{align}

\subsection{Literature Review of Pseudo-Hermiticity}

Pseudo-Hermitian operators have been investigated in the physics literature since the early 1940s, dating back to work by Dirac \cite{Dirac1942} and Pauli \cite{Pauli1943} on the quantization of fields. Dirac and Pauli observed that while the spectrum of a pseudo-Hermitian operator may not be a subset of the reals, the quadratic form $\braket{\psi|\eta A \psi}$ remains real-valued for all $\psi \in \text{Dom}(A)$. The earliest usage of the term "pseudo-Hermitian"  that I am aware of, which is consistent with \cref{pseudoHermDefn}, is in \cite{Sudarshan1961}. 

%

The Russian mathematics community investigated the properties of pseudo-Hermitian operators almost simultaneously. Mark Krein introduced what is commonly referred to as a \textit{Krein space}: an $\eta$-space where $\eta$ is the difference between a projection onto a closed linear subspace and the projection onto its orthogonal complement \cite{Gohberg1969,Azizov1995}. Predating this, Lev Pontryagin studied what is often called a \textit{Pontryagin space}: Krein spaces where range of one of its defining projections is finite-dimensional \cite{pontryagin1944hermitian,PontryaginReprint}. Generic results on pseudo-Hermitian matrices demonstrated by Western scholars in the 1960s include the works of \cite{Carlson1965,Radjavi1969}. Much of the resulting analysis of this topic is compiled in the textbooks of \cite{azizov1989linear,Bognr1974}.

Ali Mostafazadeh revitalized the physics community's interest in pseudo-Hermitian operators in  \cite{BiOrthogonal,mostafazadeh2002pseudo2,mostafazadeh2002pseudo3} by investigating connections with the increasingly popular topic of $\mathcal{PT}$-symmetry \cite{bender1998real,bender1999pt,Dorey2001}. A nice review article on pseudo-Hermiticity in quantum theory is \cite{MostaReview}. 

\subsection{Pseudo-Hermitian Matrices}
The next theorem characterizes the structure of pseudo-Hermitian matrices.
\begin{theorem} \label{pseudoEquivThm}
The following conditions on a matrix, $A \in \mathfrak{M}_n(\mathbb{C})$, are equivalent:
\begin{enumerate}
\item $A$ is pseudo-Hermitian with the invertible, Hermitian intertwiner $\eta = \eta^\dag \in {\normalfont \text{GL}}_n(\mathbb{C})$. \label{Item:PseudoHerm}
\item $A = G \eta$, where $G = G^\dag$, $\eta = \eta^\dag$, and $\eta \in {\normalfont \text{GL}}_n(\mathbb{C})$ is invertible. \label{Item:Product}
\item $A$ is similar\footnote{Similarity is reviewed in \cref{similarDefn}.} to a matrix with real entries. \label{Item:SimilarToReal}
\item $A$ has an involutive antilinear symmetry, $\Theta$, which satisfies, $\Theta^2 = \mathbb{1}$. \label{Item:AntilinearInv}
\item $A$ has an invertible antilinear symmetry. \label{Item:AntilinearSymm}
\item $A$ is similar to $A^\dag$, or equivalently, $A$ is weakly pseudo-Hermitian. \label{Item:SimilarToAdj}
\item $\sigma(A) = \sigma^*(A)$ and $\dim \ker(\lambda \mathbb{1} - A)^l = \dim \ker(\lambda^* \mathbb{1} - A)^l$ for every $l \in \gls{N}$ and $\lambda \in \mathbb{C}$. \label{Item:ComplexConjugateEvals}
\end{enumerate}
\end{theorem}
\begin{proof}
The equivalence (\ref{Item:PseudoHerm}) $\Leftrightarrow$ (\ref{Item:Product}) $\Leftrightarrow$ (\ref{Item:SimilarToReal}) was proven in \cite[Thm. 2]{Carlson1965}. The equivalence (\ref{Item:PseudoHerm}) $\Leftrightarrow$ (\ref{Item:ComplexConjugateEvals}) was proven in \cite{MostaNonDiag}.

The equivalence (\ref{Item:PseudoHerm}) $\Leftrightarrow$ (\ref{Item:SimilarToAdj}) was established in \cite{Radjavi1969,Duke1969,Djokovi1973}, re-discovered in the case where $A$ is diagonalizable in \cite{Solombrino2002}, and re-discovered in the generic case in \cite{Mostafazadeh2006}. An equivalent result that a matrix, $A$, is pseudo-Hermitian if and only if $A$ is similar to the matrix formed by taking the complex conjugate of its elements was established in \cite[Thm. 2]{Zhang2020}\footnote{Equivalence follows from the observation that every matrix is similar to its transpose \cite{Taussky1959}.}. 

The equivalence (\ref{Item:PseudoHerm}) $\Leftrightarrow$ (\ref{Item:AntilinearSymm}) $\Leftrightarrow$ (\ref{Item:AntilinearInv}) was established in \cite{Scolarici2003}. The remainder of this paragraph reviews alternative literature devoted to weaker claims. 
In the case where $A$ is diagonalizable, (\ref{Item:PseudoHerm}) $\Leftrightarrow$ (\ref{Item:AntilinearSymm}) was established in \cite{Solombrino2002,mostafazadeh2002pseudo3}. Special cases of (\ref{Item:AntilinearSymm}) $\Rightarrow$ (\ref{Item:SimilarToReal}) for particular classes of antilinear symmetries were established in \cite{Robnik}. Special cases of (\ref{Item:AntilinearInv}) $\Rightarrow$ (\ref{Item:SimilarToReal}) were provided in \cite{Lee1980,Hill1992}, although these are subsumed by the previous work of \cite{Robnik}. (\ref{Item:AntilinearSymm}) $\Rightarrow$ (\ref{Item:PseudoHerm}) was proven in the case where the antilinear symmetry is the product of an involution and the complex conjugation operator of \cref{timeReverseFinite} in \cite{Zhang2020}, which could be considered a corollary of the results in \cite{Robnik,Carlson1965}.

I will now utilize the equivalence of (\ref{Item:SimilarToReal}) $\Leftrightarrow$ (\ref{Item:SimilarToAdj}) to provide a simple proof that (\ref{Item:AntilinearInv}) and (\ref{Item:AntilinearSymm}) can be added to this list of equivalent criteria.

To prove (\ref{Item:SimilarToReal}) $\Rightarrow$ (\ref{Item:AntilinearInv}), note criteria (\ref{Item:SimilarToReal}) is saying there exists an invertible similarity transform, $S \in \text{GL}_n(\mathbb{C})$, such that $S A S^{-1} = \mathcal{T} S A S^{-1} \mathcal{T}$. Thus, $A$ has the antilinear symmetry $S^{-1} \mathcal{T} S$, which is its own inverse.

(\ref{Item:AntilinearInv}) $\Rightarrow$ (\ref{Item:AntilinearSymm}) holds trivially.

To prove (\ref{Item:AntilinearSymm}) $\Rightarrow$ (\ref{Item:SimilarToAdj}), note every matrix, $A$, with invertible antilinear symmetry, $\Theta$, is similar to $\mathcal{T} A \mathcal{T}$, since 
\begin{align}
\mathcal{T} A \mathcal{T} = \mathcal{T} \Theta A (\mathcal{T} \Theta)^{-1}.
\end{align} 
Since every matrix is similar to its transpose \cite{Taussky1959}, $\mathcal{T} A \mathcal{T}$ is similar to $A^\dag$. Since $A$ is similar to $\mathcal{T} A \mathcal{T}$, $A$ is similar to $A^\dag$.
\end{proof}
\begin{corollary} \label{antiPseudoHerm}
All pseudo-Hermitian matrices, $A$, are anti-pseudo-Hermitian.
\end{corollary}
\begin{proof}
By \cref{pseudoEquivThm}, all pseudo-Hermitian operators with intertwiner $\eta$ have an antilinear symmetry, $\Theta$. Thus,
\begin{align}
\eta \Theta H = H^\dag \eta \Theta.
\end{align}
\end{proof}
The converse of \cref{antiPseudoHerm} does not hold. Given $\mathcal{P}$ and $\mathcal{T}$ defined in \cref{typicalPTex}, all transpose-symmetric matrices are anti-pseudo-Hermitian with respect to $\tau = \mathcal{T}$. The specific transpose-symmetric matrix $\mathfrak{i} \mathcal{P}$ does not satisfy the condition (\ref{Item:ComplexConjugateEvals}) of \cref{pseudoEquivThm} required for pseudo-Hermiticity.

\begin{lemma} \label{pseudoHermTraceDeterminant}
The coefficients of the characteristic polynomial of a pseudo-Hermitian matrix are real-valued.
\end{lemma}
\begin{proof}
Reality of the coefficients of the characteristic polynomial follows the product identity of determinants,
\begin{align}
\text{det}(\lambda \mathbb{1}- A) = \text{det}(\eta^{-1} (\lambda \mathbb{1} - A^\dag)\eta) = \text{det}(\lambda \mathbb{1} - A^\dag) \text{det}(\eta^{-1}) \text{det}(\eta) = \text{det}(\lambda^* \mathbb{1} - A)^*.
\end{align}
Since the characteristic polynomial maps every $\lambda \in \mathbb{R}$ to a real number, it must have real coefficients. 
\end{proof}
A corollary of the above lemma is that the trace of an $n \times n$ pseudo-Hermitian matrix is real, since the coefficient of the $\lambda^{n-1}$ term in the characteristic polynomial is $- \text{tr}(A)$. An alternative proof of the reality of the trace, which generalizes to trace-class operators on possibly infinite-dimensional Hilbert spaces, follows from its cyclicity property,
\begin{align}
\text{Tr}(A) = \text{Tr}(\eta^{-1} A^\dag \eta) = \text{Tr}(A^\dag) = \text{Tr}(A)^*.
\end{align}

A result which is equivalent to the above lemma is that the coefficients of the characteristic polynomial of a matrix with an injective antilinear symmetry are real; this equivalent result was derived in \cite{bender2010pt}. An earlier result regarding reality of the characteristic polynomial coefficients for special cases of injective antilinear symmetries was given in \cite{Mandilara2002}.

\subsection{Pseudo-Hermitian Operators} \label{pseudoHermOperatorsSection}

%

The following theorem addresses the spectrum of a pseudo-Hermitian operator on a possibly infinite-dimensional Hilbert space.
\begin{theorem} \label{pseudoHermSpec}
Let $\eta$ be a bounded bijective self-adjoint linear operator on a Hilbert space, and $A$ be a densely-defined closed\footnote{All finite-dimensional operators are closed. An operator is closed if its graph is closed in the direct sum topology.} linear operator. Using the notation of definitions~\ref{pointSpectrum} and \ref{continuousResidual} for the decomposition of the spectrum of $A$ into point, continuum, and residual parts, the following identities hold {\normalfont \cite[p. 89]{azizov1989linear}}:
\begin{align}
\lambda \in \sigma_{p,1}(A) &\Leftrightarrow \lambda^* \in \sigma_{p,1}(\eta^{-1} A^\dag \eta) \\
\lambda \in \sigma_{r}(A) &\Leftrightarrow  \lambda^* \in \sigma_{p,2}(\eta^{-1} A^\dag \eta) \\
\lambda \in \sigma_{c}(A) &\Leftrightarrow  \lambda^* \in \sigma_{c}(\eta^{-1} A^\dag \eta).
\end{align}
Every weakly pseudo-Hermitian operator is closed {\normalfont \cite{Siegl2009,siegl2008quasi}}.
\end{theorem}

Petr Siegl remarked that the necessary conditions on the spectrum for pseudo-Hermitian and antilinear symmetric operators given in \cref{pseudoHermSpec,AntilinearSpectrumSymmetry} are not the same in the infinite-dimensional case \cite{Siegl2009,siegl2008quasi}. Following this observation, he constructed an example of a pseudo-Hermitian operator which does not have an antilinear symmetry and an operator with an antilinear symmetry which is not weakly pseudo-Hermitian. 

\section{Quasi-Hermiticity} \label{quasiHermSection}




This section is an introduction to quasi-Hermitian quantum theory. I summarize the axioms of this theory in \cref{quasiHermitianQuantumAxioms}. The text leading up to \cref{quasiHermitianQuantumAxioms} motivates these axioms and discusses mathematical properties of quasi-Hermitian linear operators, which are defined in \cref{quasiHermDefn}. In the mathematical framework of quasi-Hermitian quantum theory, as in traditional quantum theory, states define probability measures over sets of real measurement outcomes.
These probability measures are derived from the spectral decomposition of quasi-Hermitian operators. Time evolution is pseudo-unitary, in the sense of \cref{pseudoUnitary}, so states continue to correspond to probability measures while they are time evolved.
Quasi-Hermitian operator algebras include Hermitian operator algebras as a limiting case. The key insight behind the above considerations is to generalize the inner product structure via the so-called \textit{metric operator}.

Quasi-Hermiticity was originally defined in \cite{dieudonne} and the first treatment of a quantum theory built on quasi-Hermitian operator algebras was \cite{QuasiHerm92}.

\subsection{Mathematical Considerations}

\begin{defn} \label{quasiHermDefn} \cite{dieudonne}
An operator on a Hilbert space, $O$, is called \hypertarget{link:QuasiHerm}{\textit{quasi-Hermitian}} if there exists a bounded, positive-definite operator, $\eta > 0$, such that
\begin{align}
\eta O = O^\dag \eta. \label{eqn-Dieudonne}
\end{align}
$\eta$ is referred to as a \textit{metric operator} associated to $O$, and \cref{eqn-Dieudonne} is referred to as the \textit{Dieudonn{\'e} relation}.
\end{defn}

\begin{ex}
Self-adjoint operators are the special case of quasi-Hermitian operators whose metric operator is the identity.
\demo
\end{ex}

The connection between quasi-Hermitian and pseudo-Hermitian operators is realized by noting every bijective metric operator is an intertwining operator.


While sufficient, Hermiticity is not necessary to ensure a real spectrum of a linear operator. A simple example of a non-Hermitian matrix with real eigenvalues is given in \cref{PTqubit}.
As stated in \cref{quasiHermDiagonalizable}, quasi-Hermiticity is a necessary and sufficient condition for a diagonalizable\footnote{Definition~\ref{diagonalizable} reviews the notion of diagonalizable operators.} $n \times n$ matrix to have a set of real eigenvalues. A more general theorem, proven in \cite[Thm. 1]{Drazin1962}, gives the necessary and sufficient condition for the existence of $m \leq n$ linearly independent eigenvectors whose eigenvalues are real.
\begin{theorem} \label{realEigenvalsThm}
A matrix, $O \in \mathfrak{M}_n(\mathbb{C})$, has at least $m \in \mathbb{N}$ linearly independent eigenvectors whose eigenvalues are real if and only if there exists a positive-semidefinite matrix, $S \in \mathfrak{M}_n(\mathbb{C})^+$, with rank $m = {\normalfont \text{dim}}(S(\mathbb{C}^n))$ such that $S O = O^\dag S$ {\normalfont \cite[Thm. 1]{Drazin1962}}. 
\end{theorem}
\begin{corollary}  \label{quasiHermDiagonalizable}
A matrix is diagonalizable with a real spectrum if and only if it is quasi-Hermitian. 
\end{corollary}
The earliest proof that quasi-Hermitian matrices have a real spectrum which I am aware of is an immediate consequence of \cite[Thm. 5.1]{Reid1951}; other sources demonstrating the reality of eigenvalues of quasi-Hermitian matrices include \cite[p. 259]{zaanen1953linear} and \cite[p. 111]{MethodsOfMatrixAlgebra}. The earliest proofs of \cref{quasiHermDiagonalizable} which I am aware of are given in \cite[Thm. 1]{Drazin1962} and \cite[Thm. 3.3]{silberstein1962symmetrisable}. An equivalent statement to corollary \cref{quasiHermDiagonalizable} was proven slightly earlier in \cite{Taussky1960}. A paper demonstrating \cref{quasiHermDiagonalizable} that appeared forty years later, which is well-cited in the physics literature, is \cite{mostafazadeh2002pseudo2}. The implication that a diagonalizable matrix with real eigenvalues is necessarily quasi-Hermitian is a textbook problem \cite[p. 66]{Gohberg2005}.


One insightful proof of \cref{quasiHermDiagonalizable} follows from the following propositions. 
\begin{proposition} \label{generalMetricProp}
Given a matrix which is diagonalizable with real spectrum, $O = \mathcal{U} \,D\, \mathcal{U}^{-1}$, where $D$ is Hermitian and diagonal in some orthonormal basis, the most general metric operator associated to $O$ is \cite[eq. (8)]{BiOrthogonal}
\begin{align}
\eta^{-1} = \mathcal{U} \, d \,\mathcal{U}^\dag, \label{generalMetric}
\end{align}
where $d$ is a positive-definite matrix which commutes with $D$.
\end{proposition}


\begin{proposition} \cite[Thm. 2]{williams1969operators}
Let $O$ be a linear operator on a Hilbert space, $\mathcal{H}$, which is similar to its adjoint. Explicitly, 
\begin{align}
O = S^{-1} O^\dag S,
\end{align}
where $S \in \text{GL}(\mathcal{B}(\mathcal{H}))$ is bounded with bounded inverse.
Defining the \textit{numerical range} of an operator, $S$, to be the set 
\begin{align}
W(S) &:= \left\{ \frac{\braket{x|S x}}{\braket{x|x}}\,|\, x \in \text{Dom}(S) \gls{rel-comp} \{0\} \right\},
\end{align}
assume $0 \notin \text{cl}(W(S))$, where $\text{cl}$ denotes the closure of a set in the induced norm topology. Then $O$ is similar to a Hermitian linear operator. Explicitly, denoting $\tilde{S} = (S + S^\dag)^{1/2}$, we have
\begin{equation}
h := \tilde{S}^{1/2} O \tilde{S}^{-1/2} = h^\dag.
\end{equation}
\end{proposition}
\begin{corollary} \label{Williams1969Corollary}
If $O$ is a quasi-Hermitian linear operator with a bijective positive metric operator, $\eta$, then $O$ is similar to a Hermitian linear operator. Explicitly, denoting the unique positive square root of $\eta$ by $\Omega$, we have 
\begin{equation}
h := \Omega O \Omega^{-1} = h^\dag. \label{Similar}
\end{equation}
\end{corollary}


Corollary~\eqref{Williams1969Corollary} implies that there is some freedom in the representation of a Hamiltonian: it can either be represented by a quasi-Hermitian operator with metric $\eta$, or by a similar Hermitian operator with metric $\mathbb{1}$. This freedom can be generalized further; given a pair of Hamiltonian $H$ and metric $\eta$, a positive-definite metric for all similar Hamiltonians $S^{-1} H S$ is given by~\cite{kretschmer2001interpretation}, 
\begin{equation}
(H,\eta)\leftrightarrow (H',\eta')=(S^{-1}HS,\eta'=S^\dag \eta S). \label{MetricMapper}
\end{equation}
However, the properties manifest in a quasi-Hermitian representation may be lost by performing a similarity transformation. For instance, local quasi-Hermitian Hamiltonians are generically similar to nonlocal Hermitian Hamiltonians \cite{Korff2008}. 

The most general quasi-Hermitian operator associated to a bijective metric operator is a product of a Hermitian operator and the metric. For matrices, this is discussed in \cite{Carlson1965} and in \cref{pseudoEquivThm}.

Quasi-Hermitian quantum theory postulates that measurement outcomes are associated to Borel subsets of the spectrum of quasi-Hermitian linear operators. To complete the kinematical structure of quasi-Hermitian quantum theory, we need to define probability measures on the space of measurement outcomes. In Hermitian quantum theory, this was accomplished from the spectral decomposition of a Hermitian linear operator. 

In the case where $\eta$ is a positive-definite matrix acting on a finite-dimensional inner product space, the $\eta$-space it generates is a Hilbert space and quasi-Hermitian matrices, $H$, are Hermitian in this Hilbert space.
In the interest of simplicity, first consider the case where each eigenvalue $\lambda$ is simple and associated to the unique (up to rescaling) eigenvector $\ket{\lambda}$. The spectral theorem for Hermitian operators in the $\eta$-space derives a decomposition for $H$,
\begin{align}
H &= \sum_{\lambda \in \sigma(H)} \lambda \ket{\lambda}\bra{\lambda} \eta. \label{matrixQuasiHermSpectralDecompSimple}
\end{align}
In the case where the spectrum contains possibly degenerate semi-simple\footnote{Semi-simple eigenvalues and are defined in \cref{multiplicities}.} eigenvalues, a spectral decomposition of $H$ is given in terms of a set $\eta$-self-adjoint projections, $P_\lambda$, associated to each eigenvalue, $\lambda \in \sigma(H)$, such that 
\begin{align}
P_\lambda P_{\lambda'} &= \delta_{\lambda \lambda'} P_{\lambda} \\
\eta P^{}_\lambda &= P_\lambda^\dag \eta \\
\mathcal{R}(P_\lambda) &= \ker(\lambda \mathbb{1} - H)
\end{align}
and $H$ admits the spectral decomposition
\begin{align}
H = \sum_{\lambda \in \sigma(H)} \lambda P_\lambda. \label{matrixQuasiHermSpectalDecomp}
\end{align}
If the projections $P_\lambda$ are rank one, \cref{matrixQuasiHermSpectalDecomp} reduces to \cref{matrixQuasiHermSpectralDecompSimple}.

\subsection{Infinite-Dimensional Perplexities}
The spectral theorem can be applied to the special case where a quasi-Hermitian linear operator's metric is bounded and bounded below\footnote{Equivalently, due to the bounded inverse theorem \cite[Thm. 14.5.1]{Narici2010} 
, the metric is bounded and bijective. Normal operators have bounded inverse if and only if they are bounded below. 
}, where a bounded below operator is defined as follows.
\begin{defn}
An operator $A:\text{Dom}(A) \to \mathcal{H}$ is \textit{bounded below} if 
\begin{align}
\text{inf} \{||A \psi||/||\psi|| \, | \, \psi \in \text{Dom}(A) \} > 0.
\end{align}
\end{defn}
To see why, observe that for every bounded positive-definite operator on a Hilbert space, $\eta \in \mathcal{B}(\mathcal{H})^+$, the sesquilinear form $\braket{\cdot|\cdot}_\eta:\mathcal{H} \times \mathcal{H} \to \mathbb{C}$ defined in \cref{etaInnerProduct} is an inner product on $\mathcal{H}$. The norms induced by the inner products $\braket{\cdot|\cdot}$ and $\braket{\cdot|\cdot}_\eta$ are equivalent if and only if $\eta$ is bounded below. Denoting the completion of $\mathcal{H}$ in the inner product $\braket{\cdot|\cdot}_\eta$ as $\mathcal{H}_\eta$, quasi-Hermitian operators with metric $\eta \in \mathcal{B}(\mathcal{H})^+$ are self-adjoint in $\mathcal{H}_\eta$ if $\eta$ is bounded below.

If $\eta$ is not bounded below\footnote{Some physicists have incorrectly claimed that every bounded positive-definite operator has a bounded inverse, such as in \cite[App. A]{QuasiHerm92}. 
The faulty logic used in this reference is the claim that $\mathcal{R}(A) \subsetneq \mathcal{H}$ implies $A^\perp \neq \{0\}$, which only holds for finite-dimensional Hilbert spaces. A simple example of a bounded positive-definite operator which does not have a bounded inverse is $A:\ell^2 \to \ell^2$ defined by $A e_n = \frac{1}{n} e_n$, where $e_n$ an orthonormal basis of $\ell^2$. The bounded inverse theorem does not apply to $A$ since $A$ is not onto.}, the spectrum of a quasi-Hermitian linear operator need not be real. An example of a quasi-Hermitian operator with a complex spectrum was given in \cite{dieudonne} and later reprinted in \cite[p. 377-378]{Istratescu1981}. However, the eigenvalues of a quasi-Hermitian linear operator are always real. Quasi-Hermitian linear operators with complex spectrum only exist in infinite-dimensional Hilbert spaces, since every positive-definite matrix is bounded below. 

To avoid the above complications, in this thesis, I assume that the metric operator is bounded below. I do not claim this is a necessary assumption for quantum theory. For example, the densely defined operator on $L^2(\mathbb{R}, \lambda)$ given by
\begin{align}
H = - \frac{d^2}{d x^2} + \mathfrak{i} x^3,
\end{align}
has a real spectrum \cite{Dorey2001} and is quasi-Hermitian \cite{siegl2012metric}, but there does not exist a metric operator with bounded inverse \cite{siegl2012metric}.

\subsection{Quasi-Hermitian Quantum Theory}
 
I will now reformulate the Dirac-von Neumann axioms with quasi-Hermitian observables. To formally define the spectral decomposition of an observable in a possibly infinite-dimensional setting, I will introduce the notion of an $\eta$-\textit{projection-valued measure}. In the finite-dimensional case, this notion reduces to the simpler spectral decomposition given in \cref{matrixQuasiHermSpectalDecomp}.

\begin{definition} 
Let $\Sigma$ be a $\sigma$-algebra, defined in \cref{sigmaAlgAppendix}, and $\mathcal{H}$ a Hilbert space. An operator valued measure, $E$, is a map, $E: \Sigma \to \mathcal{B}(\mathcal{H})$, that is weakly\footnote{The reason for the usage of the term "weakly" here stems from the fact that $\sum_i E(B_i)$ converges in the weak operator topology for every countably collection of disjoint subsets $B_i \in \Sigma$. Intuitively, the weak operator topology is the smallest topology so that matrix elements are continuous.} countably additive. More precisely, if $B_i \in \Sigma$ is a countable collection of disjoint sets with union $\cup_i B_i = B$, then saying $E$ is weakly countably additive means
\begin{align}
\braket{E(B) x| y} = \sum_i \braket{E(B_i) x| y}.
\end{align}
for all $x, y \in \mathcal{H}$. Furthermore,
\begin{itemize}
\item $E$ is $\eta$-\textit{self-adjoint} if $E(B)$ is $\eta$-self-adjoint for all $B$.

\item $E$ is \textit{spectral} if $E(B_1 \cap B_2) = E(B_1) E(B_2)$ for all $B_1, B_2 \in \Sigma$.

\item $E$ is $\eta$-\textit{projection-valued} if it is $\eta$-self-adjoint, spectral, and $E(X) = \gls{id}$.
\end{itemize}
\end{definition}
 
\begin{defn}[Quasi-Hermitian quantum theory] \label{quasiHermitianQuantumAxioms}
\leavevmode 
\begin{enumerate} 
\item The \textit{state} of a system, $\ket{\psi}$, is an element of a separable Hilbert space, $\mathcal{H}$, which satisfies the normalization condition
\begin{align}
\braket{\psi|\eta|\psi} = 1,
\end{align} 
where $\eta \in \mathcal{B}(\mathcal{H})$ is a bounded below, positive-definite operator referred to as the \textit{metric}.
\\
\item \textit{Observables}, $A$, are quasi-Hermitian operators on $\mathcal{H}$ satisfying $\eta A = A^\dag \eta$. By the spectral theorem for self-adjoint operators \cite[Thm. 10.4]{hall2013quantum}, there exists a unique $\eta$-projection-valued measure, $P_A:\mathfrak{B}_{\sigma(A)} \to \mathcal{B}(\mathcal{H})$ such that
\begin{align}
A = \int_{\sigma(A)} \lambda \, d P_A(\lambda),
\end{align} 
where $\mathfrak{B}_{\sigma(A)}$ denotes the Borel $\sigma$-algebra generated by the spectrum $\sigma(A)$ of $A$, as defined in \cref{Borel}.
\\
\item \textit{Measurements} of observables, $A$, yield real-valued outcomes from the spectrum of $A$, $\sigma(A) \subseteq \mathbb{R}$. If a system is in the state $\ket{\psi}$, the probability of measuring an outcome from the Borel set $B \in \mathfrak{B}_{\sigma(A)}$, $
\textbf{Pr}(B)$, is 
\begin{align}
\textbf{Pr}(B) = \braket{\psi|\eta P_A(B)|\psi}. \label{QHmeasurementOutcomeProbability}
\end{align}
$\textbf{Pr}:\mathfrak{B}_{\sigma(A)} \to [0,1]$ is a probability measure. The state of the system immediately after deducing the measurement outcome lives in the set $B$ is 
\begin{align}
\ket{\psi} \rightarrow \frac{P_A(B)\ket{\psi}}{\sqrt{\braket{\psi|\eta P_A(B)|\psi}}}.
\end{align} 
Integrating \cref{QHmeasurementOutcomeProbability}, we deduce that the expectation value of the observable $A$, $\braket{A}$, is 
\begin{align}
\braket{A} &= \braket{\psi|\eta A | \psi}.
\end{align}
\\
\item \textit{Time evolution} of a state $\ket{\psi} \in \mathcal{H}$ is given by a strongly-continuous $\eta$-unitary representation of $(\mathbb{R},+)$, $U:\mathbb{R} \to \mathcal{B}(\mathcal{H}).$ Consequently, the normalization condition is preserved in time,
\begin{align}
\braket{\psi(t)|\eta|\psi(t)} = \braket{\psi(t_0)|\eta|\psi(t_0)} &\quad& \forall t, t_0 \in \mathbb{R} \label{normPreservationQH}.
\end{align}
\end{enumerate}
\end{defn} 

It is unclear to me whether an extension of Stone's theorem \cite{Stone1930,Stone1932} exists to address the case of quasi-Hermitian time evolution with a possibly time-dependent metric operator. If the unitary, $U$, happens to take the form $e^{-\mathfrak{i} H t/\hbar}$ for some possibly unbounded Hamiltonian operator $H$, then the norm preservation condition of \cref{normPreservationQH} implies
\begin{align}
\frac{d}{d t} \braket{\psi(t)|\psi(t)}_\eta = 0, \label{quasiHermUnitarity}
\end{align} 
which implies \cite{timeDependentInnerProd}
\begin{align}
\mathfrak{i} \hbar \partial_t \eta = H^\dag \eta - \eta H. \label{HamiltonianQHTheory}
\end{align}
The Hamiltonian is an observable if and only if $\eta$ is time-independent. Intuitively, the Hamiltonian corresponds to the energy observable and energy conservation is generated by time-translation invariance, so the explicit breaking time-translation invariance causes non-observability of the Hamiltonian.

Allowing for the possibility of a time-dependent inner product, \textit{any} Hamiltonian matrix can generate pseudo-unitary time evolution, including matrices with complex eigenvalues. Given a choice of Hamiltonian, \cref{HamiltonianQHTheory} defines a partial differential equation for $\eta$, which necessarily has a solution. In the case where the Hamiltonian has complex eigenvalues, the metric operator must be time-dependent. A simple example is $H = \mathfrak{i} \epsilon, \eta = e^{-2 \epsilon t/\hbar}$.

Note quasi-Hermitian quantum theory has a global phase symmetry: the normalized state vectors $\psi \in \mathcal{H}$ and $e^{\mathfrak{i} \alpha} \psi \in \mathcal{H}$ correspond to the same probability measures.

\subsection{$\mathcal{C}$-Symmetry}
Consider an operator which is pseudo-Hermitian with two intertwining operators: $\eta_1$ and $\eta_2$. Then, $\eta_2^{-1} \eta_1$ commutes with $H$ \cite{BiOrthogonal}. Following \cite{KuzhelC}, a particular symmetry of $H$ is as follows:
\begin{defn} \label{defn:C}
An $\eta$-self-adjoint operator, $H$, on the Hilbert space $\mathcal{H}$ has a $\mathcal{C}$-\textit{symmetry} if and only if $\mathcal{C} \in \mathcal{B}(\mathcal{H})$ is an involution, $\mathcal{C}^2 = \mathbb{1}$, which commutes with the Hamiltonian, $\mathcal{C} H = H \mathcal{C}$, and such that $\eta \mathcal{C}$ is positive-definite.
\end{defn}
An $\eta$-self-adjoint operator with $\mathcal{C}$ symmetry is quasi-Hermitian with the metric $\eta \mathcal{C}$. The notion of $\mathcal{C}$-symmetry of a Hamiltonian operator, $H$, was originally introduced in \cite{bender2002complex,CorrectCPT}, loosely based on a connection with quantum electrodynamics, to define a particular metric operator which could be used to define a physical inner product and unitarity.

\section{Exceptional Points and Perturbation Theory} \label{ExceptionalPointsIntro}

This section introduces perturbation theory, with an emphasis on non-Hermitian perturbations. The next paragraphs summarize what perturbation theory is, features of Hermitian perturbation problems, and properties of non-Hermitian perturbation problems which are not present in the Hermitian case.

Physicists are well acquainted with the concept of \textit{parameters}, which I would define as a tunable quantity used to classify a system. To model observables or time evolution in systems dependent on parameters in quantum theory, we must consider operators which are functions of said parameters, $H:\mathbb{C}^d \to \mathcal{B}(\mathcal{H})$. Perturbation theory answers questions about how the properties of operator-valued functions vary with their inputs.

A case familiar to many physicists is that where $H$ is Hermitian and linear in its input. Examples include the Hamiltonians of the Su-Schrieffer-Heeger model \cite{SSH}, the transverse field Ising model \cite{Pfeuty1970}, and the $XY$-model \cite{Lieb1961}. The special case $d = 2$, corresponding to two complex parameters, is referred to as a \textit{pencil}. As discussed in numerous textbooks, such as \cite{dirac1981principles,Sakurai2020,Schumacher2010,bender2013advanced,Fernandez2000,Kato1995}, perturbations of pencils which appear in quantum theory are addressed using Rayleigh–Schrödinger perturbation theory. A central result of Rayleigh–Schrödinger perturbation theory is that eigenvalues and eigenvectors admit Taylor series expansions and, thus, are differentiable. The reason why eigenvalues in the Rayleigh-Schr{\"o}dinger problem are differentiable is because the Hamiltonian is diagonalizable\footnote{A precise version of this connection is given in \cref{HellmannFeynman} \cite[chp. 2 \S 5]{Kato1995}, which appears at the end of the next section.}. Furthermore, a physical feature present in linear Hermitian perturbation problems is level repulsion \cite{wigner1929uber,WignerCollectedWorks,NEUMANN2000}. 

Non-Hermitian operators typically exhibit level attraction, leading to the coalescence of eigenvalues and eigenvectors at \textit{exceptional points}. Exceptional points are a new feature not present in Hermitian perturbation theory and will be defined precisely in \cref{EPdefn}. Due to the coalescence of eigenvectors at an exceptional point, the Hamiltonian is no longer diagonalizable, so the eigenvalues need not be differentiable. Consequently, in general, a Taylor series expansion of the eigenvalues does not exist at exceptional points. Rather, a perturbative expansion of the eigenvalues requires the use of a \textit{Puiseux series}, which is a series expansion including fractional powers. For a simple example of a Puiseux series expansion of eigenvalues at an exceptional point, see the $2 \times 2$ matrix of \cref{qubitHam} and its analysis presented in \cref{PassingThroughEP}.

\subsection{Exceptional Points}

The essence of the study of exceptional points is to characterize the stability of the \textit{multiplicities} of eigenvalues under perturbations. A definition of the multiplicities of an eigenvalue of an operator is reviewed below.

\begin{defn} \label{multiplicities}
Given an eigenvalue of an operator, $A$, on a Banach space, $\lambda \in \sigma_p(A)$, its \textit{geometric multiplicity}, $\mu_g(A;\lambda)$, is the dimension of the set of eigenvectors associated to $\lambda$, or more explicitly,
\begin{align}
\mu_g(A;\lambda) := \dim(\ker(\lambda \mathbb{1} - A)).
\end{align}
The \textit{algebraic multiplicity} of this eigenvalue, $\mu_a(A;\lambda)$, is the dimension of the set of generalized eigenvectors,
\begin{align}
\mu_a(A;\lambda) := \dim\left( \bigcup\limits_{k = 1}^\infty \ker(\lambda \mathbb{1} - A)^k \right).
\end{align}
The algebraic multiplicity of an eigenvalue is always greater than or equal to its geometric multiplicity. An eigenvalue of a complex matrix is 
\begin{itemize}
\item
\textit{semi-simple} if and only if it is algebraic and geometric multiplicities are equal.  \item \textit{simple} if and only if it is semi-simple and the algebraic multiplicity equals one.
\end{itemize}
An operator is \textit{non-degenerate} if and only if all of its eigenvalues are simple. When a particular matrix is clear from the context, I use the shorthand $\mu_{a,g}(\lambda)$.
\end{defn}


In the finite-dimensional case, a simpler definition for the algebraic multiplicity of an eigenvalue can be given:
\begin{theorem}
If $\lambda$ is an eigenvalue of a matrix, $A$, then $\mu_a(A;\lambda)$ is the multiplicity of $\lambda$ in the characteristic polynomial of $A$ {\normalfont \cite[Thm. 1.2.1]{MultiplicityBook}}.
\end{theorem}

Intuitively, exceptional points correspond to discontinuities of the multiplicities of the eigenvalues of an operator. This notion is rigorously defined below.
\begin{defn} \label{EPdefn}
Let $\mathfrak{X}$ be a topological space, $X$ be a Banach space, and consider a continuous operator-valued map, $H: \mathfrak{X} \to \mathcal{B}(X)$. 
An \textit{exceptional point}, $x \in X$, of the map $H$ is one for which the set-valued maps $\mu_a(H(\cdot);\sigma(H(\cdot))):\mathfrak{X} \to \mathbb{P}(\mathbb{N})$ and $\mu_a(H(\cdot);\sigma(H(\cdot))) - \mu_g(H(\cdot);\sigma(H(\cdot))):\mathfrak{X} \to \mathbb{P}(\mathbb{N})$ are discontinuous. A \textit{diabolical point} is a point $x \in X$ which is not an exceptional point, but for which the function $\mu_a(H(\cdot);\sigma(H(\cdot))):\mathfrak{X} \to \mathbb{P}(\mathbb{N})$ is discontinuous \cite{diabolical,diabolical2003}. 
\end{defn}
A well-cited source which introduces the terminology of an exceptional point is Tosio Kato's book \cite{Kato1995}. The definition of an exceptional point presented in his book is different than the one provided above. Specifically, Kato defines exceptional points to be branch points in the set of eigenvalues of an operator-valued function \cite{Heiss2004,Heiss2012}. This definition does not reference the geometric multiplicity and, thus, designates both exceptional and diabolical points as exceptional points. I have opted to use the above definition for two reasons. Firstly, we only gain insight by classifying discontinuities as either exceptional or diabolical. Secondly, the physics of diabolical points is not as interesting as that of exceptional points. Many physics papers implicitly define the term "exceptional points" to not include diabolical points.

To further classify exceptional points, physicists often refer to their \textit{order}. Intuitively, the order quantifies how many eigenvectors coalesce at an exceptional point. I define the order for a special case of exceptional points, namely, when the discontinuity is removable 
and the eigenvalues can be parametrized by continuous functions. One justification for the assumption of continuity of the eigenvalues is given by the following theorem  \cite{Li2019,Kato1995}.
\begin{theorem} \label{continuityEvals}
Assume $A:X \to \mathfrak{M}_n(\mathbb{C})$ is continuous, where $X$ is a connected subset of $\mathbb{R}$. Then the eigenvalues of $A$ can be parametrized by continuous functions from $X$ to $\mathbb{C}$.
\end{theorem}

\begin{defn} \label{defn:order}
Consider an operator-valued map $H:\mathfrak{X} \to \mathcal{B}(X)$ whose domain is a topological space. Suppose $x_0 \in \mathfrak{X}$ is an exceptional point, that the eigenvalues of $H$ can be parametrized by continuous functions in an open neighbourhood of $x_0$, and that the discontinuity in the difference between algebraic and geometric multiplicities is a removable discontinuity. The \textit{order} of the exceptional point, $x_0$, is a measure of the number of simultaneously coalescing eigenvectors at $x_0$, and explicitly is given by
\begin{align}
1 + \sup \left\{\lim_{x \to x_0} 
\begin{vmatrix}
\mu_a(\lambda(x)) - \mu_g(\lambda(x)) \,+\\
\mu_g(\lambda(x_0)) - \mu_a(\lambda(x_0))
\end{vmatrix}
\,:\, \lambda(x) \in \sigma(H(x)) \right\}.
\end{align}
\end{defn}
Note the order of an exceptional point is only well defined if the exceptional point is isolated from other exceptional points.

One existence theorem regarding exceptional points is the following result of \cite[Thm. 4]{Motzkin1955}.
\begin{theorem} \label{MotzkinTaussky}
Let $H(\lambda, \mu) = \lambda A + \mu B$ be a matrix pencil, where $\lambda, \mu \in \mathbb{C}$ and $A, B \in \mathfrak{M}_n(\mathbb{C})$ do not commute, so $AB \neq BA$. Then there exists at least one point $(\lambda_0, \mu_0) \in \mathbb{C}^2$ such that $H$ is a defective matrix. If there exists a pair $(\lambda, \mu) \in \mathbb{C}^2$ such that $H(\lambda, \mu)$ is diagonalizable, then $(\lambda_0, \mu_0)$ is an exceptional point {\normalfont \cite[Thm. 4]{Motzkin1955}}.
\end{theorem}
A similar result to the above theorem regards the existence of branch points of the curves of eigenvalues of matrix pencils generated by real-symmetric non-commuting matrices, and is given in \cite{Moiseyev1980}.

I conclude this section with the observation that semi-simplicity of an eigenvalue is sufficient to ensure its differentiability.
\begin{thm}[Hellmann-Feynman] \label{HellmannFeynman}
Consider a matrix-valued function of one real parameter, $H:\mathbb{R} \to \mathfrak{M}_d(\mathbb{C})$, which is differentiable at $\gamma_0 \in \mathbb{R}$ and continuous in an open neighborhood, $U$, containing $\gamma_0$. Then for every semi-simple eigenvalue, $\lambda_0 \in \sigma(H(\gamma_0))$, with algebraic multiplicity $\mu_a(\lambda_0)$, there exist $\mu_a(\lambda_0)$ functions, $\lambda_i:U \to \mathbb{C}$ with $i \in \{1, \dots, \mu_a(\lambda_0)\}$ which are differentiable at $\gamma_0$ such that $\lambda_i(\gamma)$ is an eigenvalue of $H(\gamma)$ for all $\gamma \in U$. The derivatives of these functions at $\gamma_0$ are given explicitly by
\begin{align}
\left\{\frac{d \lambda_i(\gamma)}{d \gamma} \,:\, i \in \{1, \dots, \mu_a(\lambda_0) \} \right\}= \left\{\lambda_0 + \epsilon\,:\, \epsilon \in  \sigma_{\mathcal{B}(\mathcal{R}(P_{\lambda_0}))} \left(P_{\lambda_0} \frac{dH(\gamma)}{d \gamma} P_{\lambda_0}\right) \right\},
\end{align}
where $P_{\lambda_0}$ denotes the eigenprojection associated to $\lambda_0$.
\end{thm}
The special case of \cref{HellmannFeynman} pertaining to perturbations of simple eigenvalues of Hermitian matrices is typically referred to as the \textit{Hellmann-Feynmann theorem} and has been stated many times \cite{schrodinger1926quantisierung,Gttinger1932,Hellmann1933,Feynman1939}. In this case, the perturbation of the eigenvalue $\lambda \in \sigma(H)$ whose corresponding eigenvector is $\psi \in \mathcal{H}$ is given by 
\begin{align}
\frac{d \lambda}{d \gamma} = \braket{\psi| \frac{d H}{d \gamma} \psi}.
\end{align}

\subsection{Algebraic Geometry and Eigenvalues}


Insight into both eigenvalues and exceptional points arises from considering them through the lens of algebraic geometry. Algebraic geometry studies the roots of systems of polynomial equations. A ubiquitous polynomial equation is the characteristic equation whose roots are the eigenvalues of a matrix. Motivated by this example, this section displays how algebraic geometry can be applied to the eigenvalue problem associated to a polynomial matrix function\footnote{Some references refer to these functions as $\lambda$-\textit{matrices} \cite{Hodge1994}, where $\lambda$ refers to an indeterminate.}. In particular, this section reviews the Newton-Puiseux theorem and resultants. A nice book on algebraic geometry is \cite{Hodge1994}.

The central object in algebraic geometry is the \textit{algebraic variety}, defined below.
\begin{defn}
An \textit{algebraic variety} over a field $\mathbb{F}$ is the zero set of a system of polynomial equations. A special case is the \textit{algebraic curve}, which is the zero set of a bivariate polynomial, $P:\mathbb{F}\times \mathbb{F} \to \mathbb{F}$. To be explicit, an algebraic curve the set of points, $(x, y) \in \mathbb{F}^2$, satisfying $P(x,y) = 0$. An \textit{algebraic function} is a function $f: \text{Dom}(f) \to \mathbb{F}$ with $\text{Dom}(f) \subseteq \mathbb{F}$ satisfying $P(x, f(x)) = 0$.
\end{defn}

\begin{ex}
Given a matrix polynomial, $H: \mathbb{C}^d \to \mathfrak{M}_n(\mathbb{C})$, the characteristic polynomial $\det(\lambda \mathbb{1} - H)$ is a polynomial in both $\gamma \in \mathbb{C}^d$ and $\lambda \in \mathbb{C}$. 
Thus, the graph of the spectrum of $H$,
\begin{align}
\{(\gamma, \lambda) \in \mathbb{C}^{d+1} \,|\, \lambda \in \sigma(H(\gamma)) \},
\end{align}
is an algebraic variety. For the special case of a polynomial matrix function in one variable, $H:\mathbb{C} \to \mathfrak{M}_n(\mathbb{C})$, this set is an algebraic curve.
\demo
\end{ex}

Some algebraic varieties, such as the $n$-sphere, can also be manifolds. However, a general algebraic variety is \textit{not} a manifold, due to the presence of \textit{singular points}: places where the traditional notion of tangents from calculus breaks down. A simple example of an algebraic curve with a singular point is displayed in \cref{cuspExampleFigure}.
\begin{figure}[htp!]
\centering
\includegraphics[width = 80mm]{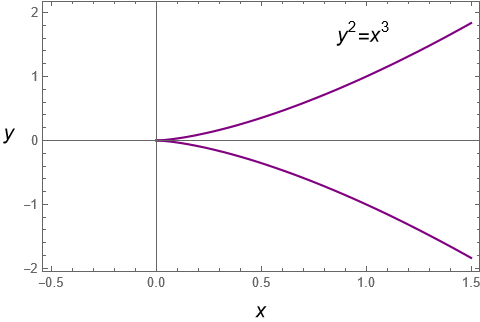}
\caption{This figure plots the \textit{semicubical parabola}, which is an algebraic curve in $\mathbb{R}^2$ defined by $y^2 = x^3$. Observe the cusp singularity at $(x,y) = (0,0)$. No Taylor series can approximate this curve at its cusp.}
\label{cuspExampleFigure}
\end{figure}

The behaviour near a singular point of an algebraic curve cannot be approximated using a Taylor series. Instead, the local behaviour of algebraic curves near singular points is given by a \textit{Puiseux series}. The formal connection between Puiseux series and algebraic curves is given by the \textit{Newton-Puiseux theorem}.
\begin{defn}
A \textit{Puiseux series} is a power series of an indeterminate raised to a rational power.
\end{defn}

\begin{thm}[Newton-Puiseux] \label{NewtonPuiseux}
Given a point in the domain of an algebraic function, there exists an open neighbourhood containing this point such that the algebraic function can be expressed as a Puiseux series for all points in this neighbourhood.
\end{thm}

\begin{defn}
Consider a polynomial matrix-valued function $H:\mathbb{C}^d \to \mathfrak{M}_n(\mathbb{C})$ with an exceptional point, $\gamma_0 \in \mathbb{C}^d$. Given a line containing $\gamma_0$, the Newton-Puiseux theorem guarantees the existence of a Puiseux series expansion of the eigenvalues in a parameter which linearly rules this line. I define the \textit{polynomial perturbative order}\footnote{Some authors, such as \cite{Hodaei2017}, define the order of an exceptional point in this way instead of in the more general notion of \cref{defn:order}.} of $\gamma_0$ to be the supremum\footnote{The supremum of a subset, $S$, of a partially ordered set, $P$, is the least element of $P$ which is greater than all elements $S$. The supremum can be viewed as a generalization of the notion of a maximum.} of the order of an eigenvalue's Puiseux series order over the set of all lines passing through the exceptional point. 
\end{defn}

The polynomial perturbative order and the order need not be the same. For instance, the $3 \times 3$ matrix given in \cref{spinOneExample} exhibits third-order exceptional points which have polynomial perturbative order equal to two.


I will now turn your attention towards a tool from algebraic geometry which I feel is underutilized by the physics community: the \textit{resultant} of two polynomials. The resultant is used to determine the intersections of algebraic curves. Furthermore, the resultant can be used to locate exceptional points. A more thorough discussion of resultants can be found in \cite{Gelfand1994}.


\begin{defn}
Consider two polynomials in an indeterminant, $\lambda$, over a field $\mathbb{F}$,
\begin{align}
f = \sum_{j = 0}^{\text{deg} f} f_j \lambda^{\text{deg}(f) - j} \\
g = \sum_{j = 0}^{\text{deg} g} g_j \lambda^{\text{deg}(g) - j},
\end{align}
where $\text{deg}$ denotes the degree of a polynomial. In particular, $f_0 \neq 0$ and $g_0 \neq 0$. Then the \textit{resultant} of $f$ and $g$ is determined by their \textit{Sylvester determinant}. Let 
\begin{align}
\tilde{f} &:= (f_j)_{j \in \{1, \dots, \text{deg} f\}}\oplus  (0)_{k \in \{1, \dots, \text{deg} g\}} \in \mathbb{C}^{\text{deg}f + \text{deg} g} \\
\tilde{g} &:= (g_j)_{j \in \{1, \dots, \text{deg} g\}}\oplus  (0)_{k \in \{1, \dots, \text{deg} f\}} \in \mathbb{C}^{\text{deg}f + \text{deg} g}
\end{align} 
denote the tuples formed by padding zeros at the end of the tuples containing the coefficients of $f$ and $g$ respectively, and $S$ denote the right shift operation on $\mathbb{C}^{\text{deg}f + \text{deg} g}$. The Sylvester determinant is
\begin{align}
\text{Res}_\lambda(f,g) := \det \begin{pmatrix}
\tilde{f}, S \tilde{f}, \dots, S^{\text{deg}(g) - 1}\tilde{f},\tilde{g}, S \tilde{g}, \dots, S^{\text{deg}(f) - 1}\tilde{g} 
\end{pmatrix}.
\end{align}
The \text{discriminant}\footnote{The term "discriminant" was introduced by James Sylvester in \cite{Sylvester1851}.}, $\text{Disc}_\lambda(f)$, is a multiple of the resultant of a polynomial and its derivative,
\begin{align}
\text{Disc}_\lambda(f) := \frac{(-1)^{\text{deg}(f) (\text{deg}(f) - 1)/2}}{f_0} \text{Res}_\lambda(f,g).
\end{align}
\end{defn}

\begin{theorem} 
The resultant of two polynomials vanishes if and only if they share a root. 
\end{theorem}
\begin{ex}
The discriminant of a quadratic is 
\begin{align}
\text{Disc}_\lambda(a \lambda^2 + b \lambda + c) = b^2 - 4 a c,
\end{align} which is a familiar expression from the quadratic formula. 
\end{ex}
\begin{ex}
If $\text{deg}(f) = 3, \text{deg}(g) = 2$, then their resultant is
\begin{align}
\text{Res}_\lambda(f,g) = \text{det} \begin{pmatrix}
f_0 & 0   & g_0 & 0   & 0   \\
f_1 & f_0 & g_1 & g_0 & 0   \\
f_2 & f_1 & g_2 & g_1 & g_0 \\
f_3 & f_2 & 0   & g_2 & g_1 \\
0   & f_3 & 0   & 0   & g_2
\end{pmatrix}.
\end{align}
\end{ex}


%

%
The following theorem regarding the number of exceptional points of a matrix is proven using the resultant.
\begin{thm}
Given a polynomial matrix function of one parameter, $H:\mathbb{C} \to \mathfrak{M}_n(\mathbb{C})$, the set of exceptional and diabolical points is of finite cardinality.
\end{thm}
\begin{proof}
Kato \cite[p. 64]{Kato1995} references a proof by Knopp \cite[p. 119-121]{Knopp2}, which involves reducing the problem to one where the characteristic polynomial, $p(\gamma, \lambda) := \det(\lambda \gls{id} - H(\gamma))$, is irreducible. I provide an alternative proof here.

Consider the sequence of polynomials
\begin{align}
R_k(\gamma) := \text{Res}_\lambda  \left(p(\gamma, \lambda) , \frac{\partial^k p(\gamma, \lambda)}{\partial \lambda^k} \right),
\end{align}
Observe $R_n(\gamma) = 1$. Thus, there exists a $k \in \{1, \dots, n\}$ which is the smallest $k$ such that $R_k(\gamma) \neq 0$. Denote this $k$ as $k_0$. The preimage $R_{k_0}^{-1}(\{0\})$ is the set of exceptional and diabolical points, which is finite since $R_{k_0}(\gamma)$ is a polynomial.
\end{proof}

\subsection{Cusps of Algebraic Curves of Exceptional Points} \label{epSurfaceSingularities}

This section displays one of the results of \cite{Barnett2023}; namely, third-order exceptional points generically occur at cusp singularities of algebraic curves consisting of second-order exceptional points. This claim is supported by a perturbative argument.

Consider a Hamiltonian which is polynomial in $d \geq 2$ complex parameters, $H: \mathbb{C}^d \rightarrow \mathfrak{M}_n(\mathbb{C})$ with a third-order exceptional point, $x_0 \in \mathbb{C}^d$, and assume there exists at least one instance of $H$ which is diagonalizable. There exists at least one eigenvalue, $\lambda_0 \in \sigma(H(x_0))$, whose algebraic multiplicity increases by 2 at $x_0$. In the vicinity of this eigenvalue, the characteristic polynomial can be written
\begin{align}
\text{det}(\lambda - H(x)) = \sum_{i=0}^n p_j(x) (\lambda - \lambda_0)^j,
\label{puiseux}
\end{align}
where the polynomials $p_j(x)$ satisfy $p_j(x_0)=0\,\forall j<3$ and $p_3(x_0)\neq 0$. As we perturb $x_0 \rightarrow x_0+\delta x$, the first-order correction $\delta\lambda=\lambda -\lambda_0$ to the eigenvalue $\lambda_0$ is found by substituting the Taylor expansions of $p_i$ in \cref{puiseux}. To simplify future calculations, assume $p_0'(x_0) \neq 0$, so the eigenvalue near $\lambda_0$ is approximated by
\begin{align}
p_3(x_0)\delta\lambda^3+ p_2(x) \delta\lambda^2+ p_1(x) \delta\lambda+\left(p_0'(x_0) \cdot \delta x \right) \approx 0.
\label{cubic}
\end{align}

Consider a perturbation along a line in parameter space passing through the exceptional point. Explicitly, let this line be $\{\theta u\,|\, \theta \in \mathbb{R}\}$ for some $u \in \mathbb{C}^d$. Given $\delta x = \theta u$, a Puiseux series expansion for a subset of eigenvalues along this line is
\begin{align}
\delta\lambda(\theta) \approx \begin{cases}
\dfrac{(p_0'(x_0) \cdot u)^{1/3}}{p_3(x_0)} \,\theta^{1/3} &\text{if } p_0'(x_0) \cdot u \neq 0 \\
\dfrac{(p_1'(x_0) \cdot u)^{1/2}}{p_3(x_0)} \,\theta^{1/2} &\text{if }  p_0'(x_0) \cdot u = 0 \text{ and } p_1'(x_0) \cdot u \neq 0 \\
\dfrac{(p_2'(x_0) \cdot u)}{p_3(x_0)} \,\theta &\text{if }  p_0'(x_0) \cdot u = p_1'(x_0) \cdot u =0 \text{ and }  p_2'(x_0) \cdot u \neq 0.
\end{cases}
\end{align}
Notably, the order of the Puiseux series expansion decreases if the line spanned by $u$ is orthogonal to the normal vector of the surface $p_0 = 0$ at $x_0$.

The set of exceptional points near $x = x_0$ is approximated by the discriminant of \cref{cubic},
\begin{align}
p_1^2 p_2^2 - 4 p_0 p_2^3 - 4 p_1^3 p_3 + 18 p_0 p_1 p_2 p_3 - 27 p_0^2 p_3^2 = 0. \label{disc}
\end{align}
The point $\delta x = 0$ is readily interpreted as a \textit{singular point} of the affine algebraic variety of exceptional points approximated by \cref{disc}, since the derivatives of \cref{disc} with respect to each coordinate $\delta x_i$ all vanish. The leading term in \cref{disc} for small $\delta x_i$ is the $p_0^2 p_3^2$ term. Assuming the leading term in the expansion of $p_1^2 p_2^2 - 4 p_0 p_2^3 - 4 p_1^3 p_3 + 18 p_0 p_1 p_2 p_3$ is an odd function of $\theta$ for a perturbation along $\delta x = \theta u$, then the line determined by $p_0(x) \approx p_0'(x_0) \delta x = 0$ is a one-directional tangent to the surface of exceptional points. Thus, we can interpret the point $x = x_0$ as a cusp singularity \cite{hilton1920plane}.

This section is concluded by discussing the relationship between higher-order exceptional points and cusp singularities of algebraic curves in two examples of $3 \times 3$ matrices. Additional examples where the aforementioned correspondence holds include the models studied in \cite{ZNOJIL2007,Ruzicka2015,Barnett2023}. 

\begin{ex}
Consider the $3 \times 3$ matrix 
\begin{align}
H = \begin{pmatrix}
\Delta + \mathfrak{i} \gamma & 1 & 0 \\
1 & 0 & 1 \\
0 & 1 & \Delta - \mathfrak{i} \gamma.
\end{pmatrix} \label{ham:3x3Example}
\end{align}
This example obeys the general rule that exceptional points of third-order are in one-to-one correspondence with cusp singularities of an algebraic curve of exceptional points of second-order. A plot of the exceptional contour is provided in \cref{fig:EP3x3Contour}. The real section of the algebraic curve of exceptional points for this model and an $n \times n$ generalization is displayed in a \cref{PTBreaking} in a subsequent chapter and \cite{Ruzicka2015,Barnett2023}. 

The characteristic polynomial is 
\begin{align}
\text{det}(\lambda \mathbb{1} - H) = \lambda^3 - 2 \Delta \lambda^2 + (\Delta^2 + \gamma^2 - 2 \gamma)\lambda + 2 \Delta.
\end{align}
A $3 \times 3$ matrix can only have third-order exceptional point when its monic characteristic polynomial is $(\lambda - \lambda_0)^3$ for some $\lambda_0$.  Comparing the zeroth and second-order terms of this cubic with the characteristic polynomial determines a constraint on the parameter $\Delta$ at third-order exceptional points,
\begin{align}
3 \lambda_0 &= 2 \Delta \label{discr3x31}\\
\lambda_0^3 &= -2 \Delta \,\Rightarrow \\
2 \Delta &\in \{-3 \sqrt{3} i, 0, 3 \sqrt{3} i\}.
\end{align}
Comparing the first term of $(\lambda - 2 \Delta/3)^3$ with the characteristic polynomial unveils the set of third-order exceptional points,
\begin{align}
(\Delta, \gamma) \text{ is an EP3 } \Leftrightarrow 2(\Delta,\gamma) \in 
\left\{
\begin{array}{l}
(0,-2\sqrt{2}), (0,2\sqrt{2}),\\
(-3 \sqrt{3} \mathfrak{i}, -\mathfrak{i}), (-3 \sqrt{3} \mathfrak{i}, \mathfrak{i}),\\
(3 \sqrt{3} \mathfrak{i}, -\mathfrak{i}), (3 \sqrt{3} \mathfrak{i}, \mathfrak{i}) 
\end{array}
\right\}.
\end{align}

Exceptional points with any order occur points where $\lambda$ is a simultaneous root of the characteristic polynomial and its derivative. These points can be determined by finding the zeros of the discriminant of the characteristic polynomial,
\begin{align}
\text{Disc}_{\lambda}(\text{det}(\lambda \mathbb{1} - H)) = 32 - 48 \gamma^2 + 24 \gamma^4 - 4 \gamma^6 + 4 \Delta^2 - 40 \gamma^2 \Delta^2 - 8 \gamma^4 \Delta^2 - 4 \gamma^2 \Delta^4. \label{disc3x3detun}
\end{align}
Singular points of the algebraic curve of exceptional points occur when $(\Delta, \gamma)$ is a simultaneous root of \cref{disc3x3detun} and its derivatives. A brief calculation reveals that the only singular points of the algebraic curve of exceptional points occur precisely at the third-order exceptional points. 

Furthermore, the third-order exceptional points are cusp points. Following the general analysis presented earlier, nearby exceptional points are approximated with the one-directional tangents determined by setting $p_0 = 0$, with the direction of this tangent determined by the sign of the remaining terms in \cref{disc}. These tangents are parametrized by $\theta \geq 0$, and the set of tangents corresponding to each third-order exceptional point is
\begin{align}
(\Delta_{EP}, \gamma_{EP}) \approx \begin{array}{l}
\{(0, \sqrt{2} - \theta)\,|\, 1 \gg \theta \geq 0\}\\
\{(0, \sqrt{2} + \theta)\,|\, 1 \gg \theta \geq 0\}\\
\left\{\frac{\mathfrak{i}}{2}
\left(3 \sqrt{3} + \theta, 1 - \theta\right)\,|\, 1 \gg \theta \geq 0 \right\}\\
\left\{\frac{\mathfrak{i}}{2}
\left(3 \sqrt{3} + \theta, -1 + \theta\right)\,|\,1 \gg \theta \geq 0 \right\}\\
\left\{\frac{\mathfrak{i}}{2}
\left(-3 \sqrt{3} - \theta, 1 - \theta\right)\,|\,1 \gg \theta \geq 0 \right\}\\
\left\{\frac{\mathfrak{i}}{2}
\left(-3 \sqrt{3} - \theta, -1 + \theta\right)\,|\,1 \gg \theta \geq 0 \right\}.
\end{array} \label{eq:3x3EP3}
\end{align} 


This example concludes with \cref{fig:EP3x3Contour}, which plots the contour of exceptional points.
\begin{figure}[!ht]
\centering
\includegraphics[width = \textwidth]{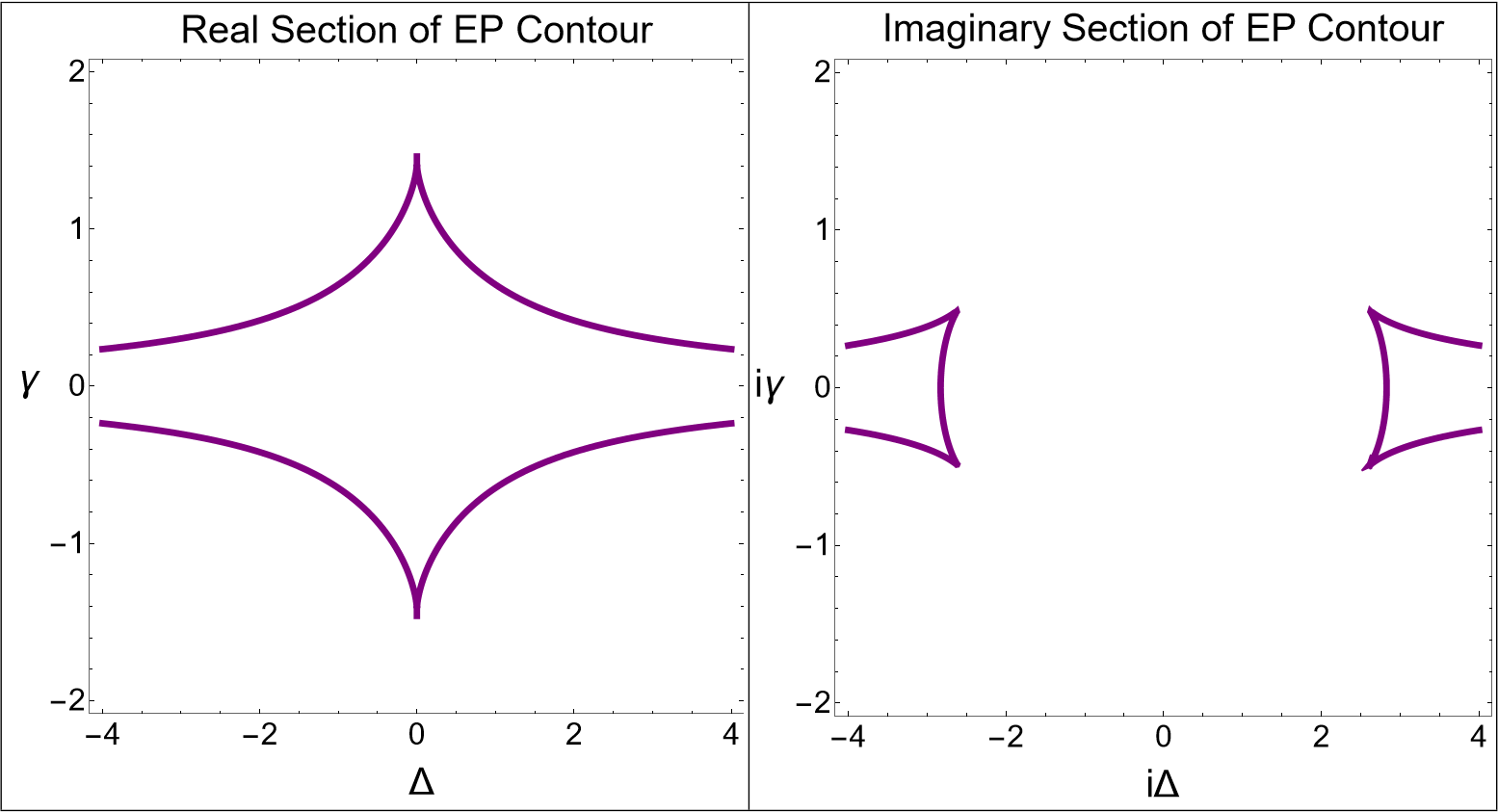}
\caption{Exceptional points of the $3 \times 3$ matrix defined in \cref{ham:3x3Example}. An analytic expression for this algebraic curve is determined by taking a discriminant, as in \cref{disc3x3detun}. The cusp points correspond to third-order exceptional points. These cusp points and their corresponding one-directional tangents are given in \cref{eq:3x3EP3}.}
\label{fig:EP3x3Contour}
\end{figure}

\demo 
\end{ex}

\begin{ex} \label{spinOneExample}
Consider the $3 \times 3$ matrix
\begin{align}
H = \begin{pmatrix}
0 & \gamma_2 & 0 \\
\gamma_1 & 0 & \gamma_2 \\
0 & \gamma_1 & 0
\end{pmatrix},
\end{align}
where at least one of the complex parameters is nonzero, $\gamma_1, \gamma_2 \neq 0$, to make things interesting.
This example violates the general rule that exceptional points of third-order are in one-to-one correspondence with cusp singularities of an algebraic curve of exceptional points of second-order. My understanding of why the general rule is violated is because of the special symmetry properties of this matrix; namely, it is an element of the dimension 3 representation of $\mathfrak{sl}(2,\mathbb{C})$.

The characteristic polynomial is
\begin{align}
\text{det}(\lambda \mathbb{1} - H) = \lambda^3 - 4 \gamma_1 \gamma_2 \lambda.
\end{align}
Thus, the curve of exceptional points is defined by $\gamma_1 \gamma_2 = 0$, and every exceptional point is a third-order exceptional point. The reason why the perturbative argument presented earlier in this section fails is because the constant term in the characteristic polynomial, what was called $p_0$ earlier, vanishes.

While these exceptional points are third-order, their polynomial perturbative order is two. This hints that a more precise quantity characterizing singular points of algebraic varieties of exceptional points is their perturbative order.
\demo
\end{ex}



\subsection{Eigenvalue Inclusion I: Bauer-Fike Theorem and Related Results}

This section will start its course on a different path from the previous subsections while remaining under the umbrella of perturbation theory. The next class of results can be categorized as \textit{eigenvalue inclusion theorems}. Eigenvalue inclusion theorems are statements regarding sets which contain the spectrum of an operator. An application of eigenvalue inclusion theorems is determining whether an antilinear symmetry is broken or unbroken \cite{Caliceti2004}. 

Perhaps the simplest eigenvalue inclusion theorem is the intermediate value theorem. If the characteristic polynomial of a matrix has different signs at two inputs, then there exists a real eigenvalue of this matrix in the open interval whose endpoints are these inputs. Following \cite{Willms2008}, one of the results presented in a later chapter of this thesis, namely \cref{inclusionTheorem}, is derived using the intermediate value theorem.

In this section, I review a well-known eigenvalue inclusion theorem, the \textit{Bauer-Fike theorem} of \cite{Bauer1960}, which relates the discrepancy between perturbed and unperturbed eigenvalues to the operator norm of the perturbation. In addition, I emphasize how the result given in \cite[Thm. 1.2]{Caliceti2004}, that pertains to domains of unbroken antilinear symmetry, can be strengthened by using the Bauer-Fike theorem. 

As observed in \cite{Michailidou2018}, the Bauer-Fike theorem is a corollary of a result proven by Alston Housenholder in \cite[p. 322]{Householder1956}, \cite[p. 66]{householder2013theory}. Housenholder's result is the $p =0$ or $p = 1$ case of \cref{generalHousenholderTheorem}. In fact, Bauer and Fike demonstrated the finite-dimensional version of the $p = 0$ version of \cref{generalHousenholderTheorem} \cite{Bauer1960}.
\begin{theorem} \label{generalHousenholderTheorem}
Given a Banach space, $(X,||\cdot||)$, consider an operator, $A: \text{Dom}(A) \to X$ with $\text{Dom}(A) \subseteq X$, and a bounded perturbation, $\delta A \in \mathcal{B}(X)$. Let $p \in [0,1]$. Then the spectrum of $A + \delta A$ is contained in the Housenholder sets of $A$, defined by 
\begin{align}
\sigma(A+\delta A) \subseteq \sigma(A) \cup \{z \in \mathbb{C}\setminus \sigma(A) \, : \, ||(z - A)^{p-1}  \delta A (z - A)^{-p} || \geq 1\}. \label{generalHousenholder}
\end{align}
\end{theorem}
\begin{proof}
Consider $\mu \in \sigma(A + \delta A) \setminus \sigma(A)$. Since $\mu \notin \sigma(A)$, the inverse $(\mu \mathbb{1} - A)^{-1} \in \mathcal{B}(X)$ exists and is bounded. Thus, the map 
\begin{align}
(\mu \mathbb{1} - A)^{p - 1} (\mu \mathbb{1} - (A + \delta A)) (\mu \mathbb{1} - A)^{-p} = \mathbb{1} - (\mu \mathbb{1} - A)^{p - 1} \delta A(\mu \mathbb{1} - A)^{-p}
\end{align}
does not have a bounded inverse. The operator $\mathbb{1} + B$ is invertible for all $B$ such that $||B|| < 1$, thus, we have demonstrated 
\begin{align}
\mu \in \sigma(A + \delta A) \setminus \sigma(A) \, \Rightarrow \, \mu \in \{z \in \mathbb{C}\setminus \sigma(A) \, : \, ||(z - A)^{p-1}  \delta A (z - A)^{-p} || \geq 1\},
\end{align}
which implies \cref{generalHousenholder}.
\end{proof}

The following definitions will be used in the Bauer-Fike theorem.
\begin{defn} \label{bauerFikePrereqs}
Let $(X, ||\cdot||_X)$ be a normed space, and let $x_i$ denote a linearly independent basis of $X$. The norm is called \textit{axis-oriented} in the basis $x_i$ if and only if the operator norm of every diagonal linear operator, $D$, is equal to its spectral radius, $||D|| = \text{sup}_{\lambda \in \sigma(D)}|\lambda|$. 
\end{defn}
\begin{ex}
The $p$-norms, defined in \cref{pnorm}, are axis-oriented in the canonical basis, $e_i$. 
\demo
\end{ex}
\begin{ex} \label{ex:indHilbSpaceAxisOriented}
An induced Hilbert space norm is axis-oriented in any orthonormal basis.
\demo
\end{ex}
\begin{defn}
Given a normed space, $X$, and a bounded operator with bounded inverse, $A \in \text{GL}(\mathcal{B}(X))$, the \textit{condition number}, denoted by $\kappa_{||\cdot||}:\text{GL}(\mathcal{B}(X)) \to [1, \infty)$, of $A$ is
\begin{align}
\kappa_{||\cdot||}(A) = ||A|| \, ||A^{-1}||.
\end{align}
\end{defn}
The product inequality, \cref{productInequality}, implies the condition number is always greater than or equal to one. Due to the $C^*$-identity, the condition number associated to the induced norm on a Hilbert space equals one for every unitary operator.

\begin{theorem}[Bauer-Fike \cite{Bauer1960}]\label{BauerFike}
Let $(X, ||\cdot||_X)$ be a Banach space\footnote{Technically, Bauer and Fike only stated the finite-dimensional version of this theorem. Proof of the infinite-dimensional case stated here is a straightforward corollary of \cref{generalHousenholderTheorem}.}. Let $H = H_0 + H_1$, where $H_0$ is a linear map on $X$ and $H_1 \in \mathcal{B}(X)$ is a bounded linear map. Suppose $H_0$ is similar to an linear map which is diagonal in the basis $x_i$, so $H_0 = S \Lambda S^{-1}$ and $\Lambda x_i = \lambda_i x_i$ for all $i$. Assume the norm is axis-oriented in the basis $x_i$. Then the spectrum is contained in a union of circles with radius $\kappa_{||\cdot||_X}(S) \, ||H_1||$ centred at the unperturbed eigenvalues, $\lambda \in \sigma(H_0)$,
\begin{align}
\sigma(H) &\subseteq \bigcup\limits_{\lambda \in \sigma(H_0)} \{\epsilon \in \mathbb{C} \, | \, |\epsilon - \lambda| < \kappa_{||\cdot||_X}(S) \, ||H_1|| \} 
\end{align}

\end{theorem}
The Bauer-Fike theorem states that the sensitivity of the eigenvalues of $H_0$ under perturbations is related to the condition number of a matrix of eigenvectors. Note the condition number is dependent on a choice of norm and on a choice similarity transform, $S$. As an extremal example, if $H_0 = \mathbb{1}$, and we choose a poorly conditioned similarity transform with $(\epsilon, \epsilon^{-1}) \in \sigma(S)$ for some $\epsilon \in \mathbb{C}$ tending to zero, by the spectral radius formula, the radius of the open balls in the Bauer-Fike theorem tends to $\infty$.

A corollary of the Bauer-Fike theorem regards the stability of real eigenvalues under perturbations. 
\begin{corollary} \label{bfCorollary}
Let $H = H_0 + H_1$, where $H_0$ and $H_1$ are linear maps on a finite-dimensional normed space, $(X, ||\cdot||_X)$. Suppose all of the eigenvalues of $H_0$ are simple and real-valued, and let $S$ denote the similarity transform mapping $H_0$ to a diagonal matrix with entries $\lambda_i \in \mathbb{R}$. Assume the norm is axis-oriented in the basis of eigenvectors of $H_0$. Assume $H_0$ and $H_1$ are simultaneously pseudo-Hermitian with the same intertwiner or that they simultaneously satisfy any of the equivalent criteria of \cref{pseudoEquivThm}. Then 
\begin{align}
\sigma(H) \subseteq \mathbb{R} \, \Leftarrow \, ||H_1|| < \frac{\inf\limits_{j \neq k} |\lambda_{j} - \lambda_k|}{2 \kappa_{||\cdot||_X}(S)}. \label{bfCorollaryIneq}
\end{align}
\end{corollary}
\begin{proof}
Suppose we find a set in the complex plane which is symmetric about reflections through the real axis and which contains exactly one eigenvalue of a pseudo-Hermitian operator. Then, that eigenvalue must be real-valued, since otherwise its inequivalent complex conjugate would also be contained in the set. Thus, the corollary will be proven if $\dim X$ such sets are constructed. 

If the inequality of \cref{bfCorollaryIneq} holds, then the open balls that are defined by the Bauer-Fike theorem,
\begin{align}
\{\epsilon \in \mathbb{C} \, | \, |\epsilon - \lambda_i| < \kappa_{||\cdot||_X}(S) \, ||H_1|| \}, \label{bfBalls}
\end{align}
are disjoint for every $i \neq j$. By continuity of the eigenvalues of $H_0 + t H_1$ for $t \in [0,1]$, each of the open balls of \cref{bfBalls} contains at least one eigenvalue. Each such eigenvalue must be real-valued by the reasoning of the previous paragraph, so the corollary is proven.
\end{proof}
The above corollary was demonstrated in the finite-dimensional case. I imagine a generalization to the infinite-dimensional case holds if $H_0$ is assumed to be diagonalizable and $H_1$ is assumed to be bounded. The criteria of \cref{pseudoEquivThm} are no longer equivalent \cite{siegl2008quasi,Siegl2009}, so one may need to take care in adapting this clause to the infinite-dimensional setting.

A theorem akin to the above corollary, which is weaker in the finite-dimensional case yet also applies to some operators in infinite-dimensional spaces, was given in \cite{Caliceti2004}.
\begin{theorem} \label{CalicetiTheorem}
Let $H = H_0 + \epsilon H_1$, where $H_0$ and $H_1$ are linear operators on a Hilbert space, and $\epsilon \in \mathbb{R}$. Suppose $H_0$ is bounded below, self-adjoint, has a discrete spectrum with 
\begin{align}
\sigma(H_0) = \{\lambda_0 < \lambda_1 < \dots \},
\end{align} 
and all the eigenvalues are simple. Suppose $H_1$ is continuous. Suppose there exists a unitary involution, $J$, such that $H_0$ and $H_1$ are pseudo-Hermitian with the intertwiner $J$. Then {\normalfont \cite{Caliceti2004}},
\begin{align}
\sigma(H) \subseteq \mathbb{R} \, \Leftarrow \, |\epsilon| < \frac{\inf\limits_{j \geq 0} (\lambda_{j+1} - \lambda_j)}{2 ||H_1|| }.
\end{align}
\end{theorem}


\subsection{Eigenvalue Inclusion II: Cassini Ovals and Geršgorin Disks}

This section summarizes two additional eigenvalue inclusion results, the \textit{Brauer-Ostrowski theorem} of \cite{Ostrowski1937,Brauer1947}, and \textit{Geršgorin's disk theorem} of \cite{gershgorin1931uber}. A more exhaustive review of this topic can be found in the textbook of \cite{Varga2004}.

\begin{defn} \label{defn:Cassini}
Given two points, $z_1, z_2 \in \mathbb{F}$, where $\mathbb{F} \in \{\mathbb{R}^2, \mathbb{C}\}$, I define a \textit{Cassini oval} to be the interior of a peculiar kind of quartic curve:
\begin{align}
C(z_1,z_2; b) = \{z \in \mathbb{C} \,|\, ||z - z_1||_2 \cdot ||z-z_2||_2 \leq b\} .
\end{align} 
I refer to the boundary of a Cassini oval as a \textit{Cassini curve}.
The points $z_1, z_2$ are referred to as the \textit{foci} of $C$. Geometrically, a Cassini curve is the locus of points such that geometric mean of the distances between these points and the foci is a constant.
\end{defn}

The set of ellipses can be defined by replacing the term "geometric mean" with "arithmetic mean" in \cref{defn:Cassini}. Similarly to how ellipses are conic sections, Cassini curves are \textit{Toric sections}. A Cassini curve with overlapping foci is a circle, and the case $||z_1 - z_2|| = 2b$ is the lemniscate of Bernoulli. Figure~\eqref{CassiniCurves} displays several Cassini curves.
\begin{figure}[!ht]
\centering
\includegraphics[width = .75 \textwidth]{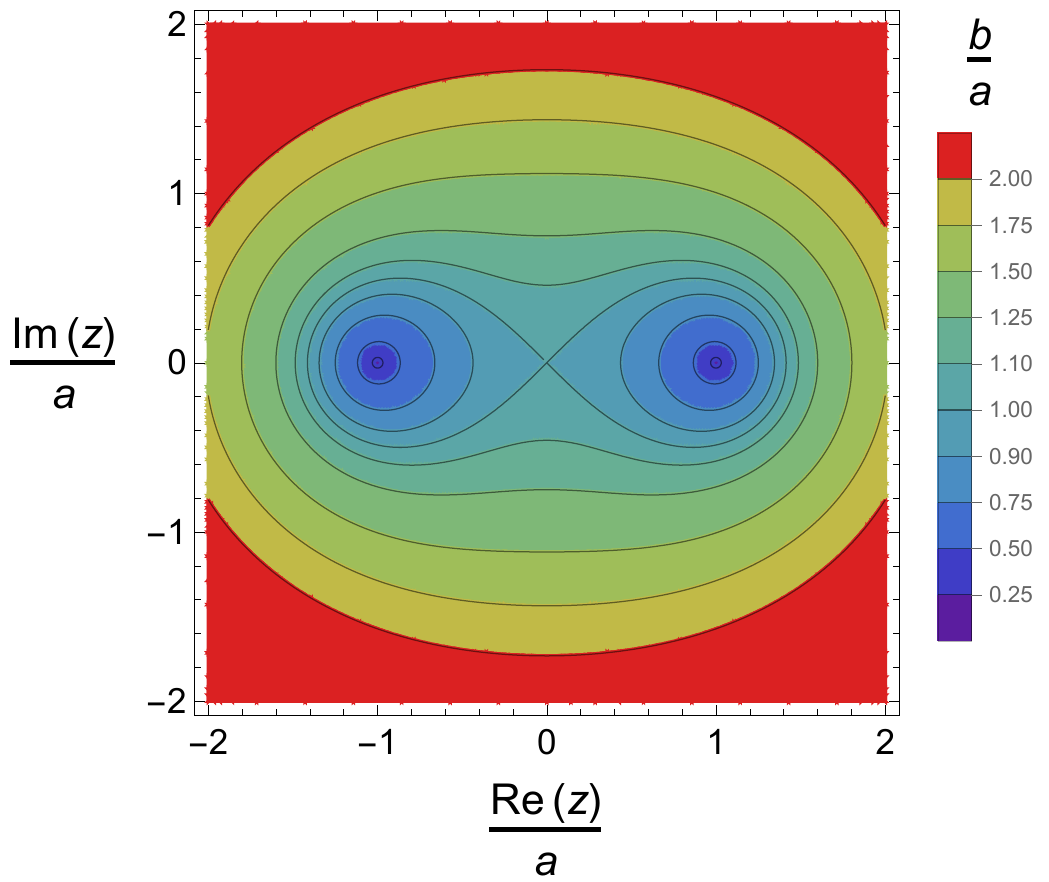}
\caption{Cassini curves with real foci, $z\in C(a, -a, b)$ with $a \in \mathbb{R}$.}
\label{CassiniCurves}
\end{figure}

\begin{theorem}[Brauer-Ostrowski \cite{Ostrowski1937,Brauer1947}] \label{BrauerOstrowski}
Given a matrix, $A \in \mathfrak{M}_n(\mathbb{C})$, its spectrum is contained in the set of its \textit{Brauer-Cassini ovals}, whose foci are the diagonal elements of $A$, 
\begin{align}
\sigma(A) \subset \cup_{i \neq j \in \{1, \dots, n\}} C(A_{ii}, A_{jj}; R_i(A) R_j(A)),
\end{align}
where 
\begin{align}
R_i(A) := \sum_{j \in \{1, \dots, n\} \setminus \{i\}} |a_{ij}|.
\end{align}
\end{theorem}
\begin{corollary} [Geršgorin disk \cite{gershgorin1931uber}] \label{Gerschgorin}
Given a matrix $A \in \mathfrak{M}_n(\mathbb{C})$, its spectrum is contained in the union of its Geršgorin disks, 
\begin{align}
\sigma(A) \subset \cup_{i \in \{1, \dots, n\}} C(A_{ii},A_{ii}; R_i(A)^2).
\end{align}.
\end{corollary}
The set of Geršgorin disks always contains the set of Brauer-Cassini ovals as a subset, so the Brauer-Ostrowski theorem is a stronger result than the Geršgorin disk theorem \cite[Thm. 2.3]{Varga2004}. Furthermore, the Geršgorin disk theorem is also a corollary of the Housenholder's inclusion result \cite{Michailidou2018,Varga2004} displayed in \cref{generalHousenholderTheorem}.
 


The continuity of eigenvalues, as presented in \cref{continuityEvals}, can be quite useful combined with \cref{BrauerOstrowski}. In particular, given $A \in \mathfrak{M}_n(\mathbb{C})$, let $\text{diag}(A)$ denote the matrix which is diagonal in the canonical basis whose diagonal entries are the same as the diagonal entries of $A$, and consider the pencil $B(t) = \text{diag}(A) + t (A - \text{diag}(A))$ with $t \in [0,1]$. In particular, $B(1) = A$, and $B$ is continuous. Suppose the union of the Brauer-Cassini ovals associated to $A$ has multiple connected components. Then, the sum of the algebraic multiplicities of the eigenvalues of $B(t)$ which are contained in each connected component is a constant independent of $t$.

\section{Non-Hermitian Qubit} \label{PTqubit} 

Sections~\ref{PTSymmetryIntroSection},\ref{PseudoHermSection},\ref{quasiHermSection}, and \ref{ExceptionalPointsIntro} introduced a variety of results which are applicable to generic classes of possibly non-Hermitian operators. This section applies these results to a simple model of a non-Hermitian qubit. 


\subsection{Algebra of $2 \times 2$ Matrices}
%

This section summarizes algebraic structures associated to $2 \times 2$ matrices, which are the mathematical tools used to investigate qubits. 
I would recommend the reader treat this section as a reference, and continue to the following subsection.

\begin{defn}
The \textit{algebra of} $2 \times 2$ \textit{matrices}, $\mathfrak{M}_2(\mathbb{F})$, over a field $\mathbb{F} \in \{\mathbb{R}, \mathbb{C}\}$ is 
\begin{align}
\mathfrak{M}_2(\mathbb{F}) = \left\{\begin{pmatrix}
a & b \\
c & d
\end{pmatrix} \,|\, a,b,c,d \in \mathbb{F} \right\}.
\end{align}
\end{defn}

Central to the analysis of qubits are the $2 \times 2$ \textit{Pauli matrices}, 
\begin{align}
\sigma_0 = \mathbb{1} &\quad&
\sigma_1 = \sigma_x = \begin{pmatrix}
0 & 1 \\
1 & 0
\end{pmatrix} &\quad& \sigma_2 = \sigma_y = \begin{pmatrix}
0 & -\mathfrak{i} \\
\mathfrak{i} & 0
\end{pmatrix} &\quad& \sigma_3 = \sigma_z = \begin{pmatrix}
1 & 0 \\
0 & -1
\end{pmatrix}.
\end{align}
The Pauli matrices are Hermitian, involutions, and unitary. Furthermore, they form a basis for $\mathfrak{M}_2(\mathbb{C})$. 
The \textit{Pauli vector}, $\vec{\sigma}$, is the 3-vector of Pauli matrices, 
\begin{align}
\vec{\sigma} = \begin{pmatrix}
\sigma_1 \\
\sigma_2 \\
\sigma_3
\end{pmatrix}.
\end{align}

The following points summarize some useful facts about the algebra of $2 \times 2$ matrices:
\begin{itemize}
\item 
$\mathfrak{M}_2(\mathbb{F})$ is a \textit{unital associative algebra} over $\mathbb{F}$. This means it is a vector space over $\mathbb{F}$; it has a bilinear product-the matrix product, which distributes over addition; and it has a two-sided identity-the identity matrix. \\
\item
$\mathfrak{M}_2(\mathbb{C})$ is a \textit{Hilbert space}\footnote{In addition, $\mathfrak{M}_2(\mathbb{C})$ is also a \textit{left Hilbert algebra} and a $C^*$\textit{-algebra}.} with the \textit{Frobenius inner product},
\begin{align}
\braket{A|B}_{HS} = \text{Tr} (A^\dag B) &\quad& A,B \in \mathfrak{M}_2(\mathbb{C}).
\end{align}
The Pauli matrices, normalized by multiplication by $1/\sqrt{2}$, are an orthonormal basis for $\mathfrak{M}_2(\mathbb{C})$. \\
\item $\mathfrak{M}_2(\mathbb{F})$ is a \textit{Clifford algebra} for $\mathbb{F}^3$ \cite{Garling2011}, which means there exists a Clifford mapping, $\iota:z \to z \cdot \vec{\sigma}$, satisfying the following properties:
\begin{enumerate}
\item $\mathbb{1} \notin \iota(\mathbb{F}^3)$ \\
\item $\iota(z) \iota(w) + \iota(w) \iota(z) = \braket{z^*|w} \mathbb{1}$ \\
\item $\mathfrak{M}_2(\mathbb{F})$ is the algebra generated by\footnote{See \cref{generated-Algebra-defn}.} $\iota(\mathbb{F}^3)$.
\end{enumerate}
\end{itemize}

\begin{defn}
The \textit{complex special linear Lie algebra}, $\mathfrak{sl}(2, \mathbb{C})$, is the subset of elements in $\mathfrak{M}_2(\mathbb{C})$ which have zero trace. The real Lie subalgebra of $\mathfrak{sl}(2, \mathbb{C})$ consisting of $2 \times 2$ skew-Hermitian matrices is referred to as $\mathfrak{su}(2)$.
\end{defn}
The imaginary unit times the Pauli matrices form a set of generators for $\mathfrak{su}(2)$. Their Lie bracket, which in this case is the commutator, is
\begin{align}
[\frac{\mathfrak{i} \sigma_i}{2}, \frac{\mathfrak{i} \sigma_j}{2}]_- = -\sum_{k = 1}^3 \epsilon_{ijk} \frac{\mathfrak{i} \sigma_k}{2},
\end{align}
where $\epsilon_{ijk}$ is the \textit{completely antisymmetric tensor},
\begin{align}
\epsilon_{ijk} &= \begin{cases}
1 & \text{if }  (i,j,k) \in \{(1,2,3), (3,1,2), (2,3,1) \} \\
-1 & \text{if } (i,j,k) \in \{(1,3,2), (2,1,3), (3,2,1) \} \\
0 & \text{otherwise}
\end{cases}. \label{antisymmetric-tensor}
\end{align}

Since the Pauli matrices are generators for both a Lie algebra and a Clifford algebra, their product satisfies the often-applied identity
\begin{align}
(\vec{a} \cdot \vec{\sigma}) (\vec{b} \cdot \vec{\sigma}) = (\vec{a} \cdot \vec{b}) \mathbb{1} + \mathfrak{i} (\vec{a} \times \vec{b}) \cdot \vec{\sigma}. \label{sl2Cproduct}
\end{align} 


\subsection{Pseudo-Hermiticity of $2 \times 2$ Matrices}


The following theorem characterizes the set of $2 \times 2$ pseudo-Hermitian matrices. The characterization of $2 \times 2$ quasi-Hermitian matrices given by \cref{quasiHerm2x2} of this theorem was given in \cite{mostafazadeh2006explicit}, although the form of the metric operator presented here is simpler. An alternative discussion of the most general way in which antilinear symmetry can be imposed on a $2 \times 2$ matrix is given in \cite{wang20132}.
\begin{theorem} \label{2x2PseudoHerm-Theorem}
\leavevmode 
\begin{enumerate}
\item
A $2 \times 2$ matrix, $A \in \mathfrak{M}_2(\mathbb{C})$, is pseudo-Hermitian if and only if it can be expressed in the form
\begin{align}
A = \frac{{\normalfont\text{tr}}(A)}{2} \mathbb{1} + (\vec{\alpha} + \mathfrak{i} \vec{\beta}) \cdot \vec{\sigma}, \label{2x2PseudoHermDecomp}
\end{align}
where the real 3-vectors $\vec{\alpha}, \vec{\beta} \in \mathbb{R}^3$ are orthogonal, so $\vec{\alpha} \cdot \vec{\beta} = 0$, and ${\normalfont \text{tr}}(A) \in \mathbb{R}$. \\
\item The most general intertwining operator associated to a $2 \times 2$ pseudo-Hermitian matrix with the decomposition \cref{2x2PseudoHermDecomp} and either $\vec{\alpha} \neq 0$ or $\vec{\beta} \neq 0$ is\footnote{As a sanity check, note $\mathbb{1}$ is only an intertwining operator in the Hermitian case $\vec{\beta} = 0$.}
\begin{align}
\eta(\zeta, \xi) = \zeta \vec{\alpha} \cdot \vec{\sigma} + \xi (\vec{\alpha} \cdot \vec{\alpha} )\mathbb{1} + \xi (\vec{\beta} \times \vec{\alpha}) \cdot \vec{\sigma} \label{2x2Intertwiner},
\end{align}
where $\zeta, \xi \in \mathbb{R}$.\\
\item \label{quasiHerm2x2}
The intertwiner $\eta(\zeta, \xi)$ is positive-definite if and only if 
\begin{align}
(\vec{\alpha} \cdot \vec{\alpha} - \vec{\beta} \cdot \vec{\beta}) \xi^2 > \zeta^2 (\vec{\alpha} \cdot \vec{\alpha}),
\end{align}
so a matrix of the form \cref{2x2PseudoHermDecomp} is quasi-Hermitian if and only if $\vec{\alpha}\cdot \vec{\alpha} > \vec{\beta} \cdot \vec{\beta}$ or $\vec{\beta} = 0$.\\
\item If $\vec{\alpha} \cdot \vec{\alpha} = \vec{\beta} \cdot \vec{\beta} > 0$, then $A$ defined in \cref{2x2PseudoHermDecomp} has a single eigenvalue, ${\normalfont \text{tr}}(A)/2$, which has geometric multiplicity equal to one.
\end{enumerate}
\end{theorem}
\begin{proof}
\leavevmode 
\begin{enumerate}
\item 
First, I show that every $2 \times 2$ pseudo-Hermitian matrix, $A \in \mathfrak{M}_2(\mathbb{C})$, admits a decomposition of the form \cref{2x2PseudoHermDecomp}. Since the set of Pauli matrices is an orthogonal basis of $\mathfrak{M}_2(\mathbb{C})$, the matrix $A$ admits a decomposition of the form
\begin{align}
A = \sum_{\nu = 0}^3 a_\nu \sigma_\nu,
\end{align}
where $a_\nu \in \mathbb{C}$. Lemma~\ref{pseudoHermTraceDeterminant} shows that the coefficients of the characteristic polynomial are real-valued. Defining $\vec{\alpha}_j := \text{Re}(a_j)$ and $\vec{\beta}_j := \text{Im}(a_j)$ for $j \in \{1, 2, 3\}$, the characteristic polynomial is
\begin{align}
\text{det}\left(\lambda \mathbb{1} - A\right) = \left(\lambda - \frac{\text{tr}(A)}{2}\right)^2 + \vec{\beta}\cdot\vec{\beta} - \vec{\alpha} \cdot \vec{\alpha} - 2 \mathfrak{i} \vec{\alpha} \cdot \vec{\beta},
\end{align}
so assuming pseudo-Hermiticity forces $\vec{\alpha} \cdot \vec{\beta} = 0$ and $\text{tr}(A) \in \mathbb{R}$. The reverse direction of the first claim follows by proving $\eta(\zeta,\xi)$ is an intertwining operator. 

\item The intertwining relation for $\xi = 0$ follows from basic facts about products in $\mathfrak{sl}(2,\mathbb{C})$, namely the identity of \cref{sl2Cproduct}.
Verifying the intertwining relation for $\xi \neq 0$ can be done either by applying the generative procedure of \cref{generative} to the $\xi = 0$ case, or by using the vector triple product,
\begin{align}
\vec{a} \times (\vec{b} \times \vec{c}\,) = (\vec{c} \cdot \vec{a}) \vec{b}  - (\vec{a} \cdot \vec{b}\,) \vec{c}.
\end{align}

\item An alternative proof that the matrix $A$ of \cref{2x2PseudoHermDecomp} is quasi-Hermitian when $\vec{\alpha}\cdot \vec{\alpha} > \vec{\beta} \cdot \vec{\beta}$ or $\vec{\beta} = 0$ can be seen from its eigenvalues,
\begin{align}
\sigma(A) = \left\{\frac{\text{tr}(A)}{2} \pm \sqrt{\vec{\alpha} \cdot \vec{\alpha} - \vec{\beta} \cdot \vec{\beta}}\right\}. \label{qubitSpectrum}
\end{align}

\item That $\text{tr}(A)/2$ is the unique eigenvalue of $A$ when $\vec{\alpha}\cdot\vec{\alpha} = \vec{\beta} \cdot \vec{\beta}$ follows from \cref{qubitSpectrum}. Since $\text{tr}(A)/2 - A$ is a nonzero matrix when $\vec{\alpha} \neq 0$, the rank-nullity theorem implies the eigenspace corresponding to the unique eigenvalue has dimension of at most one.
\end{enumerate}
\end{proof}


A pseudo-Hermitian $2 \times 2$ matrix with a choice of $\vec{\alpha}$ and $\vec{\beta}$ is unitarily similiar to one defined by a rotation acting on both $\vec{\alpha}$ and $\vec{\beta}$. In particular, this unitary similarity transform can be computed using the Rodrigues rotation formula for rotations about unit vectors. Assuming $\vec{\gamma}$ is a unit vector, $\vec{\gamma} \cdot \vec{\gamma} = 1$, then 
\begin{align}
R_{\vec{\gamma}}(\theta) :&= e^{-\mathfrak{i} \theta \vec{\gamma} \cdot \vec{\sigma}/2} \\
R_{\vec{\gamma}}(-\theta)
(\vec{\alpha} \cdot \vec{\sigma}) R_{\vec{\gamma}}(\theta) &= \cos(\theta)(\vec{\alpha} \cdot \vec{\sigma}) + \sin(\theta) (\vec{\alpha} \times \vec{\gamma})\cdot \vec{\sigma} + (\vec{\alpha} \cdot \vec{\gamma}) (\vec{\gamma} \cdot \vec{\sigma}) (1 - \cos(\theta)).
\end{align}
Thus, without loss of generality, $\alpha$ can be taken along the $x$ axis and $\beta$ can be taken along the $z$ axis. This yields a $2 \times 2$ matrix dependent on two real parameters, $\omega,t \in \mathbb{R}$,
\begin{align}
H(\omega, t) = \begin{pmatrix}
\mathfrak{i} \omega & t \\
t & -\mathfrak{i} \omega
\end{pmatrix}, \label{qubitHam}
\end{align}
which I will refer to as the qubit Hamiltonian.
To simplify subsequent calculations, I assume $t \neq 0$.
A physical system whose time evolution is modelled by a Schr{\"o}dinger equation with the non-Hermitian Hamiltonian given by \cref{qubitHam} can be implemented in coupled optical waveguides \cite{ElGanainy2007,Guo2009}. In these systems, $\omega$ represents a propagation constant corresponding to gain or loss and $t$ represents a coupling constant.

\subsection{Symmetries and Eigensystem}
The qubit Hamiltonian, $H$ of \cref{qubitHam}, is $\mathcal{PT}$-symmetric \cite{Mandilara2002}, where in this example
\begin{align}
\mathcal{P} &= \begin{pmatrix}
0 & 1 \\
1 & 0
\end{pmatrix} \\
\mathcal{T} \begin{pmatrix}
\alpha \\
\beta
\end{pmatrix} &= \begin{pmatrix}
\alpha^* \\
\beta^*
\end{pmatrix}.
\end{align}
The qubit Hamiltonain also has a \textit{Chiral symmetry}, or more explicitly, 
\begin{align}
\sigma_y H = - H \sigma_y.
\end{align}
Due to chiral and $\mathcal{PT}$-symmetry, $H$ also has an antilinear antisymmetry,
\begin{align}
\sigma_z \mathcal{T} H = - H \sigma_z \mathcal{T}.
\end{align}
Consequently, the spectrum of $H$ must be invariant under reflections about either the real or imaginary axis. Since the spectrum of $H$ contains at most two eigenvalues, the eigenvalues must  both reside in one of two axes, the real axis or the imaginary axis.

The eigenvalues of $H$ are $\epsilon_{\pm}$ \cite{Mandilara2002}, plotted in \cref{qubitEvalsFig}, where 
\begin{align}
\epsilon_{\pm}(\omega, t) = \pm \sqrt{t^2 - \omega^2},
\end{align}
which are real-valued\footnote{The reality of the spectrum for $\omega^2 < t^2$ is also a consequence of \cref{CalicetiTheorem}.} if and only if $\omega^2 \leq t^2$. $H$ is diagonalizable with real spectrum if and only if $\omega^2 < t^2$. By \cref{unbrokenRealSpectrum}, $\mathcal{PT}$-symmetry is unbroken if and only if the spectrum is real, which can be seen from the eigenspaces corresponding to $\epsilon_{\pm}$,
\begin{align}
\ker(\epsilon_{\pm} \gls{id} - H) = \text{span} \left\{\begin{pmatrix}
\mathfrak{i} \omega \pm \sqrt{t^2-\omega^2} \\
t
\end{pmatrix} \right\}. \label{qubitEigenspaces}
\end{align}
The eigenspaces are $\mathcal{PT}$-invariant for $\omega^2 \leq t^2$. As there is only a one-dimensional set of eigenvectors of the defective matrix $H(t,t)$, the points satisfying $\omega^2 = t^2$ are second-order exceptional points.

\begin{figure}[htp!]
\centering
\includegraphics[width = \textwidth]{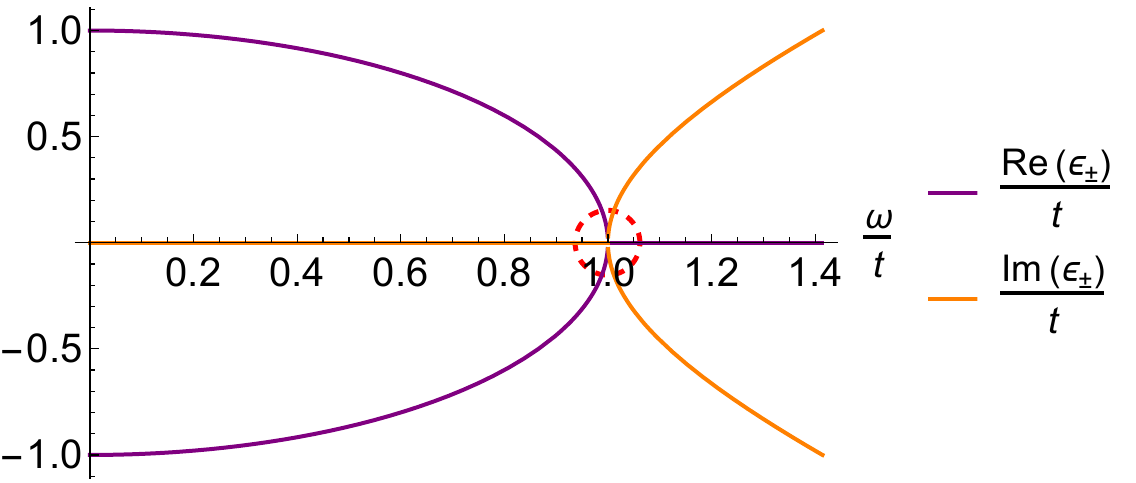}
\caption{A plot of the eigenvalues, $\epsilon_{\pm}$, of a $\mathcal{PT}$-symmetric qubit Hamiltonian, $H$, defined by \cref{qubitHam}. Encircled in red at $\omega = t$ is a second-order exceptional point. $\mathcal{PT}$-symmetry is unbroken when $\omega^2 \leq t^2$ and broken otherwise.} \label{qubitEvalsFig}
\end{figure}

By \cref{pseudoEquivThm}, $H$ must be similar to a real-valued matrix. One such similarity transform to a real matrix is given by a Clifford gate \cite{Bolt1961}, specifically the matrix representing a rotation around the $x$ axis in the Bloch sphere,
\begin{align}
R_{\hat{x}}\left(-\frac{\pi}{2}\right) = \begin{pmatrix}
1 & -\mathfrak{i}\\
1 & \mathfrak{i}
\end{pmatrix},
\end{align} 
we have
\begin{align}
R_{\hat{x}}\left(-\frac{\pi}{2}\right) H R_{\hat{x}}\left(-\frac{\pi}{2}\right)^{-1} = \begin{pmatrix}
0 & t + \gamma \\
t - \gamma & 0
\end{pmatrix},
\end{align}
which is the $2 \times 2$ case of the Hamiltonian of the Hatano-Nelson model \cite{Hatano1996}.

Every intertwining operator associated to $H$ is $\mathcal{PT}$-symmetric, which can be seen by the general formula of \cref{2x2Intertwiner}. 

\subsection{Bloch Sphere}

Our objective for this subsection is to graphically represent the eigenvectors of our qubit Hamiltonian, $H$ of \cref{qubitHam}. The resulting \cref{evecsOnBlochSphere} will display antilinear (anti)symmetry breaking, the non-orthogonality of the eigenvectors, and the exceptional point. 

An insightful depiction of the eigenspaces of $H$ follows from their \textit{Bloch sphere} representation \cite{operatorSchmidt,Schumacher2010}, a map which I will denote by $\mathcal{S}$. The Bloch sphere provides a bijective correspondence between one-dimensional vector subspaces of $\mathbb{C}^2$ and points on the unit sphere $\mathbb{S}^2 \subsetneq \mathbb{R}^3$. 
\begin{defn}[Bloch sphere]
Let $v \in \mathbb{C}^2 \setminus \{0\}$ be a nonzero vector. 
The \textit{Bloch sphere representation} of the vector space $V = \gls{span}\{v\}$ whose basis is $\{v\}$ is the point $\mathcal{S}(V) \in \mathbb{R}^3$ whose components are
\begin{align}
\mathcal{S}(V)_i := \frac{\braket{v|\sigma_i v}}{\braket{v|v}} \label{BlochInnProd}
\end{align}.
\end{defn}

\begin{lemma}
The Bloch sphere representation is a one-to-one map\footnote{A stronger statement holds, namely that $\mathcal{S}$ defines a homeomorphism from $\mathbb{S}^2$ to the complex projective line $\mathbb{CP}^1$.} whose image is the sphere $\mathbb{S}^2$.
\end{lemma}
\begin{proof}
The Bloch sphere representation of a vector space, $V$, can equivalently be computed via Frobenius inner products of the orthogonal projection onto $V$, $\mathscr{P}_V$, and Pauli matrices,
\begin{align}
\mathcal{S}(V)_i = \gls{tr}\left(\sigma_i \mathscr{P}_V\right). \label{sphere1}
\end{align}
The Pauli matrices generate an orthonormal basis of the algebra of $2 \times 2$ matrices, $\mathfrak{M}_2(\mathbb{C})$, in the Frobenius inner product. More explicitly,
\begin{align}
\frac{1}{2} \sum_{\nu = 0}^3 \sigma_\nu \gls{tr}(\sigma_\nu A) = A &\quad& \forall A \in \mathfrak{M}_2(\mathbb{C}). \label{sphere2}
\end{align}
Combining \cref{sphere1,sphere2} with the observation $\gls{tr}(\mathscr{P}_V) = 1$ yields a decomposition of $\mathscr{P}_V$ as a sum of Pauli matrices,
\begin{align}
\mathscr{P}_V = \frac{1}{2}\left(\sigma_0 + \mathcal{S}(V) \cdot \vec{\sigma} \right). \label{BlochExpansion}
\end{align}
That the image of the Bloch sphere is $\mathbb{S}^2$ follows from the properties of projections,
\begin{align}
1 &= \gls{tr}(\mathscr{P}_V) \\
&= \gls{tr}(\mathscr{P}^2_V) \\
&= \frac{1}{4} \gls{tr}\left((1 + \mathcal{S}(V) \cdot \mathcal{S}(V)) \sigma_0 + 2 \mathcal{S}(V) \cdot \vec{\sigma} \right) \Rightarrow \\
\mathcal{S}(V) \cdot \mathcal{S}(V) &= 1.
\end{align}

If two one-dimensional vector spaces are distinct, then their corresponding orthogonal projections are distinct. Combining this observation with \cref{BlochExpansion} proves the injectivity of $\mathcal{S}$.
\end{proof}

One-dimensional subspaces of $\mathbb{C}^2$ are orthogonal if and only if their Bloch sphere representations are antipodal.

A $\mathcal{PT}$-symmetric vector corresponds to an element of the \textit{equator} of the Bloch sphere, which is the set defined by $\{s \in \mathbb{S}^2 \,|\, s_3 = 0\}$.

The Bloch sphere representation of the eigenspaces of $H$ given in \cref{qubitEigenspaces} is
\begin{align}
\mathcal{S}(\ker(\epsilon_{\pm} \mathbb{1} - H))
&= 
\begin{cases}
\left(\pm \dfrac{\sqrt{t^2 - \omega^2}}{t}, 
-\dfrac{\omega}{t},
0\right) & \text{if } \omega^2 \leq t^2 \\[10pt]
\left(-\dfrac{t}{\omega}, 
0,
\mp \dfrac{\sqrt{\omega^2 - t^2}}{\omega} \right) & \text{if } \omega^2 \geq t^2.
\end{cases}
\end{align}

\begin{figure}[!ht]
\centering
\includegraphics[width = .75\textwidth]{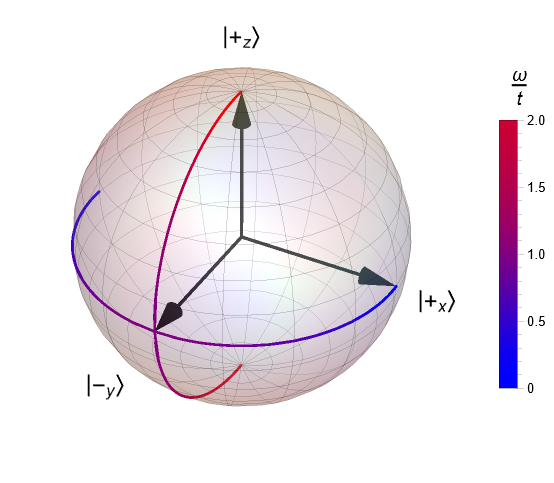}
\caption{Bloch sphere representation of the eigenvectors of the $\mathcal{PT}$-symmetric matrix given by \cref{qubitHam}. Increasing the non-Hermiticity strength $\omega$ is marked by moving from blue to red. The subspace of the Bloch sphere which is $\mathcal{PT}$-symmetric is the equator, which in polar coordinates is $\theta = \pi/2$.} \label{evecsOnBlochSphere}
\end{figure}
%


Due to the global phase symmetry of quantum theory, we can assume the state vector $\ket{\psi}$ is of the form
\begin{align}
\ket{\psi} = \begin{pmatrix}
\alpha_0 \\
\alpha_1 + \mathfrak{i} \alpha_2
\end{pmatrix},
\end{align}
where $\alpha_i \in \mathbb{R}$.
The normalization condition for $\ket{\psi}$ is 
\begin{align}
\braket{\psi|\eta| \psi} &= 1 \\
&= \zeta_0 \left(\alpha_0^2 + \alpha_1^2 + \alpha_2^2 + 2 \alpha_0 \alpha_1 \zeta_1 t + 2 \alpha_0 \alpha_2 \omega t \right). 
\end{align}

Thus, a natural interpretation of the state space is as an ellipsoid, where a state is represented by the vector $(\alpha_0, \alpha_1, \alpha_2)$. Observe how this contrasts with the traditional, Hermitian quantum theory of a qubit, where the state space is often represented as the Bloch sphere \cite{operatorSchmidt}.

\subsection{Passing Through the Exceptional Point} \label{PassingThroughEP}

Since $H$ is a sum of non-commuting matrices, \cref{MotzkinTaussky} of \cite{Motzkin1955} implies the existence of an exceptional point.

A series expansion for the eigenvalues can be computed using Newton's binomial theorem \cite{Whiteside1961}. 
In particular, given $\omega_0 \in \mathbb{R}$, we have 
\begin{align}
\epsilon_+(\omega_0 + \delta \omega, t) &= \begin{cases}
\sqrt{t^2 - \omega_0^2} \sum\limits_{n = 0}^\infty \begin{pmatrix}
\frac{1}{2} \\ n
\end{pmatrix} \left(\dfrac{2 \omega_0\, \delta \omega + \delta \omega^2}{\omega_0^2 - t^2} \right)^n& \text{ if } \omega_0^2 \neq t^2 \\
\sqrt{-2 \omega_0 \,\delta\omega} \sum\limits_{n = 0}^\infty \begin{pmatrix}
\frac{1}{2} \\ n
\end{pmatrix} \left(\dfrac{\delta\omega}{2 \omega_0} \right)^{n} & \text{ if } \omega_0^2 = t^2,
\end{cases} \\
&= \begin{cases}
\sqrt{t^2-\omega_0^2} - \dfrac{\omega_0}{\sqrt{t^2-\omega_0^2}} \, \delta \omega - \dfrac{t^2}{(t^2-\omega_0^2)^{3/2}} (\delta \omega)^2 + O(\delta \omega^3) & \text{ if } \omega_0^2 \neq t^2 \\[12pt]
\sqrt{-2 \omega_0} \,(\delta\omega)^{1/2} - \dfrac{1}{2 \sqrt{-2 \omega_0}} \, \delta \omega^{3/2} + O(\delta \omega^{5/2}) & \text{ if } \omega_0^2 = t^2,
\end{cases}
\end{align}
where 
\begin{align}
\begin{pmatrix}
\frac{1}{2} \\ n
\end{pmatrix} := \frac{1}{n!} \prod_{j \in \{0, \dots, n-1\} } \left(\frac{1}{2} - j \right)
\end{align}
denotes a generalized binomial coefficient. Note this series expansion is a Taylor series if and only if we are performing the perturbation at a point $\omega_0$ which is not an exceptional point. At the exceptional point, the series expansion for the eigenvalues is a Puiseux series, a power series involving fractional powers, instead of a Taylor series.

In conjunction with the Bauer-Fike \cref{BauerFike}, Taylor's remainder theorem derives an upper bound for the coefficient of the leading term of the Taylor expansion of the eigenvalues of $H(\omega, t)$ with respect to $\omega$. This upper bound is the condition number with respect to an axis-oriented norm of a matrix, $U$, whose columns are eigenvectors of $H$. The transpose symmetric choice for $U$ is
\begin{align}
U = \begin{pmatrix}
\mathfrak{i} \omega + \sqrt{t^2 - \omega^2} & t\\
t & -\mathfrak{i} \omega + \sqrt{t^2 - \omega^2} 
\end{pmatrix}.
\end{align}
A comparison between the condition numbers of $U$ and of the leading term in the Taylor expansion it plotted \cref{conditionNumberPlot}. I emphasize that the condition numbers of $U$ depend on the choice of columns made. Smaller condition numbers and, consequently, tighter bounds on the coefficient of the first-order term in the perturbative expansion of the eigenvalues can be obtained by optimizing the choice of columns in $U$.


\begin{figure}[!ht]
\centering
\includegraphics[width = 120mm]{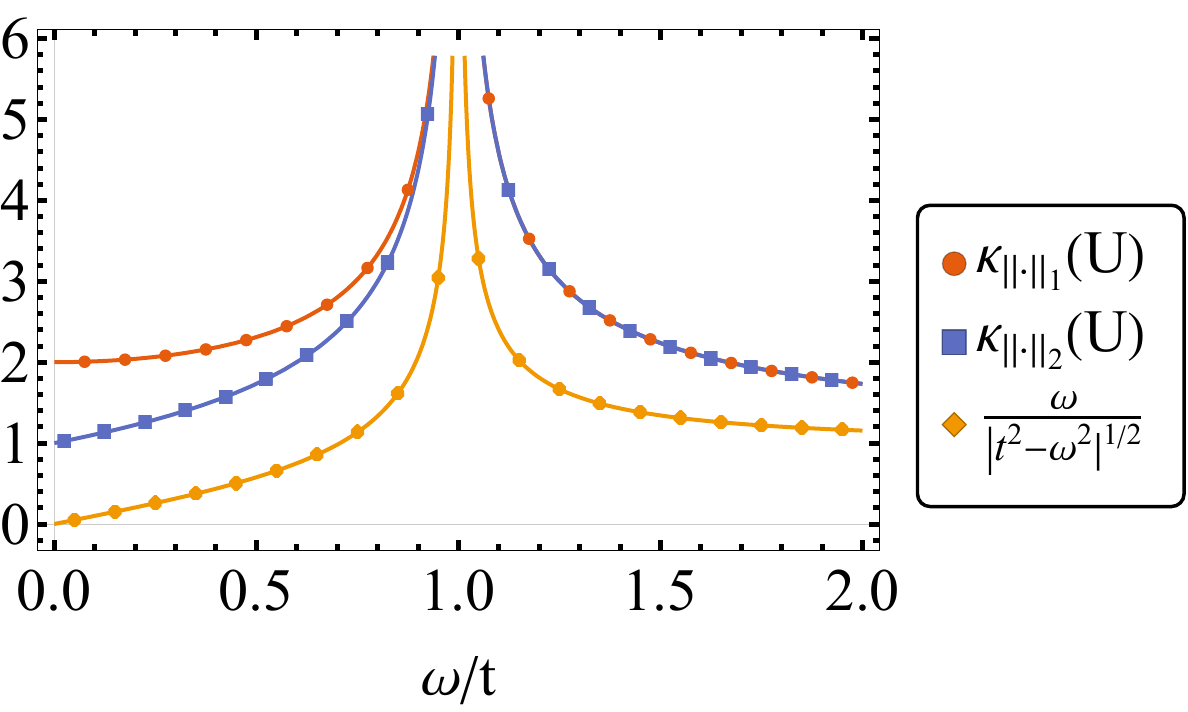}
\caption{Comparison of the first term of the Taylor series expansion of $\epsilon_+(\omega,t)$ with the upper bounds given by the Bauer-Fike \cref{BauerFike} \cite{Bauer1960}. The singularity at $\omega = t$ is a signature of the exceptional point. Due to transpose-symmetry, $\kappa_{||\cdot||_\infty}(U) = \kappa_{||\cdot||_1}(U)$. Hölder's inequality \cite{Holder} implies $\kappa_{||\cdot||_2}(U) \leq \kappa_{|\cdot||_1}(U)$. For $\omega > |t|$, $\kappa_{||\cdot||_1}(U) = \kappa_{||\cdot||_2}(U)$.}
\label{conditionNumberPlot}
\end{figure}


\subsection{Indefinite Unitary Group}

Since the Pauli matrices are both Hermitian and unitary, they could either be interpreted as observables or as unitary operators which generate time-evolution. In analogy to this observation, in a special case, the pseudo-Hermitian $2 \times 2$ qubit Hamiltonian is pseudo-unitary. To see this, note the fundamental representation of $SU(1,1)$, which is the group of all $\sigma_z$-unitaries with determinant one, is 
\begin{align}
SU(1,1) = \left\{ \begin{pmatrix}
\gamma & t \\
t^* & \gamma^*
\end{pmatrix}\,|\, (\gamma, t \in \mathbb{C}) \wedge (|\gamma|^2 - |t|^2 = 1) \right\}.
\end{align}
A physical setting where state update is given by a $\sigma_z$-unitary is the $SU(1,1)$ \textit{interferometer}, a device which was proposed in \cite{Yurke1986}. This nonlinear interferometer is constructed with an active gain element, such as a four wave mixers or parametric amplifiers, and has been constructed in laboratory settings \cite{Jing2011}.

\section{$C^*$-Algebraic Picture of Local Quantum Theory} \label{C*AlgebrasIntro}

Asking "Why do we use Hermitian operators in quantum theory?" leads to an analysis of operators with antilinear symmetry, pseudo-Hermitian operators, and quasi-Hermitian operators. A more extreme question is "Why do we use a Hilbert space?". Quantum theory can be constructed without its Hilbert space representation using $C^*$-\textit{algebras}.


In Werner Heisenberg's original work on quantum theory, he postulates that the transition amplitudes are the observables of a system \cite{Heisenberg1925,van2007sources}. Collaboration between Heisenberg, Max Born, and Pascual Jordan led to the representation of these amplitudes as matrix elements \cite{BornJordan,BornJordanHeisenberg,van2007sources}. The idea behind the $C^*$-algebraic formulation of quantum theory is to remain true to Heisenberg's vision. Observables are characterized in the abstract, through their algebraic properties alone, and transition amplitudes are computed via a mapping from this abstract space of observables into $\mathbb{C}$.

Equivalence between the $C^*$-algebraic formulation of quantum theory and the more commonplace Hilbert space formulation is a consequence of the GNS construction given in \cref{GNS-Construction} \cite{gelfand1943imbedding,Segal1947} and Stinespring's dilation \cref{StinespringThm} \cite{Stinespring1955}. Due to this equivalence, a physicist may wonder what the use of this abstraction is. Firstly, the algebraic approach to quantum field theory \cite{rejzner2016perturbative,Haag1964} has grown quite popular, especially due to the difficulties of constructive quantum field theory. Secondly, theorems which are formulated to apply to $C^*$-algebras are readily applied to the setting of quantum theory and can simplify analysis. For example, I only know how to prove the result presented in \cref{Local-Equivalence-Section} regarding local expectation values in quantum theory using the Hahn-Banach extension theorem \cite[prop. 2.3.24]{bratteli1987operator}, which is a statement regarding states in $C^*$-algebras. Lastly, reformulating a framework clarifies its axiomatic structure and suggests possibilities for generalizations.

The abstract setting of a $C^*$-algebra is defined below. Additional details on $C^*$-algebras can be located in \cref{functionalAnalysisAppendix} and in references such as \cite{Abramovich2002,KadisonRingroseI,bratteli1987operator,Bratteli1997,TakesakiI,TakesakiII}.
\begin{defn} \label{c*Defn}
An \textit{associative algebra}, $\mathfrak{A}$, over $\mathbb{C}$ is a complex vector space equipped with an associative bilinear operation, $\cdot:\mathfrak{A} \times \mathfrak{A} \to \mathfrak{A}$, referred to as the \textit{product}, which distributes over addition. A \textit{normed algebra} is an associative algebra and a normed space, where the norm must satisfy the \textit{product inequality},
\begin{align}
\forall a, b \in \mathfrak{A}, ||a \cdot b|| \leq ||a|| \, ||b||.
\end{align}
A \textit{Banach algebra} is a complete normed algebra.
An \textit{involutive complex algebra}, also referred to as a ${}^*$-\textit{algebra}, $\mathfrak{A}$, is an associative algebra over $\mathbb{C}$ equipped with a map, $*: \mathfrak{A} \rightarrow \mathfrak{A}$, which satisfies the following properties for every $a,b \in \mathfrak{A}$, $\alpha,\beta \in \mathbb{C}$:
\begin{itemize}
\item $*$ is an involution, so $(a^*)^* = a$. \\
\item $*$ is antilinear, so $(\alpha a + \beta b)^* = \overline{\alpha} a^* + \overline{\beta} b^* $, where $\overline{\alpha}$ denotes the complex conjugate of $\alpha$.\\
\item $*$ is an \textit{algebra antihomomorphism}, which means $(a \cdot b)^*= b^* \cdot a^*$.
\end{itemize}
A $C^*$-\textit{algebra} is a ${}^*$-algebra and a Banach algebra for which the $C^*$ identity holds, namely 
\begin{align}
||a^* \cdot a|| = ||a^*|| \, ||a||.
\end{align}
\end{defn}
A classic example of a $C^*$-algebra is the operator-algebra structure used in quantum theory: namely, the set of bounded operators on a Hilbert space. In this example, the product is function composition, the involution is the adjoint, and the norm is the operator norm.

\subsection{Quantum Theory} \label{C*-Quantum-Theory}

This section defines the $C^*$-algebraic framework for quantum theory. The postulates are summarized below.
\begin{itemize}
\item \textit{Observables} are the Hermitian elements of a $C^*$-algebra, $\mathfrak{A}$, which are the elements $a \in \mathfrak{A}$ such that $a = a^*$. 

\item Following \cref{stateDef}, a \textit{state}, $\varphi:\mathfrak{A} \to \mathbb{C}$, is a normalized, positive, linear functional on $\mathfrak{A}$. The \textit{expectation value} of an observable, $a = a^*$, in the state $\varphi$ is $\varphi(a)$. Reality of the expectation values of Hermitian observables follows from positivity of the state.

\item Time is parametrized by the reals. Given an initial state $\varphi_{t_0}$ at time $t_0 \in \mathbb{R}$, the state $\varphi_t$ at time $t \in \mathbb{R}$ is 
\begin{align}
\varphi_t = \varphi_{t_0} \circ \Phi_{t, t_0},
\end{align}
where $\Phi_{t, s}:\mathfrak{A} \to \mathfrak{A}$ is a completely positive map\footnote{Most treatments on the $C^*$-algebraic treatment of quantum theory instead define time-evolution in the Heisenberg picture, with the more restrictive notion of a ${}^*$-automorphism defining time evolution.}. The system is said to have \textit{time translation invariance} if $\Phi_{t, t_0} = \Phi_{t-t_0}$. 

\item Measurements $C^*$-algebraic quantum theory can be defined by using the GNS construction, summarized later in this section, to map the definition of measurement to its version in normal quantum theory.
\end{itemize}

The connection between the $C^*$-algebraic and Hilbert space formulations of quantum theory is seen through the Stinespring theorem and the GNS construction. The equivalence between the kinematic structure of both pictures of quantum theory is seen through the GNS construction.
\begin{theorem}[GNS construction \cite{gelfand1943imbedding,Segal1947}] \label{GNS-Construction}
Given a $C^*$-algebra, $\mathfrak{A}$, for every state, $\omega: \mathfrak{A} \to \mathbb{C}$, there exists a representation\footnote{Representations of $C^*$-algebras are defined in \cref{defn:Representation}.}, $\pi:\mathfrak{A} \to \mathcal{B}(\mathcal{H})$, on a Hilbert space $\mathcal{H}$ and a unit cyclic vector\footnote{A cyclic vector, $\xi \in \mathcal{H}$, associated to the subalgebra of bounded operators $\mathcal{R}(\pi)$ is one such that $\{\pi(a) \xi\,|\,a \in \mathfrak{A}\}$ is dense in $\mathcal{H}$.}, $\xi \in \mathcal{H}$, such that
\begin{align}
\omega(a) = \braket{\pi(a) \xi| \xi}.
\end{align}
\end{theorem}
Thus, expectation values of observables can either be computed using a state on a $C^*$-algebra, or a vector on a Hilbert space.

Stinespring's theorem \cite{Stinespring1955} clarifies the equivalence between time evolution as generated by a completely positive map and time evolution generated by a unitary operator.
\begin{thm}[Stinespring \cite{Stinespring1955}] \label{StinespringThm}
Let $\mathfrak{A}$ be a unital $C^*$-algebra, $\mathcal{H}$ be a Hilbert space. A map, $\Phi: \mathfrak{A} \to \mathcal{B}(\mathcal{H})$, is completely positive if and only if there exists a Hilbert space $\mathcal{K}$ and a unital $*$-homomorphism, $\pi: \mathfrak{A} \to \mathcal{B}(\mathcal{K})$ such that 
\begin{align}
\Phi(a) = V^\dag \pi(a) V,
\end{align}
where $V: \mathcal{H} \to \mathcal{K}$ is bounded.
\end{thm}

\subsection{Quasi-Hermiticity in $C^*$-algebras} \label{Quasi-Hermitian-C*}


The goal of this section is to introduce a $C^*$-algebraic framework for quasi-Hermitian quantum theory.

Suppose we are given a $C^*$-algebra, $\mathfrak{A}$, with an involution denoted by $\dag$. Every invertible element, $\Omega \in \text{GL}(\mathfrak{A})$, defines an algebra automorphism,
\begin{align}
\phi_\Omega(a) = \Omega^{-1} a \Omega. \label{quasi-Herm-C*-automorphism}
\end{align} 
This algebra automorphism can be considered an isometric $*$-isomorphism between $C^*$-algebras if the norm and involution are re-defined in the co-domain of $\phi_\Omega$. A suitable involution, $\dag_\Omega$, is the algebra anti-automorphism
\begin{align}
a^{\dag_\eta} := \eta^{-1} a^\dag \eta,
\end{align}
where 
\begin{align}
\eta = \Omega^\dag \Omega
\end{align}
is referred to as the \textit{metric}. A suitable norm is 
\begin{align}
||x||_{\Omega} := || \Omega x \Omega^{-1} ||. \label{eqn-new-norm}
\end{align}
This norm is equivalent to the original norm, since the product inequality, \cref{productInequality}, implies
\begin{align}
\frac{||x||}{||\Omega|| \, ||\Omega^{-1}||} \leq ||x||_\Omega \leq ||\Omega|| \, ||\Omega^{-1}|| \, ||x||.
\end{align}
With the above observations in mind, define $\mathfrak{A}(\Omega)$ to the $C^*$-algebra with elements, addition, and multiplication as in $\mathfrak{A}$, but with the new involution $\dag_\eta$ and the new norm $||\cdot||_\Omega$.

The identity $\mathfrak{A} = \mathfrak{A}(\gls{id})$ is a special case. 

Elements satisfying $a = a^{\dag_\eta}$ are referred to as \textit{quasi-Hermitian}. A quasi-Hermitian quantum theory is one defined in the algebra $\mathfrak{A}(\Omega)$ for some choice of $\Omega$, according to the definition given in \cref{C*-Quantum-Theory}. 

The cone of \textit{positive} elements of a $C^*$-algebra, $\mathfrak{A}$, which is denoted by $\mathfrak{A}^+$, is the set of all elements that can be written as a product $b^* b$ for some $b \in \mathfrak{A}$. In general, $\mathfrak{A}^+ \neq \mathfrak{A}(\Omega)^+$. However, if $a \in \mathfrak{A}(\Omega)^+$ is also Hermitian, then $a \in \mathfrak{A}^+$. To prove this, note
\begin{align}
a = b^{\dag_\eta} b &= \eta^{-1} b^\dag \eta b \\
&= b^\dag \eta b \eta^{-1},
\end{align}
so $a$ is the product of commuting positive elements in $\mathfrak{A}^+$. The product of commuting positive elements is positive, as is seen by applying the Gelfand-Naimark theorem to the equivalent result for bounded positive operators on a Hilbert space \cite{Wouk}.

\chapter{Non-Hermitian Lattice Models} \label{toyModelsChapter}

This chapter is devoted to matrices corresponding to one-dimensional lattice models. After introducing techniques to solve the eigenvalue problem in \cref{tridiagEsysSection}, fundamental properties of non-Hermiticity are explored in subsequent sections. Namely, intertwining operators, exceptional points, and $\mathcal{PT}$-symmetry breaking and their consequences are studied. 

A simple source of non-Hermiticity in a matrix is a pair of complex-valued elements on the diagonal, referred to as defect potentials. The original results of this chapter pertain to three types of lattice models with non-Hermitian defect potentials, presented in sections~\ref{nearestNeighbour}, \ref{SSHSection}, and \ref{uniformSection} respectively. The results for each case are summarized at the beginning of their corresponding sections and a global summary can be found in \cref{Toy Models Summary} of the summary chapter at the beginning of this thesis. The discussion preceding these models serves to introduce notation and to review and apply known properties of tridiagonal matrices and their ilk.

The central mathematical structure of this chapter is a class of square matrix polynomials, $H_n: \mathbb{C}^n \times \mathbb{C}^n \times \mathbb{C}^n \to \mathfrak{M}_n(\mathbb{C})$, where the number of rows is finite and greater than one ($n \geq 2$). The matrix elements of $H_n$ in the canonical basis of $\mathbb{C}^n$ are 
\begin{align}
H_n(\alpha,\beta,z)_{ij} &= \alpha_j \delta^i_{j+1} + \beta_i \delta^j_{i+1} + z_i \delta^i_{j} + \alpha_n \delta^i_{1} \delta^j_{n} + \beta_n \delta^i_{n} \delta^j_{1}, \label{tridiagElements}
\end{align}
where $\delta$ is the \textit{Kronecker delta}, satisfying 
\begin{align}
\delta^i_{j} := \begin{cases}
1 & \text{ if } i = j \\
0 & \text{ if } i \neq j.
\end{cases}
\end{align}
A matrix representation of $H_n$ for $n = 6$ is 
\begin{align}
H_6(\alpha,\beta,z) = \begin{pmatrix}
z_1 & \beta_1  & 0        & 0   	  & 0  		 & \alpha_6 \\
\alpha_1 & z_2 & \beta_2  & 0   	  & 0   	 & 0 		\\
0        & \alpha_2 & z_3 & \beta_3  & 0 	     & 0 		\\
0        & 0        & \alpha_3 & z_4 & \beta_4  & 0 		\\
0        & 0   		& 0   	   & \alpha_4 & z_5 & \beta_5  \\
\beta_6  & 0  		& 0   	   & 0   	  & \alpha_5 & z_6
\end{pmatrix}. \label{Hscr}
\end{align}
When a choice for the inputs of $H_n$ is clear from the context, I will use the shorthand 
\begin{align}
\mathscr{H}_n := H_n(\alpha,\beta,\gamma)
\end{align} 
for the value of $H_n$. 
$\mathscr{H}_n$ is referred to as \textit{tridiagonal} when $\alpha_n = \beta_n = 0$ and enforcing the condition $\alpha_n = \beta_n = 0$ is referred to as imposing \textit{open boundary conditions}. When $\alpha_n$ or $\beta_n$ is nonzero, I will refer to $\mathscr{H}_n$ as a tridiagonal matrix with \textit{perturbed corners}. A tridiagonal matrix is called \textit{irreducible} if $\alpha_i \neq 0$ and $\beta_i \neq 0$ for all $i \in \{1, \dots, n-1\}$.


Some applications of matrices akin to $H_n$ include the following: 



\begin{itemize}
\item Tridiagonal matrices with perturbed corners are used in many discretization schemes for linear partial differential equations of first or second order \cite{garcia2000numerical}. Examples of such equations which appear in physics include the Schr{\"o}dinger equation, the advection-diffusion equation, Poisson's equation, the heat equation, and the wave equation; and an example of a scheme yielding tridiagonal matrix eigenvalue problems is the Lax-Wendroff scheme \cite{LaxWendroff}.

\item The adjacency matrices of the path and cycle graphs 
with $n$ vertices are included in the range of $H_n$. Adjacency matrices of graphs are used to model random walks of particles on said graphs \cite{Kac1947}. Physically, random walks can be interpreted as models of particles hopping on a graph. Consequently, the parameters $\alpha, \beta$ in $\mathscr{H}_n$ are typically referred to as \textit{hopping amplitudes}.

\item A simple model of a material is to assume its constituent particles are fixed in location on a lattice. Each particle can have local degrees of freedom, such as, for example, a magnetic dipole moment. Perturbing the material from equilibrium causes the local degrees of freedom to evolve time. In quantum theory, this time evolution is modelled via the Schr{\"o}dinger equation with a Hamiltonian. In the simplest quantum many-body problems, there exists a reduction from the full problem to its \textit{first-quantized} form. Typically, first-quantized Hamiltonians associated with time-evolution which is local to a lattice are analysed through the eigenvalue problem of tridiagonal matrices \cite{rutherford1948xxv,Lieb1961}. 
\end{itemize}
%

I will now introduce some notation: the canonical basis of $\mathbb{C}^n$ is $e_i \in \mathbb{C}^n$, where 
\begin{align}
(e_i)_j := \delta^i_{j}.
\end{align}
The vector whose components in the canonical basis are all equal to one is $\textbf{e} \in \mathbb{C}^n$,
\begin{align}
\textbf{e} = \sum_{i = 1}^n e_i. \label{bold e vector}
\end{align}
Given $m,n \in \mathbb{R}$ such that $m-n \in \mathbb{N}$, \gls{stMaryRd} is a shorthand for
\begin{align}
\dbrac{m,n} &:= \{m, \dots, n\}.
\end{align}
Given $\psi \in \mathbb{C}^n$ for every subset $S \subseteq \dbrac{1, n}$, I will define
\begin{align}
\psi_S = \sum_{i \in S} \psi_i e_i.
\end{align}
Lastly, I will use the shorthand
\begin{align}
\overline{m} := n-m+1.
\end{align}

\section{Eigensystems of One-Dimensional Lattice Models} \label{tridiagEsysSection}

This section analyses the eigenvalue problem of $\mathscr{H}_n$,
\begin{align}
\mathscr{H}_n \psi = \lambda \psi,
\end{align}
where $\lambda \in \mathbb{C}$ and $\psi \in \mathbb{C}^n$. Explicitly, the eigenvalue problem is
\begin{align}
\beta_1 \psi_2 &= (\lambda - z_1) \psi_1 - \alpha_{n} \psi_{n} \nonumber \\
\beta_i \psi_{i+1} &= (\lambda - z_i) \psi_i - \alpha_{i-1} \psi_{i-1} & \forall i \in \dbrac{2,n-1} \nonumber \\
\beta_n \psi_1 &= (\lambda - z_n) \psi_n - \alpha_{n-1} \psi_{n-1}. \label{recurrence}
\end{align}

The solution to the eigenvalue problem of an operator is crucial to understanding its physical interpretation in quantum theory, as the interpretation of an operator as an observable is only possible with an understanding of the set of measurement outcomes associated to said observable. 

\subsection{Similarity Transformations}

Similarity transformations are often used to map eigenvalue problems of matrices to previously solved eigenvalue problems. If $\lambda \in \sigma(A), v\in \ker(\lambda \mathbb{1} - A)$ is an eigenvalue/vector pair for the matrix $A$, then $(\lambda, S v)$ is an eigenvalue eigenvector pair for the similar matrix $S A S^{-1}$. 

This section showcases similarity transformations which map matrices in $H_n(\mathbb{C}^n,\mathbb{C}^n,\mathbb{C}^n)$ to matrices in $H_n(\mathbb{C}^n,\mathbb{C}^n,\mathbb{C}^n)$. Two such similarity transformations are given by the \textit{exchange matrix}, $\mathcal{P}_n$, the \textit{shift matrix}, $S_n$. Explicitly, define the matrices $E_n, \mathcal{P}_n, S_n \in \mathfrak{M}_n(\mathbb{C})$ by their elements
\begin{align}
(E_n)_{ij} :&= (-1)^{i} \delta^i_{j} \label{diagAltSign}\\
(\mathcal{P}_n)_{ij} :&= \delta^{i}_{\overline{j}} \\
(S_n)_{ij} :&= \begin{cases}
1 & \text{if } i-j \equiv 1 \mod n\\
0 & \text{otherwise}
\end{cases} \nonumber \\
&= H_n(\textbf{e}, 0, 0),
\end{align}
and define the vector $s_n \in \mathbb{C}^n$ to have the components 
\begin{align}
(s_n)_i := \begin{cases}
1 & \text{ if } i < n \\
(-1)^n & \text{ if } i = n.
\end{cases}
\end{align}
Both $\mathcal{P}$ and $S$ are permutation matrices. The following similarity transforms can be computed,
\begin{align}
E_n H_n(\alpha,\beta,z) E_n &= H_n(-s_n \alpha, -s_n \beta, z) \label{tridiagStagger} \\
\mathcal{P}_n H_n(\alpha, \beta, z) \mathcal{P}_n &= H_n(\mathcal{P}_n S_n \beta, \mathcal{P}_n S_n \alpha, \mathcal{P}_n z) \label{paritySimilarity} \\
S_n H_n(\alpha,\beta,z) S_n^{-1} &= H_n(S_n \alpha, S_n \beta, S_n z).
\end{align}
The tridiagonal case of \cref{tridiagStagger} was known to \cite{kahan1966accurate}, physical applications of this formula can be found in \cite{Valiente2010,Joglekar2010}. 

A special case is the \textit{centrohermitian} or $\mathcal{PT}$-\textit{symmetric} case, where $\mathcal{T}$ denotes the complex-conjugation operation in the canonical basis, as in \cref{centrohermitianDefn}. The similarity transformation given by $\mathcal{P}$, \cref{paritySimilarity}, implies that $\mathscr{H}_n$ is $\mathcal{PT}$-symmetric if and only if $z = \mathcal{P}_n z^*$ and $\alpha = \mathcal{P}_n S_n \beta$.

A powerful corollary of \cref{tridiagStagger} is \cite[Thm. 3.1]{Dyachenko2021}, a result which we generalize to cases with $\alpha_n, \beta_n \neq 0$ below.
\begin{theorem} \label{Dyachenko2021Thm}
Suppose either $\alpha_n = \beta_n = 0$ or $n$ is even and $z \in \mathbb{C}^n$ is 2-periodic, so that $z_j = z_k$ whenever $j \equiv k \mod 2$. Then 
\begin{align}
\sigma(H_n(\alpha,\beta,z)) = \left\{ \frac{z_1 + z_2}{2} \pm \sqrt{\lambda^2 + \left(\frac{z_1 - z_2}{2} \right)^2} \,|\, \lambda \in \sigma(H_n(\alpha,\beta,0)) \right\}.
\end{align}
If $u(\lambda) \in \mathbb{C}^n$ is an eigenvector of $H_n(\alpha,\beta,0)$ with eigenvalue $\lambda$, then two corresponding eigenvectors, $v_{\pm}(\lambda)\in \mathbb{C}^n$, of $H_n(\alpha,\beta,z)$ are 
\begin{align}
v_{\pm}(\lambda) &= \left(\lambda \pm \sqrt{\lambda^2 + \left(\frac{z_1 - z_2}{2} \right)^2} \right) u(\lambda) + \frac{z_2 - z_1}{2} E_n u(\lambda) \\
H_n(\alpha,\beta,z) v_{\pm}(\lambda) &= \left(\frac{z_1 + z_2}{2} \pm \sqrt{\lambda^2 + \left(\frac{z_1 - z_2}{2} \right)^2} \right) v_{\pm}(\lambda).
\end{align}
\end{theorem}
\begin{proof}
The key insight comes from the identity 
\begin{align}
H_n(\alpha,\beta,z) = \frac{z_1 + z_2}{2}\mathbb{1} + \frac{z_2 - z_1}{2} E_n + H_n(\alpha,\beta, 0).
\end{align}
By \cref{tridiagStagger}, $E_n$ and $H_n(\alpha,\beta, 0)$ anticommute, so the theorem statement follows from \cref{eqn-anticommut-evecs}, a result proven in a later chapter.
\end{proof}

Next, consider the irreducible case, where $\alpha_j \neq 0$ and $\beta_j \neq 0$ for all $j \in \dbrac{1,n-1}$. A diagonal similarity transform which maps an irreducible tridiagonal matrix to a transpose symmetric counterpart is \cite{Rzsa1969}, 

\begin{align}
\mathcal{D}_{ij} &:= \prod^{i-1}_{k = 1} \sqrt{\frac{\alpha_k}{\beta_k}} \delta^i_{j} \label{triDiagToSymmetric}\\
(\mathcal{D}^{-1} \mathscr{H}_n \mathcal{D})_{ij} &= \sqrt{\alpha_j \beta_j} \delta^i_{j+1} + \sqrt{\alpha_i \beta_i} \delta^j_{i+1} + z_i \delta^i_{j} \nonumber \\
&+ \alpha_n \left(\prod^{n-1}_{k = 1} \sqrt{\frac{\alpha_i}{\beta_i}} \right)\delta^i_{1} \delta^j_{n} + \beta_n \left(\prod^{n-1}_{k = 1} \sqrt{\frac{\beta_i}{\alpha_i}} \right)\delta^i_{n} \delta^j_{1} \\
&= (\mathcal{D}^{-1} \mathscr{H}_n \mathcal{D})_{ji}.
\end{align}

\subsection{Tridiagonal Case}
This section examines the eigenvalue problem and the problem of determining the inverse associated to $\mathscr{H}_n$ in the tridiagonal case $\alpha_n = \beta_n = 0$. 

A peculiar feature of tridiagonal matrices, $\mathscr{H}_n$, is a correspondence between their eigenvectors and the \textit{minors} of $\lambda \mathbb{1} - \mathscr{H}_n$, as defined in \cref{minorsDef}. Furthermore, when it exists, the inverse of $\mathscr{H}_n$ can be calculated using the minors of $\lambda \mathbb{1} - \mathscr{H}_n$.
\begin{definition} \label{minorsDef}
Given a matrix, $A \in \mathbb{C}^{m \times n}$, a \textit{minor} of $A$ is the determinant of a square submatrix of $A$. Explicitly, given two collections of indices with equal cardinality, $I \subseteq \dbrac{1,m}, J \subseteq \dbrac{1,n}$ such that $\text{card}(I) = \text{card}(J)$, a minor is
\begin{align}
\det_{I, J}(A) := \det A_{IJ},
\end{align}
where $A_{IJ} \in \mathfrak{M}_{\text{card}(I)}(\mathbb{C})$ is the submatrix with elements
\begin{align}
(A_{I J})_{ij} := A_{f_I(i) f_J(j)},
\end{align}
and $f_{I}:\dbrac{1, \text{card}(I)} \to I$ and $f_J:\dbrac{1, \text{card}(J)} \to J$ are the unique strictly monotone functions with their respectively defined domains and codomains.
If $A\in \mathfrak{M}_n(\mathbb{C})$ is a square matrix and $I = \dbrac{1, n} \setminus \{i\}, J = \dbrac{1, n} \setminus \{j\}$, then $\det_{I, J}(A)$ is referred to as the $(i, j)$-th \textit{cofactor} of $A$. 
The \textit{adjugate} of a square matrix, which is sometimes referred to as the \textit{classical adjoint} \cite[p. 159]{hoffmann1971linear} or the \textit{adjunct}, is \cite{strang2006linear}
\begin{align}
\text{adj}(A)_{ij} := (-1)^{i+j} \det_{\dbrac{1, n} \setminus \{j\},\dbrac{1, n} \setminus \{i\}}(A).
\end{align}
\end{definition}

The determinant of a tridiagonal matrix is referred to as a \textit{continuant}.
For the specific tridiagonal matrix $\mathscr{H}_n$, consider minors associated to the first and last $i$ rows respectively,
\begin{align}
\theta_i(\lambda) &:= \begin{cases}
\det\limits_{\dbrac{1, i}, \dbrac{1, i}}(\lambda \mathbb{1} - \mathscr{H}_n) & \text{if } i \in \dbrac{1,n} \\
1 & \text{if } i = 0 \\
0 & \text{if } i = -1 
\end{cases} \label{principalMinor}
\\
\phi_i(\lambda)  &:= 
\begin{cases}
\det\limits_{\dbrac{i, n}, \dbrac{i, n}}(\lambda \mathbb{1} - \mathscr{H}_n) & \text{if } i \in \dbrac{1,n}\\
1 & \text{if } i = n+1 \\
0 & \text{if } i = n+2
\end{cases}. \label{backwardsMinor}
\end{align}
In particular, $\theta_n(\lambda) = \phi_1(\lambda)$ is the characteristic polynomial of $\mathscr{H}_n$. Continuants satisfy linear recurrence relations, explicitly
\begin{align}
\theta_i(\lambda) &= (\lambda - z_i) \theta_{i-1}(\lambda) - \alpha_{i-1} \beta_{i-1} \theta_{i-2}(\lambda) &\quad& \forall i \in \dbrac{1,n}, \label{thetaRecurrence} \\
\phi_i(\lambda) &= (\lambda - z_i) \phi_{i+1}(\lambda) - \alpha_{i+1} \beta_{i+1} \phi_{i+2}(\lambda)&\quad& \forall i \in \dbrac{1,n}. \label{phiRecurrence}
\end{align}





%

The following theorem concerning the eigenvectors of tridiagonal matrices, proven in \cite[p. 69]{Gantmacher2002}, immediately follows by comparing the recurrence relations of \cref{recurrence,thetaRecurrence}:
\begin{theorem} \label{charPolyEvec}
Suppose $\lambda \in \sigma(H_n(\alpha,\beta,z))$ and $\alpha_n = \beta_n = 0$.
Then an eigenvector corresponding to $\lambda$ is 
\begin{align}
\psi(\lambda) := \sum_{i = 1}^n \,\dfrac{\theta_{i-1}(\lambda)}{\prod^{i-1}_{j = 1} \beta_{j}} \,e_i \in \ker(\lambda \mathbb{1} - H_n(\alpha,\beta,z)), \label{psi(l)}
\end{align}
where $\theta_i(\lambda)$ is defined in \cref{principalMinor}.
\end{theorem}

An interesting interlacing result for irreducible Hermitian tridiagonal matrices follows from the relationship between eigenvectors and determinants presented in \cref{charPolyEvec} and the \textit{Cauchy interlacing theorem}, which is a beautiful property regarding the eigenvalues of Hermitian matrices.
\begin{theorem}[Cauchy Interlacing \cite{Hwang2004}] \label{thm:CauchyInterlacing}
Let $J = J^\dag \in \mathfrak{M}_n(\mathbb{C})$ be an $n \times n$ Hermitian matrix with $n \geq 2$, and let $J_{n-1} \in \mathfrak{M}_{n-1}(\mathbb{C})$ denote the principle minor formed by removing the $n$-th row and column from $J$. Since $J, J_{n-1}$ are Hermitian, their eigenvalues are real {\normalfont \cite{Hermite1855}} and, thus, can be sorted as $\lambda_{1}(J) \leq \lambda_2(J) \leq \dots \leq \lambda_n(J)$ and $\lambda_{1}(J_{n-1}) \leq \lambda_2(J_{n-1}) \leq \dots \leq \lambda_{n-1}(J_{n-1})$ respectively, where eigenvalues are included in these sorted list a number of times equal to their algebraic multiplicity. Then for all $i \in \{1, \dots, n-1\}$,
\begin{align}
\lambda_{i}(J) \leq \lambda_{i-1}(J_{n-1}).
\end{align}
\end{theorem}
\begin{corollary} \label{cor:CauchyTridiag}
The interlacing inequalities given in \cref{thm:CauchyInterlacing} for irreducible tridiagonal matrices are strict inequality.
\end{corollary}
\begin{proof}
Suppose $\lambda$ was a root of a tridiagonal matrix, $\mathscr{H}_n$, and its principal minor, $\mathscr{H}_{n-1}$. Then the linearity property for determinants implies that the characteristic polynomials for all minors associated with the first $j \in \dbrac{1,n}$ rows and columns of $\mathscr{H}_n$ evaluate to zero,
\begin{align}
\lambda \in \sigma(\mathscr{H}_{n}) \cap \sigma(\mathscr{H}_{n-1})\,\Rightarrow\,\theta_i(\lambda) = 0 &\quad& \forall i \in \dbrac{1,n}.
\end{align}
Due to the relationship between eigenvectors and $\theta_i$ given in \cref{charPolyEvec}, this would imply $e_1$ is an eigenvector of $\mathscr{H}_n$, a clear contradiction.
\end{proof}

A further corollary of the Cauchy interlacing theorem is that the eigenvalues of irreducible Hermitian tridiagonal matrices are simple. One particularly useful result regarding the multiplicity of eigenvalues for more general cases of tridiagonal matrices is given in \cite[Thm. 4.2]{elliott1953characteristic}, which we summarize here\footnote{Joseph Elliot proved this for transpose symmetric tridiagonal matrices. The similarity transform of \cref{triDiagToSymmetric} reduces the proof of the general case to the transpose symmetric case.}.
\begin{thm} \label{tridiagNondegen}
Given a tridiagonal matrix, $\mathscr{H}_n$ satisfying $\alpha_n = \beta_n = 0$, if an eigenvalue of $\mathscr{H}_n$ has geometric multiplicity $k$, then $\alpha_i \beta_i = 0$ for at least $k-1$ values of $i \in \dbrac{1,n-1}$. Thus, in the irreducible case, where $\alpha_i \neq 0$ and $\beta_i \neq 0$ for all $i \in \dbrac{1, n-1}$, if the characteristic polynomial has less than $n$ distinct roots, then $\mathscr{H}_n$ is defective.
\end{thm}

The following theorem, given in \cite{Lewis1982,UsmaniDetailed,Usmani1994}, characterizes the inverses of tridiagonal matrices.
\begin{theorem} \label{tridiagInverse}
Assuming $\det \mathscr{H}_n \neq 0$, the matrix elements of the inverse of $\mathscr{H}_n$ are 
\begin{align}
(\mathscr{H}_n^{-1})_{ij} = \frac{1}{\theta_n(0)}\begin{cases}
\theta_{i-1}(0) \phi_{j+1}(0) \prod_{k = i}^{j-1} \beta_k & \text{if } i < j \\
\theta_{i-1}(0) \phi_{i+1}(0) & \text{if } i = j \\
\theta_{j-1}(0) \phi_{i+1}(0) \prod_{k = j}^{i-1} \alpha_k & \text{if } i > j
\end{cases},
\end{align}
where $\theta_{i}(\lambda)$ and $\phi_i(\lambda)$ are defined in \cref{principalMinor,backwardsMinor}.
\end{theorem}
\begin{proof}
A well known result from linear algebra is an explicit formula for the inverse of a matrix, $A \in \text{GL}_n(\mathbb{C})$,
\begin{align}
A^{-1} = \frac{\text{adj}(A)}{\det A}.
\end{align}
The relationship between the matrix elements of $\text{adj}(A)$ and the continuants $\theta_i(0), \phi_j(0)$ follows from elementary properties of the determinant.
\end{proof}



\subsection{Solved Examples}


This section showcases a tactic for solving eigensystems of one-dimensional lattice models. First, the characteristic polynomial is computed using a linear recurrence relation. The eigenvalues are the roots of this polynomial. Then, an ansatz for the eigenvectors corresponding to the eigenvalues is given by treating the eigenvalue equation as a linear recurrence relation. More detailed properties of the eigensystems, such as cases where the eigenvalues have closed-form expressions, are left to later parts of this chapter.

The final example of this section takes a different approach, instead relying on the geometry of the inhomogenous Lorentz group to analyse generalizations of the Sylvester-Kac matrices.

One benefit of compiling solved examples is the ability to readily investigate their generalizations. For instance, the linearity property of determinants on rows and columns can be used to express the characteristic polynomial of a complicated matrix in terms of closed-form characteristic polynomials associated with simpler matrices. As another example, the main result of \cite{Chu2010} can be derived combining the algebraic considerations in \cref{su(2)tridiagExample} with the chiral-symmetry result in \cref{Dyachenko2021Thm}.

The \textit{Chebyshev polynomials} are particularly useful in the analysis of certain tridiagonal matrix eigenvalue problems. Some key results regarding the Chebyshev polynomials are compiled in \cref{Chebyshev Appendix}. For now, I simply define the Chebyshev polynomials of the second kind. {}{}{}{}
\begin{defn}
The \textit{Chebyshev polynomials of the second kind} are the unique solution to the linear recurrence relation
\begin{align}
U_{j+1}(x) := 2 x U_j(x) - U_{j-1}(x), \label{Chebyshev-Recurrence}
\end{align}
where $j \in \mathbb{Z}$ and the initial conditions are 
\begin{align}
U_{-1}(x) = 0 &\quad& U_0(x) = 1.
\end{align}
\end{defn}
When the argument of a Chebyshev polynomial is a real number, this polynomial can be expressed using trigonometric functions. For example,
\begin{align}
U_{j}(x) &:= \begin{cases}
\dfrac{\sin ((j+1) \arccos x)}{\sin (\arccos x)} & \text{if } x^2 < 1 \\[10pt] 
(\pm 1)^j (j+1) &\text{if } x = \pm 1 \\[10pt]
\dfrac{\sinh ((j+1) \text{arccosh} \, x)}{\sinh (\text{arccosh}\, x)} & \text{if } x^2 > 1.
\end{cases}\label{ChebyshevSecondKind}
\end{align}

Examples of tridiagonal matrices whose eigenvalue problems admit closed-form solutions which are not discussed in this section can be found in \cite{Eberlein1964,Chang2009,Chu2019,Alazemi2021,KILI2016}.

\begin{ex} \label{Toeplitz}
Consider the adjacency matrix of the path graph with $n$ vertices, which corresponds to $\mathscr{H}_n$ with the parameters $\alpha = \beta = \textbf{e}_{\dbrac{1,n-1}}$ and $z = 0$. This eigenvalue problem has found applications in \cite{rutherford1948xxv}. 

As summarized in \cite{Muir1906,muir1920theory}, 
J. Wolstenholme and J. W. L. Glaisher determined the characteristic polynomial of $\mathscr{H}_n$ \cite{Wolstenholme1874}, 
\begin{align}
\text{det}(\lambda \mathbb{1} - \mathscr{H}_n) = U_i\left(\frac{\lambda}{2}\right).
\end{align}
The roots of the Chebyshev polynomials are
\begin{align}
U_n\left( \cos \left(\frac{k \pi}{n+1} \right)\right) = 0 &\quad& \forall k \in \dbrac{1, n}, \label{ChebyshevSecondKindRoots}
\end{align}
so the spectrum of $\mathscr{H}_n$ is
\begin{align}
\sigma(\mathscr{H}_n) = \left\{2 \cos \left(\frac{k \pi}{n+1} \right) : k \in \dbrac{1,n} \right\}.
\end{align}
By \cref{charPolyEvec,tridiagNondegen}, the eigenvector $\psi:\sigma(A) \to \ker(\lambda \mathbb{1} - \mathscr{H}_n)$ corresponding to an eigenvalue $\lambda \in \sigma(\mathscr{H}_n)$ is
\begin{align}
\psi(\lambda) = \sum_{i = 1}^n \psi_1 U_{i-1}\left(\frac{\lambda}{2} \right) e_i.
\end{align}
Applying \cref{tridiagInverse}, the inverse of $\mathscr{H}_{2n}$ is
\begin{align}
(\mathscr{H}_{2n}^{-1})_{jk}= (-1)^{n+j+k} \sin\left(\frac{\min\{j,k\} \pi}{2}\right) \sin\left(\frac{(2 n-\max\{j,k\}+1) \pi}{2}\right).
\end{align} 
\demo
\end{ex}

\begin{ex} \label{Ortega}
$\alpha = \beta = \textbf{e}_{\dbrac{1,n-1}}, z = z_{m_1} e_{m_1} + z_{m_2} e_{m_2}, m_1 < m_2$. The full solution of the eigensystem problem for $m_2 = n-m_1+1$ appears in \cite{ortega2019mathcal}, special cases of which have appeared in \cite{rutherford1948xxv,Babbey}. The characteristic polynomial is found by using the linearity property of determinants on the rows and columns indexed by $m_1, m_2$, resulting in an expression containing the characteristic polynomial of \cref{Toeplitz},
\begin{align}
\det(\lambda \mathbb{1} - \mathscr{H}_n)
&= U_n \left(\frac{\lambda}{2} \right) 
+ z_{m_1} z_{m_2} U_{n - m_2} \left(\frac{\lambda}{2} \right) U_{m_2 - m_1 - 1} \left(\frac{\lambda}{2} \right) U_{m_1-1} \left(\frac{\lambda}{2} \right) \nonumber \\
&- z_{m_1} U_{n-m_1}\left(\frac{\lambda}{2} \right) U_{m_1-1}\left(\frac{\lambda}{2} \right) 
- z_{m_2} U_{n-m_2} \left(\frac{\lambda}{2} \right) U_{m_2 - 1} \left(\frac{\lambda}{2} \right).
\end{align}
In the special case where $m_2 = n-m_1+1$, the characteristic equation simplifies to
\begin{align}
\text{det}(\lambda \mathbb{1} - \mathscr{H}_n) &= U_{n} \left(\frac{\lambda}{2}\right)
+ z_{m_1} z_{m_2} U_{n-2 m_1} \left(\frac{\lambda}{2} \right) U_{m_1-1}^2\left(\frac{\lambda}{2} \right) \\
 &- (z_{m_1} + z_{m_2}) U_{n-m_1} \left(\frac{\lambda}{2} \right) U_{m_1-1} \left(\frac{\lambda}{2} \right). \label{ortegaCharEq}
\end{align}

The eigenvalues of $\mathscr{H}_n$ are the roots of \cref{ortegaCharEq}. By \cref{charPolyEvec,tridiagNondegen}, the unique eigenvector corresponding to an eigenvalue $\lambda \in \mathscr{H}_n$ is (up to a normalization constant, $\psi(\lambda)_1$)
\begin{align}
\frac{\psi(\lambda)_{i+1}}{\psi(\lambda)_1} &= \begin{cases}
U_{i}\left(\frac{\lambda}{2} \right) & \text{if } i \in \dbrac{0, m_1-1} \\[10pt]
U_i\left(\frac{\lambda}{2} \right) - z_{m_1} U_{i-m_1} \left( \frac{\lambda}{2} \right)U_{m_1-1} \left( \frac{\lambda}{2} \right) & \text{if }  i \in \dbrac{m_1, m_2-1}\\[10pt]
\begin{array}{l}
U_{i} \left(\frac{\lambda}{2} \right) +
z_{m_1} z_{m_2} U_{i - m_2} \left(\frac{\lambda}{2} \right) U_{m_2 - m_1} \left(\frac{\lambda}{2} \right) U_{m_1-1} \left(\frac{\lambda}{2} \right) \\
- z_{m_1} U_{i-m_1}\left(\frac{\lambda}{2} \right) U_{m_1-1}\left(\frac{\lambda}{2} \right) \\
- z_{m_2} U_{i-m_2} \left(\frac{\lambda}{2} \right) U_{m_2 - 1} \left(\frac{\lambda}{2} \right)
\end{array} & \text{if } i \in \dbrac{m_2, n-1}.
\end{cases}
\end{align}

In the special case where $n+1 = m_1 + m_2$, since $\mathscr{H}_n$ is centrohermitian, by \cref{AntilinearSpectrumSymmetry}, a simpler expression for $\psi_i$ is 
\begin{align}
\frac{\psi(\lambda)_{i+1}}{\psi(\lambda)_1} &= \begin{cases}
U_{i}\left(\frac{\lambda}{2} \right) & \text{if } i \in \dbrac{0, m_1-1} \\[10pt]
U_i\left(\frac{\lambda}{2} \right) - \gamma_{m_1} U_{i-m_1} \left( \frac{\lambda}{2} \right)U_{m_1-1} \left( \frac{\lambda}{2} \right) & \text{if }  i \in \dbrac{m_1, m_2-1}\\[10pt]
U_{n-i-1}\left(\frac{\lambda^*}{2} \right) & \text{if } i \in \dbrac{m_2, n-1}
\end{cases}.
\end{align}
The roots of this characteristic polynomial admit closed-form expressions for special cases of the potential strength $z_i$, as summarized in \cref{closedFormEvalsTable}.

\demo
\end{ex}

One technique for solving linear recurrence relations is to reduce the problem to computing matrix powers. For our applications, we will use the following formula for the powers of a $2 \times 2$ matrix, $A$,
\begin{theorem}
If $A \in \mathfrak{M}_2(\mathbb{C})$, then given $n \in \gls{Z+}$,
\begin{align}
A^n &= \begin{cases}
(\det A)^{(n-1)/2} U_{n-1} \left(\frac{{\normalfont \text{tr}} A}{2 (\det A)^{1/2}} \right) A - (\det A)^{n/2} U_{n-2} \left(\frac{{\normalfont \text{tr}} A}{2 (\det A)^{1/2}} \right) \mathbb{1} & \text{ if } \det A \neq 0 \\
({\normalfont \text{tr}} A )^{n-1} A & \text{ if } \det A = 0
\end{cases}. \label{2x2Powers}
\end{align}
\end{theorem}
\begin{proof}
A proof for the case where $A$ is invertible can be found in \cite{Ricci1975}.
For a 2 $\times$ 2 matrix, $A$, the Cayley-Hamilton theorem \cite{Cayley1858,Frobenius1878} implies
\begin{align}
A^2 &= (\text{tr} A )A - (\det A) \mathbb{1} \, \Rightarrow \\
A^{n+2} &= (\text{tr} A) A^{n+1} - (\det A) A^n.
\end{align}
The solution to this recurrence relation ($n \geq 1$) is \cref{2x2Powers}.
\end{proof}
Amusingly, it does not matter which square root of $\det A$ is used in \cref{2x2Powers}, since 
\begin{align}
U_n(x) = (-1)^n U_n(-x).
\end{align}

The following example generalizes \cref{Toeplitz} by allowing the corner elements $\alpha_n, \beta_n$ to be nonzero.
\begin{ex} \label{nontrivialBoundaryExample}
Consider $\mathscr{H}_n$ with the parameters
$\alpha = \textbf{e}_{\dbrac{1,n-1}} + \alpha_n e_n, \beta = \textbf{e}_{\dbrac{1,n-1}} + \beta_n e_n, \gamma = 0$. The special case $\alpha_n = \beta_n = 1$ is an example of a circulant matrix, whose eigensystem was determined in \cite{got1911demonstration}. A solution to the problem where $\alpha_n, \beta_n$ are left arbitrary is given in \cite{YUEH2008}. 

The eigenvalue equation for sites $i \in \dbrac{2, n-1}$ can be recast as 
\begin{align}
\begin{pmatrix}
\psi_{i+1} \\
\psi_i
\end{pmatrix} &= \begin{pmatrix}
\lambda & -1 \\
1 & 0
\end{pmatrix} \begin{pmatrix} \psi_{i} \\ \psi_{i-1} \end{pmatrix} &\quad& \forall i \in \dbrac{2,n-1}. \label{PBCStepOne}
\end{align}
Recursively applying \cref{PBCStepOne} on the vector $(\psi_i, \psi_{i-1})^T$ results in a relationship between $\psi_i$, $\psi_{1}$ and $\psi_2$,
\begin{align}
\begin{pmatrix}
\psi_{i} \\
\psi_{i-1}
\end{pmatrix} &= \begin{pmatrix}
\lambda & -1 \\
1 & 0
\end{pmatrix}^{i-2} \begin{pmatrix}
\psi_2 \\
\psi_1
\end{pmatrix} \\
&= \begin{pmatrix}
U_{i-2} \left(\frac{\lambda}{2} \right) & -U_{i-3} \left(\frac{\lambda}{2} \right) \\
U_{i-3} \left(\frac{\lambda}{2} \right) & -U_{i-4} \left(\frac{\lambda}{2} \right)
\end{pmatrix} \begin{pmatrix}
\psi_2 \\
\psi_1
\end{pmatrix} \Rightarrow \\
\psi_{i} &= 
U_{i-2} \left(\frac{\lambda}{2} \right) \psi_2 - 
U_{i-3} \left(\frac{\lambda}{2} \right) \psi_1. \label{PBCevec}
\end{align}
The eigenvalue equations at sites $i = 1, n$, which are interpreted as boundary conditions, provide two additional relationships between $\psi_n$, $\psi_{1}$ and $\psi_2$,
\begin{align}
\alpha_n \psi_n &= \lambda \psi_1 - \psi_2 \label{PBCSiteOne} \\
\beta_n \psi_1 &= \lambda \psi_n - \psi_{n-1}.\label{PBCSiteTwo}
\end{align}
Substituting \cref{PBCevec} into \cref{PBCSiteOne,PBCSiteTwo}, we find
\begin{align}
\begin{pmatrix}
\alpha_n U_{n-2}\left(\frac{\lambda}{2} \right) + 1& -\alpha_n U_{n-3}\left(\frac{\lambda}{2} \right)-\lambda\\
-U_{n-1}\left(\frac{\lambda}{2} \right) & \beta_n +U_{n-2} \left(\frac{\lambda}{2} \right)
\end{pmatrix} \begin{pmatrix}
\psi_2 \\
\psi_1
\end{pmatrix} = 0. \label{PBCStepTwo}
\end{align}
A nontrivial eigenvector must have $(\psi_2, \psi_1)^T \neq 0$; thus, the determinant of the $2 \times 2$ matrix in \cref{PBCStepTwo} must vanish. The following identity for Chebyshev polynomials, 
\begin{align}
U_{a-1}(x) U_{b}(x) - U_{b-1}(x) U_a(x) = U_{a-b-1}(x),
\end{align} 
simplifies the emergent determinant equation, resulting in the characteristic polynomial,
\begin{align}
\gls{det} \left(\lambda I - \mathscr{H}_n\right) &= -\lambda U_{n-1} \left(\frac{\lambda}{2} \right) + \alpha_n \lambda + U_{n-2} \left(\frac{\lambda}{2} \right) + \alpha_n \beta_n U_{n-2}\left(\frac{\lambda}{2} \right) + \beta_n \\
&= U_{n} \left(\frac{\lambda}{2} \right) + \alpha_n \beta_n U_{n-2}\left(\frac{\lambda}{2} \right) + \alpha_n + \beta_n, \label{PBCcharPoly}
\end{align}
which is a special case of \cite[eq. (3.10)]{YUEH2008}. Special cases with explicit eigenvalues are given in \cite[Thm. 3.4]{YUEH2008}. A simpler way of deriving \cref{PBCcharPoly} would be to apply the linearity property of determinants to the columns indexed by $1$ and $n$, which reduces the problem to \cref{Toeplitz}.

If the matrix on the left hand side of \cref{PBCStepTwo} is identically zero for a given eigenvalue, then this eigenvalue has geometric multiplicity of two.
\demo
\end{ex}

\begin{ex}
Consider a tridiagonal matrix, $\alpha_n = \beta_n = 0$, where the remaining elements are $p$-periodic. Explicitly, given $p \in \gls{Z+}$, then 
\begin{align}
(i,j \in \dbrac{1,n-1} )\wedge (i \equiv j \mod p) \,&\Rightarrow \, \alpha_i = \alpha_j, \beta_i = \beta_j \\
(i,j \in \dbrac{1,n} )\wedge (i \equiv j \mod p)\,&\Rightarrow \, z_i = z_j.
\end{align}
The characteristic polynomial was originally solved in \cite{Elsner1967,Rzsa1969}. The solution was later rediscovered in \cite{Beckermann1995}, applied to the eigenvalue problem in the $p = 2$ case in \cite{Gover1994}, and applied to the eigenvalue problem for the general $p$ case in \cite{Gilewicz1999}. The matrix $\mathscr{H}_n$ for the $p = 2$ case coincides with the Hamiltonian of the Su-Schrieffer-Heeger model \cite{SSH}.

The remaining text in this example displays a derivation of the characteristic polynomial for the $p = 2$ case, presenting the result in \cref{2periodicCharPoly}.
A linear recurrence relation for the characteristic polynomial,
\begin{align}
D_n\left(
\begin{pmatrix} \alpha_1 \\ \alpha_2 \end{pmatrix},
\begin{pmatrix} \beta_1 \\ \beta_2 \end{pmatrix},
\begin{pmatrix} z_1 \\ z_2 \end{pmatrix}
\right) := \det(\lambda \mathbb{1} - H_n(\alpha,\beta, z)),
\end{align}
can be derived using linearity properties of determinants. Using the shorthand 
\begin{align}
\mathscr{D}_n &:= D_n\left(
\begin{pmatrix} \alpha_1 \\ \alpha_2 \end{pmatrix},
\begin{pmatrix} \beta_1 \\ \beta_2 \end{pmatrix},
\begin{pmatrix} z_1 \\ z_2 \end{pmatrix}
\right) \\
\tilde{\mathscr{D}}_n &:= D_n \left(
\begin{pmatrix} \alpha_2 \\ \alpha_1 \end{pmatrix},
\begin{pmatrix} \beta_2 \\ \beta_1 \end{pmatrix},
\begin{pmatrix} z_2 \\ z_1 \end{pmatrix}
\right),
\end{align}
we find
\begin{align}
\mathscr{D}_n &= 
(\lambda-z_1) \tilde{\mathscr{D}}_{n-1} - \alpha_1 \beta_1  \mathscr{D}_{n-2} \\
&= 
\left((\lambda-z_1)(\lambda - z_2) - \alpha_1 \beta_1 - \alpha_2 \beta_2 \right) \mathscr{D}_{n-2} + \alpha_2 \beta_2 \left( \mathscr{D}_{n-2} - (\lambda - z_1) \tilde{\mathscr{D}}_{n-3}
\right) \\
&= 
\left((\lambda-z_1)(\lambda - z_2) - \alpha_1 \beta_1 - \alpha_2 \beta_2 \right) \mathscr{D}_{n-2} -\alpha_1 \beta_1 \alpha_2 \beta_2 \mathscr{D}_{n-4}. \label{DnRecurrence}
\end{align}
The recurrence relation of \cref{DnRecurrence} could be considered an alternative definition of $D_n$ when supplanted with the initial conditions 
\begin{align}
\mathscr{D}_{2} &= (\lambda - z_1)(\lambda - z_2) - \alpha_1 \beta_1 \\
\mathscr{D}_{1} &= (\lambda - z_1) \\
\mathscr{D}_{0} &= 1\\
\mathscr{D}_{-1} &= 0.
\end{align}
In particular, the recurrence relation allows for a self-consistent definition of $D_{-2}$,
\begin{align}
\mathscr{D}_{-2} &= -\frac{1}{\alpha_1 \beta_1}.
\end{align}


The recurrence relation can be rewritten using matrix powers. Letting $n = 2k + r$ with $k \in \mathbb{N}$ and $r \in \{0,1\}$, assume $\alpha_1 \beta_1 \alpha_2 \beta_2 \neq 0$, and let 
\begin{align}
Q = \frac{(\lambda - z_1)(\lambda-z_2) - \alpha_1 \beta_1 - \alpha_2 \beta_2}{2 \sqrt{\alpha_1 \beta_1 \alpha_2 \beta_2}}.
\end{align}
We find
\begin{align}
\begin{pmatrix}
\mathscr{D}_n \\
\mathscr{D}_{n-2}
\end{pmatrix}
&= \begin{pmatrix}
(\lambda - z_1)(\lambda-z_2) - \alpha_1 \beta_1 - \alpha_2 \beta_2 & -\alpha_1 \beta_1 \alpha_2 \beta_2 \\
1 & 0
\end{pmatrix}^k 
\begin{pmatrix}
\mathscr{D}_r \\
\mathscr{D}_{r-2}
\end{pmatrix} \\
&= (\alpha_1 \beta_1 \alpha_2 \beta_2)^{k/2} \begin{pmatrix}
U_{k}(Q) & -(\alpha_1 \beta_1 \alpha_2 \beta_2)^{1/2} U_{k-1}(Q) \\
(\alpha_1 \beta_1 \alpha_2 \beta_2)^{-1/2} U_{k-1}(Q) & U_{k-2}(Q)
\end{pmatrix} 
\begin{pmatrix}
\mathscr{D}_r \\
\mathscr{D}_{r-2}
\end{pmatrix}.
\end{align}

Consequently,
\begin{align}
\begin{pmatrix} 
\det(\lambda \mathbb{1} - \mathscr{H}_{2k+1}) \\
\det(\lambda \mathbb{1} - \mathscr{H}_{2k}) 
\end{pmatrix} &= 
(\alpha_1 \beta_1 \alpha_2 \beta_2)^{k/2} 
\begin{pmatrix}
(\lambda - z_1) U_k(Q) \\
U_k(Q) + \sqrt{\frac{\alpha_2 \beta_2}{\alpha_1 \beta_1}} U_{k-1}(Q)
\end{pmatrix}. \label{2periodicCharPoly}
\end{align}
\demo
\end{ex}

\begin{ex}
Consider the case of $\mathscr{H}_n$ with 2-periodic, Hermitian off-diagonal elements perturbed at its four corners. Explicitly, $z_i = z_{i} (\delta^i_{1} + \delta^i_{n})$, and if $i, j \in \dbrac{1,n-1}$ satisfy $i \equiv j \mod 2$, then
$\alpha_i = \alpha_{j}, \beta_i = \alpha^*_i$. Analysis and a brief literature review for the special case $\gamma = 0, \alpha_i = \beta_i = 1$ is presented in \cref{nontrivialBoundaryExample}.
The characteristic equation in the case $\alpha_n = \beta_n$ was known to \cite[p. 64]{da2007characteristic}; the general result is presented in \cref{charPolyTable} \cite{Barnett2023}. For simplicity, denote 
\begin{align}
Q &= \frac{\lambda^2 - |\alpha_1|^2 - |\alpha_2|^2}{2 |\alpha_1 \alpha_2|}
\end{align}
\begin{table*}[htp!]
\centering
\begingroup
\setlength{\tabcolsep}{8pt} 
\renewcommand{\arraystretch}{2} 
\begin{tabular}{|l|l|}
\hline
Constraints & $|\alpha_1 \alpha_2|^{-\floor{n/2}} \text{det} (\lambda I - \mathscr{H}_n)$ \\
\hhline{|=|=|}
$n = 2k $ & 
$\begin{array}{l}
U_{k}(Q) + \left(\dfrac{z_1 z_n - \alpha_n \beta_n}{|\alpha_2|^2} \right)U_{k-2}(Q) \\
+ \left(\dfrac{|t_2|^2 - \lambda (z_1 + z_n) + z_1 z_n - \alpha_n \beta_n}{|\alpha_1 \alpha_2|} \right) U_{k-1}(Q)  - \dfrac{\alpha_n + \beta_n}{|\alpha_2|} 
\end{array} $ \\
\hline
$n = 2k +1$ & 
$\begin{array}{l}
\left(\lambda - z_1 - z_n\right) U_k(Q) -(\alpha_n + \beta_n) \\
+  \left(\dfrac{\lambda(z_1 z_n - \alpha_n \beta_n)- z_1 |\alpha_1|^2 - z_n |\alpha_2|^2}{|\alpha_1 \alpha_2|} \right)U_{k-1}(Q)
\end{array}.$ \\
\hline
\end{tabular}
\endgroup
\caption{Characteristic polynomial of the Hamiltonian, $\mathscr{H}_n$, of a Su-Schrieffer-Heeger chain with non-Hermitian perturbed corners.}
\label{charPolyTable}
\end{table*}

The eigenvectors corresponding to an eigenvalue are computed as a function of their first two entries,
\begin{align}
\begin{pmatrix}
\psi_{2k} \\
\psi_{2k-1}
\end{pmatrix} &= \begin{pmatrix}
\frac{\lambda}{\alpha^*_1} & -\frac{\alpha_2}{\alpha^*_1} \\
1 & 0
\end{pmatrix} \begin{pmatrix}
\psi_{2k-1} \\
\psi_{2k-2}
\end{pmatrix} \\
&= \begin{pmatrix}
\frac{\lambda}{\alpha^*_1} &- \frac{\alpha_2}{\alpha^*_1} \\
1 & 0
\end{pmatrix}\begin{pmatrix}
\frac{\lambda}{\alpha^*_2} & -\frac{\alpha_1}{\alpha^*_2} \\
1 & 0
\end{pmatrix} \begin{pmatrix}
\psi_{2k-2} \\
\psi_{2k-3}
\end{pmatrix}\\
&= \begin{pmatrix}
\frac{\lambda^2}{\alpha^*_1 \alpha^*_2}-\frac{\alpha_2}{\alpha^*_1} & -\frac{\lambda \alpha_1}{\alpha_1^* \alpha^*_2} \\
\frac{\lambda}{\alpha^*_2} & -\frac{\alpha_1}{\alpha^*_2}
\end{pmatrix}^{k-1} \begin{pmatrix}
\psi_{2} \\
\psi_{1}
\end{pmatrix} \\
&= \left(\sqrt{\frac{\alpha_1 \alpha_2}{\alpha_1^* \alpha_2^*}}\right)^{k-1}\begin{pmatrix}
\frac{\lambda^2-|\alpha_2|^2}{|\alpha_1 \alpha_2|} & -\frac{\lambda \alpha_1}{|\alpha_1 \alpha_2|} \\
\frac{\lambda \alpha_1^*}{|\alpha_1 \alpha_2|} & -\frac{|\alpha_1|}{|\alpha_2|}
\end{pmatrix}^{k-1} \begin{pmatrix}
\psi_{2} \\
\psi_{1}
\end{pmatrix} \\
&= 
\left(\sqrt{\frac{\alpha_1 \alpha_2}{\alpha_1^* \alpha_2^*}}\right)^{k-1}
\left( U_{k-2}\left(Q \right)
\begin{pmatrix}
\frac{\lambda^2-|\alpha_2|^2}{|\alpha_1 \alpha_2|} & -\frac{\lambda \alpha_1}{|\alpha_1 \alpha_2|} \\
\frac{\lambda \alpha_1^*}{|\alpha_1 \alpha_2|} & -\frac{|\alpha_1|}{|\alpha_2|}
\end{pmatrix}
- U_{k-3}\left(Q \right) 
\begin{pmatrix}
1 & 0 \\
0 & 1
\end{pmatrix} \right)
\begin{pmatrix}
\psi_{2} \\
\psi_{1}
\end{pmatrix} \\
&= 
\left(\sqrt{\frac{\alpha_1 \alpha_2}{\alpha_1^* \alpha_2^*}}\right)^{k-1}
\begin{pmatrix}
U_{k-1}(Q) + \frac{|\alpha_1|}{|\alpha_2|} U_{k-2}(Q) & -\left(\frac{\lambda \alpha_1}{|\alpha_1 \alpha_2|}\right) U_{k-2}(Q) \\
\left(\frac{\lambda \alpha_1^*}{|\alpha_1 \alpha_2|}\right) U_{k-2}(Q) & -\frac{|\alpha_1|}{|\alpha_2|} U_{k-2}(Q) - U_{k-3}\left(Q \right) 
\end{pmatrix} 
\begin{pmatrix}
\psi_{2} \\
\psi_{1}
\end{pmatrix} \label{SSHevecAnsatz}
\end{align}
The relationship between $\psi_1$ and $\psi_2$ is found inserting \cref{SSHevecAnsatz} into the eigenvalue equation at site $1$
\begin{align}
\frac{\lambda \psi_1 - z_1 \psi_1 -\alpha_1^* \psi_2}{\alpha_n}= 
\begin{cases}
 \left(\sqrt{\frac{\alpha_1 \alpha_2}{\alpha_1^* \alpha_2^*}}\right)^{k-1} 
 \left[\begin{array}{l
 }\left(U_{k-1}(Q) + \frac{|\alpha_1|}{|\alpha_2|} U_{k-2}(Q) \right)\psi_2 \\
 -\left(\frac{\lambda \alpha_1}{|\alpha_1 \alpha_2|}\right) U_{k-2}(Q) \psi_1 
 \end{array}\right]
 &\text{if } n = 2k \\
 \left(\sqrt{\frac{\alpha_1 \alpha_2}{\alpha_1^* \alpha_2^*}}\right)^{k} \left[\begin{array}{l}
 \left(\frac{\lambda \alpha_1^*}{|\alpha_1 \alpha_2|}\right) U_{k-1}(Q) \psi_2 \\
 - \left(\frac{|\alpha_1|}{|\alpha_2|} U_{k-1}(Q) + U_{k-2}\left(Q \right) \right) \psi_1 
 \end{array}\right]
 &\text{if } n = 2k+1
\end{cases},
\end{align}
thereby determining the eigenvector modulo normalization. If, for a given eigenvalue $\lambda$, the above identity holds trivially due to vanishing prefactors of $\psi_1$ and $\psi_2$, then the eigenvector corresponding to that $\lambda$ is doubly degenerate.



The characteristic polynomial and eigensystem for the special case $\alpha_1 = \alpha_2 = 1$ can be simplified to the results presented in \cref{nontrivialBoundaryExample}, by observing $Q = T_2\left(\frac{\lambda}{2} \right)$ and using the composition identity of \cite[Lemma 3]{Zhang2004}\cite[Thm. 5]{Brandi2020}
\begin{align}
U_{mk-1}(x) = U_{k-1} \left(T_m (x) \right) U_{m-1} \left(x \right). \label{chebyshevComposed}
\end{align}

%
\demo
\end{ex}

\begin{ex} \label{su(2)tridiagExample}
The matrices considered in this example are displayed in \cref{generalsl2CElement}, which appears after a brief digression into the representations of $\mathfrak{isu}(2)$. This example considers the structure of $\mathfrak{isu}(2)$, which is a Lie algebra over the reals with the generators $\{L_1,L_2,L_3\}$ whose Lie bracket satisfies
\begin{align}
[L_i,L_j] = \sum_{k = 1}^3 i \epsilon_{ijk} L_k,
\end{align}
where $\epsilon_{ijk}$ denotes the completely antisymmetric tensor, as defined in \cref{antisymmetric-tensor}. Particularly illuminating in understanding $\mathfrak{isu}(2)$ are its representations, which are linear maps $\pi_n:\mathfrak{isu}(2) \to \mathfrak{M}_n(\mathbb{C})$ that map the Lie bracket into the commutator. 
Denote
\begin{align}
S_i := \pi_n(L_i).
\end{align}
This example will take $\pi_n$ to be the \textit{spin} $s$ representation, where $n = 2s + 1$ and $2s \in \mathbb{N}$. In this representation, there exists a basis such that $S_3$ is a diagonal matrix, 
\begin{align}
S_3 e_j = (j-s-1) e_j &\quad& \forall j \in \dbrac{1,2s+1}.
\end{align}
A general algebra element of $\pi_n(\mathfrak{isu}(2))$ is a tridiagonal matrix of the form 
\begin{align}
\vec{a} \cdot \vec{S} = H_n\left(\alpha, \beta, z\right),
\end{align}
where $\vec{a} \in \mathbb{R}^3$ and
\begin{align}
\alpha &= \left(\frac{a_1 + i a_2}{2}\right) \sum_{m = -s+1}^{s} c_{s,m} e_{m+s} \\
\beta &= \left(\frac{a_1 - i a_2}{2}\right)\sum_{m = -s+1}^{s} c_{s,m} e_{m+s} \\
z &= a_3 \sum_{m = -s}^s m e_{m+s+1} ,
\end{align}
where $c_{s,m}$ is the ladder coefficient given by 
\begin{align}
c_{s,m} = \sqrt{(s+m)(s+1-m)}.
\end{align}
A diagonal similarity transform $\mathcal{D}$ of \cref{triDiagToSymmetric} maps $\vec{a} \cdot \vec{S}$ into
\begin{align}
\mathcal{D}^{-1} (\vec{a} \cdot \vec{S}) \mathcal{D} = H_n(\alpha',\beta',z'), \label{diagSimTransformSL2C}
\end{align}
where 
\begin{align}
\alpha' &= \left(\frac{a_1 + i a_2}{2}\right) \sum_{j = 1}^{n-1} j e_{j} \\
\beta' &= \left(\frac{a_1 - i a_2}{2}\right)\sum_{j = 1}^{n-1} (n-j) e_{j} \\
z' &= a_3 \sum_{m = -s}^s m e_{m+s+1}.
\end{align}

Consider an element of the spin $s$ representation of the complexified Lie algebra,
\begin{align}
A := \vec{a} \cdot \vec{S} + i \vec{b} \cdot \vec{S} \in \mathfrak{sl}(2,\mathbb{C}) = \mathfrak{su}(2) \oplus \mathfrak{isu}(2), \label{generalsl2CElement}
\end{align}
where $\vec{a}, \vec{b} \in \mathbb{R}^3$. Matrices which are of this form, or which can be mapped into this form via \cref{diagSimTransformSL2C}, include the \textit{Sylvester-Kac matrices} \cite{sylvester1854theoreme,Kac1947,Clement1959}, their generalizations \cite[\S 576]{muir2003treatise}, and the model of \cite{Graefe2008}. 

Given a unit vector $\vec{c}$ satisfying $\vec{c} \cdot \vec{a} = 0$ and $\vec{c} \cdot \vec{c} = 1$, a corresponding similarity transformation which maps $A$ to another tridiagonal matrix in the representation of the complexified Lie algebra is given by the Rodrigues rotation formula,
\begin{align}
e^{i \theta \vec{c} \cdot \vec{S}/2} (\vec{a} \cdot \vec{S}) e^{- i \theta \vec{c} \cdot \vec{S}/2} = \cos(\theta) (\vec{a} \cdot \vec{S}) + \sin(\theta) (\vec{a} \times \vec{c}) \cdot \vec{S} + (\vec{a} \cdot \vec{c}) (\vec{c} \cdot \vec{S}) (1-\cos(\theta)).
\end{align}
By using a suitable similarity transform which maps $A$ to a multiple of $S_3$, its spectrum can be obtained,
\begin{align}
\sigma(A) = \left(\sqrt{\vec{a}\cdot \vec{a} - \vec{b} \cdot \vec{b}} \right) \dbrac{-s,s}.
\end{align}
$A$ is defective when $\vec{\alpha} \cdot \vec{\alpha} = \vec{\beta} \cdot \vec{\beta} > 0$, where $n$ eigenvalues coalesce to the zero mode. Unlike the generic case of a linear matrix function for which $n$ eigenvalues coalesce, the perturbative expansion for the eigenvalues at the exceptional point can be done with a second-order Puiseux series. However, if the Hamiltonian at the exceptional point is perturbed by a matrix with no connection to the algebra $\mathfrak{sl}(2, \mathbb{C})$, then the order of the Puiseux series can increase back to $n$ \cite{Hodaei2017}.
\demo
\end{ex}


\section{Nearest Neighbour Impurities in an Open Chain} \label{nearestNeighbour}
\subsection{Model}
Intuitively, the Hamiltonian studied in this section is $\mathcal{PT}$-symmetric, local to a one-dimensional open chain with an even number of sites, irreducible, and has a non-Hermitian pair of impurities potentials at the center of the chain. As an example, the Hamiltonian for $n = 6$ is 
\begin{align}
\mathscr{H}_6 = \begin{pmatrix}
\text{Re}(z_1) & t_5^*  & 0        & 0   	  & 0  		 & 0 \\
t_1 & \text{Re}(z_2) & t_4^*  & 0   	  & 0   	 & 0 		\\
0        & t_2 & \Delta + \mathfrak{i} \gamma & t_3^* & 0 	     & 0 		\\
0        & 0        & t_3 & \Delta - \mathfrak{i} \gamma & t_2^*  & 0 		\\
0        & 0   		& 0   	   & t_4 & \text{Re}(z_2) & t_1^*  \\
0  & 0  		& 0   	   & 0   	  & t_5 & \text{Re}(z_1)
\end{pmatrix}.
\end{align}
Explicitly, this section studies the tridiagonal case of $\mathscr{H}_n$ with the parametric restrictions
\begin{align}
\alpha_n &= \beta_n = 0 \nonumber\\
n &= 2m \nonumber\\
z_m &= z^{*}_{m+1} \nonumber \\
z_j &= z^*_{n-j+1} \in \mathbb{R} &\quad& \forall j \in \dbrac{1, m-1} \nonumber\\
t_j &:= \alpha_j = \beta^*_{n-j} &\quad& \forall j \in \dbrac{1, n-1} \nonumber \\
t_j t^*_{n-j} &>0. \label{nearestNeighborParams}
\end{align}
The last constraint is not necessary for $\mathcal{PT}$-symmetry; it is instead enforced to help ensure reality of the spectrum.
To simplify select equations, I employ the parametrization
\begin{align}
(\Delta, \gamma) = (\text{Re}(z_m), \text{Im}(z_m)).
\end{align}
The parameter $\Delta$ will be referred to as \textit{detuning}.

A graphical representation of a model corresponding to this choice of parameters is displayed in \cref{openChainFig}.

\begin{figure}[htp!] 
\centering
\includegraphics[width = \textwidth]{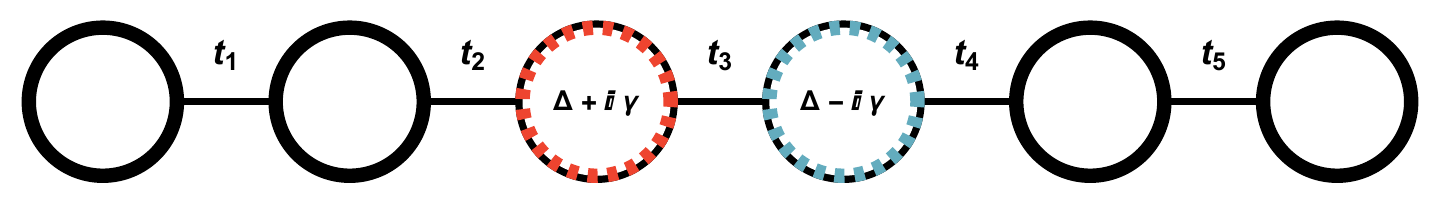}
\caption{Graphical representation of the model studied in this section for $n = 6$. References which study this model include \cite{Babbey,MyFirstPaper}.} \label{openChainFig}
\end{figure}

\subsection{Results}

Pseudo-Hermiticity is addressed in \cref{pseudoHermOpenChainSection}. The $\mathcal{PT}$-unbroken and broken domains can be compactly summarized as 
\begin{align}
|\gamma|\leq |t_m| \, &\Rightarrow \, \sigma(H) \subseteq \mathbb{R} \\
|\gamma|> |t_m| \, &\Rightarrow \, \sigma(H) \cap \mathbb{R} = \emptyset.
\end{align}
A one-parameter family of Hermitian intertwining operators is constructed. This family includes a positive-definite metric operator in the case $|\gamma|<|t_m|$, implying the reality of the spectrum and the existence of a $\mathcal{C}$-symmetry in the sense of \cref{defn:C}. Corresponding similar Hermitian Hamiltonians are constructed in \cref{Equivalent Hamiltonian Section}. Section~\ref{maximalPTOpenChain} will establish that \textit{every} eigenvalue has an imaginary part when $|\gamma| > |t_m|$. Consequently, $|\gamma| = |t_m|$ is an exceptional point. The exceptional point is examined in \ref{PTBreakingOpenChainSection}. Proposition~\ref{EP2OpenChain} states that at the exceptional point, there are exactly $m$ eigenvectors corresponding to eigenvalues with algebraic multiplicity two. 

A tool which simplifies subsequent calculations is the diagonal similarity transformation which maps tridiagonal matrices into transpose symmetric counterparts, given in \cref{triDiagToSymmetric}.

\subsection{Pseudo-Hermiticity} \label{pseudoHermOpenChainSection}

\begin{proposition} \label{pseudoHermOpenChainThm}
An intertwining operator, $M$, for the matrix $\mathscr{H}_n$ of \cref{Hscr}, satisfying the constraints of \cref{nearestNeighborParams}, is determined by the matrix elements
\begin{align}
M(Z) &= \mathcal{D} \begin{pmatrix}
\mathbb{1}_m & \frac{Z^*}{|t_m|} \mathcal{P}_m \\
\frac{Z}{|t_m|} \mathcal{P}_m & \mathbb{1}_m
\end{pmatrix} \mathcal{D}, \label{homomorphismMetric} 
\end{align}
where  $\mathbb{1}_m$ is the $m \times m$ identity matrix, $Z$ is a constant with arbitrary real part and $\text{Im}(Z) = \gamma$, and $\mathcal{D}$ is a Hermitian diagonal similarity transform defined in \cref{triDiagToSymmetric}. $M$ is positive-definite when $|Z|< |t_m|$. $M$ is the only intertwining operator for $H$ which is a sum of the identity matrix and an antidiagonal matrix. 
\end{proposition}
\begin{proof}
Analysis can be performed by assuming transpose symmetry of $\mathscr{H}_n$ without loss of generality, due to the similarity transform $\mathcal{D}$ and the metric mapping identity of \cref{MetricMapper}. 

The proof proceeds by induction. For $n=2$, the most general intertwining operator was determined in \cref{2x2Intertwiner} and \cite{wang20132}. Up to a normalization factor, this intertwining operator is 
\begin{equation}
M = \begin{pmatrix}
1 & \frac{Z^*}{|t_m|}\\
\frac{Z}{|t_m|} & 1
\end{pmatrix}.
\end{equation}

Suppose $M$ is an intertwining operator for $m = m_0$. Then the operator of \cref{homomorphismMetric} for $m_0+1$ satisfies
\begin{align}
\sum^{m_0}_{j=2} M_{i\,j} H_{jk} = \sum^{m_0}_{j=2} H^{\dag}_{ij} M_{j\,k}, &\quad& \forall \,\, i,k \in \dbrac{2,n-1}.
\end{align} 
Assuming the intertwining operator is a sum of diagonal and anti-diagonal parts, the quasi-Hermiticity condition for the remaining elements implies
\begin{align}
\begin{pmatrix}
M_{1,1} & M_{1,n} \\
M_{n,1} & M_{n,n}
\end{pmatrix} &= 
\begin{pmatrix}
M_{2,2} & M_{2,n-1} \\
M_{n-1,2} & M_{n-1,n-1}
\end{pmatrix}.
\end{align}
Note $M$ is positive-definite if and only if $|Z|^2 < |t_m|^2$.
\end{proof}

Proposition~\eqref{pseudoHermOpenChainThm} was demonstrated in \cite{Barnett_2021,Barnett2023}. Previous literature determined the special case with $n = 2$ \cite{mosta2003equivalence,wang20132} and the case with $\gamma_i = 0, t_j = t \,\, \forall j \in \dbrac{2, n-1}$ \cite{Znojil2009}. The deeper reason for why an elegant expression for an intertwining operator exists in this case has to do with representation theory, and will be revisited in \cref{section:PseudoHermFromRep}.

Importantly, $M$ is a block matrix when written over parity-symmetric blocks. A physical implication is that in quasi-Hermitian theories with metric $M$, the operators $A_{ij} = \sum_{k \in S} $\gls{delta}$ \delta^{i}_{k}$ are quasi-Hermitian observables whenever the collection of sites $S$ is parity symmetric (so if $i \in S$, $\bar{i} \in S$). When second quantized, this operator takes the role of a number operator determining the number of particles in the lattice subset $S$. 

\subsection{Equivalent Hermitian Hamiltonian} \label{Equivalent Hamiltonian Section}
This section concerns the domain of quasi-Hermiticity, $|Z|^2 < |t_m|^2$. Following \cref{Williams1969Corollary}, Hermitian Hamiltonians, $h$, which are similar to $\mathscr{H}_n$ are determined via the unique positive square root associated to a metric operator, $M(Z)$. Defining
\begin{align}
\zeta &:= \sqrt{1-|Z|/|t_m|} + \sqrt{1+|Z|/|t_m|},
\end{align}
this square root is
\begin{align}
\Omega &:= M^{1/2} \\
&= \begin{pmatrix}
\frac{\zeta}{2} \mathbbm{1}_m & \frac{Z_m^*}{\zeta |t_m|} \mathcal{P}_m \\
\frac{Z_m}{\zeta |t_m|} \mathcal{P}_m & \frac{\zeta}{2} \mathbb{1}_m
\end{pmatrix}=\frac{\zeta}{2}\mathbbm{1}_{2m}+\frac{1}{\zeta |t_m|}\left(\text{Re}(Z) \sigma_x+\text{Im}(Z)\sigma_y\right)\otimes\mathcal{P}_m.
\end{align}
A one-parameter family of Hermitian Hamiltonians similar to $\mathscr{H}_n$ are
\begin{align}
h(Z)_{ij} &:= (\Omega \mathscr{H}_n \Omega^{-1})_{ij} \\
&= \sum^{n-1}_{i=1} \left(t'_i \delta^j_{i+1} + {t'_i}^*\delta^i_{j+1}\right) + \sum^n_{i = 1} \left(\text{Re}(z)_i \delta^i_{j} \right)\\
t'_i &:= \begin{cases}
\sqrt{|t_i t_{n-i}|} & \, \text{if }i \neq m \\
\dfrac{\zeta |t_m| + i \gamma \sqrt{|t_m|^2 - |Z|^2}}{Z} & \,\text{if }i = m
\end{cases}\label{equivHermHam} 
\end{align}
The equivalent Hamiltonian of \cref{equivHermHam} remains tridiagonal and is interpreted as \textit{local} to a one-dimensional chain. This is in contrast with the generic case, where a local quasi-Hermitian operator is similar to a nonlocal Hermitian operator \cite{Korff2008}. In addition, the spectrum is a function of only $t'_i$ and $z_i$. The equivalent hopping amplitudes are real valued, $t'_{i} = t'_{n-i} \in \mathbbm{R}$, for the Hamiltonian $h(i \gamma)$.

There are two linearly independent intertwining operators generated by \cref{homomorphismMetric}: $M(i \gamma)$ and
\begin{align}
M' := \mathcal{D} \mathcal{P}_n \mathcal{D}.
\end{align}
Consequently, a symmetry\footnote{An expanded discussion of the topic of $\mathcal{C}-$symmetry can be found in the introduction chapter, specifically in \cref{pseudoHermOperatorsSection}.} of the Hamiltonian is given by \cite{BiOrthogonal}
\begin{align}
\mathcal{C} = \frac{1}{\sqrt{|t_m|^2 - \gamma^2}} \mathcal{D}^{-1} \mathcal{P} \mathcal{D}^{-1} M(i \gamma) = \frac{1}{\sqrt{|t_m|^2-\gamma^2}} \mathcal{D}^{-1} \left(i \gamma \sigma_z \otimes \mathcal{P}_m + |t_m| \mathcal{P}_n \otimes\mathcal{P}_m\right) \mathcal{D}
\end{align}

Due to $\mathcal{C}$-symmetry and \cref{tridiagNondegen}, the energy eigenvectors of $\mathscr{H}_n$ are elements of the eigenspaces, $V_{\pm}$, of $\mathcal{C}$, 
\begin{align}
V_{\pm} = \mathcal{D}^{-1} \left(\text{span} \left\{(i \gamma \pm \sqrt{t^2 - \gamma^2})  e_j +t e_{\bar{j}} : j \in \dbrac{1, m} \right\} \right),
\end{align}
where $(e_i)_j = \delta^i_{j}$. The coalescence of the eigenspaces $V_{\pm}$ as $\gamma^2 \nearrow |t_m|^2$ is a signature of an exceptional point.

\subsection{Maximally Broken $\mathcal{PT}$-Symmetry} \label{maximalPTOpenChain}
This section presents two proofs that every eigenvalue has a nonzero imaginary part in the $\mathcal{PT}$-unbroken domain $|\gamma| > |t_m|$, generalizing the result of \cite{MyFirstPaper} to the case with a real parity symmetric potential. 

The first proof follows from considering the characteristic polynomial of $\mathscr{H}_n$. Let $p_A$ denote the monic characteristic polynomial of a matrix $A$ and let $\mathscr{H}_i$ be the submatrix formed by taking the first $i$ rows and columns of $\mathscr{H}_n$. Equivalently, using the notation of \cref{minorsDef},
\begin{align}
\mathscr{H}_i := \mathscr{H}_{\dbrac{1,i} \dbrac{1,i}}.
\end{align} 
The linearity property of determinants implies a decomposition for $p_{\mathscr{H}_n}$,
\begin{align}
p_{\mathscr{H}_n}(\lambda) = (\gamma^2-|t_m|^2) p_{\mathscr{H}_{m-1}}^2(\lambda) + \left(\Delta p_{\mathscr{H}_{m-1}}(\lambda) + p_{\mathscr{H}_m}(\lambda) \right)^2. \label{charPolyDecomp}
\end{align}
When $\gamma^2  > |t_m|^2$, $p_H$ is the sum of squares of two monic polynomials of degree $m$. Suppose a real eigenvalue exists: $\lambda \in \sigma(\mathscr{H}_n) \cap \mathbb{R}$. This eigenvalue must be a simultaneous root of $p_{\mathscr{H}_{m-1}}$ and $p_{\mathscr{H}_{m}}$. Since $\mathscr{H}_n$ is non-degenerate \cite{elliott1953characteristic}, the unique eigenvector associated to $\lambda$ must be a linear multiple of the nonzero vector given in \cref{psi(l)}. Since $\lambda$ is a simultaneous root of $p_{\mathscr{H}_{m-1}}$ and $p_{\mathscr{H}_{m}}$, and the components of $\psi$ are given by \cref{psi(l)}, the components $\psi_m(\lambda)$ and $\psi_{m+1}(\lambda)$ vanish. Thus, the eigenvalue equations imply $\psi$ is identically equal to zero: a contradiction! 

A second proof follows directly from the eigenvalue equations. This proof is presented in \cite{Barnett2023} and was performed in the case with zero detuning in \cite{MyFirstPaper}. For simplicity, I will implicitly perform the similarity transform to the transpose symmetric case given by $\mathscr{D}$.
Suppose a given eigenvalue is real: $\lambda \in \sigma(\mathscr{H}_n) \in \mathbbm{R}$. Without loss of generality, $\psi_1$ can be assumed to be real-valued. The eigenvalue equations imply 
\begin{align}
\psi_j \in \mathbb{R} &\quad& \forall j \in \dbrac{1,m}
\end{align}
for all sites $j$ in the left half of the lattice. 
Since $\mathscr{H}_n$ is non-degenerate \cite{elliott1953characteristic}, the corresponding eigenstate, $\psi$, is $\mathcal{PT}$-unbroken.
Due to unbroken $\mathcal{PT}$-symmetry, there exists a phase $\chi \in \mathbb{R}$ such that 
\begin{align}
\psi_{\overline{j}} e^{\mathfrak{i} \chi} = \psi_j &\quad&\forall j \in \dbrac{1,m}.
\end{align}
With these observations in mind, the eigenvalue equations at sites $(m,m+1)$ are equivalent to
\begin{equation}
\begin{pmatrix}
(z_m - \lambda) \psi_m + |t_{m-1}| \psi_{m-1} & |t_m| \psi_m \\
|t_m| \psi_m & (z_m^* - \lambda) \psi_m + |t_{m-1}| \psi_{m-1}
\end{pmatrix}
\begin{pmatrix}
1 \\ e^{i \chi}
\end{pmatrix} 
=
0.
\end{equation}
For this matrix to have a nontrivial kernel, its determinant must vanish; thus
\begin{align}
(|t_{m-1}| \psi_{m-1} + (\Delta - \lambda) \psi_m)^2 + (\gamma^2 - |t_m|^2) \psi_m^2 = 0.
\end{align}
However, if $|\gamma| > |t_m|$, the determinant is strictly positive. The contradictory assumption was taking $\lambda \in \mathbbm{R}$; hence, every eigenvalue is complex for $|\gamma| > |t_m|$. 

\subsection{Exceptional Points} \label{PTBreakingOpenChainSection}
This section will demonstrate that $|\gamma| = |t_m|$ is a second-order exceptional point corresponding to a defective Hamiltonian. Owing to the existence of a positive \textit{semi-}definite intertwining operator given by \cref{homomorphismMetric}, I will prove that the spectrum is real and each eigenvalue has algebraic multiplicity equal to two and geometric multiplicity equal to one. 

\begin{proposition} \label{EP2OpenChain}
When $|\gamma| = |t_m|$, $H$ has exactly $m$ mutually orthogonal eigenvectors corresponding to real eigenvalues, $\sigma(H) \subset \mathbb{R}$. Each eigenvalue has algebraic multiplicity equal to two and geometric multiplicity equal to one. 
\end{proposition}

\begin{proof}
First, we prove that $H$ has at most $m$ linearly independent eigenvectors. I know of at least two ways of doing this. One is a direct application \cref{realEigenvalsThm} of \cite{Drazin1962}, choosing the intertwiner of \cref{homomorphismMetric} for $S$. A second proof follows from the decomposition of the characteristic polynomial of $H$ given in \cref{charPolyDecomp}. When $\gamma^2 = |t_m|^2$, the characteristic polynomial is the square of a monic polynomial of degree $m$, so each eigenvalue has an algebraic multiplicity of at least two. Since $H$ is non-degenerate \cite{elliott1953characteristic}, there are at most $m$ linearly independent eigenvectors.


To prove that $H$ has exactly $m$ eigenvectors when $\gamma^2 = |t_m|^2$, consider the action of $\mathcal{D}^{-1} H \mathcal{D}$ on $\ker M$. An orthonormal basis of $\ker M$ is  
\begin{align}
\ker M &= \text{span} \{\tilde{e}_j\,|\,j \in \dbrac{1, m} \} \\
\tilde{e}_j &= \frac{\mathfrak{i} \gamma}{\sqrt{2} |t_m|} e_j + \frac{1}{\sqrt{2}} e_{\bar{j}}.
\end{align}
Then
\begin{align}
H \tilde{e}_j &= \begin{cases}
t_{1} \tilde{e}_2 & \text{ if } j = 1\\
t_{j-1} \tilde{e}_{j-1} + t_{j} \tilde{e}_{j+1} & \text{ if } j \in \dbrac{2,  m-1} \\
t_{m-1} \tilde{e}_{m-1} & \text{ if } j = m
\end{cases}. \label{tildeH}
\end{align}
Thus, $\ker M$ is an invariant subspace of $H$. Define $\tilde{H}:\ker M \to \ker M$ by the condition $\tilde{H}(v) = H(v)$ for all $v \in \ker M$. Equation~\eqref{tildeH} implies that $\tilde{H}$ is Hermitian, so it has $m$ orthogonal eigenvectors whose corresponding eigenvalues are real. Since these eigenvectors are also of eigenvectors $H$, $H$ has at least $m$ eigenvectors whose corresponding eigenvalues are real.
\end{proof}

\section{Su-Schrieffer-Heeger Model with Edge Defects} \label{SSHSection}

\subsection{Model}

This model studies a non-Hermitian perturbation of the \textit{Su-Schrieffer-Heeger} model \cite{SSH} with open boundary conditions. 

The Su-Schrieffer-Heeger (SSH) model was invented to model electrical conductivity in a doped polyacetylene polymer chain \cite{SSH}. The model is a tight-binding model describing hopping on a one-dimensional chain with alternating bond strengths. The ratio of bond strengths dictates the phase of the system. When the outer bond strength is larger than the inner bond strength, the system is in the topologically nontrivial phase, marked by the existence of edge states, and the system behaves as a topological insulator \cite{Batra2020}. When the inner bond strength is larger than the outer bond strength, there are no edge states and the system is in the topologically trivial phase. Since its inception, the topological phases of the SSH model have been experimentally realized \cite{Meier2016}.

Expanding upon the works of \cite{Lieu2018,Klett2017,Zhu2014,Jin2017,Yao2018,Turker2019,Mochizuki2020}, this section examines a non-Hermitian perturbation of the SSH chain. To expand upon previous work, a detuning potential is introduced on sites with a non-Hermitian defect at the edges of the chain. Since the edge states are localized to sites with a non-Hermitian impurity potential, the phases dictated by topology and $\mathcal{PT}$-symmetry are intertwined: in the thermodynamic limit, the $\mathcal{PT}$-unbroken phase is the topologically trivial phase.

Intuitively, the Hamiltonian studied in this section is $\mathcal{PT}$-symmetric, is local to a one-dimensional open chain, has a non-Hermitian pair of defect potentials at the edges of the graph, and has 2-periodic hopping. Explicitly, this section studies the tridiagonal case of $\mathscr{H}_n$ with the parametric restrictions
\begin{align}
z &= z_{\{1,n\}} \nonumber\\
t_j &:= \alpha_j = \beta_{j} &\quad& \forall j \in \dbrac{1, n-1} \nonumber \\
t_j &= t_k &\quad& \text{if } i \equiv j \mod 2. \label{SSHParams}
\end{align}
To simplify select equations, I employ the parametrization
\begin{align}
(\Delta, \gamma) = (\text{Re}(z_1), \text{Im}(z_1)).
\end{align}
I will also define $(t_L, t_R) = (\alpha_n,\beta_n)$.
For $n = 6$, the Hamiltonian  is
\begin{align}
\mathscr{H}_6 = \begin{pmatrix}
z_1 & t_1  & 0        & 0   	  & 0  		 & t_L \\
t_1 & 0 & t_2  & 0   	  & 0   	 & 0 		\\
0        & t_2 & 0 & t_1 & 0 	     & 0 		\\
0        & 0        & t_1 & 0 & t_2  & 0 		\\
0        & 0   		& 0   	   & t_2 & 0 & t_1  \\
t_R  & 0  		& 0   	   & 0   	  & t_1 & z_n
\end{pmatrix}.
\end{align}

A graphical representation of a model corresponding to this choice of parameters is displayed in \cref{SSHChainFig}.

\begin{figure}[!ht] 
\centering
\includegraphics[width = .75\textwidth]{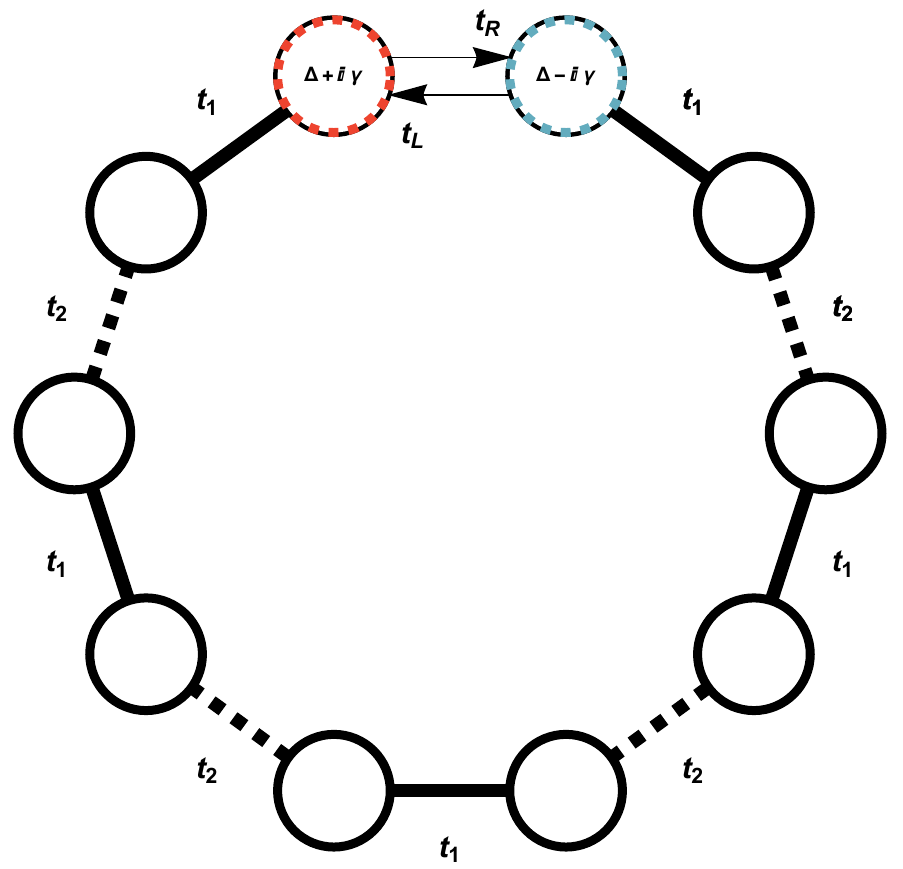}
\caption{Graphical representation of the model studied in this section, an SSH chain with non-Hermitian defects on the boundary for $n = 10$.} \label{SSHChainFig}
\end{figure}

\subsection{Reducible Case}

Before addressing the general case, I will briefly comment on the reducible case $t_1 t_2 = 0$, with $n \geq 4$ even.

If $t_1 = 0$, then the Hamiltonian decomposes into a direct sum of $2\times 2$ blocks over pairs of adjacent sites. The spectrum is 
\begin{align}
\sigma(\mathscr{H}_{n}) = \left\{-t_2, t_2, \frac{z_1 + z_n \pm \sqrt{4 t_L t_R + (z_1 - z_n)^2}}{2} \right\}.
\end{align}
An exceptional point corresponding to a defective Hamiltonian occurs when $(z_1 - z_n)^2 = -4 t_L t_R$.

If $t_2 = 0$, then the Hamiltonian block decomposes into one $4\times 4$ block and $l-2$ $2\times 2$ blocks. A subset of the spectrum is $\{-t_1, t_1 \} \subset \sigma(H)$. In the case $z_1 = z^*_n$ and $t_L t_R \in \mathbb{R}$, the results from \cref{nearestNeighbour} imply that the spectrum is purely real when $|\gamma| < \sqrt{|t_L t_R|}$. An exceptional point occurs at $|\gamma| = \sqrt{|t_L t_R|}$, where the sum of geometric multiplicities of the eigenvalues of $\mathscr{H}_n$ is $n-2$. There are four eigenvalues with a nonzero imaginary part when $|\gamma| \geq \sqrt{|t_L t_R|}$.

Thus, the remaining treatment of the perturbed SSH chain will assume $t_1 t_2 \neq 0$ without loss of generality.

\subsection{Eigenvalue Inclusion Results}

This section is devoted to finding subsets of the complex plane which contain some or all of the eigenvalues of $H$. As a consequence, we will find a subset of the $\mathcal{PT}-$unbroken and $\mathcal{PT}$-broken domains.  

A subset of the $\mathcal{PT}$-unbroken domain is found by applying the intermediate value theorem to the characteristic polynomial. To simplify results, we define
\begin{align}
\mu_k=| t_1+ t_2 e^{(2i\pi/n)k}|\geq 0
\end{align}
and denote the intervals with endpoints $(-1)^{s_1} (t_1 + (-1)^{s_2} t_2)$ and $(-1)^{s_1} \text{sgn} (t_1 + (-1)^{s_2} t_2) \mu_1$ with $s_1, s_2 \in \{0,1\}$  as $I( (-1)^{s_1}, (-1)^{s_2} )$.

\begin{proposition} \label{inclusionTheorem}
Consider an even chain with  $n = 2k$  and assume $t_L = - t_R$. If $t_2^2=z_1 z_n - t_L t_R$, then
\begin{align}
\sigma(\mathscr{H}_n) = \left\{ \pm \mu_j\,|\, j \in \dbrac{1, k-1} \right \}\cup\left\{\frac{z_1 + z_n}{2}\pm \sqrt{t_1^2 - t_2^2 + \left(\frac{z_1 + z_n}{2} \right)^2}\,\right\}. \label{exactSpectrum}
\end{align}
If $t_2^2 \neq z_1 z_n - t_L t_R$, then the intervals $(\mu_{j+1}, \mu_j)$ and $(-\mu_j, -\mu_{j+1})$ each contain one eigenvalue of $\mathscr{H}_n$ for $j \in \dbrac{1, k-1}$. Constraining other parameters as specified below guarantees the existence of additional real eigenvalues in corresponding intervals, 
\begin{align}
1 + k \left(1 \pm_2 \frac{t_2}{t_1} \right)\frac{\left(\Delta \mp_1 t_2\right)^2 + \gamma^2 -t_L t_R}{t_2^2 - \Delta^2 - \gamma^2 + t_L t_R}\geq 0 \, &\Rightarrow \, \sigma(\mathscr{H}_n) \cap I(\pm_1 1,\pm_2 1) \neq \emptyset \label{ineq}. 
\end{align}
\end{proposition}

\begin{proof}
We utilize an alternative expression to the characteristic polynomial. Using the identity
\begin{align}
U_{n\pm 1}(x) &= x U_n(x) \pm T_{n + 1}(x),
\end{align}
where $T_n(x)$ is the Chebyshev polynomial of the first kind, $T_n(x) = \cos (n \arccos x)$, we can rewrite the characteristic polynomial, with $n=2k$, as
\begin{align}
(t_1 t_2)^{-k} \text{det} (\lambda I - H) &= T_k(Q) \left(1- \frac{z_1 z_n - t_L t_R}{t_2^2}  \right) - \frac{t_L + t_R}{t_2} \nonumber \\
&+ \left(1 + \frac{z_1 z_n - t_L t_R}{t_2^2}\right)Q U_{k-1}(Q) \nonumber \\
&+ \left(\frac{t_2^2 - \lambda (z_1 + z_n) + z_1 z_n - t_L t_R}{t_1 t_2}  \right) U_{k-1}(Q).
\label{altPoly}
\end{align}
Equation~\eqref{exactSpectrum} follows from the observation that the first term of \cref{altPoly} vanishes when $z_1 z_n - t_L t_R = t_2^2$ and the second one vanishes when $t_R=-t_L$. 

Next, consider the case with $z_1 z_n - t_L t_R \neq t_2^2$ and $n\geq 4$. The sign of the characteristic polynomial is different at the endpoints of the intervals $(\mu_j, \mu_{j+1})$ and $(-\mu_{j+1},-\mu_j)$ for all $j \in \dbrac{ 1, k-2 }$, so there exists a real eigenvalue of $H$ inside each of these intervals. The inequalities of \cref{ineq} follow from considering the sign of the characteristic polynomial at the endpoints of the intervals $I(\pm_1 1, \pm_2 1)$.
\end{proof}

The inequalities of \cref{ineq} when $t_1^2=t_2^2$ were known to  \cite[eq. (63-64)]{Willms2008}. Special cases of \cref{exactSpectrum} were presented in \cite{rutherford1948xxv,Willms2008,Korff2008,guo2016solutions}. Two of the inequalities \cref{ineq} are satisfied if $t_2^2 \geq z_1z_n-t_Lt_R=\Delta^2 + \gamma^2-t_L t_R$. All  four inequalities are satisfied if, additionally, $t_1 \geq t_2$. Thus, $\sigma(H) \subset \mathbb{R}$ when $\Delta^2 + \gamma^2 - t_L t_R \leq t_2^2 \leq t_1^2$. 

Now, we focus on the complex part of the spectrum. A subset of the $\mathcal{PT}$-broken domain is found as an application of the Brauer-Ostrowski ovals \cref{BrauerOstrowski} \cite{Brauer1947,Ostrowski1937,Varga2004}. Let $C(w_1,w_2; b) = \{w \in \mathbb{C} \,|\, |w - w_1|  \cdot |w-w_2| \leq b\}$ denote a Cassini oval. By the Brauer-Ostrowski ovals theorem, all eigenvalues of $H$ are elements of the union of  Cassini ovals, specifically
\begin{align}
\sigma(H) \subseteq &C(0,0;(t_1+t_2)^2) \cup C(0,z_n; (t_1+t_2)(t_1+|t_R|)) \cup \nonumber \\
&C(z_1, 0; (t_1+t_2)(t_1+|t_L|)) \cup C(z_1,z_n;(t_1+|t_L|)(t_1+|t_R|)). \label{Cassini}
\end{align}
Since eigenvalues are  continuous in the arguments of a continuous matrix function \cite{Kato1995}, if the union of the Cassini ovals in \cref{Cassini} contains disjoint components, then each component contains at least one eigenvalue of $H$. In particular, if both of the inequalities
\begin{align} 
\left[|z_1|^2-(t_1+t_2)^2\right] \gamma  &> |z_1| (t_1+|t_L|)(t_1+|t_R|), \label{unbrokenIneq1} \\
2(r_1+t_2) &< \left(|z_1| + \sqrt{|z_1|^2 - 4t_1^2 - 4 \min\{|t_L|^2, |t_R|^2\} }\right)\label{unbrokenIneq2}
\end{align}
hold, then there exist disjoint components containing the points $z_1$ and $z_n=z_1^*$, implying the existence of at least two eigenvalues with nonzero imaginary parts. 



\subsection{Topological Phases}
The eigenvectors of tight-binding models are characterized as either \textit{bulk} or \textit{edge} states based on how their inverse participation ratio scales with the chain size $n$~\cite{JoglekarSaxena,Joglekar2011z}. Roughly, the bulk eigenstates  are spread over most of the chain irrespective of the chain size, whereas edge states remain exponentially localized within a few sites even with increasing chain size. Observing
\begin{align}
U_n(Q) = \frac{(Q + \sqrt{Q^2 - 1})^{n+1} - (Q - \sqrt{Q^2 - 1})^{n+1}}{2 \sqrt{Q^2 - 1}},
\end{align}
we see the sequence of Chebyshev polynomials $U_n(Q)$ is oscillatory in $n$ for $|Q|\leq 1$ and scales exponentially with $n$ when $\min_{q \in [-1,1]} |Q - q| = \mathcal{O}(1)$, where asymptotics are defined as $n \rightarrow \infty$. Thus, the non-trivial  topological phase with edge-localized states is equivalent to existence of  eigenvalues which do not satisfy $|Q|\leq 1$. Therefore, for this particular model, in the thermodynamic limit, the $\mathcal{PT}$-broken phase is equivalent to  topologically nontrivial phase, as complex eigenvalues correspond to edge states.

\subsection{Surface of Exceptional Points (Open Chain)} \label{EP Surface Section}

Consider the SSH Hamiltonian with detuned defects $z_1 = z_n^*$ as a function of three dimensionless parameters, $\Delta/t_2, \gamma/t_2,$ and $t_1/t_2$. In this case, the set of exceptional points form a 2-dimensional surface, with ridges that correspond to third-order exceptional points.  At $\Delta=0$, these ridges intersect, giving rise to cusp singularities corresponding to fourth-order exceptional points. 


\begin{figure}[!ht]
   \centering
   \begin{subfigure}[b]{0.45\textwidth}
      \centering
      \includegraphics[width = \textwidth]{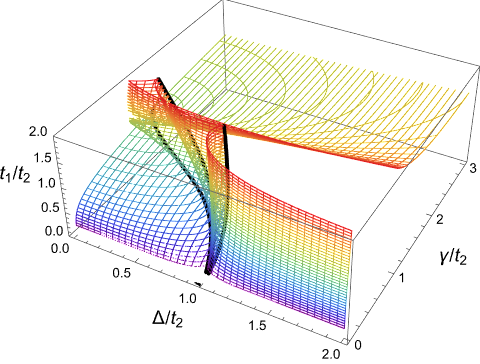}
      \caption{}
   \end{subfigure}
   \hfill
   \begin{subfigure}[b]{0.45\textwidth}
      \centering
      \includegraphics[width = \textwidth]{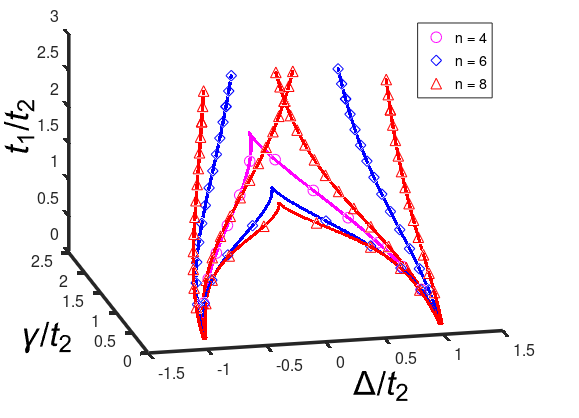}
      \caption{}
   \end{subfigure}
   \caption{(a) Exceptional points (EPs) of an open SSH chain with end-defects ($m = 1$, $\alpha_n = \beta_n = 0$) and $n = 8$. The domain satisfying the inequalities of \cref{ineq} lies in between this EP2 surface and the $\gamma = 0$ plane. Ridges of this surface correspond to EP3s, plotted in \cref{thirdOrderFig}(b). As one passes through this surface along a ray originating at $\gamma = 0$, we pass from the $\mathcal{PT}-$unbroken region to the $\mathcal{PT}$-broken region. The $\mathcal{PT}$-broken region is a subset of the topologically nontrivial phase, marked by the existence of edge states.
(b) Contours of EP3s from \cref{thirdOrderFig}(a) show cusp singularities at $\Delta = 0$ and are fourth-order exceptional points (EP4s).}
   \label{thirdOrderFig}
\end{figure}

\section{Uniform Chain} \label{uniformSection}


This section addresses a model with open boundary conditions, uniform hopping, and a pair of defect-potentials at parity-symmetric sites. This corresponds to the Hamiltonian $\mathscr{H}_n$ with the parametric restrictions
\begin{align}
\alpha &= \beta = t \textbf{e}_{\dbrac{1,n-1}} \neq 0, \nonumber \\
z &= z_{m} e_{m} + z_{\bar{m}} e_{\bar{m}}, \label{uniformParameters}
\end{align} 
where I remind the reader that $\textbf{e}$ was defined in \cref{bold e vector}. For $n = 6$ and $m = 2$, the Hamiltonian  is
\begin{align}
\mathscr{H}_6 = \begin{pmatrix}
0 & t  & 0        & 0   	  & 0  		 & 0 \\
t & z_2 & t  & 0   	  & 0   	 & 0 		\\
0        & t & 0 & t & 0 	     & 0 		\\
0        & 0        & t & 0 & t  & 0 		\\
0        & 0   		& 0   	   & t & z_5 & t  \\
0  & 0  		& 0   	   & 0   	  & t & 0
\end{pmatrix}.
\end{align}
For this case, the characteristic polynomial and eigenvector ansatz were determined in example~\ref{Ortega} \cite{ortega2019mathcal}. In the notation in that example, this section is the special case $m := m_1 = \overline{m_2}$. Introducing the dimensionless defect strength
\begin{align}
z'_i := z_i/t
\end{align}
will simplify subsequent formulas. A similar simplification follows from using the rescaled characteristic polynomial,
\begin{align}
P_{n,m}(x; z') :&= \gls{det}(2x I - H/t) \nonumber\\
&= U_n\left(x\right) - (z'_m + z'_{\overline{m}}) U_{n-m}\left(x\right)  U_{m-1}\left(x\right) \nonumber \\&+ z'_m z'_{\overline{m}} U_{n-2m}\left(x\right) U_{m-1}\left(x\right)^2. \label{charPolySpecific}
\end{align} 
A root of the rescaled characteristic polynomial derives an eigenvalue via $x/2 \in \sigma(\mathscr{H}_n)$.

Due to \cref{tridiagNondegen} \cite[Thm. 4.2]{elliott1953characteristic}, there are no diabolical points.

As in previous sections, I will introduce the parametrization
\begin{align}
\Delta &= z_m + z_{\overline{m}} \\
\gamma &= \mathfrak{i}(z_{\overline{m}} - z_m).
\end{align}


\subsection{Exact Eigenvalues} \label{ExactEvalues}

Special cases of the Hamiltonian have a closed-form spectrum, as summarized in \cref{closedFormEvalsTable}. 

\begin{table*}[!ht]
\centering
\begingroup
\setlength{\tabcolsep}{8pt} 
\renewcommand{\arraystretch}{1.5} 
\begin{tabular}{|l|c|l|l|}
\hline 
Source & $m$ & Constraint on $z_i$ & Eigenvalues of $\mathscr{H}_n$, $\sigma(\mathscr{H}_n)$ \\
\hhline{|=|=|=|=|}
\hline 
$\begin{array}{l}\text{\cite[eq.(40)]{Willms2008}} \end{array}$
& 1 
& $\begin{array}{l} z_1 z_{n} = t^2 \end{array}$  
& $\begin{array}{l}\left\{
2 t \cos\left(\dfrac{j \pi}{n}\right), j \in \dbrac{1,n-1} \right\} \\\cup
\left\{z_1 + z_n \right\}
\end{array}
$\\[10pt]
\hline 
$\begin{array}{l}\text{\cite[eq.(8)]{vEgervry1928}} \end{array}$
& 1 
& $\begin{array}{l} z = 0 \end{array}$  
& $\begin{array}{l}\left\{
2 t \cos\left(\dfrac{j \pi}{n+1}\right), j \in \dbrac{1,n} \right\}
\end{array}
$\\[10pt]
\hline 
$\begin{array}{l} \text{\cite[p.230]{rutherford1948xxv}} \\ \text{\cite[p.28]{elliott1953characteristic}} \end{array}$
& 1 
& $\begin{array}{l}
z_{1,n} = \pm t\\ z_{n,1} = 0 \end{array}$  
& $\begin{array}{l} \left\{
\mp 2 t \cos\left(\dfrac{2j \pi}{2n+1}\right), j \in \dbrac{1,n} \right\} 
\end{array}$\\[10pt]
\hline 
$\begin{array}{l} \text{\cite[p.239]{Rutherford1952}} \\ \text{\cite[p.28]{elliott1953characteristic}} \end{array}$
& 1 
& $\begin{array}{l} z_{1,n} = t,\\ z_{n,1} = -t \end{array}$  
& $\begin{array}{l} \left\{
2 t \cos\left(\dfrac{(2j-1) \pi}{2n}\right), j \in \dbrac{1,n} \right\} \end{array}$\\[10pt]
\hline 
$\begin{array}{l} \text{\cite[p.3]{Babbey}} \end{array}$
& $\dfrac{n}{2}$ 
& $\begin{array}{l}
z_m = z^*_{m+1} \\
\,\,\,\,\,\,\,\,\,= \pm \mathfrak{i} t \end{array}$ & $\begin{array}{l}
\left\{
2 t \cos \left(\dfrac{j \pi}{m+1}\right), j \in \dbrac{1,m}
\right\}
\end{array}$\\[10pt]
\hline 
\end{tabular}
\endgroup
\caption{
This table lists cases of tridiagonal matrices whose eigenvalues have closed-form expressions. In all cases, the matrix elements are given by \cref{tridiagElements} with $\alpha = \beta = t \textbf{e}_{\dbrac{1,n-1}}, z = z_{\{m,\overline{m}\}}$. If $m = n/2$, $n$ is even. The case shown in the first row can be derived by applying the similarity transform of \cref{triDiagToSymmetric} to the result in \cite[\S 14]{Ledermann1954}. Special cases of the first row were determined by physicists \cite{rutherford1948xxv,Korff2008,guo2016solutions}. Any case where there are less than $n$ distinct roots is a second-order exceptional point, corresponding to a defective matrix $\mathscr{H}_n$. The eigenvectors can be computed using the ansatz of \cite{Babbey} or by calculating characteristic polynomials \cite{Gantmacher2002}.
}
\label{closedFormEvalsTable}
\end{table*}

\addtocounter{table}{-1}

\begin{table}[!ht]
\centering
\begingroup
\setlength{\tabcolsep}{8pt} 
\renewcommand{\arraystretch}{1.5} 
\begin{tabular}{|l|c|l|l|}
\hline 
Source & $m$ & Constraint on $z_i$ & Eigenvalues of $\mathscr{H}_n$, $\sigma(\mathscr{H}_n)$ \\
\hhline{|=|=|=|=|}
$\begin{array}{l}\text{\cite[table 1]{Barnett2023}} \end{array}$ 
& $\dfrac{n}{2}$ 
& $\begin{array}{l} z_m = z^*_{m+1} \\
\,\,\,\,\,\,\,\,\,= e^{\pm \mathfrak{i} \pi/3}  t
\end{array}$  
& $\begin{array}{l} \left\{ 2 t\cos \left(\dfrac{j \pi}{m+1} \right), j \in \dbrac{1,m} \right\} \cup\\ \left\{2 t \cos \left(\dfrac{(2j-1) \pi }{2m + 1}\right), j \in \dbrac{1,m}  \right\} \end{array}$\\[10pt]
\hline
$\begin{array}{l}\text{\cite[table 1]{Barnett2023}} \end{array}$ 
& $\dfrac{n}{2}$ 
&$\begin{array}{l} z_m = z^*_{m+1} \\
\,\,\,\,\,\,\,\,\,= e^{\pm 2 \pi \mathfrak{i}/3} t 
\end{array}$  
& $\begin{array}{l} \left\{2t \cos \left(\dfrac{j \pi}{m+1}\right), j \in \dbrac{1,m} \right\} \cup\\ \left\{2 t\cos \left(\dfrac{2j \pi}{2m+1}\right), j \in \dbrac{1,m} \right\} \end{array}$ \\
\hline
$\begin{array}{l}\text{\cite[table 1]{Barnett2023}} \end{array}$ 
& $\dfrac{n}{2}$ 
& $\begin{array}{l} z_m = z^*_{m+1}\\
\,\,\,\,\,\,\,\,\,= (-1\pm \mathfrak{i})t 
\end{array}$  
& $\begin{array}{l} \left\{
2 t \cos\left( \dfrac{2j \pi }{2m + 1} \right), j \in \dbrac{1,m} \right\} 
\end{array}$\\[10pt]
\hline 
$\begin{array}{l}\text{\cite[table 1]{Barnett2023}} \end{array}$ 
& $\dfrac{n}{2}$ 
& $\begin{array}{l} z_m = z^*_{m+1} \\
\,\,\,\,\,\,\,\,\,= (1\pm \mathfrak{i}) t
\end{array}$  
& $\begin{array}{l}\left\{
2 t \cos \left(\dfrac{(2j-1)\pi}{2m+1} \right), j \in \dbrac{1,m} \right\}
\end{array}
$ \\[10pt]
\hline 
\end{tabular}
\endgroup
\caption{\textit{\textbf{{(continued.)}}}
This table summarizes the tridiagonal matrix eigenvalue problems solved in \cite[table 1]{Barnett2023}, and summarized in examples~\ref{ex:cp2},\ref{ex:cp3},\ref{ex:cp4},\ref{ex:cp5}.
}
\end{table}

The remainder of this section provides detailed derivations of the final five entries of \cref{closedFormEvalsTable}. The feature that all of the entries have in common is that $n = 2m$ and $z_m = z^{*}_{m+1}$. Applying the identity
\begin{align}
U_{2m}(x) &= U_{m}^2(x) - U_{m-1}^2(x),
\end{align}
to the characteristic polynomial yields
\begin{align}
\text{det}(2x I - H/t) &= U_{m}^2(x) - 2 \text{Re}(z'_m) U_{m}\left(x\right)  U_{m-1}\left(x\right) + (|z'_m|^2-1) U_{m-1}\left(x\right)^2.
\end{align}
Thus, on the unit disk $|z'_m| = 1$, the characteristic polynomial contains a factor of $U_m(x)$. The roots of $U_m$ have a closed-form expression, given in \cref{ChebyshevSecondKindRoots}, so a corresponding subset of the eigenvalues is given by
\begin{align}
|z'_m| = 1 &\, \Rightarrow \, \left\{
2 t \cos \left(\dfrac{j \pi}{m+1}\right), j \in \dbrac{1,m}
\right\} \subseteq \sigma(H). \label{nearestNeighbourEvalueSubset}
\end{align}

\begin{ex} 
$z_m = z^*_{m+1} = \pm \mathfrak{i} t$. This example is the simplest of those presented in this section. The characteristic polynomial is $U_m^2(x)$, so the only roots are given by \cref{nearestNeighbourEvalueSubset}.
\end{ex}
The remaining examples utilize the Chebyshev polynomials of the \textit{third} and \textit{fourth} kinds, which can be defined as
\begin{align}
V_m(x) &= U_m(x) - U_{m-1}(x) \\
W_m(x) &= U_m(x) + U_{m-1}(x).
\end{align}
The roots of these polynomials are
\begin{align}
V_n\left(\cos \left(\dfrac{(2k-1)\pi}{2n+1} \right) \right) &= 0 &\quad& \forall k \in \dbrac{1,n} \\
W_n\left(\cos \left(\dfrac{2k \pi}{2n+1} \right) \right) &= 0 &\quad& \forall k \in \dbrac{1,n}.
\end{align}

\begin{ex} \label{ex:cp2}
$z_m = z^*_{m+1} = (e^{\pm \pi \mathfrak{i}/3}) t$. The characteristic polynomial factorizes as 
\begin{align}
\text{det}(2x I - H/t) = U_m(x) V_m(x).
\end{align}
\end{ex}

\begin{ex} \label{ex:cp3}
$z_m = z^*_{m+1} = (e^{\pm 2 \pi \mathfrak{i}/3}) t$. The characteristic polynomial factorizes as 
\begin{align}
\text{det}(2x I - H/t) = U_m(x)W_m(x).
\end{align}
\end{ex}

\begin{ex} \label{ex:cp4}
$z_m = z^*_{m+1} = (1\pm \mathfrak{i}) t$. The characteristic polynomial factorizes as
\begin{align}
\text{det}(2x I - H/t) = V_m(x)^2.
\end{align}
\end{ex}

\begin{ex} \label{ex:cp5}
$z_m = z^*_{m+1} = (-1\pm \mathfrak{i}) t$. The characteristic polynomial factorizes as
\begin{align}
\text{det}(2x I - H/t) = W_m(x)^2.
\end{align}
\end{ex}

\subsection{Constant Eigenvalues} \label{Constant Evalues}

The simplest functional dependence eigenvalues can have on a parameter is independence. For generic perturbations of operators, no eigenvalues remain invariant under the perturbation. However, symmetry considerations can force a subset of the eigenvalues to be a constant which is independent of the parameters. A simple example is the set of matrices with a chiral symmetry and an odd number of elements, all of which have a zero mode. 

In this section, we note if the defects in the lattice model of example~\eqref{Ortega} are placed appropriately, there exist constant eigenvalues. The set of constant eigenvalues is determined in its entirety and the characteristic polynomial admits a corresponding factorization. This generalizes the result reported in \cite{ortega2019mathcal} to cases where 
$z_{m} \neq - z_{\overline{m}}$.

For notational simplicity, denote
\begin{align}
g_j := \text{gcd}(n+1, j m),
\end{align}
where $j \in \{1,2\}$ and $\text{gcd}$ is the greatest common divisor.

\begin{theorem}
Consider the tridiagonal matrix studied in example~\eqref{Ortega}, whose elements are given by \cref{tridiagElements} with the parametric restrictions $\alpha_{\dbrac{1 n-1}} = \beta_{\dbrac{1 n-1}} = 1, \alpha_n = \beta_n = 0, z_i = z_i (\delta^i_{m} + \delta^i_{\overline{m}})$, and let $m \in \dbrac{1, \floor{n/2}}$ without loss of generality. The following properties hold:

\begin{enumerate}
\item A subset of the spectrum is independent of $z_i$, namely 
\begin{align}
\left\{2 \cos \left(\frac{\pi r}{g_1} \right) \,|\, r \in \dbrac{1, g_1-1} \right\} = \bigcap\limits_{(z_m, z_{\overline{m}}) \in \mathbb{C}^2} 
\sigma(H_n(\textbf{e}_{\dbrac{1,n-1}}, \textbf{e}_{\dbrac{1,n-1}}, z_{\{m,\overline{m}\}})). 
\label{opaqueEvals}
\end{align}
This subset is the only subset of eigenvalues which is independent of $z_i$. 
The corresponding eigenvectors are also independent of $z_i$. 

\item The remaining eigenvalues are the roots of the following polynomial of degree $n-g_1+1$,
\begin{align}
\lambda &\in \sigma(\mathscr{H}_n) \setminus \left\{2 \cos \left(\frac{\pi r}{g_1} \right) \,|\, r \in \dbrac{1, g_1-1} \right\} \Leftrightarrow \nonumber \\
0 &= U_{g_1^{-1}(n+1)-1} \left( T_{g_1} \left(\frac{\lambda}{2}\right) \right)
- (z_{m} + z_{\overline{m}}) U_{n-m} \left(\frac{\lambda}{2} \right) U_{g_1^{-1}m-1} \left( T_{g_1} \left(\frac{\lambda}{2} \right) \right)\nonumber\\
&+ z_{m} z_{\overline{m}} U^2_{g_1-1}\left(\frac{\lambda}{2} \right)U_{n-2 m} \left(\frac{\lambda}{2} \right) U_{g_1^{-1} m-1} \left( T_{g_1}\left(\frac{\lambda}{2} \right) \right).
\end{align}
\end{enumerate}
\end{theorem}
\begin{proof}
The key is to apply the Chebyshev polynomial identity \cref{chebyshevComposed} to the characteristic equation \cref{ortegaCharEq}. Every eigenvalue $\lambda \in \sigma(\mathscr{H}_n)$ satisfies


\begin{align}
0 &= U_{g_1 g_1^{-1}(n+1)-1} \left(\frac{\lambda}{2}\right) 
- (z_{m} + z_{\overline{m}}) U_{n-m} \left(\frac{\lambda}{2} \right) U_{g_1 g_1^{-1}m-1} \left(\frac{\lambda}{2} \right) \nonumber \\
&+ z_{m} z_{\overline{m}} U_{n-2 m} \left(\frac{\lambda}{2} \right) U_{g_1 g_1^{-1} m-1}^2\left(\frac{\lambda}{2} \right) \\
&= U_{g_1-1}\left(\frac{\lambda}{2} \right) U_{g_1^{-1}(n+1)-1} \left( T_{g_1} \left(\frac{\lambda}{2}\right) \right) \nonumber \\
&- (z_{m} + z_{\overline{m}}) U_{g_1-1}\left(\frac{\lambda}{2} \right) U_{n-m} \left(\frac{\lambda}{2} \right) U_{g_1^{-1}m-1} \left( T_{g_1} \left(\frac{\lambda}{2} \right) \right) \nonumber \\
&+ z_{m} z_{\overline{m}} U^2_{g_1-1}\left(\frac{\lambda}{2} \right)U_{n-2 m} \left(\frac{\lambda}{2} \right) U_{g_1^{-1} m-1}^2 \left( T_{g_1}\left(\frac{\lambda}{2} \right) \right).
\end{align}
Thus, the characteristic polynomial contains $U_{g_1-1} \left(\frac{\lambda}{2} \right)$ as a divisor for all $(z_{m}, z_{\overline{m}})\in\mathbb{C}^2$. This divisor is trivial when $g_1 = 1$. For $g_1 > 1$, the corresponding roots of $U_{g_1-1}$ are the left hand side of \cref{opaqueEvals}. As a consequence, the identity 
\begin{align}
\left\{2 \cos \left(\frac{\pi r}{g_1} \right) \,|\, r \in \dbrac{1, g_1-1} \right\} \subseteq \bigcap\limits_{(z_m, z_{\overline{m}}) \in \mathbb{C}^2} \sigma(H_n(\textbf{e}_{\dbrac{1,n-1}}, \textbf{e}_{\dbrac{1,n-1}}, z_{\{m,\overline{m}\}}))
\end{align}
holds. 

To prove these are the only eigenvalues which are constant, we use the Hellmann-Feynman \cref{HellmannFeynman}. The derivatives of the eigenvalues with respect to $z_m, z_{\bar{m}}$ must be zero. Consider first the perturbation of the eigenpair $E_k, \psi(k)$ in $z_{m}$
\begin{align}
\frac{d E_{k}}{d z_{m}} &= \braket{\psi(k)| \frac{d \mathscr{H}_n}{d z_m} \psi(k)} \\
&= |\psi_{m}(k)|^2,
\end{align}
thus,
\begin{align}
\left. \frac{d E_k}{d z_m} \right|_{z_m = 0} &= \left(U_{m-1}\left(\cos \left(\frac{k \pi}{n+1} \right) \right)\right)^2.
\end{align}
The only way this derivative can equal zero is if $\frac{k \pi}{n+1} \in \{\frac{l \pi}{m} \,|\, l \in \dbrac{1, m-1} \}$, which yields the desired roots. The calculation for the $m_2$ perturbation is identical.
\end{proof}

In the special case $z_{m} = - z_{\overline{m}}$, there are even more constant eigenvalues. Given $\lambda \in \sigma(\mathscr{H}_n)$, factorizing the characteristic polynomial yields
\begin{align}
0 &= U_{g_2 g_2^{-1}(n+1)-1} \left(\frac{\lambda}{2}\right) - z^2_{m} U_{g_2 g_2^{-1}(n + 1 -2 m) - 1} \left(\frac{\lambda}{2} \right) U_{m-1}^2\left(\frac{\lambda}{2} \right) \\
&= U_{g_2-1} \left(\frac{\lambda}{2}\right) \left(U_{g_2^{-1}(n+1)-1} \left( T_{g_2}\left(\frac{\lambda}{2}\right) \right) - z_m^2 U_{g_2^{-1}(n + 1 -2 m) - 1} \left( T_{g_2}\left(\frac{\lambda}{2} \right) \right)U_{m-1}^2\left(\frac{\lambda}{2} \right) \right),
\end{align}
which implies
\begin{align}
z_{m} = - z_{\overline{m}} \,\,\,\,\, \Rightarrow \,\,\,\,\,
\left\{2 t \cos \left(\frac{\pi r}{g_2} \right) \,|\, r \in \dbrac{1, g_2-1} \right\} \subsetneq \sigma(\mathscr{H}_n).
\end{align}

When $n$ is odd, $g_2$ is an even number which is larger than one. Consequently, zero is an constant eigenvalue. That zero is a constant eigenvalue in the case of a lattice with an odd number of sites can also be derived from the chiral symmetry displayed in \cref{tridiagStagger}.

%

\subsection{Exceptional Contours} \label{EPContours Uniform Chain}

This section will yield an analytic expression for the set of exceptional points when the non-Hermitian defects are at the edges of our one-dimensional uniform chain ($m = 1$). The characteristic polynomial in this case simplifies to \cite[eq. (2.3)]{rutherford1948xxv}
\begin{align}
P_{n,1}(x; z') :&= U_n\left(x\right) - (z'_1 + z'_{n}) U_{n-1}\left(x\right) + (z'_1 z'_{n} )U_{n-2}\left(x\right). 
\end{align} 

As summarized in \cref{closedFormEvalsTable}, there exist choices for the rescaled defect strength, $z'$, such that the spectrum of the Hamiltonian, $\mathscr{H}_n$, consists of only simple eigenvalues. Thus, the set of exceptional points of $\mathscr{H}_n$ is the algebraic curve, $C_{\text{EP}} \subsetneq \mathbb{C}$, defined by the zeroes of $P_{n,1}$ and $\dot{P}_{n,1}$,
\begin{align}
C_{\text{EP}} := \left\{(z'_1,z'_n) \in \mathbb{C}^2 \,|\, \exists x \in \mathbb{C}: P_{n,1}(x,z') = 0 \text{ and } \frac{\partial}{\partial x} P_{n,1}(x,z') = 0\right\}.
\end{align}
Using the derivatives of the Chebyshev polynomials, such as 
\begin{align}
U_n'(x) &= \begin{cases}
\dfrac{-nx U_n(x)+(n+1) U_{n+1}(x)}{1-x^2}& \text{if } x \neq \pm 1 \\
(\pm 1)^{n+1}\left( \dfrac{n(n+1)(n+2)}{3}\right)& \text{if } x = \pm 1
\end{cases},
\end{align}
the derivative of $P_n$ can be computed analytically,
\begin{align}
(1-x^2) \frac{\partial}{\partial x} P_{n,1}(x,z') &= A(x,z') P_{n,1}(x,z') + B(x,z') U_{n-1}(x,z') \\
A(x,z') :&= n \frac{z'_1 + z'_n - x (1+z'_1 z'_n)}{1-z'_1 z'_n }\\
B(x,z') :&= \left(\begin{array}{c}
2 z'_1 z'_n + (n+1)(1-z'_1 z'_n) - (z'_1 + z'_n)x\\
+ \left(\frac{n(2 x-z'_1-z'_n)}{1-z'_1 z'_n} \right) \left(2x z'_1 z'_n -z'_1-z'_n\right)
\end{array}  \right).
\end{align}
There are only three possibilities for when a simultaneous zero of $P_{n,1}$ and $\dot{P}_{n,1}$ could exist:
\begin{itemize}
\item $P_{n,1}(x,z') = U_{n-1}(x,z') = 0$. A straightforward application of the recurrence relation, \cref{Chebyshev-Recurrence}, implies $z'_1 z'_n = 1$ for this case. As summarized in \cref{closedFormEvalsTable} \cite[eq.(40)]{Willms2008}, the spectrum is exactly solvable in this case. The exceptional points are given by 
\begin{align}
z'_1 + z'_n \in \left\{2 \cos\left(\frac{j \pi}{n}\right)\,|\, j \in \dbrac{1,n-1}\right\} \, \Rightarrow \, z' \in C_{\text{EP}}.
\end{align}
\item $x = \pm 1$. The corresponding exceptional points are 
\begin{align}
z'_1 = {z'_n}^* \in \left\{\frac{2-2n^2 + \mathfrak{i} \sqrt{3 n^2 - 3}}{(2n-1)(n-1)}, \frac{2n^2-2+ \mathfrak{i} \sqrt{3 n^2 - 3}}{(2n-1)(n-1)} \right\}\,\Rightarrow\, z' \in C_{\text{EP}}.
\end{align}
\item $P_{n,1}(x,z') = B(x,z') = 0$. Denoting the zeroes of $B(\cdot, z')$ as $b_\pm$, corresponding exceptional points are determined by the zero set of 
\begin{align}
P_{n,1}(b_+,z')P_{n,1}(b_-,z') = 0 \, &\Rightarrow\, z' \in C_{\text{EP}}. \label{EP-Surface-Simple}
\end{align} 
\end{itemize}
An alternative derivation of the set of exceptional points follows from computing a resultant. Using the resultant formula for Chebyshev polynomials given in \cite[eq.~(1.2)]{Jacobs2011}\cite[Thm.~4]{Louboutin2013},
\begin{align}
&\text{Res}_x\left(U_n(x), U_m(x)\right)\nonumber\\ &= \begin{cases}
0 & \text{if } \text{gcd}(m+1,n+1) > 1 \\
(-1)^{mn/2} * 2^{mn} & \text{if } \text{gcd}(m,n) = 1
\end{cases},
\end{align} 
a somewhat lengthy derivation yields
\begin{align}
P_{n,1}(1,z')P_{n,1}(-1,z')\text{Res}_{x}\left(P_{n,1}(x,z'),\frac{dP_{n,1}(x,z')}{dx}\right)\propto
P_{n,1}(b_+,z')P_{n,1}(b_-,z'),
\end{align}
where the constant prefactor can be ignored since the only physical content contained in this resultant is whether it is or is not zero.

The exceptional points for zero detuning, which satisfy \cref{EP-Surface-Simple}, were known to \cite{Korff2008,jin2009solutions}. They are 
\begin{align}
z'_1 = {z'_n}^* = \begin{cases}
\pm \mathfrak{i} \sqrt{1+1/n} & \text{if } n \text{ is odd} \\
\pm \mathfrak{i} & \text{if } n \text{ is even}
\end{cases} \,\Rightarrow\, z' \in C_{\text{EP}}.
\end{align}
This exceptional point corresponds to a zero mode with heightened algebraic multiplicity, which is three in the odd case and two in the even case.

A real section of the exceptional contour, $\Delta, \gamma \in \mathbb{R}$, is displayed in the following \cref{PTBreaking}.

\begin{figure}[!ht]
\centering
\includegraphics[width = \textwidth]{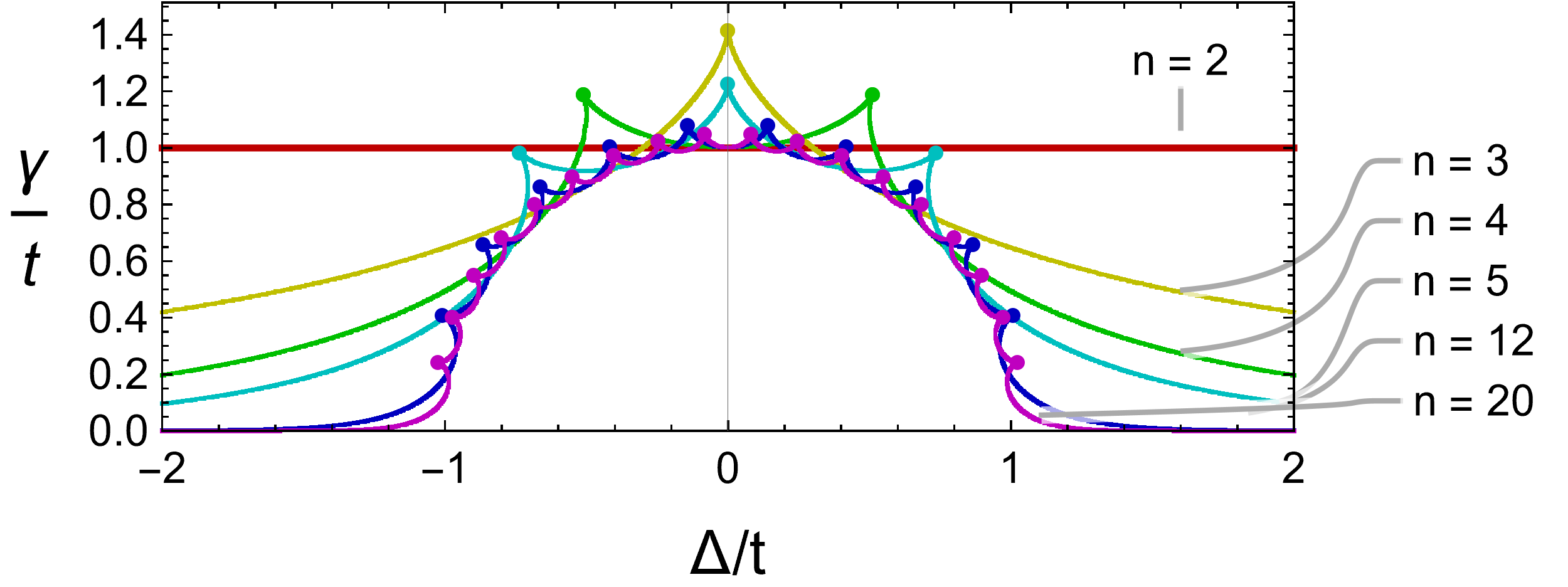}
\caption{Exceptional points (EPs) of a $\mathcal{PT}$-symmetric Hamiltonian with uniform hopping and non-Hermitian defect potentials $\Delta \pm \mathfrak{i} \gamma$ at the ends of a one-dimensional chain with open boundary conditions, uniform hopping, and variable lattice sizes $n$. $\mathcal{PT}$-symmetry breaks as the $\gamma$ is tuned above the EP contour. Additional features of these curves are summarized in \cref{epContourConjecture}. This Hamiltonian is explicitly described by $H_n$ of \cref{tridiagElements} with the parameters of \cref{uniformParameters} and $m = 1$. The curves for $n = 3,4,5$ can be found in \cite[Figure II]{Ruzicka2015}.}
\label{PTBreaking}
\end{figure}

The numerically constructed contours share several properties which I hypothesize hold for any $n$, but do not prove analytically:

\begin{conjecture} \label{epContourConjecture}
Considering each contour of exceptional points as an algebraic curve in $\mathbb{R}^2$, there are $n-2$ cusp singularities in the upper half-plane, $\Delta \in \mathbb{R}$, $\gamma > 0$. Cusp singularities correspond to third-order exceptional points, all other exceptional points are second-order. The exceptional point contour has no acnodes or crunodes. The exceptional point contours are one-to-one functions of $\text{arg} (z_1) \in (0, \pi) \cup (\pi, 2 \pi)$. As $n \rightarrow \infty$, the $\mathcal{PT}$-unbroken region approaches union of the real line, $\gamma = 0$, with the unit disk, $|z'| = 1$.
\end{conjecture}

When $m = 1$, the only singularities of the exceptional point contour are cusp singularities. As exhibited in the following plot, an example of an exceptional point contour with acnodes and crunodes is given by the uniform chain with $(n,m) = (5,2)$.

\begin{figure}[!ht]
\centering
\includegraphics[width = 80mm]{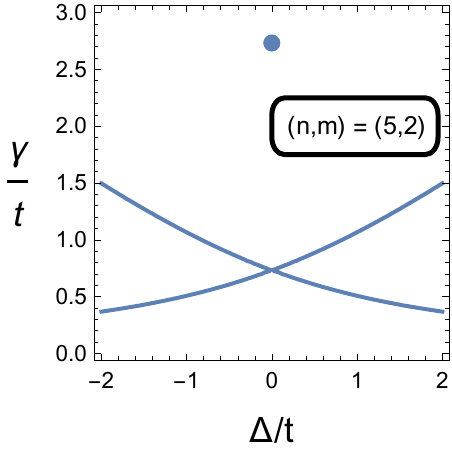}
\caption{Exceptional points for a five-site uniform chain with open boundary conditions. Explicitly, the Hamiltonian is $\mathscr{H}_5$ with $t_i = t$, $\alpha_n = \beta_n = 0$, $(n,m) = (5,2)$. The points $z_2 = (\sqrt{3} \pm 1) \mathfrak{i} t$ are a crunode $(-)$ and an acnode $(+)$, respectively, with eigenvalues $\sigma(H) = \{0, -(\pm 1)^{1/2}3^{1/4}, (\pm 1)^{1/2} 3^{1/4}\}$. The zero eigenvalue is a constant eigenvalue. The nonzero eigenvalues at the singular points $z_2 = (\sqrt{3} \pm 1) \mathfrak{i} t$ have algebraic multiplicity 2 and geometric multiplicity 1.
}
\end{figure}

\subsection{Perturbative Analysis} \label{Perturbative Section}

In this section, perturbation theory is applied to gain insight into the exceptional contour. Two features are analysed: perturbations of the eigenvalues at the exceptional points and the asymptotic behaviour of the exceptional contour in the large detuning limit.

\subsubsection{Eigenvalue Perturbations}

As a consequence of the Newton-Puiseux theorem, given an $n \times n$ matrix which is a polynomial in one parameter, $\theta$, the eigenvalues, $\lambda_i$ can be expanded as a Puiseux series in $\theta$,
\begin{align}
\lambda_i(\theta)=\lambda_i(\theta_0)+ \sum^\infty_{j = 1} \epsilon_{ij}(\theta_0) (\theta-\theta_0)^{j/k(i)},
\end{align}
where $k(i) \in \mathbb{Z}_+$. To guarantee a real spectrum in a neighbourhood of $\theta_0$, as is the case for a Hermitian matrix, the condition $k(i)=1$ is necessary and sufficient. On the other hand, if $\lambda_i(\theta_0)$ is an exceptional point with polynomial perturbative order $N$, then $k(i)=N$. The sensitivity of the spectrum to perturbations in $\theta_0$ is quantified by 
\begin{align}
\tau(\theta_0) :=\max\{\epsilon_{i1}(\theta_0)|k(i)=\sup(k(1),\cdots,k(n))\}
\end{align}
If $\theta_0$ parametrizes an exceptional contour which contains exceptional points with different values of $k$, the corresponding $\tau(\theta_0)$ must diverge as one approaches a point with an increased $k$ value. We now consider perturbations of the eigenvalues of the $\mathcal{PT}$-symmetric case of $H$ for $m = 1$ at the exceptional points, \cref{PTBreaking}. If the tangent to an algebraic curve of exceptional points is unique and 2-directional, then perturbations along the tangents to the contour have $k = 1$ and result in real eigenvalues. In the orthogonal direction, there exists exactly one pair of eigenvalues which displays a real-to-complex-conjugates transition, so $k=2$. Only at the cusp singularities is $k = 3$ satisfied. To show this, in \cref{tau} we plot the coefficient of the square-root term $\tau(\theta)$ as a function of angle in the $(\Delta,\gamma)$ plane, i.e. $\theta=\text{arg}(z_1)$, for $\theta\in[0,\pi/2]$. As is expected, $\tau(\theta)$ diverges at non-uniformly distributed cusp points. 

\begin{figure}[!ht]
\centering
\includegraphics[width =0.5\textwidth]{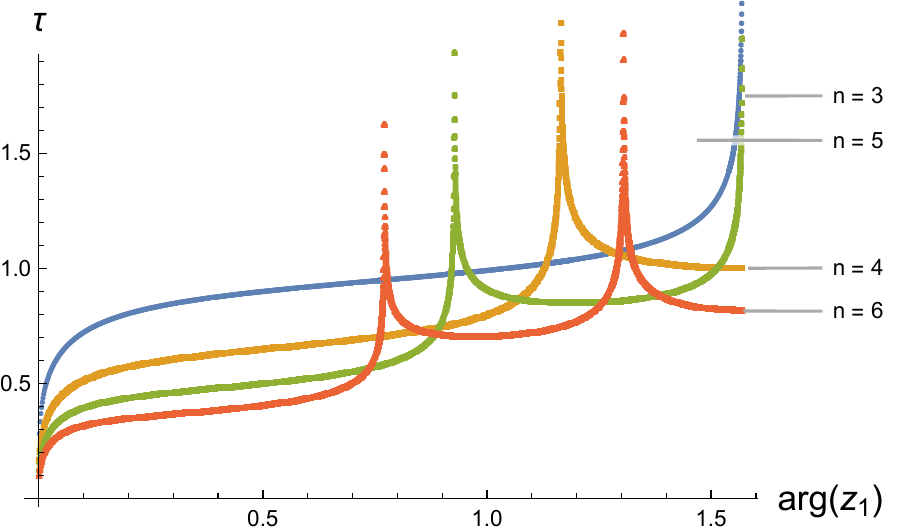}
\caption{Coefficient $\tau(\theta)$ of the square-root term in the perturbative expansion of the algebraically degenerate eigenvalue as a function of $\theta=\text{arg}(z_1)$. Since  $\tau(\theta)=\tau(\pi-\theta)=\tau(-\theta)$, the range is confined to $[0,\pi/2]$. The divergences  are signatures of third-order exceptional points where the eigenvalue expansion is expressed through cube-roots instead of square-roots.}
\label{tau}
\end{figure}

\subsubsection{Exceptional Contours for Large Detuning}

Our next perturbative result regards the behaviour of the exceptional contour in the large detuning limit. To simplify intermediate calculations, I will introduce the rescaled detuning and impurity strength, $\Delta' = \Delta/t$ and $\gamma' = \gamma/t$

A perturbative analysis of the set of exceptional points in the limit of large detuning $\Delta \to \infty$ and small defect strength $\gamma \to 0$ follows from applying the method of dominant balance, as reviewed in \cite{bender2013advanced}, to the expression $P_{n,1}(b_+) P_{n,1}(b_-)$, which appears in \cref{EP-Surface-Simple}. 

Firstly, note an approximation of the roots $b_{\pm}$ for $\Delta \gg 1$ and $\gamma \ll 1$ is 
\begin{align}
b_+ &= \frac{\Delta'}{2} + \frac{1 + (1-n) {\gamma'}^2}{2 \Delta'} + O\left(\frac{{\gamma'}^2}{{\Delta'}^3} \right) \\
b_- &= \frac{(n-1)\Delta'}{2n} + O\left(\frac{1}{\Delta'}\right).
\end{align}
Furthermore, note when $\gamma = 0$, the exact formula $b_+ = \frac{1 + {\Delta'}^2}{2 \Delta'}$ holds. Thus, applying 
\begin{align}
U_{n-1}\left(\frac{x+x^{-1}}{2} \right)&= \frac{x^{n}-x^{-n}}{x-x^{-1}}. \label{Joukowski}
\end{align} 
to $P_{n,1}(b_+)$ yields a zeroth-order approximation in $\gamma$,
\begin{align}
P_{n,1}(b_+) = ({\Delta'})^{-n}(1-({\Delta'})^2) + O\left(\frac{\gamma^2 \Delta^{n-2}}{t^n} \right).
\end{align}
The coefficient of the $\gamma^2$ term is determined by calculating a second derivative of $P_{n,1}(b_+)$, the result is
\begin{align}
P_{n,1}(b_+) = {\Delta'}^{-n}(1-\Delta^2) + {\gamma'}^2 {\Delta'}^{n-2} + O(\Delta^{-n}, \gamma^4 \Delta^{n-4}).
\end{align}
To determine the leading asymptotic behaviour of the exceptional contour, we only need to consider the leading term of $P_{n,1}(b_-)$, which results from the approximation $U_{n}(x) = 2^n x^n + O(x^{n-1}))$ for large $x$. Finally, we find
\begin{align}
P_{n,1}(b_+) P_{n,1}(b_-) = \frac{\Delta^2}{n^2 t^2}\left(1-\frac{1}{n}\right)^{n-2}\left[\frac{\gamma^2\Delta^{2(n-2)}}{t^{2n}}-1\right] + O(1,\gamma^4 \Delta^{2n-4}). \label{rhoInf}
\end{align}

It follows that the exceptional points determined by vanishing of  \cref{rhoInf} satisfy $\gamma_\text{EP}=t^{n-1}/\Delta^{(n-2)}$ in the limit $\Delta/t\gg 1$. The scaling predicted by \cref{rhoInf} is confirmed via the following log-log plot of the exceptional contour for large detuning.
\begin{figure}[!ht]
\centering
\includegraphics[width = 120mm]{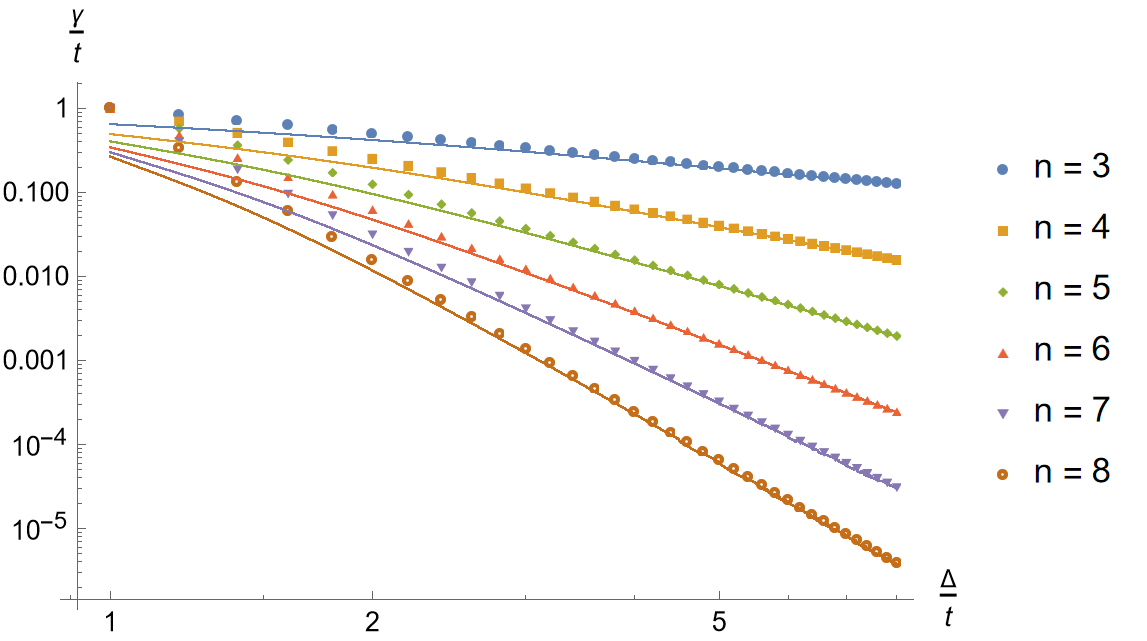}
\caption{Log-log plot of the exceptional contours (thin lines), contrasted with the predicted asymptotic behaviour $\gamma^2 \Delta^{2(n-2)}= t^n$ (dotted lines).}
\end{figure}

\subsection{Pseudo-Hermiticity} \label{pseudoHermCriticalChain}

A closed-form expression for an intertwining operator is known for $m = 1$\cite{farImpurityMetric,Ruzicka2015},
\begin{equation}
M_{ij}(z') = \begin{cases}
\,\,\,\,\,1 &  i = j \\
-\mathfrak{i} \,\frac{\gamma}{t}\, \left(\frac{\Delta - \mathfrak{i} \gamma}{t}\right)^{j-i-1} &  i < j \\
\,\,\,\,\,\mathfrak{i} \,\frac{\gamma}{t}\, \left(\frac{\Delta + \mathfrak{i} \gamma}{t}\right)^{i-j-1} &  i > j
\end{cases}. \label{not positive}
\end{equation}
The spectrum of this intertwining operator satisfies the following symmetries,
\begin{align}
\sigma(M(z')) &= \sigma(M({z'}^*)) \label{metricxReflect} \\
&= \sigma(M(-z')) \label{metricxyReflect}.
\end{align}
The first of these identities, \cref{metricxReflect}, follows from the identity $\sigma(A) = \sigma(A^T)$ which for all linear maps $A$, where $A^T$ denotes the transpose of $A$ \cite{Taussky1959}. The second, \cref{metricxyReflect}, follows from the diagonal similarity transform 
\begin{align}
M(-z')_{ij} = (-1)^{i+j} M (z'),
\end{align}
which is transformation $E_n$ encountered in \cref{diagAltSign} and \cite{kahan1966accurate,Valiente2010,Joglekar2010}.

%
While $M$ is a Hermitian intertwining operator, $M$ is not a metric operator for the entire $\mathcal{PT}$-unbroken domain due to a lack of positive-definiteness. The domain where $M$ is positive-definiteness is plotted in \cref{positive figure}. This domain is contrasted with the $\mathcal{PT}$-unbroken domain of the Hamiltonian $\mathscr{H}_n$. For zero detuning, $\Delta = 0$, in an even lattice, $n/2 \in \mathbb{Z}_+$, the numerical results displayed in \cref{positive figure,lambda1} suggest that $M(\mathfrak{i} \gamma)$ is positive-definite for all $-t < \gamma < t$. At the exceptional point $\gamma = \pm t$, $M$ is a positive semi-definite intertwiner, since the kernel of $M(\pm \mathfrak{i})$ has dimension $n-1$ and $M$ has one nonzero eigenvalue equal to $n$. 

\begin{figure}[!ht] 
\centering
\includegraphics[width = \textwidth]{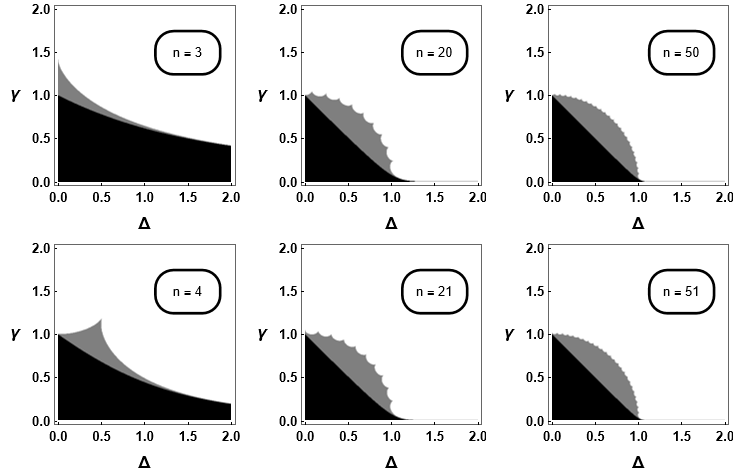} 
\caption{The domain of $z' \in \mathbb{C}$ such that $M(z')$ of \cref{not positive} is positive is indicated in black. The grey domain is the $\mathcal{PT}$-unbroken domain, for which there exists a positive-definite metric operator which is not $M(z')$. As $n \rightarrow \infty$, $M(z')$ is only positive-definite inside of a domain with a linear boundary, which is found below.}
\label{positive figure}
\end{figure}

The boundary in the complex plane for $\gamma$ for which the metric of \cref{not positive} is positive-definite is bounded by a linear region as $n \rightarrow \infty$. Our numerical search suggests that the metric is positive-definite inside this region as $n \rightarrow \infty$, a claim we state without analytic proof. An analytical expression for this linear region is determined by showing $\braket{\psi|M \psi} < 0$ in the thermodynamic limit  $n \rightarrow \infty$ for some choices of vectors $\psi$. In particular, choosing $\psi^+ = \frac{1}{\sqrt{n}} \textbf{e}$, $\psi^- = \sum_{j=1}^n \frac{(-1)^j}{\sqrt{n}} e_j$, the associated quadratic forms determined by $M$ are 
\begin{align}
\lim_{n\rightarrow \infty} \sum_{ij} \psi^+_i M_{ij} \psi^+_j &= \frac{(\Delta - 1)^2-\gamma^2}{(z-1)(z^*-1)}\\
\lim_{n\rightarrow \infty} \sum_{ij} \psi^-_i M_{ij} \psi^-_j &= \frac{(\Delta + 1)^2-\gamma^2}{(z+1)(z^*+1)}.
\end{align}
For $M$ to be positive-definite, both of the above quantities must be positive, which can only happen when $|\Delta - 1|<|\gamma|$ and $|\Delta+ 1|<|\gamma|$.

In the remainder of this section, we discuss positivity of $M$ along two contours in the complex plane: the imaginary axis and the circle $(\Delta-1/2)^2 + \gamma^2 = 1/4$. The reason why this circle was chosen was to add to the results presented in \cite{farImpurityMetric}.

%

The following figure addresses positivity of $M$ for the zero detuning case, $\Delta = 0$. 
\begin{figure}[!ht]
\centering
\includegraphics[width = .8\textwidth]{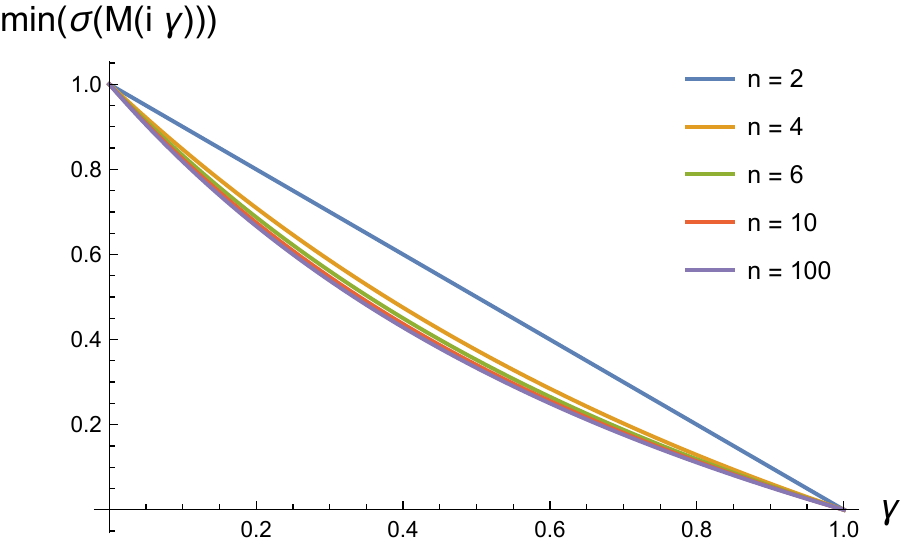}
\caption{This figure plots the smallest eigenvalue of $M(\mathfrak{i} \gamma')$ for $\gamma' \in [0,1]$. Since $\min(\sigma(M(\mathfrak{i} \gamma'))) \in [0,1]$, our numerics imply of positive-definiteness of $M$ for all $\gamma' \in (-1,1)$.} \label{lambda1}
\end{figure}

Positivity of $M$ for the case $(\Delta-1/2)^2 + \gamma^2 = 1/4$ was addressed for small $n$ in \cite{farImpurityMetric}. The following figure examines larger values of $n$.

\begin{figure}[!ht]
\centering
\includegraphics[width = \textwidth]{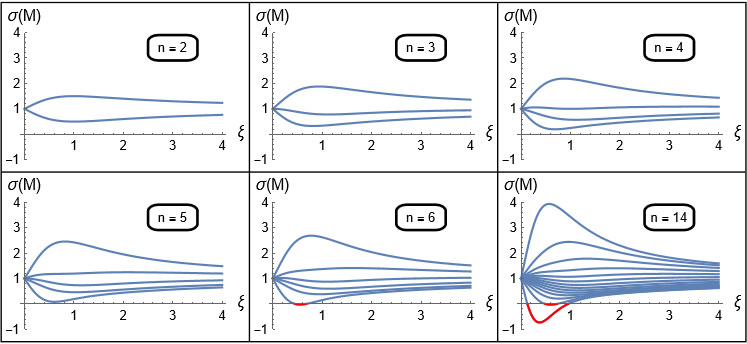}
\caption{This figure displays the spectrum of $M(\frac{1 + \mathfrak{i} \xi}{1 + \xi^2})$ for $\xi \in \mathbb{R}^+$. The spectrum is only strictly positive for $n < 6$. Positive eigenvalues are depicted in blue, while negative eigenvalues are depicted in red.} 
\end{figure}

\chapter{Tactics for Generating Pseudo-Hermitian Operators}
\label{Tactics}

Examples elucidate key features of mathematical objects. 
This chapter introduces two new techniques for generating nontrivial examples of pseudo-Hermitian operators with closed-form intertwining operators.


Given an arbitrary operator, it can be difficult to ascertain whether it is pseudo-Hermitian, and if so, what a corresponding intertwining operator would be. 
A straightforward approach to determining whether a matrix is pseudo-Hermitian is to first solve the eigenvalue problem, and then to use \cref{generalMetric} \cite[eq. (8)]{BiOrthogonal}. However, due to the Abel–Ruffini theorem \cite{topProofAbel,Ayoub1980} 
, this is often impossible without resorting to numerical or approximative procedures. Furthermore, numerical methods typically exhibit undesirable scaling with system size. The $QR$-algorithm is one such numerical technique for determining the eigensystem of an $N \times N$ matrix \cite{Francis1961,Francis1962,Kublanovskaya1962}, which has a computational complexity of $O(N^3)$. On a more extreme note, in many physical applications, such as the quantum many-body problem, the number of rows and columns of relevant matrices scale exponentially with system size. 

Even in cases where an operator is pseudo-Hermitian with a closed-form expression for the intertwining operator, by the same reasoning mentioned above, it can be difficult to determine whether this intertwining operator is positive-definite. As elaborated on in \cref{pseudoHermCriticalChain}, an example includes the matrix displayed in \cite{Ruzicka2015}. In this case, all the intertwining operators are known in closed-form, but a positive-definite intertwining operator is only known for a subset of the domain with real spectrum.

Section~\ref{anticommutPencilSection} finds an intertwining operator associated to matrices which are a linear combination of two anticommuting invertible generators. Physically, anticommuting matrices appear in models with chiral or particle-hole symmetries \cite{Chiu2016}, where one matrix is interpreted as Hamiltonian, and the other is interpreted as the symmetry. 


The second technique introduced in this chapter utilizes the representation theory of $C^*$-algebras. A minuscule portion of this rich mathematical theory is reviewed in \cref{section:Homomorphisms}. Given a predetermined set of pseudo-Hermitian elements of a $C^*$-algebra, $\mathfrak{A}$, as discussed in \cref{representationTheoryTheorem}, a new pseudo-Hermitian operator is generated by applying a direct sum of representations to the predetermined set. Furthermore, this new pseudo-Hermitian operator can be perturbed by a self-adjoint element in the commutant of the range of the direct sum of representations without altering the intertwining operator. These pseudo-Hermitian operators include quasi-Hermitian operators in a suitable special case. In a case where the predetermined pseudo-Hermitian operators are mutually commuting, a subset of the domain of broken antilinear symmetry is determined in \cref{PT-Breaking-Theorem-Representation-Theory}. Examples are given, and I discuss how existing results such as \cite{Shi2022,MyFirstPaper,Barnett2023} can be derived using this new technique.

A third technique is available from the definition. Consider a fixed intertwining matrix, $\eta$. The most general pseudo-Hermitian matrix associated to $\eta$ is the product of a Hermitian matrix on the left and $\eta$ on the right \cite{Carlson1965}. When $\eta$ is a sparse matrix, the set of $\eta$-self-adjoint matrices can be computed with relative ease.

 
\section{Pencils with Anticommuting Generators} 
\label{anticommutPencilSection}


The mathematical setting of this section stems from a pair of anticommuting linear operators. Physically, if one of these operators is interpreted as an observable and the other is an involution, then the involution is interpreted as a \textit{chiral symmetry} of this observable.

\subsection{Chiral Symmetry}

This section summarizes some well-known examples in physics that have a notion of chiral symmetry. First, I define what I mean by chiral symmetry.
\begin{defn}
An operator, $J$, has a \textit{chiral symmetry}, $E$, if and only if $J E = - E J$ and $E$ is an involution, so $E^2 = \mathbb{1}$.
\end{defn}
The spectrum of operators with chiral symmetry is characterized in the following lemmas. The first lemma considers a relaxed scenario of two anticommuting operators, where neither of them is required to be an involution.
\begin{lemma} \label{lemma:anti-commute}
Consider two operators, $J$ and $E$, whose domain is a vector space over a field, $\mathbb{F}$. Suppose $J$ and $E$ anticommute. Then,
\begin{align}
E\left(\ker\left(J - \lambda \mathbb{1}\right)\right) \subseteq \ker\left(J + \lambda \mathbb{1}\right) &\quad& \forall \lambda \in \mathbb{F}.
\end{align}
Thus, if $E$ is injective\footnote{A linear map is injective if and only if it its kernel consists of only the zero vector.}, for every eigenvalue of $J$, $\lambda \in \sigma_p(J)$, another eigenvalue of $J$ is $-\lambda \in \sigma_p(J)$. If $E$ is bijective, then $E^{-1}$ also anticommutes with $J$, which implies 
\begin{align}
E\left(\ker\left(J - \lambda \mathbb{1}\right)\right) = \ker\left(J + \lambda \mathbb{1}\right) &\quad& \forall \lambda \in \mathbb{F}.
\end{align}
Furthermore, in the case where $E$ is bijective, the identity $\lambda \in \sigma(J) \, \Leftrightarrow -\lambda \in \sigma(J)$ holds since the inverse of $\mu\mathbb{1} + J$ admits a closed-form expression for all elements $\mu$ of the resolvent set of $J$,
\begin{align}
(\mu \mathbb{1} + J)^{-1} = E(\mu \mathbb{1} - J)^{-1} E^{-1} &\quad& \forall \mu \in \mathbb{F} \setminus \sigma(J).
\end{align}
\end{lemma}
The following lemma addresses the case of a diagonalizable operator with a chiral symmetry.
\begin{lemma}
A diagonalizable operator, $J$, whose domain is a vector space over the field $\mathbb{F}$ has a chiral symmetry, $E$, if and only if 
\begin{align}
\dim \ker(\lambda \mathbb{1} - J) = \dim \ker(\lambda \mathbb{1} + J) &\quad& \forall \lambda \in \mathbb{F}.
\end{align}
The chiral symmetry operator can always be written as
\begin{align}
E = \bigoplus_{\lambda \in \sigma_p(J)} E_\lambda,
\end{align}
where $E_\lambda:\ker\left(J - \lambda \mathbb{1}\right) \to \ker\left(J + \lambda \mathbb{1}\right)$ is an arbitrary linear bijection satisfying $E^{}_\lambda = E_{-\lambda}^{-1}$. 
\end{lemma}
An example of a physical system exhibiting chiral symmetry is a system of massless, relativistic fermions in an even number of space-time dimensions \cite{langacker2017standard}. 
The remainder of this section examines a different example; chiral symmetry emerges in models where the Hamiltonian is local to a bipartite graph.

\begin{defn}
A \textit{directed graph}, $G$, is a set of \textit{vertices}, $V$, equipped with a set of \textit{edges}, $E \subseteq V \times V$,
\begin{align}
G = (V, E).
\end{align} 
Edges link pairs of vertices. A directed graph is called \textit{bipartite} if there exists a decomposition of the vertex set into disjoint subsets, $V = V_A \cup V_B$ with $V_A \cap V_B = \emptyset$, such that every edge links elements of $V_A$ with $V_B$. More explicitly, the edge sets of directed graphs satisfy
\begin{align}
E \subseteq (V_A \times V_B) \cup (V_B \times V_A).
\end{align}
Perhaps more intuitively, a bipartite graph is one such that only two colours are needed for a vertex colouring. 
\end{defn}

A pictorial representation of some examples of bipartite graphs is given in \cref{bipartiteGraphs}.

\begin{figure}[!ht]
\centering
\includegraphics[width=\textwidth]{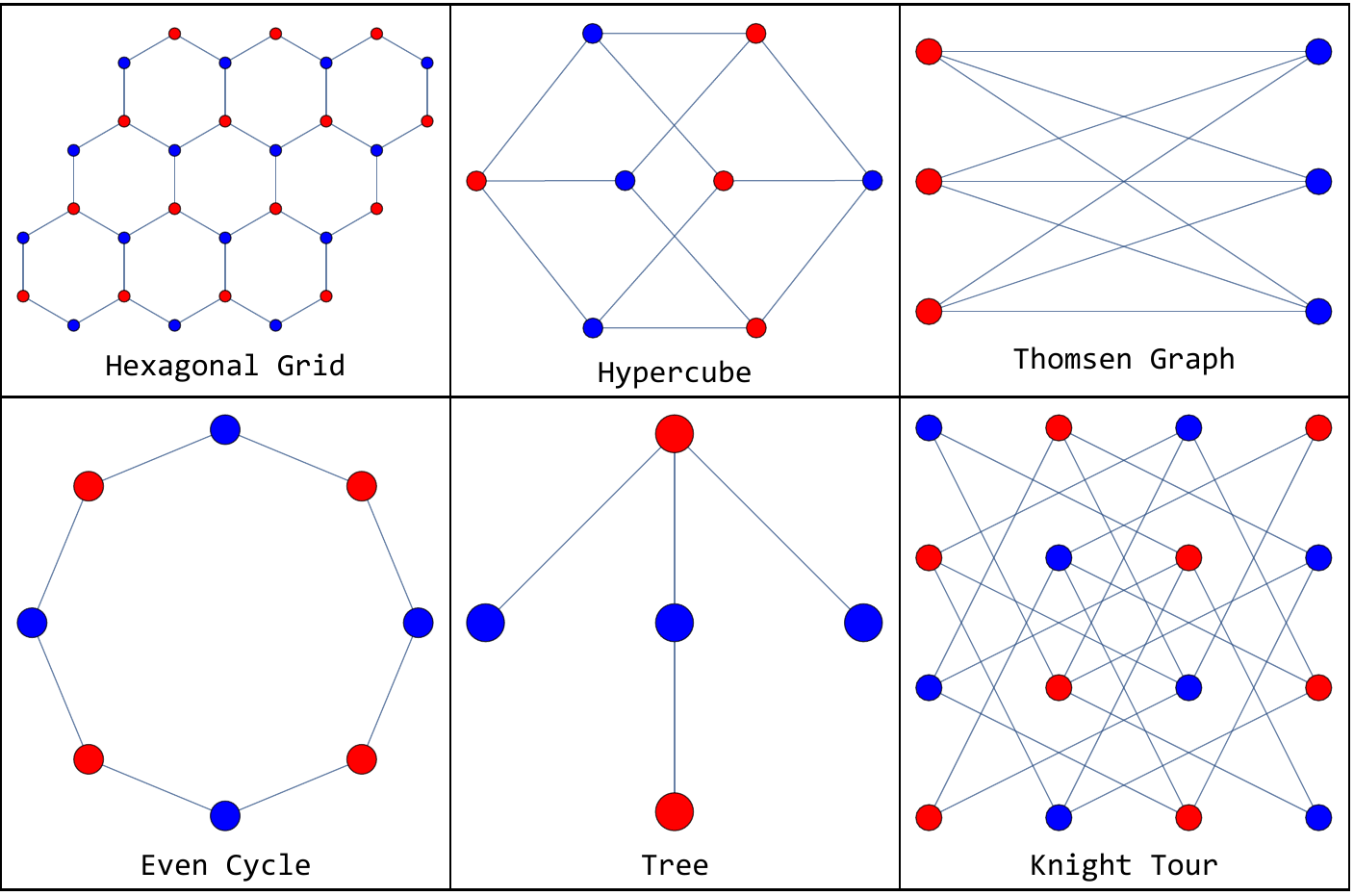}
\caption{Examples of bipartite graphs. Vertices are depicted as circles and edges are depicted as lines. Observe how vertices with the same colour never share an edge.}
\label{bipartiteGraphs}
\end{figure}


\begin{ex} [Tight-Binding Models]
An inner product space associated to a vertex set, $V$, can be realized via an injective map, $e:V \to X$, whose codomain is a vector space. An inner product on the vector space generated by $e$, $\text{span}(e(V))$, is defined by
\begin{align}
\braket{e(v) | e(w)} = \delta_{vw}.
\end{align}
A Hilbert space, $\mathcal{H}_V$, associated to $V$ is, thus, the closure 
\begin{align}
\mathcal{H}_V := \text{cl}(\text{span}(e(V))).
\end{align} 

Given a directed graph, $G = (V,E)$, a Hamiltonian, $H$, on $\mathcal{H}_V$ is called \textit{local} to the graph $G$ if 
\begin{align}
\braket{e(v)|H e(w)} \neq 0 \, \Rightarrow \, (v,w) \in G.
\end{align}
Consider a Hamiltonian, $H$, which is local to a bipartite graph. Associated to a bipartite graph are Hilbert subspaces $\mathcal{H}_{V_{A,B}} \subseteq \mathcal{H}$. One chiral symmetry of $H$ is the difference between orthogonal projections onto the subspaces $\mathcal{H}_{V_{A,B}}$,
\begin{align}
E_{G} = \mathscr{P}_{\mathcal{H}_{V_A}} - \mathscr{P}_{\mathcal{H}_{V_B}} \, \Rightarrow \, [H, E_G]_+ = 0. \label{eqn-bipartite-chiral-symmetry}
\end{align}
\demo
\end{ex}

\subsection{Non-Hermitian Pencils with Anticommuting Generators}

Following the advice of Olga Taussky \cite{Taussky1988}, whenever we have two "interesting" operators, we should study the pencil generated by them. Throughout this section, $J$ and $E$ will denote two anticommuting operators on a Hilbert space, $\mathcal{H}$, satisfying
\begin{align}
[J,E]_+ = 0.
\end{align}
Since these operators are interesting, consider an element of their pencil,
\begin{align}
H = J + \mathfrak{i} \gamma E,
\end{align}
with $\gamma \in \mathbb{C}$. Let's start by characterizing the spectrum of $H$.

\begin{proposition}
A set containing the spectrum of $H$ is 
\begin{align}
\sigma(H) \subseteq \left\{\pm \sqrt{\lambda_{J}^2 - \gamma^2 \lambda_{E}^2} \,|\, (\lambda_{J},\lambda_E) \in \sigma(J) \times \sigma(E) \right\}. \label{anticommutPencilSpectrum}
\end{align}
In the case where $E$ is an involution, given an eigenvalue of $J$, $\lambda \in \sigma_p(J)$, and a corresponding eigenvector, $u(\lambda) \in \ker(\lambda \mathbb{1} - J)$,
two eigenvectors of $H$ with eigenvalues $\pm\sqrt{\lambda^2 - \gamma^2}$ respectively are
\begin{align}
v_{\pm}(\lambda) := \left(\lambda \pm \sqrt{\lambda^2 - \gamma^2} \right) u(\lambda) + \mathfrak{i} \gamma E u(\lambda) \in \ker\left(\pm \sqrt{\lambda^2 - \gamma^2} \mathbb{1} - H \right). \label{eqn-anticommut-evecs}
\end{align}
If $J$ is diagonalizable and $E$ is an involution, then $H$ is diagonalizable for all $\gamma \in \mathbb{R} \setminus \sigma(J)$ and the set inequality \cref{anticommutPencilSpectrum} is an equality. 
\end{proposition}
\begin{proof}
This proof will only display a proof of \cref{anticommutPencilSpectrum} as the rest of the proposition statement is self-evident.
The anticommutation relation implies a commutation relation,
\begin{align}
[J, E^2]_- = 0.
\end{align}
Thus, $H^2$ is a sum of commuting elements
\begin{align}
H^2 &= J^2 - \gamma^2 E^2 \label{anticommutPencilSquare},
\end{align}
so its spectrum satisfies
\begin{align}
\sigma(H^2) \subseteq \{\lambda_{J}^2 - \gamma^2 \lambda_{E}^2 \,|\, (\lambda_{J},\lambda_E) \in \sigma(J) \times \sigma(E)\}. 
\end{align}
The proof is completed by applying the spectral mapping \cref{spectralCalculus} to the last equation.
\end{proof}
Coalescence of the eigenvectors $v_+(\gamma_{\text{EP}}) = v_-(\gamma_{\text{EP}})$ occurs at the exceptional points of $H$, which are $\gamma_{\text{EP}} \in \sigma_p(J)$. Since $(J + \mathfrak{i} \lambda E)u(\lambda) = \lambda v_{\pm}(\lambda)$, at the exceptional points, the dimension of the generalized eigenspace associated to 0 is twice the dimension of the kernel of $H$.

Despite the fact that $H$ is generically non-Hermitian, it can have a real spectrum. This can be understood through the lens of pseudo-Hermiticity, which is clarified in the following proposition.
\begin{proposition} 
Assume $J, E$ satisfy the Dieudonn{\'e} relation with the same intertwining operator, $\eta_0$, or more explicitly,
\begin{align}
\eta_0 J = J^\dag \eta_0 &\quad& \eta_0 E = E^\dag \eta_0.
\end{align} 
Additionally, assume $J$ is bijective. Then, an intertwining operator for $H$ is $\eta_0 J^{-1}$. Thus, applying the generative procedure of \cite{bian2019time}, 
\begin{align}
\eta = \eta_0 + \mathfrak{i} \gamma \eta_0 J^{-1} E \label{anticommuteMetric}
\end{align}
is an intertwining operator for $H$ satisfying
\begin{align}
\eta H = H^\dag \eta.
\end{align}
If $\gamma \in \mathbb{R}$ and $\eta_0 = \eta_0^\dag$, then $\eta$ is Hermitian.
\end{proposition}



Since pseudo-Hermiticity of $H$ has been established, we are in a position to address quasi-Hermiticity of $H$. Establishing quasi-Hermiticity of $H$ requires finding a positive-definite metric operator, which is performed in \cref{prop:Pencil-QuasiHerm}. The proof requires some light machinery regarding positivity for Hilbert spaces associated to nontrivial metric operators, which will be introduced now. Following \cite[p. 40]{azizov1989linear} and \cref{etaSpace}, given a strictly-positive operator, $\eta_0 \in \text{GL}(\mathcal{B}(\mathcal{H}))^+$, define the $\eta_0$-\textit{space}, $\mathcal{H}_{\eta_0}$, to be the vector space of $\mathcal{H}$ endowed with the positive-definite inner product $\braket{\cdot|\eta_0 \cdot}$. An $\eta_0$-\textit{(strictly-)positive} operator\footnote{Some references, such as \cite[p.107]{azizov1989linear}, refer to what I call $\eta_0$-positive operators as $\eta_0$-non-negative.} is a (strictly-)positive operator on $\mathcal{H}_{\eta_0}$. Results which characterise (strictly-)positive operators, such as those presented in \cref{section:Positive}, apply straightforwardly to their $\eta_0$-(strictly)-positive counterparts. The norm in an $\eta_0$ space is
\begin{align}
||\psi||_{\eta_0} = ||\eta_0^{1/2} \psi|| &\quad& \forall \psi \in \mathcal{H}_{\eta_0}.
\end{align}

\begin{proposition} \label{prop:Pencil-QuasiHerm}
Assume $J$ is a bijection, $E$ is an involution, and that $J$ and $E$ satisfy the Dieudonn{\'e} relation with the same strictly-positive metric operator, $\eta_0 \in \text{GL}(\mathcal{B}(\mathcal{H}))^+$, or more explicitly,
\begin{align}
\eta_0 J = J^\dag \eta_0 &\quad& \eta_0 E = E^\dag \eta_0.
\end{align} 
Given $\gamma \in \mathbb{R}$, define $\eta$ as in \cref{anticommuteMetric}. Then, $\eta$ is strictly positive if $|\gamma| < ||J^{-1}||_{\eta_0}$. 
If $J$ is diagonalizable, then $\eta$ is strictly-positive if and only if $|\gamma| < ||J^{-1}||_{\eta_0}$.
\end{proposition}
\begin{proof}

Before proceeding with the proof, I remark on the quantity $||J^{-1}||_{\eta_0}$. The next equation follows from the $C^*$-identity, that the involution of a $C^*$-algebra is an isometry \cite[Thm. 16.1]{Doran2018}, and the spectral radius \cref{thm:Spectral-Radius},
\begin{align}
||J^{-1}||^2_{\eta_0} = ||J^{-1} E||^2_{\eta_0} = \text{spr}(J^{-2}).
\end{align}

A proof that $\eta$ is strictly positive whenever $|\gamma| < ||J^{-1}||_{\eta_0}$ follows from understanding $\eta$ as a product of two $\eta$-strictly-positive operators.
\begin{itemize}
\item $\eta_0$ is trivially $\eta_0$-strictly-positive.\\
\item One proof that $\mathbb{1} + \mathfrak{i} \gamma J^{-1} E$ is $\eta_0$-strictly-positive follows from its $\eta_0$-self-adjointness, its bijectivity, and the positivity of its spectrum. $\eta_0$-self-adjointness follows from the $\eta_0$-self-adjointness of $J^{-1}$ and $E$ and their anticommutation relation. Bijectivity follows from the construction of a bounded inverse, which is given by a Neumann series whenever $|\gamma| \, ||J^{-1} E||_{\eta_0} < 1$\cite[\S 2.2]{bratteli1987operator}. Positivity of the spectrum follows from the spectral radius theorem.
\end{itemize}
Since $\eta$ is the product of two $\eta_0$-strictly-positive operators, the infimum of the spectrum of $\eta$ is nonzero and positive \cite[Thm. 3]{Wigner1963}\cite{WignerCollectedWorks}. Thus, $\eta$ is strictly-positive.

To prove that when $J$ is diagonalizable, strict-positivity of $\eta$ implies $|\gamma| < ||J^{-1}||_{\eta_0}$, note the strict-positivity of $\eta$ implies $H$ is quasi-Hermitian. Thus, by \cref{Williams1969Corollary} \cite{williams1969operators}, $H$ is diagonalizable with a real spectrum. By \cref{anticommutPencilSpectrum}, this implies $|\gamma| < ||J^{-1}||_{\eta_0}$.
\end{proof}

\begin{ex}
Let $J \in \text{GL}_n(\mathbb{C})$ satisfy the constraint that $J_{ij} = 0$ whenever $i - j$ is a positive even integer. Furthermore, assume $J$ is invertible. Considered as a Hamiltonian, $J$ is local to a complete bipartite graph, namely, the graph whose vertex set is the integers ranging from $1$ to $n$ and whose edge set is the set of pairs of vertices which are both even or both odd. Then, due to \cref{{eqn-bipartite-chiral-symmetry}}, one chiral symmetry $J$ possesses is a diagonal matrix, 
\begin{align}
E_{ij} = \begin{cases}
\delta^i_{j} & \text{if } i \text{ is even} \\
-\delta^i_{j} & \text{if } i \text{ is odd},
\end{cases}
\end{align} 
where $\delta$ denotes the Kronecker delta.
One example of an invertible $J$ is an irreducible tridiagonal matrix with $n$ even. In this case, the spectrum of $H$ given by \cref{anticommutPencilSpectrum} was deduced by \cite{Dyachenko2021}. Furthermore, the corresponding Hamiltonian in this case has been studied in, for instance, \cite{Lieu2018}. Consider the special case where $J$ is the Toeplitz tridiagonal matrix with the matrix elements
\begin{align}
J_{ij} = \delta^i_{j+1} + \delta^j_{i+1}.
\end{align}
The inverse of $J$ can be computed explicitly \cite{Huang1997}. Thus, the matrix elements of the metric of \cref{anticommuteMetric} can also be computed exactly, they are 
\begin{align}
\eta_{jk}= \delta^j_k + (-1)^{n+j} \mathfrak{i} \gamma \sin\left(\frac{\min\{j,k\} \pi}{2}\right) \sin\left(\frac{(2 n-\max\{j,k\}+1) \pi}{2}\right).
\end{align} 
Using a diagonal similarity transform, as in \cref{triDiagToSymmetric} and \cref{MetricMapper}, allows for a generalization to the case where $J$ is the Hamiltonian of the Hatano-Nelson model \cite{Hatano1996},
\begin{align}
J_{ij} = \alpha \delta^i_{j+1} + \beta \delta^j_{i+1} &\quad& \alpha \beta \neq 0.
\end{align}
In this case, we have 
\begin{align}
\eta_{jk}&= \left(\delta^j_k + (-1)^{n+j} \mathfrak{i} \gamma \sin\left(\frac{\min\{j,k\} \pi}{2}\right) \sin\left(\frac{(2 n-\max\{j,k\}+1) \pi}{2}\right) \right) \nonumber \\
&\times \left(\frac{\beta^*}{\alpha^*} \right)^{(j-1)/2}\left(\frac{\beta}{\alpha} \right)^{(k-1)/2}.
\end{align} 
\demo
\end{ex}

Given a finite sequence, $(J_k, E_k)_{k \in \{1, \dots, n\}}$, of operators which anticommute, a new pair of operators which anticommute can be constructed
\begin{align}
    \mathcal{E}   &= \otimes_{k \in \{1, \dots, n\}} E_k, \\
    \mathcal{J}_k &= \otimes_{j \in \{1, \dots, n\}} \begin{cases}
    \mathbb{1}_j & \text{if } j \neq k\\
    J_k & \text{if } j = k,
    \end{cases} \\
    \left[\sum_{k=1}^n \alpha_k \mathcal{J}_k, \mathcal{E} \right]_+ &= 0,
\end{align}
with $\alpha_k \in \mathbb{C}$. If each $J_k$ is local to a one-dimensional chain, then the sum $\sum_{k} \alpha_k \mathcal{J}_k$ is a Hamiltonian that is local to an $n$-dimensional hypercube. 




\section{Primer on Representations of $C^*$-Algebras} \label{section:Homomorphisms}

This section summarizes some basic properties of reducible representations of $C^*$-algebras.

\begin{defn} \label{defn:Representation}
A ${}^*$-\textit{homomorphism}, $\phi$, is a linear map whose domain and codomain are ${}^*-$algebras such that 
\begin{itemize}
\item $\phi$ is an \textit{algebra homomorphism}, so that $\phi(a b) = \phi(a) \phi(b)$ for all $a, b \in \text{Dom}(\phi)$. 

\item $\phi$ preserves the involution, so $\phi(x^*) = \phi(x)^*$.
\end{itemize}
A \textit{representation}, $\phi$, of a $C^*$-algebra, $\mathfrak{A}$, on a Hilbert space, $\mathcal{H}$, is a $*$-homomorphism whose codomain is the $C^*$-algebra of bounded operators on $\mathcal{H}$.
\end{defn}

Representations can further be classified by their invariant subspaces as defined below.
\begin{defn}
An \textit{invariant subspace}, $\mathcal{K}$, of a representation of a $C^*$-algebra, $\phi: \mathfrak{A} \to \mathcal{B}(\mathcal{H})$, is a subspace $\mathcal{K} \subseteq \mathcal{H}$ such that for every $a \in \mathfrak{A}$, $\mathcal{K}$ is an invariant subspace\footnote{The invariant subspace of an operator is defined in \cref{invariantSubspaceOperatorDefn}.} of the operator $\phi(a)$. $\phi$ is said to be \textit{reducible} if there exists an invariant subspace other than $\{0\}$ or $\mathcal{H}$, and \textit{irreducible} otherwise. 
\end{defn}

My understanding of why a representation is called "reducible" follows from \cref{representationDecomposition}: Reducible representations can be expressed as a sum of representations whose ranges are strict subsets of their parent representation.

\begin{lemma} \label{Lemma:-Invariant-Subspace-Homomorphism}
If $\mathcal{K}$ is an invariant subspace of a $*$-homomorphism, $\phi: \mathfrak{A} \to \mathcal{B}(\mathcal{H})$, then $\mathcal{K}^\perp$ is an invariant subspace as well. $\mathcal{K}$ is a closed invariant subspace of a $*$-homomorphism, $\phi: \mathfrak{A} \to \mathcal{B}(\mathcal{H})$ if and only if $\mathcal{K}^\perp$ is an invariant subspace as well.
\end{lemma}
\begin{proof}
To prove the first statement, suppose $\psi \in \mathcal{K}^\perp, \varphi \in \mathcal{K}$. Then, for every $a \in \mathfrak{A}$, we have
\begin{align}
\braket{\psi|\phi(a^*) \varphi} &= 0 \\
&= \braket{\phi(a^*)^\dag \psi|\varphi} \\
&= \braket{\phi(a) \psi|\varphi}. \label{invariantSubspaceComplement}
\end{align}
Since \cref{invariantSubspaceComplement} holds for every $\varphi \in \mathcal{K}$, it follows that $\phi(a) \psi \in \mathcal{K}^\perp$ for all $a \in \mathfrak{A}$ and for all $\psi \in \mathcal{K}^\perp$. Thus, $\mathcal{K}^\perp$ is an invariant subspace of $\phi$. 

The second statement follows from the first and the identity $((\mathcal{K})^\perp)^\perp = \text{cl}(\mathcal{K})$.
\end{proof}

\begin{proposition} \label{representationDecomposition}
Consider a representation of a $C^*$-algebra, $\phi:\mathfrak{A} \to \mathcal{B}(\mathcal{K})$. Suppose there is a decomposition of the Hilbert space into a finite number, $n$, of orthogonal closed invariant subspaces of $\phi$, $\mathcal{K}_{i \in \{1, \dots, n\}} \subseteq \mathcal{K}$. More explicitly, assume
\begin{align}
\mathcal{K} = \gls{span}(\cup_{i} \mathcal{K}_i),
\end{align}
where $\mathcal{K}_i$ are closed invariant subspaces of $\phi$ such that 
\begin{align}
\mathcal{K}_i \subseteq \mathcal{K}_j^\perp &\quad& \forall i \neq j.
\end{align}
Define $\phi_{i}:\mathfrak{A} \to \mathcal{B}(\mathcal{K})$ as
\begin{align}
\phi_{i}(a) := \phi(a) \mathscr{P}_{\mathcal{K}_i} &\quad& \forall a \in \mathfrak{A}, \label{phi-i}
\end{align}
where $\mathscr{P}_V$ denotes the unique orthogonal projection whose range is the closed vector space $V$. Let $h: \{1, \dots, n\} \to \mathfrak{A}$ be a function which associates an element of $\mathfrak{A}$ to every invariant subspace $\mathcal{K}_i$.
The maps $\phi_{i}$ satisfy the following properties for all $i,j \in \{1, \dots, n\}$ and $a,b \in \mathfrak{A}$:
\begin{enumerate}
\item $\phi = \sum_{k = 1}^n \phi_{k}$. \label{reducibleRepItem1} \\
\item $\mathscr{P}_{\mathcal{K}_i}$ commutes with\footnote{Another version of this statement is that $\mathscr{P}_{\mathcal{K}_i}$ is in the commutant of the range of $\phi$ for all $i$.} $\phi(a)$.\\
\item $\phi_{i}(a) \phi_j(b) = \delta_{ij} \phi_i (ab)$. \label{reducibleRepProductItem} \\
\item $\phi_i$ is a representation of $\mathfrak{A}$ on $\mathcal{K}$ \\
\item $||\sum_{k=1}^n \phi_k(h(k)) || = \sup_{k \in \{1, \dots, n\}} || \phi_k(h(k)) ||$.
\end{enumerate} 
\end{proposition}
\begin{proof}
\leavevmode 
\begin{enumerate}
\item Since the $\mathcal{K}_i$ are mutually orthogonal and generate $\mathcal{K}$, the projections $\mathscr{P}_{\mathcal{K}_i}$
satisfy\footnote{More generally, the map $P: \mathbb{P}(\{1, \dots, n\})$ defined by $P(S) := \sum_{i \in S} \mathscr{P}_{\mathcal{K}_i}$ is a projection-valued measure, in the sense of  \cref{operator-Valued-Measures}.}
\begin{align}
\sum_{i = 1}^n \mathscr{P}_{\mathcal{K}_i} = \mathbb{1}. \label{projectorsSumToOne}
\end{align}
The equation listed as \cref{reducibleRepItem1} follows from this.

\item Since $\mathcal{K}_i$ is an invariant subspace of $\phi$, it must be the case that $\phi(a):\mathcal{K}_i \to \mathcal{K}_i$ for all $i \in \{1, \dots, n\}$. Thus,
\begin{align}
\phi(a) \mathscr{P}_{\mathcal{K}_i} = \mathscr{P}_{\mathcal{K}_i} \phi(a) \mathscr{P}_{\mathcal{K}_i} &\quad& \forall a \in \mathfrak{A}, i \in \{1, \dots, n\}.
\end{align}
That $\mathscr{P}_{\mathcal{K}_i}$ commutes with $\phi(a)$ follows from the previous equation and \cref{projectorsSumToOne},
\begin{align}
\mathscr{P}_{\mathcal{K}_i} \phi(a) &= \sum_{j = 1}^n \mathscr{P}_{\mathcal{K}_i} \phi(a) \mathscr{P}_{\mathcal{K}_j} \\
&= \mathscr{P}_{\mathcal{K}_i} \phi(a) \mathscr{P}_{\mathcal{K}_i} \\
&= \phi(a) \mathscr{P}_{\mathcal{K}_i} &\quad& \forall a \in \mathfrak{A}, i \in \{1, \dots, n\}.
\end{align}

\item This follows from the definition of $\phi_i$ and the observation that $\mathscr{P}_{\mathcal{K}_i}$ commutes with $\phi(a)$,
\begin{align}
\phi_{i}(a) \phi_j(b) &= \phi_{i}(a) \mathscr{P}_{\mathcal{K}_i} \phi_j(b) \mathscr{P}_{\mathcal{K}_j} \\
&= \phi_{i}(a) \mathscr{P}_{\mathcal{K}_i}  \mathscr{P}_{\mathcal{K}_j}\phi_j(b) \\
&= \delta_{ij} \phi_{i}(a) \phi_j(b) &\quad& \forall a,b \in \mathcal{A}.
\end{align}

\item Proving that $\phi_i$ is a representation requires showing that it preserves the involution and the product. Preservation of the involution follows from the commutation property of the orthogonal projections $\mathscr{P}_i$,
\begin{align}
\phi_i(a^*) &= \phi(a^*) \mathscr{P}_{\mathcal{K}_i} \\
&= \phi(a)^* \mathscr{P}_{\mathcal{K}_i} \\
&= (\phi(a) \mathscr{P}_{\mathcal{K}_i})^* &\quad& \forall a \in \mathfrak{A}, i \in \{1, \dots, n\}.
\end{align}
Preservation of the product follows from the previously demonstrated \cref{reducibleRepProductItem}.

\item
    Given $\psi \in \mathcal{K}$, let $\psi_i := \mathscr{P}_{\mathcal{K}_i} \psi$. Applying the Pythageorean identity\footnote{See \cref{Pythagoras}.} for Hilbert spaces, 
    \begin{align}
        \norm{ \sum_{i = 1}^{n} \phi_{i}(h_i) } &= \sup_{\psi \in \mathcal{K} \setminus \{0\}} \dfrac{\norm{ \sum_{i = 1}^{n} \phi_{i}(h_i) \psi_i }}{\norm{ \sum_{j = 1}^{n} \psi_j }} \\
        &= \sup_{\psi \in \mathcal{K}\in \mathcal{K} \setminus \{0\}} \sqrt{\dfrac{\sum_{i = 1}^{n} \norm{\phi_{i}(h_i) \psi_i}^2}{\sum_{j = 1}^{n} \norm{\psi_j }^2}} \\
        &\leq \sup_{i \in S} \norm{ \phi_i(h_i) }.
    \end{align}
    Equality is achieved with a suitable choice of $\psi \in \mathcal{K}_i$, where $i$ is one of the indices which extremizes the supremum.
    
    \end{enumerate}
\end{proof}

A trivial example of a decomposition of the form used in the above proposition is the decomposition $\mathcal{K}_1 = \mathcal{K}$, $\mathcal{K}_2 = \{0\}$. Lemma~\ref{Lemma:-Invariant-Subspace-Homomorphism} implies that every finite-dimensional $\mathcal{K}$ can be decomposed as a direct sum of irreducible representations.

The following example of a representation of the algebra of $2 \times 2$ matrices will be used in subsequent sections.
\begin{ex} \label{2x2-Matrix-Algebra-Representation}
Let $\mathcal{H}$ be a Hilbert space, and $\mathcal{U} \in \mathcal{B}(\mathcal{H})$ be a unitary on $\mathcal{H}$. This unitary defines a representation of the $C^*$-algebra of $2 \times 2$ complex matrices, $\mathfrak{M}_2(\mathbb{C})$, on the algebra of $2 \times 2$ matrices with bounded operator-valued elements\footnote{See \cref{matrixAlgebraExample} for more details.}, $\mathfrak{M}_2(\mathcal{B}(\mathcal{H})) \simeq \mathcal{B}(\mathcal{H}) \oplus \mathcal{B}(\mathcal{H})$. Specifically, this representation is
\begin{align}
\phi_{\mathcal{U}}\left(\begin{pmatrix}
a & b \\
c & d
\end{pmatrix} \right) = \begin{pmatrix}
a \,\mathbb{1} & b \,\mathcal{U} \\
c \,\mathcal{U}^\dag & d \,\mathbb{1}
\end{pmatrix}.
\end{align}
Suppose $V$ is an invariant subspace of $\mathcal{U}$. Since $V^\perp$ is an invariant subspace\footnote{See \cref{Invariant-Subspace-Of-Adjoint}.} of $\mathcal{U}^\dag$, $V^\perp \oplus V$ is an invariant subspace of $\phi_{\mathcal{U}}$.
\demo
\end{ex}
\subsection{Unital Maps}
This section reviews some facts about unital maps that I use later in this chapter.

\begin{defn}
A \textit{unital} map, $\phi$, is one whose domain and codomain are unital associative algebras\footnote{More details on unital algebras can be found in \cref{Unital-Algebras-Section}.} such that $\phi(\mathbb{1}) = \mathbb{1}$.
\end{defn}

An interpretation of an operator as pseudo-Hermitian or quasi-Hermitian requires associating them with a corresponding bijective or positive-definite intertwining operators respectively. One nice feature about unital ${}^*$-homomorphisms is that they preserve the notions of invertibility and positivity, as demonstrated in the following lemma.
\begin{lemma}
If $\phi$ is a unital ${}^*$-homomorphism whose domain and co-domain are $C^*$-algebras, then $\phi$ maps invertible positive elements into invertible positive elements. In particular, if $a \in \text{GL}(\text{Dom}(\phi))$ is invertible, then
\begin{align}
\phi(a)^{-1} = \phi(a^{-1}).
\end{align}
\end{lemma}


\section{Pseudo-Hermiticity and Representation Theory} \label{section:PseudoHermFromRep}


Given a set of $m$-self-adjoint elements in a $C^*$-algebra and representation of this algebra on a Hilbert space, $\mathcal{K}$, the theorems presented in this section generate pseudo-Hermitian operators and their corresponding intertwining operators on this Hilbert space. Furthermore, unital representations will generate quasi-Hermitian operators if $m$ is positive and invertible.

To avoid repetitious definitions, I introduce the setting of the subsequent theorems and examples here:
\begin{defn} \label{Representation-Theory-Setting}
\leavevmode 
\begin{itemize}
\item $\mathfrak{A}$ is a $C^*$-algebra and $\mathcal{H}$ is a separable Hilbert space. \\
\item $\phi$ is a representation of $\mathfrak{A}$ on a separable Hilbert space, $\mathcal{K}$.\\
\item $\mathcal{K}$ decomposes as a direct sum over a finite number, $n$, of mutually orthogonal closed invariant subspaces of $\phi$, denoted by $\mathcal{K}_{i \in \{1, \dots, n\}}$. Explicitly, 
\begin{align}
\mathcal{K}_i &\subseteq \mathcal{K}^\perp_j &\quad& \forall i \in \{1, \dots, n\} \setminus \{j\} \\
\mathcal{K} &= \gls{span}\left(\bigcup\limits_{i \in \{1, \dots, n\}} \mathcal{K}_i \right).
\end{align}
\item $\phi_i:\mathfrak{A} \to \mathcal{B}(\mathcal{K})$ are representations defined by $\phi_i(a) = \phi(a) \mathscr{P}_{\mathcal{K}_i}$, where $\mathscr{P}_V$ denotes the orthogonal projection onto the closed linear subspace $V$. The properties of the representations $\phi_i$ are detailed in \cref{representationDecomposition}. \\
\item $\mathcal{R}$ denotes the range of a linear map and $'$ denotes the commutant of a subset in $\mathcal{B}(\mathcal{K})$. 
\item $h_{k \in \{1, \dots, n\}} \in \mathfrak{A}$ is a sequence of elements in $\mathfrak{A}$. The central object of interest is the operator 
\begin{align}
H := A + \sum_{k=1}^n \phi_k(h_k).\label{homoHam}
\end{align}
The operator $A$ is assumed to be in the commutant of the range of $\phi$, or more explicitly, $A \in \mathcal{R}(\phi)'.$
\end{itemize}
\end{defn}

\begin{theorem} \label{representationTheoryTheorem}
Consider the setting outlined in \cref{Representation-Theory-Setting}. Let $\Omega, \Omega^\dag \in \phi_i(\mathfrak{A})'$ for all $i \in \{1, \dots, n\}$. Assume the sequence $h_{i \in \{1, \dots, n\}} \in \mathfrak{A}$ satisfies the pseudo-Hermiticity condition $m h_i = h_i^* m$ for some invertible Hermitian element $m = m^*$, and assume $A \in \mathcal{R}(\phi)'$ satisfies the intertwining relation
\begin{align}
\Omega^\dag \Omega A = A^\dag \Omega^\dag \Omega.
\end{align} 
Then,
\begin{enumerate}
\item $H$ satisfies the intertwining relation
\begin{align}
\Omega^\dag \phi(m) \Omega H = H^\dag \Omega^\dag \phi(m) \Omega. \label{pseudoHerm-Representation}
\end{align}
\\
\item
If $\Omega$ is injective, $m$ is positive and invertible, and $\phi$ is a unital representation, then $\phi(m)$ is bijective and positive, and $H$ is quasi-Hermitian.
\end{enumerate}
\end{theorem}

\begin{proof}
\leavevmode 
\begin{enumerate}
\item
Observe that 
\begin{align}
\Omega^\dag \phi(m) \Omega A = A^\dag \Omega^\dag \phi(m) \Omega
\end{align}
immediately follows from the fact that $A$ and $\Omega$ are in the commutant, $A \in \mathcal{R}(\phi)'_{sa}$. Thus, establishing the intertwining relation \cref{pseudoHerm-Representation} is equivalent to showing $H-A$ is $\Omega^\dag \phi(m) \Omega$-self-adjoint. Since $\Omega$ is in the commutants $\phi_i(\mathfrak{A})$ for all $i \in \{1, \dots, n\}$, this is equivalent to showing $H-A$ is $\phi(m)$-self-adjoint. This follows from the properties of $\phi$ discussed in \cref{representationDecomposition},
\begin{align}
\phi(m) \left(\sum_{j=1}^n \phi_j(h_j)\right) &= \left(\sum_{i=1}^n \phi_i(m)\right) \left(\sum_{j=1}^n \phi_j(h_j) \right)  \\
&= \sum_{i=1}^n \sum_{j=1}^n \delta_{ij} \phi_j(m h_j) \\
&= \sum_{i=1}^n \phi_i(h^*_i m) \\
&= \sum_{i=1}^n \phi_i(h_i)^\dag \phi_i(m) \\
&= \sum_{i=1}^n \sum_{j=1}^n \delta_{ij} \phi_i(h_i)^\dag \phi_j(m) \\
&= \left(\sum_{j=1}^n \phi_j(h_j)\right)^\dag \phi(m).
\end{align}

\item
If $\phi$ is unital, then the inverse of $\phi(m)$ is $\phi(m^{-1})$. Since $m^{-1}$ is bounded and unital representations are isometries, $\phi(m)$ is bounded and has bounded inverse, so $\phi(m)$ is bijective. Positivity of $\phi(m)$ follows from positivity of $m$ and the assumption that $\phi$ is unital. 
\end{enumerate}
\end{proof}

The following examples provide applications of \cref{representationTheoryTheorem} to models of physical interest, namely the models studied in \cite{Shi2022,Barnett2023,MyFirstPaper}. Both examples will assume $\Omega = \mathbb{1}$.

\begin{ex} \label{ex:Shi2022}
This example will build a $2n \times 2n$-dimensional pseudo-Hermitian matrix and intertwiner from a set of $2 \times 2$ matrices. Following \cref{2x2-Matrix-Algebra-Representation}, consider the representation of the algebra of $2 \times 2$ matrices on $\mathbb{C}^n \oplus \mathbb{C}^n$,
\begin{align}
\eta := \phi_{\mathbb{1}_n}: \begin{pmatrix}
a & b \\
c & d
\end{pmatrix} \to 
\begin{pmatrix}
a \mathbb{1}_n & b \mathbb{1}_n \\
c \mathbb{1}_n & d \mathbb{1}_n
\end{pmatrix},
\end{align}
where $\mathbb{1}_n$ denotes the identity matrix in $\mathbb{C}^n$. 
This representation can be written as a direct sum of irreducible representations over the invariant subspaces
\begin{align}
\mathcal{K}_{k \in \{1, \dots, n\}} := \gls{span}\{e_k, e_{n+k} \}.
\end{align}
In particular\footnote{An irreducible representation defined by $\eta_k$ is its restriction to the subspace $\gls{span}\{e_k, e_{n+k}\}$.},
\begin{align}
\eta_k \left(\begin{pmatrix}
a & b \\
c & d
\end{pmatrix} \right) \sum_{j = 1}^{2n} \alpha_j e_j := (a \alpha_k + b \alpha_{n+k}) e_k + (c \alpha_k + d \alpha_{n+k}) e_{n+k}.
\end{align}
The commutant of $\eta$ is the set of all maps of the form $B \oplus B$ where $B \in \mathfrak{M}_n(\mathbb{C}^d)$. Let 
\begin{align}
h_k := \kappa_k \begin{pmatrix}
0 & 1 \\
\gamma^2 & 0
\end{pmatrix} \label{Shi-Pre-Hamiltonian}
\end{align}
for $k \in \{1, \dots, d\}$, $\kappa_k \in \mathbb{R}$, $\gamma^2 \in \mathbb{R}$. All of these $h_k$ have a common intertwining operator,
\begin{align}
M = \begin{pmatrix}
\gamma^2 & 0 \\
0 & 1 
\end{pmatrix},
\end{align}
which is positive-definite whenever $\gamma \in \mathbb{R} \setminus \{0\}$. Thus, given a $d \times d$ Hermitain matrix, $J = J^\dag \in \mathfrak{M}_d(\mathbb{C})$, the matrix considered by \cite{Shi2022}
\begin{align}
H = J \oplus J + \sum_k \eta_k(h_k) \label{Shi-Hamiltonian}
\end{align}
is quasi-Hermitian for all $\gamma \in \mathbb{R} \setminus \{0\}$ with the metric operator
\begin{align}
\eta(M) = (\gamma^2 \mathbb{1}_n )\oplus (\mathbb{1}_n).
\end{align}
\demo
\end{ex}

\begin{ex} \label{2011Example}
In the same spirit as the previous example, this example constructs a $2n \times 2n$-dimensional pseudo-Hermitian operator given a set of $2 \times 2$ matrices. This time, consider the representation, $\eta := \phi_{\mathcal{P}}$, defined by the unitary exchange matrix, $\mathcal{P}_n \in \mathfrak{M}_n(\mathbb{C})$ with the elements
\begin{align}
\mathcal{P}_n e_i = e_{n-i+1}.
\end{align} 
The invariant subspaces of $\eta$ are 
\begin{align}
\mathcal{K}_{k \in \{1, \dots, n\}} := \gls{span} \{e_k, e_{2n-k+1}\},
\end{align}
and the corresponding maps $\eta_k$ are 
\begin{align}
\eta_k \left(\begin{pmatrix}
a & b \\
c & d
\end{pmatrix} \right) \sum_{j = 1}^{2n} \alpha_j e_j := (a \alpha_k + b \alpha_{2n-k+1}) e_k + (c \alpha_k + d \alpha_{2n-k+1}) e_{2n-k+1}.
\end{align}
Let $h_k$ be the most general traceless $2 \times 2$ pseudo-Hermitian matrix, as characterized by \cref{2x2PseudoHerm-Theorem} from the introduction chapter,
\begin{align}
h_k = \delta_{kn} \left(\vec{\alpha} \cdot \vec{\sigma} + \mathfrak{i} \vec{\beta} \cdot \vec{\sigma} \right),
\end{align}
where $\vec{\alpha}, \vec{\beta} \in \mathbb{R}^3$ satisfy $\vec{\alpha} \cdot \vec{\beta} = 0$, and $\vec{\sigma}$ is the Pauli vector. A one-parameter family of $2 \times 2$ intertwiners, $M(\zeta)$, associated to $h_k$ is
\begin{align}
M(\zeta) = \zeta \vec{\alpha} \cdot \vec{\sigma} + (\vec{\alpha} \cdot \vec{\alpha} )\mathbb{1} + (\vec{\beta} \times \vec{\alpha}) \cdot \vec{\sigma}.
\end{align}
$M(\zeta)$ is positive-definite if $(1-\zeta^2) (\vec{\alpha} \cdot \vec{\alpha}) > (\vec{\beta} \cdot \vec{\beta}))$. Then the $2n \times 2n$ matrix
\begin{align}
H = \eta_n(H_n) + J \oplus (P J P)
\end{align}
is pseudo-Hermitian with the intertwiner $\eta \circ M(\zeta)$ for all $J = J^\dag$. This intertwiner is positive-definite whenever $M(\zeta)$ is positive-definite. If $J$ is tridiagonal and $h_n = \mathfrak{i} \gamma \vec{\sigma}_3 + \vec{\sigma}_1$, then $H$ is the Hamiltonian studied in \cref{nearestNeighbour} and presented in the references \cite{MyFirstPaper,Barnett2023}, and $M$ is the metric determined in \cref{homomorphismMetric}.
\demo
\end{ex}

\section{A Commutative Case} \label{Commutative-Section}

The previous section addressed the question of finding pseudo-Hermitian operators using representation theory. A subset of the domain of unbroken antilinear symmetry is given by the quasi-Hermitian case, which is associated to a positive-definite metric operator. This section discusses scenarios where broken antilinear symmetry occurs. In particular, the central result of this section, \cref{PT-Breaking-Theorem-Representation-Theory}, determines a necessary condition for the operator $H$ of \cref{Representation-Theory-Setting} to have a purely imaginary spectrum.

In the interest of generality, theorem~\ref{PT-Breaking-Theorem-Representation-Theory} constructs a Hamiltonian using the reducible nature of a generic representation. A simpler result, given in the next lemma~\ref{lemma:imSpec}, with a far simpler proof, applies if this reducible nature is not utilized. 

\begin{lemma} \label{lemma:imSpec}
Suppose $\phi$ is a unital ${}^*$-homomorphism whose domain and co-domain are $C^*$-algebras. Let $h \in {\normalfont \text{Dom}}(\phi)$ and let $A$ be self-adjoint and in the commutant of the range of $\phi$, so $A \in \mathcal{R}(\phi)'$. Then
\begin{align}
\sigma(A + \phi(h)) \subseteq \left\{\epsilon + \lambda\,|\,(\epsilon,\lambda) \in \sigma(h) \times \sigma(A) \right\}.
\end{align}
\end{lemma}
\begin{proof}
This follows from two facts: first, that $A$ and $\phi(h)$ commute; second, that $\sigma(\phi(h)) \subseteq \sigma(h)$ for all unital ${}^*$-homomorphisms \cite[Thm. 4.1.8]{KadisonRingroseI}.
\end{proof}

The above lemma~\ref{lemma:imSpec} can be applied to a previously introduced example, as summarized below.

\begin{ex}[continues=ex:Shi2022]
Consider the model introduced in \cite{Shi2022}. Assume $\kappa_k = \kappa \neq 0$ and $\gamma^2 < 0$. The spectrum of $h_k$ defined in \cref{Shi-Pre-Hamiltonian} is $\pm \gamma$, which in this case is purely imaginary. Thus, the imaginary part of every eigenvalue in spectrum of $H$ defined in \cref{Shi-Hamiltonian} is $\pm \gamma$.
\demo
\end{ex}

The following theorem characterizes the invariant subspaces, and consequently, the symmetries of a commutative case of $H$.
\begin{theorem} \label{Rep-Theory-Direct-Sum-Theorem}
Consider the setting of \cref{Representation-Theory-Setting}. Assume $\phi$ is unital with the domain $\mathfrak{A} = \mathcal{B}(\mathcal{H})$ and suppose the operators $h_k$ are simultaneously diagonalizable
\begin{align}
h_k = S \Lambda(k) S^{-1},
\end{align}
where $\Lambda(k)$ is diagonal the orthonormal basis $\ket{x_j} \in \mathcal{H}$, and $S, S^{-1} \in \mathcal{B}(\mathcal{H})$. 
Let $A = A^\dag \in \mathcal{R}(\phi)'$ be a bounded self-adjoint operator. Then for every $\ket{x_j}$, the spaces 
\begin{align}
R_j := \mathcal{R}(\phi(S \ket{x_j} \bra{x_j} S^{-1})) \subset \mathcal{K}
\end{align}
are invariant subspaces of $H$. 
Consequently, defining the restrictions $H_j = H|_{R_j}$, $H$ admits a factorization as a direct sum,
\begin{align}
H = \oplus_j H|_{R_j}.
\end{align}
\end{theorem}
\begin{proof} 
Let $\ket{\psi} \in R_j$ be arbitrary. 
The operators $\phi(S \ket{x_j} \bra{x_j} S^{-1})$ are mutually commuting projectors.
Since $A$ commutes with each of these projectors by supposition,
\begin{align}
A \ket{\psi} &= A \phi(S \ket{x_j} \bra{x_j} S^{-1}) \ket{\psi} \\
&= \phi(S \ket{x_j} \bra{x_j} S^{-1}) A \ket{\psi} \in R_j.
\end{align}
Thus, we only need to show that $R_j$ is an invariant subspace of $\sum_{k = 1}^n \phi_k(h_k)$ to show that it is an invariant subspace of $H$. This follows from the properties of $\phi_k$ listed in \cref{representationDecomposition}. 

That $H = \oplus_j H|_{R_j}$ follows from the facts that the spaces $R_j$ are linearly independent and span the Hilbert space, $\mathcal{H} = \oplus R_j$.
\end{proof}
\begin{corollary}
$\sum_j \alpha_j \phi(S \ket{x_j} \bra{x_j} S^{-1})$ is a linear symmetry of $H$.
\end{corollary}

\begin{theorem} \label{PT-Breaking-Theorem-Representation-Theory}
Consider the setting introduced in \cref{Representation-Theory-Setting}.
Suppose $\phi$ is unital representation of $\mathfrak{A} = \mathcal{B}(\mathcal{H})$ on $\mathcal{K}$ such that $\mathcal{H}$ and $\mathcal{K}$ are finite-dimensional. Consider a sequence of simultaneously diagonalizable operators whose spectrum is a subset of the imaginary axis,
\begin{align}
h_k = S \Lambda(k) S^{-1},
\end{align}
such that $\sigma(\mathfrak{i} h_k) \in \mathbb{R}$ and $\Lambda(k)$ is diagonal in an orthonormal basis, $\ket{x_j}$, of $\mathcal{H}$. For every fixed $j$, assume either $\braket{x_j|\Lambda(k) | x_j} \geq 0$ for all $k \in \{1, \dots, n\}$, or $\braket{x_j|\Lambda(k) | x_j} \leq 0$ for all $k \in \{1, \dots, n\}$.
Let $A = A^\dag \in \mathcal{R}(\phi)'$. 
Let 
\begin{align}
\tilde{\Lambda}_j := \max_{k \in \{1, \dots, n\}} \braket{x_j|\Lambda(k) |x_j}.
\end{align}
Then, a set containing the spectrum of $H$ is given in \cref{Rep-Theory-Inclusion}. A consequence is that every real eigenvalue of $H : = A + \sum_{k=1}^n \phi_k(h_k)$
is simultaneously a real eigenvalue of $A$.
\end{theorem}
\begin{proof}
The strategy used to prove this theorem is to treat $H$ as a perturbation of a simpler operator, $\tilde{H}$, that has a closed-form expression for the imaginary part of its eigenvalues. The Bauer-Fike \cref{BauerFike} \cite{Bauer1960} will then be used to approximate the spectrum of $H$ to be "near" the spectrum of $\tilde{H}$. Before doing this, I will express $H$ as a direct sum of operators on subspaces of $\mathcal{K}$.

The bound on eigenvalues given by the Bauer-Fike theorem contains a factor of the condition number of a diagonalizing similarity transform corresponding to the unperturbed operator; to avoid this prefactor, it is wise to work the similar operator 
\begin{align}
H' := \phi(S)^{-1} H \phi(S),
\end{align}
where $\phi(S)^{-1} = \phi(S^{-1})$ exists since $\phi$ is a unital representation. Since $A$ is in the commutant of the range of $\phi$, and due to the properties of $\phi_k$ listed in \cref{representationDecomposition}, another expression for $H'$ is 
\begin{align}
H' &= A + \sum_{k = 1}^n \phi_k(S^{-1} h_k S) \\
&= A + \sum_{k = 1}^n \phi_k(\Lambda(k)).
\end{align}
By \cref{Rep-Theory-Direct-Sum-Theorem}, we can write
\begin{align}
H' = \oplus H'|_{R_j},
\end{align}
where 
\begin{align}
R_j := \mathcal{R}(\phi(\ket{x_j} \bra{x_j})).
\end{align}

Define the operators $\tilde{H}_j:R_j \to R_j$ by
\begin{align}
\tilde{H}_j = A|_{R_j} + \tilde{\Lambda}_j \mathbb{1}_{R_j}.
\end{align}
Since $A$ is self-adjoint, and the real part of every matrix element of $\Lambda(k)$ equals zero, every eigenvalue of $\tilde{H}_j$ has real and imaginary parts of the form
\begin{align}
\lambda \in \sigma(\tilde{H}_j) \, &\Rightarrow \, \text{Re}(\lambda) \in \sigma(A) \\
\lambda \in \sigma(\tilde{H}_j) \, &\Rightarrow \, \text{Im}(\lambda) = \text{Im} \left( \tilde{\Lambda}_j \right). \label{HtildeSpec}
\end{align}

Consider the decomposition
\begin{align}
H'|_{R_j} = \tilde{H}_j + (H'|_{R_j} - \tilde{H}_j).
\end{align}
An upper bound for the norm of the second term may be deduced using the product inequality of the operator norm, given in \cref{productInequality}, the properties of $\phi_k$ given in \cref{representationDecomposition}, the fact that $\phi$ is a contraction\footnote{See \cref{homomorphism-Contraction}.}, and the fact that the operator norm is axis-oriented in an orthonormal basis. Explicitly,
\begin{align}
\norm{H'|_{R_j} - \tilde{H}_j} &= \norm{\sum_{k = 1}^n \phi_k(\Lambda(k) - \tilde{\Lambda}_j \mathbb{1})|_{R_j}} \\
&= \norm{\sum_{k = 1}^n \phi_k(\Lambda(k) - \tilde{\Lambda}_j \mathbb{1}) \phi(\ket{x_j} \bra{x_j})} \\
&\leq \norm{\sum_{k = 1}^n \phi_k(\Lambda(k)\ket{x_j} \bra{x_j} - \tilde{\Lambda}_j\ket{x_j} \bra{x_j})} \\
&\leq \sup_{k \in \{1, \dots, n\}} ||\phi_k(\Lambda(k)\ket{x_j} \bra{x_j} - \tilde{\Lambda}_j\ket{x_j} \bra{x_j})|| \\
&\leq \sup_{k \in \{1, \dots, n\}} ||\phi(\Lambda(k)\ket{x_j} \bra{x_j} - \tilde{\Lambda}_j\ket{x_j} \bra{x_j}) \mathscr{P}_{\mathcal{K}_k}|| \\
&\leq \sup_{k \in \{1, \dots, n\}} ||\Lambda(k)\ket{x_j} \bra{x_j} - \tilde{\Lambda}_j\ket{x_j} \bra{x_j}|| \\
&\leq |\tilde{\Lambda}_j|,
\end{align}
where the last inequality follows from the sign condition which was imposed on $\Lambda$.


 Using this norm inequality in the Bauer-Fike \cref{BauerFike} along with the spectral inclusion result for $\tilde{H}_j$ given in \cref{HtildeSpec}, we find a set which contains all the eigenvalues of $H'$,
\begin{align}
\sigma(H') = \sigma(H) \subseteq \bigcup_j \bigcup_{\lambda_A \in \sigma(A)} \{z \in \mathbb{C} \,|\,|z - \lambda_A - \tilde{\Lambda}_j|  \leq |\tilde{\Lambda}_j| \} \label{Rep-Theory-Inclusion}.
\end{align}
That every real eigenvalue of $H$ is contained in the spectrum of $A$ is a consequence of the circular geometry of the inclusion set given above.
\end{proof}

I will apply the previous theorem to specific models in \cref{General-Maximal-Breaking-Section}. I will conclude this section by analysing an example where \cref{PT-Breaking-Theorem-Representation-Theory} does not hold. This example emphasizes necessity of the assumption of commutativity of $h_k$ in \cref{PT-Breaking-Theorem-Representation-Theory}. 

\begin{ex} \label{non-commuting-representation-example}
Consider a special case, $\eta:\mathfrak{M}_2(\mathbb{C}) \to \mathfrak{M}_4(\mathbb{C})$, of the homomorphism defined in \cref{2x2-Matrix-Algebra-Representation},
\begin{align}
\eta\left(\begin{pmatrix}
a & b \\
c & d
\end{pmatrix} \right) &= \begin{pmatrix}
a & 0 & 0 & b \\
0 & a & b & 0 \\
0 & c & d & 0 \\
c & 0 & 0 & d
\end{pmatrix}.
\end{align}
This representation is reducible. Following \cref{representationDecomposition}, it admits a decomposition
\begin{align}
\eta &= \eta_1 + \eta_2, \\
\eta_1\left(\begin{pmatrix}
a & b \\
c & d
\end{pmatrix} \right)&:=  \begin{pmatrix}
a & 0 & 0 & b \\
0 & 0 & 0 & 0 \\
0 & 0 & 0 & 0 \\
c & 0 & 0 & d
\end{pmatrix}, &\quad&
\eta_2\left(\begin{pmatrix}
a & b \\
c & d
\end{pmatrix} \right):=  \begin{pmatrix}
0 & 0 & 0 & 0 \\
0 & a & b & 0 \\
0 & c & d & 0 \\
0 & 0 & 0 & 0
\end{pmatrix}.
\end{align}
This representation will generate a $4 \times 4$ pseudo-Hermitian Hamiltonian from a pair of $2 \times 2$ pseudo-Hermitian Hamiltonians. For the $2 \times 2$ Hermitian intertwining operator, I will choose
\begin{align}
m &= \begin{pmatrix}
1 & -\mathfrak{i} \gamma \\
\mathfrak{i} \gamma & 1
\end{pmatrix},
\end{align}
where $\gamma \in \mathbb{R}$ is a parameter.
The following two $2 \times 2$ matrices\footnote{$h_1$ is the qubit Hamiltonian exhibited in the introduction chapter in \cref{PTqubit}.} are pseudo-Hermitian
\begin{align}
h_1 = \begin{pmatrix}
0 & 1 \\
1 & 0
\end{pmatrix} m = \begin{pmatrix}
\mathfrak{i} \gamma & 1 \\
1 & -\mathfrak{i} \gamma
\end{pmatrix} &\quad& h_2 = \begin{pmatrix}
0 & -\mathfrak{i} \\
\mathfrak{i} & 0
\end{pmatrix} m = \begin{pmatrix}
\gamma & -\mathfrak{i} \\
\mathfrak{i} & \gamma
\end{pmatrix},
\end{align}
since they are a Hermitian operator composed with the intertwiner \cite{Carlson1965}. Their spectrum is readily deduced,
\begin{align}
\sigma(h_1) = \sigma(h_2) = \left\{-\sqrt{1 - \gamma^2}, \sqrt{1 + \gamma^2} \right\}.
\end{align}
One Hermitian element in the commutant of the range of $\eta$ is  
\begin{align}
A = \begin{pmatrix}
0 & t & 0 & 0 \\
t & 0 & 0 & 0 \\
0 & 0 & 0 & t \\
0 & 0 & t & 0
\end{pmatrix} \in \mathcal{R}(\eta)',
\end{align}
where $t \in \mathbb{R}$ is a parameter.
Thus, by \cref{representationTheoryTheorem}, the Hamiltonian given in \cref{Representation-Theory-Setting},
\begin{align}
H &= A + \eta_1(h_1) + \eta_2(h_2)  \\
  &= \begin{pmatrix}
  i \gamma & t & 0 & 1 \\
  t & \gamma & -\mathfrak{i} & 0 \\
  0 & \mathfrak{i} & \gamma & t \\
  1 & 0 & t & -\mathfrak{i} \gamma
  \end{pmatrix}, \label{4x4Hex}
\end{align}
is pseudo-Hermitian with the intertwiner $\eta(m)$. Since $m$ is positive-definite if and only if $|\gamma| < 1$, $H$ is quasi-Hermitian and has a real spectrum for $|\gamma| < 1$. 

Despite the fact that $h_1$ and $h_2$ have a purely imaginary spectrum for $|\gamma| > 1$, $H$ can still have a real spectrum!

\begin{figure}[!ht]
\centering
\includegraphics[width = \textwidth]{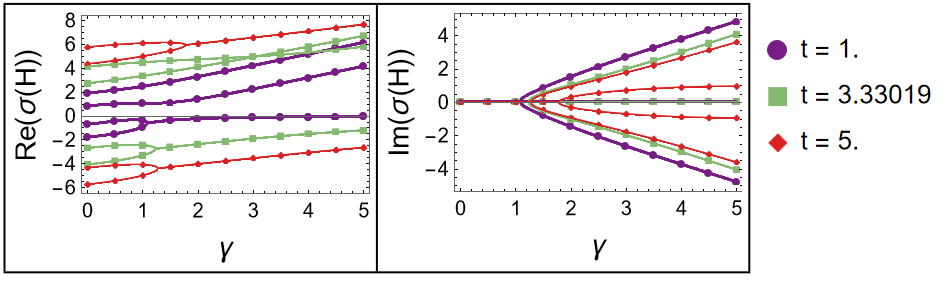}
\caption{Real and imaginary parts of the spectrum of $H$ given by \cref{4x4Hex}. Observe that the spectrum is real for all $|\gamma| < 1$, as predicted by \cref{representationTheoryTheorem}.}
\end{figure}

Figure~\ref{fig:EPContourTactics} will display the exceptional contour of $H$. This exceptional contour can be determined in closed-form by computing the discriminant of the characteristic polynomial; the resulting algebraic curve is
\begin{align}
0 &= 4 t^6 \gamma^2 (6 - 5\gamma^2) + 4 t^8 (\gamma^2-2)^2 - 4 t^2 \gamma^2 (\gamma^2 - 4) (\gamma^2 - 1)^2 \nonumber \\
  &- 4 \gamma^4 (\gamma^2 - 1)^3 + t^4 (16 - 40 \gamma^2 + 45 \gamma^4 - 22 \gamma^6 + \gamma^8). \label{Tactics-EP-Contour}
\end{align}

\begin{figure}[!ht]
\centering
\includegraphics[width = .5 \textwidth]{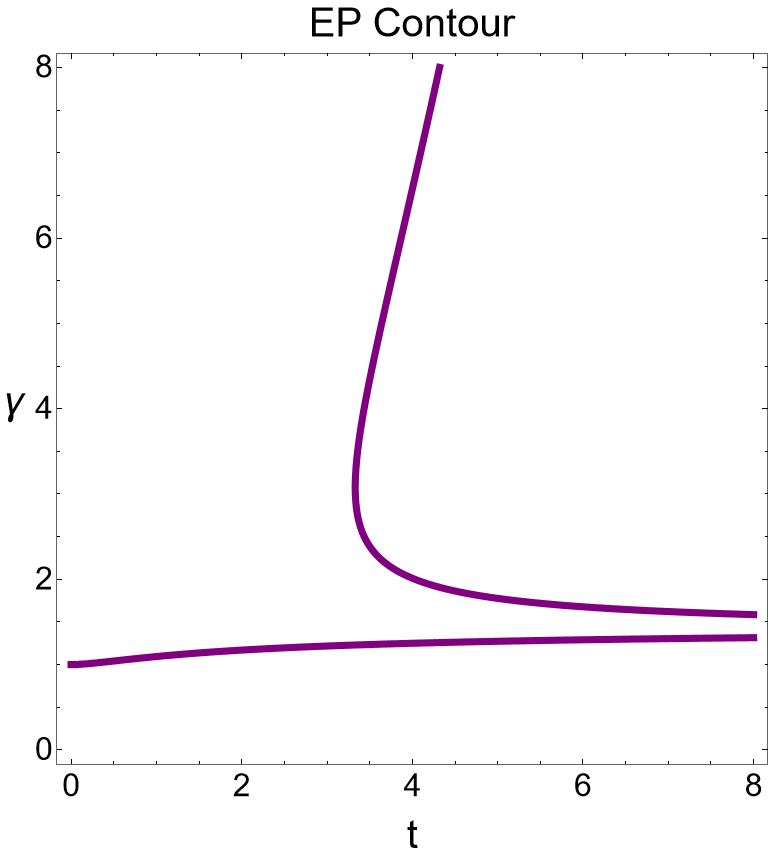}
\caption{Exceptional contour of $H$. This is the zero set of \cref{Tactics-EP-Contour}. Note the exceptional point at $(t, \gamma) = \left(\sqrt{(11 + 5 \sqrt{5})/2},\sqrt{5 + 2 \sqrt{5}}\right)$ has a vertical tangent and the inclusion of the limit $(t,\gamma) = (0,1)$.}
\label{fig:EPContourTactics}
\end{figure}
\demo
\end{ex}

\subsection{Maximal Symmetry Breaking of Even Matrices} \label{General-Maximal-Breaking-Section}

This section applies \ref{PT-Breaking-Theorem-Representation-Theory} from the last section to some matrices with an even number of rows and columns. 

\begin{ex}[continues=2011Example]
This example satisfies all the assumptions required for \cref{PT-Breaking-Theorem-Representation-Theory} to hold. Thus, in the regime of broken antilinear symmetry, where $\text{disc}_x(\text{det}(x \mathbb{1}_2 - h_n)) < 0$, the only real eigenvalues of $H$ are those which are independent of the perturbation $h_n$. 

An alternative way to derive this is directly from the characteristic polynomial. Let $J_{n-1} \in \mathfrak{M}_{n-1}(\mathbb{C})$ denote the submatrix formed by removing the last row and column, which corresponds to the index $n$, from $J$. Then, the characteristic polynomial of $H$ is 
\begin{align}
\text{det}(\lambda \mathbb{1}_{2n} - H) &= -\text{disc}_x(\text{det}(x \mathbb{1}_2 - h_n)) \left(\text{det}(\lambda \mathbb{1}_{n-1} - J_{n-1}) \right)^2 \nonumber \\
&+\left(\text{tr}(h)/2 \, \text{det} (\lambda \mathbb{1}_{n-1} - H_{n-1}) + \text{det}(\lambda \mathbb{1}_{n-1} - J_{n}) \right)^2.
\end{align}
When the eigenvalues of $h_n$ are purely imaginary, or equivalently, when $\text{disc}_x(\text{det}(x \mathbb{1}_2 - h_n)) < 0$, $\lambda$ is a real eigenvalue of $H$ if and only $\lambda$ is eigenvalue of both $J$ and $J_{n-1}$. Alternatively, $\lambda$ is a real eigenvalue of $H$ only if $\lambda$ is an eigenvalue of $J \oplus P J P$, as predicted by \cref{PT-Breaking-Theorem-Representation-Theory}. Equivalently, $\lambda$ is a real eigenvalue if and only if it remains in the spectrum for every choice of $h_n$. 

Suppose $J$ has an eigenvalue, $\lambda \in \sigma(J)$, with algebraic multiplicity greater than one. By the Cauchy interlacing \cref{thm:CauchyInterlacing} \cite{Hwang2004}, $\lambda$ must also be in $\sigma(J_{n-1})$, so $\lambda$ is an eigenvalue which is insensitive to changes in $h_n$.
\demo
\end{ex}

The remainder of this section summarizes two cases where the constant eigenvalues of the Hamiltonian from \cref{2011Example} can be determined analytically.

\begin{ex}[continues=2011Example]
Let $J$ be an irreducible tridiagonal matrix. By \cref{cor:CauchyTridiag}, $\sigma(J) \cap \sigma(J_{n-1}) = \emptyset$, so there are no constant eigenvalues.
\demo
\end{ex}

The next example considers circulant matrices, whose spectrum has a simple expression.
\begin{defn}
A \textit{circulant} matrix, $A \in \mathbb{C}^{n \times n}$, is one of the form
\begin{align}
A = \begin{pmatrix}
c_0 & c_{n-1} & \dots & c_2 & c_1 \\
c_1 & c_0 & c_{n-1} &  & c_1 \\
\vdots & c_1 & c_0 & \ddots & \vdots \\
c_{n-2} & & \ddots & \ddots & c_{n-1} \\
c_{n-1} & c_{n-2} & \dots & c_1 & c_0
\end{pmatrix}. \label{eq:Circulant-Defn}
\end{align}
To give a definition which is perhaps more explicit, consider the congruence relation $i \sim j \, \Leftrightarrow (i-j)/n \in \mathbb{Z}$, and let $[i]_n \in \mathbb{Z}/n \mathbb{Z}$ denote the set of integers in the equivalence class of $i \in \mathbb{Z}$. Then a matrix is circulant if and only if it can be expressed in the form
\begin{align}
A_{ij} = c_{[i-j]_n}.
\end{align}
\end{defn}

As a tangent, the Cauchy interlacing theorem for circulant matrices is a consequence of the fact that the derivatives of a function interlace the function. To realize this, understand the relationship between the derivative of the characteristic polynomial of a circulant matrix and its principle minors \cite[Thm. 2.1]{Kushel2016}.
\begin{theorem} \label{circulantDerivative}
Given a circulant matrix, $A \in \mathfrak{M}_n(\mathbb{C})$, let $A_{n-1} \in \mathfrak{M}_{n-1}(\mathbb{C})$ be a matrix formed by deleting a row and column sharing the same index. Then, letting $p_A$ denote the characteristic polynomial of $A$, we have 
\begin{align}
n p_{A_{n-1}}(x) = \frac{d}{dx} p_A(x).
\end{align} 
\end{theorem}

\begin{ex}[continues=2011Example]
Let $J$ be a circulant matrix with matrix elements given in \cref{eq:Circulant-Defn}. The eigenvalue problem for circulant matrices has a closed-form solution \cite{got1911demonstration}. Explicitly, let 
\begin{align}
\omega_{n,j} := e^{2 \mathfrak{i} j \pi/n} &\quad& j \in \{0, \dots, n-1\}
\end{align}
denote the set of $n$-th roots of unity. Then $\psi_j(k) = \omega^j_{n,k}$ is an eigenvector of $A$, with the eigenvalue 
\begin{align}
\lambda_k = c_0 + \sum_{j = 1}^{n-1} c_{n-j} \omega_{n,k}^{j}.
\end{align}
By Hermiticity of $J$, it follows that the eigenvalues of $J$ associated to $j \in \left\{1, \dots, \floor{\frac{n-1}{2}} \right\}$ always have an algebraic multiplicity greater than one and, thus, correspond to constant eigenvalues of $H$. The eigenvalue associated to $j = 0$ has algebraic multiplicity equal to one and, therefore, corresponds to an imaginary eigenvalue of $H$ when the spectrum of $h_n$ is complex.
\demo
\end{ex}


\chapter{Local Observables} \label{LocalityChapter}

Can quasi-Hermitian frameworks for isolated quantum systems predict physical phenomena outside the scope of traditional, Hermitian models?

Every prediction made by a quasi-Hermitian quantum theory on a finite-dimensional Hilbert space can be made in the context of a Hermitian quantum theory. This equivalence is derived by either performing a similarity transform which maps the algebra of observables to an algebra of Hermitian observables, as in \cref{Williams1969Corollary} of \cite{williams1969operators}, or by re-defining the physical inner product to be one defined by a positive-definite metric operator, as in \cref{etaSpace} and \cite{QuasiHerm92}.

However, an open question is whether \textit{local} quasi-Hermitian quantum theory contains physics beyond the capabilities of local Hermitian quantum theory. The previously mentioned methods for converting quasi-Hermitian quantum theory into Hermitian quantum theory alter notions of locality. For instance, the similarity transform defined by the square root of the metric operator is a nonlocal map, and generically maps local quasi-Hermitian Hamiltonians into nonlocal Hermitian Hamiltonians \cite{Korff2008}. Additionally, the metric operator can be \textit{entangled}. Thus, quasi-Hermitian inner product structures can generalize traditionally utilized algebraic structures associated with a notion of locality, such as an inner product generated by a tensor product factorization.

This chapter analyses local observable algebras in quasi-Hermitian quantum theories. My findings are summarized below:

\begin{itemize}
\item Local observable algebras are constructed in both the tensor product model of locality and the fermionic model of locality defined in \cite{BravyiKitaev}. The generic feature of non-Hermitian observable algebras is that they have smaller dimensions than their Hermitian counterparts and do not contain nontrivial observables for some spatial subsystems. In particular, the fermionic analysis is applied to two of the toy models introduced earlier, in \cref{toyModelsChapter}. This analysis is presented in \cite{Barnett_2021}.

\item I demonstrate how the expectation values of local quasi-Hermitian observables can be equivalently computed by applying a modified state to local Hermitian observables in \cref{Local-Equivalence-Section}. Thus, the Tsirelson bound \cite{Cirelson1980} for Bell's inequality violations of quantum theory holds for quasi-Hermitian theories as well. Nonlocal games generalize the types of correlations studied by Bell \cite{nonlocalGames} and are reviewed for Hermitian quantum theory \cref{Nonlocal-Games-Review}. A more general consequence of the equivalence of local expectation values between quasi-Hermitian and Hermitian quantum theories is that the value of a quasi-Hermitian strategy for a nonlocal game is equivalent to the value of a Hermitian strategy of a nonlocal game. I emphasize that the Hermitian state is not similar to a quasi-Hermitian state, and thus, the expectation values of nonlocal operators cannot be equivalently computed with this state.
\end{itemize}

\section{Introduction to Locality in Quantum Theory} \label{localityIntroSection}


Portrayals of reality must assign meaning to the concept of space. The Dirac-von Neumann axioms of quantum theory, as presented in \cref{DiracVonNeumann} and \cite{dirac1981principles,vonNeumannBook,vonNeumannBookEnglish}, are defined in general terms that do not directly refer to space. In this section, I seek to codify some implementations of the concept of space in quantum theory. 

Locality principles in physics argue that for sufficiently short time-scales, the behaviour of objects is dictated their immediate surroundings. In many-body systems, objects and their surroundings can be interpreted as subsystems of the total system. Thus, implementing locality principles in mathematical models of a system can be considered a special case of the problem of decomposing this system into subsystems.

A \textit{lattice} introduces natural algebraic relationships between subsystems. Intuition demands clarifying the meaning of the operations of intersecting and merging two subsystems. Mathematically, these operations are realized as the \textit{meet} and \textit{join} operations on a lattice, as defined below.
\begin{defn}
A \textit{lattice} is a set, $L$, together with two binary commutative operations, $\wedge$ and $\vee$, referred to as the \textit{meet} and \textit{join}, respectively. These operations satisfy \textit{absorption laws},
\begin{align}
a \vee (a \wedge b) = a &\quad& \forall a,b \in L, \\
a \wedge (a \vee b) = a &\quad& \forall a,b \in L.
\end{align}
\end{defn}
Note that lattices in this sense need not be discrete.

A partial order can be defined on every lattice by
\begin{align}
a \leq b  \text{ if } a = a \wedge b.
\end{align}
Intuitively, the partial ordering of subsystems characterizes whether one subsystem can be contained in another; if two subsystems satisfy $S \leq T$, then $S$ is contained in $T$. Since the entire system contains every subsystem, lattices associated to physical systems must have a \textit{greatest element}, denoted by $\Lambda$, that corresponds to the entire system. Furthermore, we will assume the existence of a \textit{least element}, $0_L$, that corresponds to the empty subsystem. Explicitly, these elements satisfy
\begin{align}
0_L \leq S \leq \Lambda &\quad& \forall S \in L.
\end{align}
and the lattice is assumed to be \textit{bounded}.

Observables are measurable quantities which characterize a system. Often, an observable is best interpreted as being a property of a subsystem as opposed to the entire system. For example, when modelling the universe, the observable which corresponds to "the sum of the momentum of all your house cats" is best thought of as being a property of the subsystem which is "your home". When an observable pertains to a subsystem, this observable is referred to as \textit{local} to that subsystem. A measurement of an observable local to the subsystem $S$ is referred to as a \textit{local measurement} in $S$.

In the algebraic formulation of quantum theory, observables are elements of a $C^*$-algebra\footnote{$C^*$\textit{-algebras} are defined in \cref{c*Defn}. In practice, it is sufficient to think of a $C^*$-algebra as a collection of bounded operators on a Hilbert space.}. The $C^*$-algebra which is generated\footnote{See \cref{generated-Algebra-defn}.} by the observables which are local to a subsystem, $S$, is called the \textit{algebra of observables local to} $S$ and will be denoted by $\mathfrak{A}_S$. Two physical considerations limit the possibilities for what kinds of observable algebras model locality:
\begin{enumerate}
\item \textbf{Isotony}:
Given two spatial subsystems, $S, T \in L$, such that $S$ is a subsystem of $T$, all local measurements that can be performed in $S$ can also be performed in $T$. Equivalently, all observables in $S$ must be observables in $T$. This constraint is referred to as \textit{isotony}, and it is imposed mathematically on the algebras of observables as
\begin{align}
S \leq T \, &\Rightarrow \,\mathfrak{A}_S \subseteq \mathfrak{A}_T. \label{isotony}
\end{align}

\item \textbf{Subsystem independence}:
Observers local to spatially disjoint subsystems must be able to simultaneously perform local measurements without their results influencing one another. To ensure this, everywhere except in \cref{HaagKastler}, we impose the \textit{subsystem independence constraint} on the algebras $\mathfrak{A}_S$: 
\begin{align}
S \wedge T = 0_L \, &\Rightarrow \, \forall a_S \in \mathfrak{A}_S, a_T \in \mathfrak{A}_T, [a_S, a_T]_- = 0. \label{disjointCommute}
\end{align}
Consequently, the only way for spatially disjoint observers to influence each other is through interaction. 
\end{enumerate}
Local observable algebras satisfying \cref{isotony} and \cref{disjointCommute} are generally referred to as the \textit{commuting operator model}. A more refined construction is that of \textit{quasi-local algebras} \cite[\S 2.6]{bratteli1987operator}.

An observable is called \textit{extensively local} to $S$, if for every $T \leq S$ such that $T \neq S$, the observable is not local to $T$. 

The remainder of this section compiles a collection of commuting operator models. With the exception of \cref{HaagKastler}, all examples work with a particular type of lattice, a $\sigma$-\textit{algebra}, defined below.
\begin{definition} \label{sigmaAlg}
A $\sigma$-\textit{algebra} over the set $X$ is a collection of subsets, $\Sigma \subseteq \mathbb{P}(X)$, satisfying
\begin{enumerate}
\item $X \in \Sigma$ 
\item $E \in \Sigma \, \Rightarrow \, X \setminus E \in \Sigma$ 
\item If $\{E_i\}_{i \in \gls{N}} \subseteq \Sigma$, then $\cup_{i \in \gls{N}} E_i \in \Sigma$.
\end{enumerate}
\end{definition}
In a $\sigma$-algebra, the meet and join are the set intersection and union, respectively. Furthermore, the greatest and least elements of a $\sigma$-algebra are $\Lambda = X$ and $0_\Sigma = \emptyset$, respectively.

\begin{ex} \label{finite-lattice-locality}
\textbf{Particle on a Finite Lattice}: The Hilbert space is $\mathbb{C}^n \otimes \mathcal{H}_{\text{int}}$. Physically, this example models a particle localized to a collection of $n$ spatial sites, with an internal degree of freedom represented by $\mathcal{H}_{\text{int}}$. Space is labelled by a subset of integers, $\Lambda = \{1, \dots, n\}$. The set of subsystems is the set of all subsets of space, $\Sigma = \mathbb{P}(\Lambda)$. The algebras of local observables are
\begin{align}
\mathfrak{A}_S = \{M \in \mathfrak{M}_n(\mathbb{C}) \,|\, M_{ij} \neq 0 \text{ only if } i, j \in S \} \otimes \mathcal{B}(\mathcal{H}_{\text{int}}),
\end{align}
where $\mathfrak{M}_n(\mathbb{C})$ denotes the algebra of $n \times n$ matrices with complex entries and $S \in \Sigma$. For example, for  $n = 2$ and $\mathcal{H}_{\text{int}} = \mathbb{C}$, there are four local observable algebras, associated to the subsets $\Sigma \in \{ \emptyset, \{1\}, \{2\}, \{1,2\} \}$. These algebras are 
\begin{align}
\mathfrak{A}_\emptyset &=\{0 \} \\
\mathfrak{A}_{\{1\}} &= \text{span} \left\{\begin{pmatrix}
1 & 0 \\
0 & 0
\end{pmatrix} \right\} \\
\mathfrak{A}_{\{2\}} &= \text{span} \left\{\begin{pmatrix}
0 & 0 \\
0 & 1
\end{pmatrix} \right\} \\
\mathfrak{A}_{\{1,2\}} &= \mathfrak{M}_2(\mathbb{C}).
\end{align}

\demo
\end{ex}

\begin{ex}
\textbf{One-Dimensional Continuum Particle}: Space is parametrized by the reals, $\Lambda = \mathbb{R}$. The set of spatial subsystems, $\Sigma$, is the $\sigma$-algebra of measurable subsets of the reals satisfying the Carathéodory condition, as introduced in \cite{caratheodory1914lineare,Edgar2019}. In particular, $\Sigma$ contains the Borel $\sigma$-algebra generated by the Euclidean topology on $\mathbb{R}$. Denote the Lebesgue measure as $\lambda$, as introduced in \cite{Lebesgue1902}. Let $\mathcal{L}^2(\mathbb{R})$ denote the space of square integrable functions, that is, measurable maps $\psi \in \text{Bor}(\mathbb{R}, \mathbb{C})$ satisfying
\begin{align}
\int_{\mathbb{R}} |\psi|^2 \, d\lambda < \infty.
\end{align}
Let $L^2(\mathbb{R})$ denote the set of equivalence classes, $[\psi]$, of $\mathcal{L}^2(\mathbb{R})$, where two functions are defined to be equivalent if they are equal almost everywhere with respect to $\lambda$. $L^2(\mathbb{R})$ is a Hilbert space with the inner product
\begin{align}
\braket{[f]\,|\,[g]} = \int_{\mathbb{R}} f^* \cdot g \, d \lambda.
\end{align}
The physical Hilbert space is $L^2(\mathbb{R}) \otimes \mathcal{H}_{\text{int}}$, where $L^2(\mathbb{R})$ models the spatial degree of freedom for our one-dimensional particle and $\mathcal{H}_{\text{int}}$ models internal degrees of freedom.

For every subset $S \in \Sigma$, let $L^2(S)$ denote the subset of $L^2(\mathbb{R})$ whose elements are equivalent to a function whose set-theoretic support\footnote{The set-theoretic support of a function whose codomain is $\mathbb{C}$ is the preimage of $\mathbb{C} \setminus \{0\}$.} is contained in $S$. The subalgebra of linear operators local to $S$, $\mathfrak{A}_S$, is 
\begin{align}
\mathfrak{A}_S = \{\mathscr{P}_{L^2(S)} A \mathscr{P}_{L^2(S)} : A \in \mathcal{B}(L^2(\mathbb{R}))\} \otimes \mathcal{B}(\mathcal{H}_{\text{int}}),
\end{align}
where $\mathscr{P}_{L^2(S)}$ denotes the orthogonal projection onto $L^2(S)$.
\demo
\end{ex}

\begin{ex} \label{LocalityExampleTPM}
\textbf{Tensor Product Model}:
In this thesis, the greatest element of every lattice in the tensor product model, $\Lambda$, is assumed to be finite. The $\sigma$-algebra of subsystems is the power set of $\Lambda$. The Hilbert space factorizes as a tensor product, $\mathcal{H} = \otimes_{i \in \Lambda} \mathcal{H}_i$. For each subsystem, $S \subseteq \Lambda$, the algebra of local operators is
\begin{align}
\mathfrak{A}_S := \bigotimes_{i \in \Lambda} \begin{cases}
\mathcal{B}(\mathcal{H}_i) & \text{if } i \in \Lambda \\
\text{span} \{\gls{id}_i\} & \text{if } i \in \Lambda \setminus S,
\end{cases} \label{TPM-Local-Observables}
\end{align}
where $\gls{id}_{i}$ is the identity operator acting on $\mathcal{H}_{i}$. A discussion on the relation between the subsystem independence constraint and the \textit{no-signalling principle} in the tensor product model can be found in \cite{noSignalling}.

The tensor product model has a more confined structure than the generic commuting operator model. However, if the local observable algebras can be realized on a finite-dimensional Hilbert space, a tensor product factorization of the local observable algebras can be derived. Instead of assuming the existence of a tensor product factorization, we can instead suppose that:
\begin{enumerate}
\item The local observable algebras satisfy the isotony and subsystem independence constraints of \cref{isotony,disjointCommute}. \\
\item The local observable algebras generate the entire space of operators on $\mathcal{H}$, or more explicitly,
\begin{align}
\vee_{i \in \Lambda} A_{\{i\}} = \mathcal{B}(\mathcal{H}). \label{algebras-generating-Hilbert-Space}
\end{align} 
Physically, this constraint implies that there are no superselection rules. \\
\item Every local observable algebra is unital with the same identity element. \\
\item The local observable algebras $\mathfrak{A}_S$ and $\mathfrak{A}_{\Lambda \setminus S}$ are \textit{statistically independent} \cite{Haag1964}, which means for every state $\phi_S$ on $\mathfrak{A}_S$ and every state $\phi_{\Lambda \setminus S}$ on $\mathfrak{A}_{\Lambda \setminus S}$, there exists a unique state, $\phi_\Lambda$ on $\mathfrak{A}_\Lambda$ such that the restrictions of $\phi_\Lambda$ to the domains $\mathfrak{A}_S$ and $\mathfrak{A}_{\Lambda \setminus S}$ are $\phi_S$ and $\phi_{\Lambda \setminus S}$, respectively.
\end{enumerate}
If the above suppositions hold, then the local observable algebras are isomorphic to a set of local observable algebras in the tensor product model \cite{Roos1970,ObservableLocality}.
\demo
\end{ex}

The algebras in the tensor product model are examples of type $I$ von Neumann factors. Thus, more general notions of locality are obtained by considering type $II$ or type $III$ von Neumann factors, is often done in Algebraic quantum field theory \cite{Witten2018}.

\begin{ex} \label{Lattice-Fermion-Locality-Defn}
\textbf{Lattice Fermions}:
In this thesis, the greatest element of every lattice, $\Lambda$, for systems of fermions assumed to be a set with finite cardinality. The $\sigma$-algebra of subsystems is the power set of $\Lambda$. For every site $i \in \Lambda$, we assume the existence of a \textit{creation} operator, $\hat{a}^\dag_i \in \mathfrak{A}$, and an \textit{annihilation} operator, $\hat{a}_i \in \mathfrak{A}$, satisfying the \textit{canonical anticommutation relations},\footnote{Basic facts about the creation and annihilation operators are reviewed in \cref{FermionsIntro}.}
\begin{align}
[\hat{a}^\dag_i, \hat{a}_j]_+ = \delta_{ij} \gls{id} &\enspace& [\hat{a}_i, \hat{a}_j]_+ = 0. \label{CAR-normal}
\end{align}
Intuitively, the role of these operators is to insert or remove fermions from the spatial site corresponding to their index.
An explicit construction of such operators is given by the Jordan-Wigner transformation; see \cref{JordanWignerProposition} and \cite{JordanWigner,WignerCollectedWorks} for details. The Jordan-Wigner transformation unitarily maps the creation and annihilation operators to a Hilbert space with a tensor product decomposition, so we may apply \cref{LocalityExampleTPM} in this Hilbert space to induce a notion of locality. An alternative definition, which follows directly from the canonical anticommutation relations, was given in \cite{BravyiKitaev}.
The local observable algebras associated to $S$ are generated by a linear combinations of terms that are products of creation and annihilation operators with indices in $S$. In each term, the sum of the number of creators and the number of annihilators is required to be even, 
\begin{align} 
\mathfrak{A}_{S} = \text{span} \left\{ 
\left(\prod\limits_{i \in A} \hat{a}^\dag_i \right)\left( \prod\limits_{j \in B} \hat{a}^{}_j \right):  \begin{array}{l} A,B \subseteq \mathbb{P}(S) \, \text{ and }\\ \gls{card}(A)+\gls{card}(B) \equiv 0 \Mod{2} \end{array} \right\}.
\end{align}
The reason why only products with an even number of terms are considered is to achieve subsystem independence; anticommuting operators may not commute, but even products of them do.
\demo
\end{ex}

\begin{ex} \label{HaagKastler}
\textbf{Algebraic Quantum Field Theory}:
A successful concept in modern physics has been the unification of space and time into one object: \textit{spacetime}. Spacetime is typically equipped with a \textit{causal structure}, which classifies pairs of points as \textit{spacelike}, \textit{timelike}, or \textit{lightlike} separated from one another. Consider the example of Minkowski spacetime, which is the manifold $\mathbb{R}^{1,d}$ for $d \in \gls{Z+}$ equipped with an indefinite metric
\begin{align}
\eta = \text{diag}(-c^2, 1, \dots, 1),
\end{align}
where $c$ denotes the speed of light.
Given a curve $\gamma: \mathbb{R} \to \mathbb{R}^{1,d}$, we say $\gamma$ is \textit{causal} if for every $\tau \in \mathbb{R}$, $\braket{\gamma(\tau)| \eta \gamma(\tau)} \leq 0$. Two subsets, $S,T \subseteq \mathbb{R}^{1,d}$, are \textit{spacelike separated} if there does not exist a causal curve with one point in $S$ and one point in $T$.

The \textit{Haag-Kastler axioms}, introduced in \cite{Haag1964} and summarized in \cite{rejzner2016perturbative}, are a set constraints on local observable algebras associated with Minkowski spacetime. The Haag-Kastler axioms are defined on subsets which are elements of a net as opposed to elements of a $\sigma$-algebra. In addition to isotony, the following assumptions on the algebras of observables are imposed:

\begin{itemize}
\item \textbf{Causality}: There is no causal relation between measurements of observables local to spacelike separated domains. Consequently, given two observables localized to spacelike separated domains, it must be possible to associate definite measurement outcomes to both of them. Thus,  
\begin{align}
S \text{ is spacelike separated from } T\,\Rightarrow\, [a_S, a_T]_- = 0 &\quad& \forall a_S \in \mathfrak{A}_S, a_T \in \mathfrak{A}_T\label{einsteinLocality}.
\end{align}
In the non-relativistic limit $c \to \infty$, assuming $\mathfrak{A}_S$ has no dependence on time reduces \cref{einsteinLocality} to the subsystem independence constraint of \cref{disjointCommute}. 

\item \textbf{Poincaré Covariance}: A strongly continuous unitary representation\footnote{This is defined in \cref{DiracVonNeumann}.}, $U$, of the Poincaré group exists on $\mathcal{H}$ such that 
\begin{align}
\mathfrak{A}_{g S} = U(g) \mathfrak{A}_S U(g)^ \dag &\quad& \forall g \in \mathbb{R}^{1,d} \rtimes O(1,d).
\end{align}

\item \textbf{Spectrum Condition}: The spectrum of the generator of space-time translations is contained in the closed forward lightcone.

\item \textbf{Time-Slice Axiom}: The algebra of an open neighborhood of a Cauchy surface of a region coincides with the algebra of the entire region. Physically, this axiom implies that the initial value problem is well-posed.

\item \textbf{Vacuum}: There exists a vector, $\Omega \in \mathcal{H}$, which is Poincaré invariant and \textit{cyclic} for all $\mathfrak{A}_S$ with $S \neq \emptyset$, which means $\{a \Omega \,|\, a \in \mathfrak{A}_S\}$ is dense in $\mathcal{H}$.
\end{itemize}
\demo
\end{ex}

%
%

\subsection{Quasi-Hermitian Locality} \label{sec:QH-Locality}

Every model of locality naturally admits quasi-Hermitian generalizations. This section displays such generalizations. 

Let $\eta \in \mathfrak{A}^+_\Sigma$ be strictly-positive. Then, the algebra whose elements are 
\begin{align}
\mathfrak{A}_S(\eta) := \mathfrak{A}_S \cap \eta^{-1} \mathfrak{A}_S \eta
\end{align}
is a $C^*$-algebra with the involution $\dag_\eta$ defined by 
\begin{align}
a^\dag_{\eta} := \eta^{-1} a^\dag \eta &\quad& \forall a \in \mathfrak{A}_S(\eta)
\end{align}
and the norm defined by 
\begin{align}
||a||_{\eta} := ||\eta^{1/2} a \eta^{-1/2}|| &\quad& \forall a \in \mathfrak{A}_S(\eta).
\end{align}
If $\mathfrak{A}_S$ satisfies either isotony or the subsystem independence constraint, then the algebras $\mathfrak{A}_S(\eta)$ will satisfy isotony or subsystem independence, respectively. Local quasi-Hermitian observable algebras are, thus, consistently defined by $\mathfrak{A}_S(\eta)$. Intriguingly, $\mathfrak{A}_S(\eta)$ is the algebra generated by the set of quasi-Hermitian operators contained in $\mathfrak{A}_S$.

\begin{ex}[continues=finite-lattice-locality] \label{quasi-Hermitian-Lattice-Locality} 
\textbf{Quasi-Hermitian Particle on a Finite Lattice}: Recall that the Hilbert space is $\mathbb{C}^n \otimes \mathcal{H}_{\text{int}}$; the lattice is a subset of integers, $\Lambda = \{1, \dots, n\}$; and the set of spatial subsystems is $\Sigma = \mathbb{P}(\Lambda)$. Let $\eta \in \text{GL}_n(\mathbb{C})^+$ be a positive-definite matrix. 
The local observable algebras are defined using vector spaces associated to $\eta$. To define these spaces, let $\gls{ei}$ denote the canonical basis of $\mathbb{C}^n$, which is the set of basis vectors with elements $(e_i)_j = \delta_{ij}$, where $\delta$ is the Kronecker delta. 
For every $S \in \Sigma$, let $V_S(\eta)$ denote the direct sum of every invariant subspace of $\eta$ contained in $\text{span}\{e_i\,|\,i \in S\}$. A more explicit definition of $V_S(\eta)$ is 
\begin{align}
V_S(\eta) = \ker \left(\mathscr{P}_{\text{span}\{e_i\,|\,i \in \Lambda \setminus S\}} \eta|_{\text{span}\{e_i\,|\,i \in S\}} \right),
\end{align}
where $\mathscr{P}_X$ denotes the orthogonal projection onto a closed linear subspace $X$.

The local observable algebras are\footnote{$\mathcal{B}(V_S(\eta), \mathbb{C})$ is the dual space of $V_S(\eta)$, which is the set of continuous linear functional on $V_S(\eta)$.}
\begin{align}
\mathfrak{A}_S := \mathcal{B}(V_S(\eta)) \otimes \mathbb{1}_{\Lambda \setminus S}.
\end{align}
The real vector space of quasi-Hermitian elements of $\mathfrak{A}_S$ is
\begin{equation}
\text{span}_{\mathbb{R}} \left\{ \ket{\psi} \bra{\psi} \eta\,:\, \ket{\psi} \in V_S(\eta)
\right\},
\end{equation}
where $\text{span}_{\mathbb{R}}$ denotes the real span.
\demo
\end{ex}
\section{Nonlocal Games} \label{Nonlocal-Games-Review}

By design, measurements local to one subsystem cannot influence measurement results local to a disjoint subsystem. However, it is still possible that the results of measurements local to disjoint subsystems may be correlated. John Bell famously demonstrated that quantum theory exhibits correlations that are not present in any locally deterministic theory, such as classical mechanics \cite{Bell1964}. The CHSH inequality \cite{Clauser1969} is an example of a bound on what correlations can be produced in locally deterministic theories. In this context, Boris Tsirelson showed that quantum theory satisfies a weaker bound on correlations \cite{Cirelson1980}. 

The framework of nonlocal games is a platform designed to address what kinds of correlations between disjoint subsystems are possible \cite{nonlocalGames}. The CHSH inequality can be realized as a particular nonlocal game, the \textit{CHSH game}. 

Consider two spatially separated observers, Alice and Bob, who are the players of a nonlocal game. In the setting of a nonlocal game, Alice and Bob are asked to work independently to accomplish a desired outcome. 

A referee approaches Alice and Bob separately, asking them questions $s$ and $t$, respectively. The questions are randomly sampled from input sets $\mathcal{I}_A$ and $\mathcal{I}_B$, where the probability of receiving the pair $s,t$ is $\pi(s,t)$. After receiving their respective questions, Alice and Bob return answers $a$ and $b$ to the referee, which are elements of answer sets $\mathcal{O}_A, \mathcal{O}_B$. When choosing their answers, Alice and Bob are not allowed to communicate; any collaboration on their part must take place before they know what questions they will be asked.
For the combined question, $(s,t)$, there are combinations of answers, $(a,b)$, which 
are "correct" and combinations which are "incorrect". More formally, there is a function, $V:\mathcal{I}_A \times \mathcal{I}_B \times \mathcal{O}_A \times \mathcal{O}_B \rightarrow \{0,1\}$, where the combinations $V(s,t;a,b) = 1$ are correct and the combinations $V(s,t;a,b) = 0$ are incorrect. 

The goal of the players is to produce correct outcomes as often as possible. The players are given time to prepare a strategy. Thus, they know what questions they could be asked, $\mathcal{I}_A$, $\mathcal{I}_B$; they know what their chances of being asked a particular combination of questions is, $\pi$; and they know the correct answers for each pair of questions, or equivalently the function $V$. As part of their strategy, they are allowed to prepare a shared quantum state, $\psi$, but they cannot communicate with each other once they know which question they have to address. Their answers can, thus, be informed by measurements on the shared quantum states, which we realize as commuting projective operator valued measures $A^a_s \in \mathcal{B}(\mathcal{H})$ on Alice's subsystem and $B^b_t \in \mathcal{B}(\mathcal{H})$ on Bob's subsystem. 

Therefore, the players aim to develop a strategy which optimizes the \textit{quantum value}, $\omega$, of the game, $G$,
\begin{align}
\omega_{\text{co}}(G) = \sup_{A, B, \psi} \sum_{s,t} \sum_{a,b} \pi(s,t) V(s,t|a,b) \braket{\psi|A^a_s  B^b_t \psi},
\end{align}
where even the Hilbert space structure, $\mathcal{H}$, is allowed to vary (For instance, Alice and Bob could share an entangled qutrit as opposed to an entangled qubit).

For example, the CHSH game is represented with binary questions and answers, 
\begin{align}
\mathcal{I}_A = \mathcal{I}_B = \mathcal{O}_A = \mathcal{O}_B = \{0,1\},
\end{align} 
and
\begin{align}
V(s,t;a,b) = \begin{cases}
1 & \text{ if } a \oplus b = s \wedge t \\
0 & \text{ otherwise}.
\end{cases}
\end{align}
A restatement of the possibility quantum violations of Bell's inequality is that the quantum value of the CHSH game, which is $\cos^2(\pi/8)$, is larger than this game's classical value, which is $3/4$ \cite{nonlocalGames}.

A special case of quantum strategies to nonlocal games arises from assuming the local operations performed by the players can be realized in the tensor product model. I denote the corresponding quantum value associated with these strategies as $\omega_{\text{ten}}$. One rather illuminating result in the study of nonlocal games is the existence of a game, $G$, such that the strict inequality $\omega_{\text{ten}}(G) < \omega_{\text{co}}(G)$ holds \cite{Ji2021}.

A central result of this thesis is that the values of nonlocal games in quasi-Hermitian quantum theory cannot exceed the values of nonlocal games in traditional, Hermitian quantum theory; this is proven in \cref{Local-Equivalence-Section}.


\section{Quasi-Hermitian Tensor Product Model} \label{TPMSection}

This section generalizes the tensor product of locality, presented in \cref{LocalityExampleTPM}, to quasi-Hermitian quantum theories.

The Hilbert space is assumed to factorize as a tensor product over a finite lattice, $\Lambda$, 
\begin{align}
\mathcal{H} = \bigotimes\limits_{j \in \Lambda} \mathcal{H}_j.
\end{align} 
For notational simplicity, for every subset, $S \subseteq \Lambda$, I will denote
\begin{align}
\mathcal{H}_S := \bigotimes\limits_{j \in S} \mathcal{H}_j.
\end{align}
This Hilbert space is endowed with a bijective positive metric operator, $\eta = \Omega^\dag \Omega$. Following \cref{sec:QH-Locality}, the local quasi-Hermitian observable algebras are the algebras generated by local quasi-Hermitian operators.

I define the \textit{quasi-local algebra} to the algebra generated by algebras associated to individual sites,
\begin{align}
\mathfrak{A}_{\text{loc}} := \vee_{i \in S} \mathfrak{A}_{\{i\}}.
\end{align} 
In contrast with the Hermitian tensor product model \cite{ObservableLocality}, $\mathfrak{A}_{\text{loc}}$ need not be equal to $\mathcal{B}(\mathcal{H})$.

A useful characterization of local quasi-Hermitian observable algebras follows from the \textit{operator-Schmidt decomposition} \cite{operatorSchmidt} of $\eta$,
\begin{align}
\eta = \sum_i \chi_i \bigotimes\limits_{j \in \Lambda} \eta_{i, j}, \label{operator-Schmidt-69}
\end{align}
where $\eta_{i,j}$  are Hermitian and are orthonormal under the Frobenius inner product\footnote{The Frobenius inner product is defined in \cref{hilbertSchmidtHilbertAlg}.},
\begin{align}
\text{Tr} (\eta_{i,k} \eta_{j,k}) = \delta_{i,j}.
\end{align} 
The number of nonzero terms in the operator-Schmidt decomposition of an operator is referred to as its \textit{Schmidt rank}.
In general, an operator-Schmidt decomposition of $\eta$ such that $\eta_{i,j}$ are strictly-positive operators for all $i,j$ does not exist. Examples of such $\eta$ include entangled operators, as defined in \cite{werner1989quantum}. In spite of the generic lack of a tensor product factorization of the metric, adjoints can be computed locally, as displayed in the following theorem.
\begin{theorem} \label{Adjoint-Equivalence-Theorem}
Let $\varphi_{i}:\mathfrak{A}_{\{i\}} \to \mathbb{C}$ for $i \in \Lambda$ be states, which are normalized positive linear functionals, as defined in \cref{stateDef}. Suppose $\eta$ is a strictly-positive element of $\mathfrak{A}$. Define 
\begin{align}
\xi_j &:= \sum_i \left(\prod_{k \in \Lambda \setminus \{j\}} \varphi_k(\eta_{i, k}) \right)\eta_{i, j}, \\
\xi &:= \bigotimes\limits_{j \in \Lambda} \xi_j.
\end{align}
Then for all $i \in \Lambda$ and $a \in \mathfrak{A}_{\{i\}}$,  
\begin{align}
\eta^{-1} a^\dag \eta = \xi^{-1} a^\dag \xi. \label{adjoint-Equality}
\end{align}
Furthermore, $\xi$ is positive-definite. Thus, $\mathfrak{A}_{\text{loc}}$ is a $C^*$-algebra either when the adjoint and norm are defined with $\eta$ or when the adjoint and norm defined with $\xi$.
\end{theorem}
\begin{proof}
The proof starts by proving \cref{adjoint-Equality}.

Consider first the case of a local quasi-Hermitian operator, $a = \eta^{-1} a^\dag \eta \in \mathfrak{A}_{\{i\}}$. This $a$ admits a tensor product decomposition,
\begin{align}
a = \bigotimes\limits_{j \in \Lambda} \begin{cases}
a_i & j = i \\
\mathbb{1} & j \neq i
\end{cases}
\end{align}
for some $a_i \in \mathcal{B}(\mathcal{H}_i)$. From the quasi-Hermiticity condition,
\begin{align}
\left(\otimes_{j \in \Lambda} \begin{cases}
\mathbb{1} & j = i \\
\omega_{j} & j \neq i
\end{cases} \right) (a_i) &= 
\left(\otimes_{j \in \Lambda} \begin{cases}
\mathbb{1} & j = i \\
\omega_{j} & j \neq i
\end{cases} \right) (\eta^{-1} a_i^\dag \eta) \, \Leftrightarrow\\
a_i &= \xi_i^\dag a_i^\dag \xi_i. 
\end{align}
Thus, the corresponding $a = \eta^{-1} a^\dag \eta \in \mathfrak{A}_{\{i\}}$ satisfies
\begin{align}
\eta^{-1} a^\dag \eta &= \xi^{-1} a^\dag \xi.
\end{align}
Consider now a general local quasi-Hermitian observable algebra element, $a \in \mathfrak{A}_{\{i\}}$. This element can be expressed as a sum of a quasi-Hermitian and an anti-quasi-Hermitian part,
\begin{align}
a &= \left(\dfrac{a + a^\dag}{2}\right) + \left(\mathfrak{i} \dfrac{a - a^\dag}{2 \mathfrak{i}} \right).
\end{align}
Since \cref{adjoint-Equality} was demonstrated for each term in the above expansion, this identity must hold for every $a$.

That $\xi$ is positive-definite follows from showing $\xi_i$ is positive-definite. Considering the definition of a positive-definite operator on a Hilbert space with the Gelfand-Naimark theorem in mind, we understand that $\xi_i$ is positive-definite if and only if the valuation $\omega(\xi_i)$ is greater than zero for every state $\omega:\mathfrak{A}_i \to \mathbb{C}$. The valuation satisfies
\begin{align}
\omega(\xi_i) = \left( \bigotimes\limits_{j \in \Lambda} \begin{cases}
\omega & j = i \\
\varphi_j & j \neq i
\end{cases}\right)(\eta),
\end{align}
which is greater than zero since states are completely positive and $\eta$ is positive-definite. Thus, $\xi$ is positive-definite.

That $\mathfrak{A}_{\text{loc}}$ is $C^*$-algebra with the adjoint defined by $\xi$ follows from the equivalence of adjoints, which was established already, and equivalence of norms defined by $\eta$ and $\xi$, which is discussed in\footnote{Technically, \cref{Quasi-Hermitian-C*} established the equivalence of the norms defined by $\eta$ and $\mathbb{1}$ and the equivalence of the norms defined by $\mathbb{1}$ and $\xi$. The equivalence between $||\cdot||_{\eta^{1/2}}$ and $||\cdot||_{\xi^{1/2}}$ is, thus, established by transitivity.} \cref{Quasi-Hermitian-C*}.

\end{proof}

The above theorem tells us about the structure of local quasi-Hermitian observable algebras, but it does not describe what elements are contain in them. Every such algebra trivially contains the span of the identity operator, but such an observable contains no physical information, as a measurement of such an observable always yields the degenerate eigenvalue. The nontrivial case is addressed in the next theorem.

\begin{theorem} \label{quasilocal_theorem}
Assume the metric admits an operator-Schmidt decomposition of the form \cref{operator-Schmidt-69}. Consider a spatial subsystem $S \subseteq \Lambda$.
The following are equivalent:
\begin{enumerate}
\item $a$ is a quasi-Hermitian element in $\mathfrak{A}_S$ which is not a multiple of the identity operator.
\item There exists a simultaneous solution, $a_S \in \mathcal{B}(\mathcal{H}_S)$, to the operator equations
\begin{align}
\left(\bigotimes\limits_{j \in S}\eta_{i,j} \right) a_S = a_S^\dag \left(\bigotimes\limits_{j \in S}\eta_{i,j} \right)&\quad& \forall i: \chi_i > 0, \label{block metrics}
\end{align}
where $a_S$ is not a multiple of the identity operator. Furthermore, $a_S \otimes \mathbb{1}_{\Lambda \setminus S}$ is quasi-Hermitian. 
\end{enumerate}
In the case where $a_S$ is diagonalizable, a third equivalent condition is:

3. Given any operator-Schmidt decomposition of $\eta$, the operators $\{\otimes_{j \in S}\eta_{i,j}\}$ are simultaneously reducible under a invertible transformation of the form
\begin{align}
U^{\dag} \left(\otimes_{j \in S} \eta_{i,j} \right) U = 
\begin{pmatrix}
\eta^i_R & 0 \\
0 & \eta^i_{S\setminus R}
\end{pmatrix}, \label{reduce metric}
\end{align}
where $U$ is an invertible operator on $\mathcal{H}_S$.
\end{theorem}
\begin{proof}
Proof of $\textit{1}\Leftrightarrow \textit{2}$: Assume $a \in \mathfrak{A}_S$.
Let $\phi_{i,S}: \mathcal(B)(\mathcal{H}_{S}) \rightarrow \mathbb{C}$ denote a linear functional satisfying 
\begin{equation}
\phi_{i,S} \left(\bigotimes\limits_{k \in \Lambda}\eta_{j,k} \right) = \delta_{ij}.
\end{equation}
Applying the map $\phi_{i,S} \otimes \mathbbm{1}_{\Lambda \setminus S}$ to both sides of the quasi-Hermiticity condition for $a$
imposes the constraints \cref{block metrics} on $a_S$. In the case where $\eta$ is a Hilbert-Schmidt operator, this map is realized as taking a partial trace over $\Lambda \setminus S$ after multiplication by ${\eta_{i,\Lambda \setminus S}}^{\dag} \otimes \mathbbm{1}$. Conversely, if $a_S$ satisfies \cref{block metrics}, $a = a_S \otimes \mathbb{1}_{\Lambda \setminus S}$ is quasi-Hermitian with respect to $\eta$.

Proof of $\textit{1} \Rightarrow \textit{3}$: Using \cref{diagonalizable}, and noting that $a$ is diagonalizable if and only if $a_S$ is diagonalizable
, $a_S$ has a diagonalization $a_S = S D S^{-1}$, with $D = D^{\dag}$. Substituting this into \cref{block metrics},
\begin{equation}
U^{\dag} \left(\otimes_{j \in S} \eta_{i,j} \right) U D = D U^{\dag} \left(\otimes_{j \in S} \eta_{i,j} \right) U.
\end{equation}
For $a$ to be distinct from a multiple of the identity, there must be at least two distinct diagonal elements in $D$. Denoting two such elements as $d_1 \neq d_2$, given two associated eigenvectors $\ket{d_1}$ and $\ket{d_2}$ satisfying $D \ket{d_i} = d_i$ for $i \in \{1,2\}$, the matrix elements of $U^{\dag} \left(\otimes_{j \in S} \eta_{i,j} \right) U$ vanish
\begin{equation}
\braket{d_1|U^{\dag} \left(\otimes_{j \in S} \eta_{i,j} \right) U D|d_2} = \braket{d_1|D U^{\dag} \left(\otimes_{j \in S} \eta_{i,j} \right) U|d_2} = 0.
\end{equation}
Thus, $U^{\dag} \left(\otimes_{j \in S} \eta_{i,j} \right) U$ is reducible to the eigenspaces of $D$, completing this direction of the proof.
 
Proof of $\textit{3} \Rightarrow \textit{1}$: If $S$ satisfying \cref{reduce metric} exists, then the equation a local quasi-Hermitian observable satisfies, \cref{block metrics}, can be rewritten as 
\begin{equation}
\begin{pmatrix}
\eta^i_R & 0 \cr
0 & \eta^i_{S \setminus R}
\end{pmatrix} (U^{-1} a_S U) = (U^{-1} a_S U)^{\dag} \begin{pmatrix}
\eta^i_R & 0 \cr
0 & \eta^i_{S \setminus R}
\end{pmatrix}.
\end{equation}
A nontrivial solution for $a_S$ can be constructed with distinct eigenvalues $d_1, d_2$,
\begin{equation}
a_S = U \begin{pmatrix}
d_1 \mathbbm{1}_R & 0 \\
0 & d_2 \mathbbm{1}_{S \setminus R}
\end{pmatrix} U^{-1}.
\end{equation}
\end{proof}

While \cref{quasilocal_theorem} tells us about nontrivial observables, it does not tell us what the criterion for the existence of extensive observables is.

\begin{corollary} \label{corollary:Schmidt-Rank}
Assuming $\dim(\mathcal{H}_S) < \infty$, if the Schmidt rank of the metric is greater than $( \dim(\mathcal{H}_S)-1)^2 + 1$, no local observables exist in the algebra $\mathfrak{A}_S(\eta^{1/2})$ other than multiples of the identity.
\end{corollary}
\begin{proof}
The proof proceeds by contradiction. Assume a nontrivial local observable exists in $\mathfrak{A}_S(\eta^{1/2})$. By \cref{quasilocal_theorem}, the Schmidt operators, $\otimes_{i \in S}\eta_{i,S}$, must be simultaneously reducible. Let $|R|$ denote the dimension of the one of the invariant subspaces of $U^\dag \eta_{i,S} U$ as given in \cref{reduce metric}. The decomposition fixes $2(\dim(\mathcal{H}_S)-|R|)|R|$ matrix elements of each $\eta^i_B$ in a suitable basis, which is minimized by a block of size $|R| = 1$. This leaves $(\dim(\mathcal{H}_S)-1)^2+1$ unfixed parameters in $\otimes_{i \in S}\eta_{i,S}$. If the dimension of $\text{span} \{ \otimes_{i \in S}\eta_{i,S} \}$ exceeds this bound, the Schmidt operators must be linearly dependent, a contradiction.
\end{proof}

\subsection{Schmidt Rank Less Than Three}
This section demonstrates that local quasi-Hermitian observables always exist when the metric operator, $\eta \in \mathcal{B}(\mathcal{H}_A \otimes \mathcal{H}_B)$, has Schmidt rank less than three. 

The case where the Schmidt rank of the metric equals one is trivial. Suppose the Schmidt rank of the metric equals two. Then, by \cite{Cariello2014,DelasCuevas2019}, the metric is separable and can be expressed as 
\begin{align}
\eta = \sigma_1 \otimes \tau_1 + \sigma_2 \otimes \tau_2,
\end{align}
where $\sigma_i, \tau_i$ are positive semi-definite. Furthermore, at least one of each pair of $\sigma_1, \tau_2$ and $\sigma_2, \tau_1$ must be positive-definite. If not, then consider $v \otimes w$ to be an element of the null space of $\sigma_1 \otimes \tau_2$ or $\sigma_2 \otimes \tau_1$, respectively. Then, $\eta v \otimes w = 0$, a contradiction. 

Assume without loss of generality that $\sigma_1$ and $\tau_2$ are positive-definite. Then, the operators $\mathbb{1} \otimes \tau_2^{-1} \tau_1$ and $\sigma_1^{-1} \sigma_2 \otimes \mathbb{1}$ are local quasi-Hermitian observables.

\subsection{$\mathcal{PT}$-Symmetry and Local Observables}

Suppose the metric operator has an injective antilinear symmetry, $\mathcal{PT}$. This holds for the common choice of metric operator $\eta = \mathcal{PC}$ associated with a $\mathcal{C}$-symmetry \cite{bender2002complex}. If $\mathcal{PT} \mathfrak{A}_{S_1} \mathcal{PT} \subseteq \mathfrak{A}_{S_2}$ for some subsets $S_1, S_2 \in \Lambda$ and if $\mathfrak{A}_{S_1}$ contains a nontrivial quasi-Hermitian observable, $O$, then the nontrivial quasi-Hermitian observable $\mathcal{PT} O \mathcal{PT}$ is local to $S_2$.

Later examples consider a spatial system with $n$ sites, $\Lambda = \{1, \dots, n\}$, where $\mathcal{PT}$ maps $\mathcal{H}_i$ to $\mathcal{H}_{n-i+1}$.


\section{Local Equivalence of Quasi-Hermitian and Hermitian Kinematics} \label{Local-Equivalence-Section}

This section demonstrates how local expectation values in either the commuting operator model or in the tensor product model of quasi-Hermitian theory can be equivalently computed using Hermitian observables and states in the same model of locality. An immediate corollary of this is that introduction of quasi-Hermitian strategies cannot increase the values of nonlocal games; this is readily deduced from the definition of the value of a nonlocal game.

\subsection{Commuting Operator Model}

Consider the commuting operator model, where the only requirements placed on the local observable algebras are the isotony and subsystem independence constraints of \cref{isotony,disjointCommute}. The isometric $*$-isomorphism defined by
\begin{align}
\phi_{\eta^{-1/2}}(a) := \eta^{1/2} a \eta^{-1/2} \label{isometry2}
\end{align}
that was considered in the introduction chapter in \cref{quasi-Herm-C*-automorphism} is all that is needed to construct the equivalent state. This state is 
\begin{align}
\varphi_H = \varphi \circ \phi_{\eta^{1/2}}.
\end{align}
The algebra automorphism $\phi_{\eta^{-1/2}}:\mathfrak{A}(\eta^{1/2}) \to \mathfrak{A}$ maps commuting operator algebras into commuting operator algebras.

\subsection{Tensor Product Model} \label{TPM-Equiv-State-Section}

Proving local equivalence in the tensor product model is trickier than the demonstration for the commuting operator model given in the previous section. The similarity transform $\phi_{\eta^{-1/2}}$, of \cref{isometry2}, maps algebras which factorize as a tensor product into algebras which no longer factorize with this tensor product. 

The following table summarizes the steps performed to demonstrate local equivalence in the tensor product model.
\begin{table}[htp!]
\centering
\begin{tabular}{ll}
$\varphi$     & $:\mathfrak{A}(\Omega) \to \mathbb{C}$        \\
$\downarrow$  & (restriction)                                 \\
$\varphi|_{\mathfrak{A}_{\text{loc}}(\Omega)}$ & $:\mathfrak{A}_{\text{loc}}(\Omega) \to \mathbb{C}$ \\
$\downarrow$  & (Hahn-Banach)                                 \\
$\varphi_\xi$ & $:\mathfrak{A}(\xi^{1/2}) \to \mathbb{C}$     \\
$\downarrow$  & (local similarity)                            \\
$\varphi_H$   & $:\mathfrak{A} \to \mathbb{C}$
\end{tabular}
\caption{Summary of steps performed to prove local equivalence.}
\end{table}

The following discussion fills in the details for the steps summarized above.
\begin{itemize}
\item For a quasi-Hermitian state $\varphi$ on $\mathfrak{A}(\Omega)$, consider the restricted functional 
$\varphi|_{\mathfrak{A}_{\text{loc}}(\Omega)}: \mathfrak{A}_{\text{loc}}(\Omega) \rightarrow \mathbb{C}$.
Note, by \cref{Adjoint-Equivalence-Theorem}, adjoints can be equivalently computed using either the involution defined by $\Omega$ or an involution defined by a positive-definite operator admitting a tensor product factorization, $\xi$. The restricted state $\varphi|_{\mathfrak{A}_{\text{loc}}(\Omega)}$, thus, maps positive elements of the form $a^{\dag_\xi} a$ with $a \in \mathfrak{A}_{\text{loc}}(\Omega)$ to positive numbers. Therefore, $\varphi|_{\mathfrak{A}_{\text{loc}}} (\Omega)$ be interpreted as a state on a $C^*$-algebra with an involution defined by $\xi$.
\item
The Hahn-Banach theorem \cite[prop. 2.3.24]{bratteli1987operator} guarantees the existence of an extension of the state $\varphi|_{\mathfrak{A}_{\text{loc}}(\Omega)}$ on $\mathfrak{A}(\xi^{1/2})$, $\varphi_\xi$, which is also a state. I emphasize that for all $a \in \mathfrak{A}_{\text{loc}}(\Omega)$, the expectation values $\varphi_\xi(a) = \varphi(a)$ are equivalent.
\item
Consider the isometric $*$-isomorphism $\phi_{\xi^{-1/2}}:\mathfrak{A}(\xi) \to \mathfrak{A}$ given in \cref{isometry2}. This isomorphism is interpreted as a \textit{local} similarity transform, since $\xi$ admits a tensor product factorization. Following \cite{williams1969operators}, note $\phi_{\xi^{-1/2}}$ maps a local quasi-Hermitian element with metric $\xi$ to a local Hermitian element. Thus, I define a state on the Hermitian algebra,
\begin{align}
\varphi_H := \varphi_\xi \circ \phi_{\xi^{1/2}}.
\end{align}
This state is equivalent to $\varphi_\xi$ in the sense that the expectation value of a local observable $a \in \mathfrak{A}_{\text{loc}}(\Omega)$ can be equivalently computed as the expectation value of a local Hermitian observable, $\varphi_H(\phi_{\xi^{-1/2}}(a))$.
\end{itemize}
In summary, the expectation value of any local quasi-Hermitian observable can equivalently be computed as the expectation value of the same local operator in the new algebra, $\mathfrak{A}(\xi^{1/2})$, using the state $\varphi_\xi$.

%
%
%

\subsection{Construction of the Equivalent State}

The proof of the existence of the equivalent Hermitian state in the tensor product model given in \cref{TPM-Equiv-State-Section} is non-constructive. This section discusses under what circumstances an explicit equivalent state can be computed.

Suppose there exists an onto, completely-positive projection, $\Phi:\mathfrak{A}(\xi^{1/2}) \to \mathfrak{A}_{\text{loc}}$. The onto condition forces $\Phi(a) = a$ whenever $a \in \mathfrak{A}_{\text{loc}}$.
Then, given a state $\varphi:\mathfrak{A}(\Omega) \to \mathbb{C}$, an explicit construction for the Hermitian state would be
\begin{align}
\omega_H := \omega \circ \Phi \circ \phi_{\xi^{1/2}}.
\end{align}
While this certainly would be nice, $\Phi$ only exists if $\mathfrak{A}_{\text{loc}}$ is injective in the category whose objects are operator systems and whose morphisms are completely positive maps \cite[prop. 15.1]{paulsen2002completely}. Not all $C^*$ algebras are injective \cite{Huruya1985}. However, Arveson's extension theorem guarantees the injectivity of the $C^*$-algebra of bounded operators on a Hilbert space, $\mathcal{B}(\mathcal{H})$ \cite{Arveson1969}.

A simple example where a completely-positive projection admits a closed-form expression occurs when the domain and codomain are algebras of matrices over $\mathbb{C}$. Explicitly, consider the canonical basis elements, $E_{ij}$, of the algebra of $n \times n$ matrices,
\begin{align}
(E_{ij})_{kl} := \delta_{ik} \delta_{jl}.
\end{align}
This basis is orthonormal under the Frobenius inner product. Thus, complete-positivity of the map $\Phi:\mathfrak{M}_n(\mathbb{C}) \to \mathfrak{M}_m(\mathbb{C})$ with $m \leq n$ defined by
\begin{align}
\Phi(a) := \sum_{i, j = 1}^m E_{ij} \text{Tr}(E_{ji} a)
\end{align}
follows from Choi's theorem\footnote{A similar computation was performed in \cite{Breuer}.}~\ref{Choi-Theorem} \cite{Choi1975}.

\section{Locality in Quasi-Hermitian Free Fermions} \label{FermionObservableAlgebraSection}

This section considers the framework of local fermions given in \cref{Lattice-Fermion-Locality-Defn} and \cite{BravyiKitaev}.

\subsection{Primer on Fermions} \label{FermionsIntro}

The building blocks of fermionic models in quantum theory are presented in this section. A reader who is familiar with this topic can safely skip this section, although I advise reviewing the notation of \cref{annihilation-vector} and \cref{defn:vacuum}. A more thorough introduction to free fermions in Hermitian quantum theory is given in \cite{nielsen2005fermionic} and an introduction geared towards the mathematically inclined reader with a curiosity for infinite-dimensional settings is given in \cite[\S 5.2]{Bratteli1997}. The latter introduction is based on the original work of \cite{Cook1953}.

Suppose the objective of modelling a single particle which resides in some physical space has been achieved. A subsequent goal would be to model a system with more than one particle residing in the same space. In the case where the particles are \textit{indistinguishable}, their states need to satisfy \textit{exchange statistics}. Fermions are particles which satisfy \textit{Fermi-Dirac statistics}, which means a multi-particle state picks up a negative sign when a pair of its single-particle states are swapped. A physical consequence of these statistics is that fermions obey the \textit{Pauli exclusion principle} \cite{Pauli1925,Pauli1940}: two fermions cannot co-exist in the same state. 

A mathematical setting adapted to fermionic systems is the \textit{canonical anticommutation relations algebra}, or \textit{CAR algebra} for short, over a complex inner product space \cite{Cook1953}. In this framework, multi-particle states are generated by applying \textit{creation operators} associated to single-particle states on the \textit{vacuum}. The reason why the CAR algebra implements the Pauli exclusion principle is because acting with a creation operator associated with the same state twice will result in the zero vector, so by design, no two particles can be in the same state.

Rather than directly working in the CAR algebra, this thesis only utilizes representations of the CAR algebra as operators on a Hilbert space. That means the following definition of the CAR algebra is sufficient:
\begin{defn} \label{defn-CAR}
Let $X$ denote a complex inner product space. Then, a \textit{representation of the CAR algebra} over $X$ on a Hilbert space, $\mathcal{H}$, is an antilinear map, $\hat{a}:X \to \mathcal{B}(\mathcal{H})$, which satisfies
\begin{align}
\hat{a}(v) \hat{a}(w) + \hat{a}(w) \hat{a}(v) &= 0, \label{general-CAR-1} \\
\hat{a}(v)^\dag \hat{a}(w) + \hat{a}(w) \hat{a}(v)^\dag &= \braket{v|w} \mathbb{1} \label{general-CAR-2}
\end{align}
for all $v, w \in X$. An operator in $\mathcal{R}(\hat{a})$ is referred to as an \textit{annihilation operator} and an operator in $\mathcal{R}(\hat{a})^\dag$ is referred to as a \textit{creation operator}.
\end{defn}

Explicit representations of the CAR algebra are given on \textit{Fock space} \cite{Fock1932}; for details, see \cite[Chp. 5]{Bratteli1997}.

The most basic state in a many-body Hilbert space is the \textit{vacuum}, which is defined for fermionic systems below.
\begin{defn} \label{defn:vacuum}
A \textit{vacuum} is a state in the kernel of $\hat{a}(v)$ for all $v \in \text{Dom}(\hat{a})$. Every CAR algebra representation has at least one vacuum. Starting in \cref{section-second-quantize} and continuing for the rest of this thesis, the vacuum is assumed to be \textit{unique} and will be denoted by $\ket{\emptyset}$. 
\end{defn}

The remainder of this section considers the finite-dimensional case with $X = \mathbb{C}^n$. Letting \gls{ei} denote the canonical basis of $\mathbb{C}^n$, I define the shorthand 
\begin{align}
\hat{a}_i := \hat{a}(e_i), &\quad& \hat{a}^\dag_i := \hat{a}(e_i)^\dag,
\end{align}
where $i$ is an element of $\dbrac{n}$, which is
\begin{align}
\dbrac{n} := \{1, \dots, n\}.
\end{align}
The anticommutation relations of \cref{general-CAR-1,general-CAR-2} applied to $\hat{a}_i$ generates the familiar canonical-anticommutation-relations, as displayed in \cref{CAR-normal}. By linearity, the action of $\hat{a}$ on an arbitrary vector is a linear combination of the $\hat{a}_i$,
\begin{align}
\hat{a}\left(\sum_{i = 1}^n v_i e_i \right) = \sum_{i = 1}^n v^*_i \hat{a}_i. \label{annihilation-vector}
\end{align}

Representations of the CAR algebra over finite-dimensional Hilbert spaces are constructed using the \textit{Jordan-Wigner transform} \cite{JordanWigner,WignerCollectedWorks}, which is discussed below. 
\begin{proposition}[Jordan-Wigner Transform] \label{JordanWignerProposition}
A representation of the CAR algebra, $\hat{a}:\mathbb{C}^n \to \mathcal{B}(\mathcal{H})$, exists on a finite-dimensional Hilbert space, $\mathcal{H}$, if and only if the dimension of $\mathcal{H}$ is an integer multiple of $2^{n}$. Furthermore, denoting this integer multiple as $m$, there exists a unitary, $U:\mathcal{H} \to \mathbb{C}^{m 2^{n}}$, referred to as the \textit{Jordan-Wigner transform}, such that 
\begin{align}
U \hat{a}_j U^\dag &= \left(\prod_{k < j} Z_k \right) \sigma_j \otimes \mathbb{1}_{\mathbb{C}^{(m-1) 2^{n}}}\\
U \hat{a}^\dag_j U^\dag &= \left(\prod_{k < j} Z_k \right) \sigma^\dag_j \otimes \mathbb{1}_{\mathbb{C}^{(m-1) 2^{n}}},
\end{align}
where 
\begin{align}
Z_k := \bigotimes_{j = 1}^{n} \begin{cases}
\begin{pmatrix}
-1 & 0 \\
0 & 1
\end{pmatrix} & \text{ if } j = k \\
\mathbb{1} & \text{ if } j \neq k,
\end{cases} &\quad& 
\sigma_k := \bigotimes_{j = 1}^n \begin{cases}
\begin{pmatrix}
0 & 1\\
0 & 0
\end{pmatrix} & \text{ if } j = k \\
\mathbb{1} & \text{ if } j \neq k.
\end{cases}
\end{align}
\end{proposition}

The Jordan-Wigner transform provides a correspondence between a system of qubits and a system of free fermion system and, thus, can be applied to models of coupled qubits such as the XY model \cite{Lieb1961} and the transverse field Ising model \cite{Pfeuty1970}. There is some freedom in choosing the ordering of which qubit operators correspond to which fermion operators. For example, for every permutation, $p: \dbrac{n} \to \dbrac{n}$, there exists another transform, $U_p$, such that
\begin{align}
U_p \hat{a}_{p(i)} U^\dag_p &= \left(\prod_{j < i} Z_j \right) \sigma_i \otimes \mathbb{1}_{\mathbb{C}^{(m-1) 2^n}} \label{JP1} \\
U_p \hat{a}^\dag_{p(i)} U^\dag_p &= \left(\prod_{j < i} Z_j \right) \sigma^\dag_i \otimes \mathbb{1}_{\mathbb{C}^{(m-1) 2^n}}. \label{JP2}
\end{align}
Given a representation of the CAR algebra over a finite-dimensional Hilbert space with dimension $n$, if this representation has only one vacuum, by \cref{JordanWignerProposition}, the second-quantized Hilbert space is isomorphic to $\mathbb{C}^{2^n}$.

The set of creation and annihilation operators generate the entire space of operators on $\mathcal{H}$,
\begin{align}
\mathcal{B}(\mathcal{H}) = \text{span} \left \lbrace \left(\prod_{i \in S_1} \hat{a}^\dag_i \right) \left(\prod_{j \in S_2} \hat{a}^{}_j\right): S_1, S_2 \subseteq \dbrac{n}\right \rbrace. \label{OperatorSpanBasic}
\end{align}
Equivalently, the space of operators on $\mathcal{H}$ can be generated from linear combinations of creation and annihilation operators, so long as the linear combinations arise from vectors which form a linearly independent basis of $\mathbbm{C}^n = \text{span}\{v^\mu|\mu \in \dbrac{n} \} = \text{span} \{{w^\nu|\nu \in \dbrac{n} }\}$,
\begin{align}
\mathcal{B}(\mathcal{H}) = \text{span} \left \lbrace
\left(\prod\limits_{\mu \in S_1} \hat{a}(v^\mu)^\dag\right) \left( \prod\limits_{\nu \in S_2} \hat{a}(w^\nu)\right): S_1, S_2 \subseteq \dbrac{n} \right \rbrace. \label{OperatorSpan}
\end{align}

\subsection{Second-Quantized Pseudo-Hermiticity} \label{section-second-quantize}

Second-quantization is a procedure that takes a fixed quantum theory, which is referred to as \textit{first-quantized}, and maps it into a \textit{second-quantized} theory. The quintessential application of second-quantization is to generate quantum theories of non-interacting, or free, particles. While a first-quantized quantum theory describes the physics governing a single particle, its second-quantization describes a physical theory whose states are superpositions of multi-particle states. Second-quantization can generate models with either Bosonic or Fermionic statistics. I only analyse the fermionic case in this thesis.

In contrast with a generic many-body problem, many physical features of free fermionic models can be computed with computational resources which scale polynomially, rather than exponentially, with system size. Examples of such features include the ground state energy \cite{nielsen2005fermionic}, correlation functions, and entanglement entropies \cite{FreeFermionEntanglement}. Since free fermions can be solved using significantly less work than their interacting counterparts, their analysis can generate an understanding of more complicated models, such as the Hubbard model \cite{Hubbard}, via perturbation theory. 

The objective of this section is to find intertwining operators of second-quantized pseudo-Hermitian Hamiltonians. This was essentially achieved in \cite[prop. 4]{Mintchev1980}, where the problem was first solved distinguishable particles in Fock space, then solved for bosonic particles using bosonic symmetrization procedure. In contrast to \cite[prop. 4]{Mintchev1980}, I derive the intertwining operator for second-quantized fermionic Hamiltonians directly from the CAR algebra; I do not assume an explicit representation, such as one on a symmetrized Fock space.

Denote the Hilbert space of the first-quantized theory by $\mathcal{H}_1$. Second-quantization is performed using a representation of the CAR algebra over $\mathcal{H}_1$ on $\mathcal{H}$, as defined in \cref{defn-CAR}, and two maps, $d \Gamma$ and $\Gamma$, with domain $\mathcal{B}(\mathcal{H}_1)$ and codomain $\mathcal{B}(\mathcal{H})$. The map $d\Gamma$ defines Hamiltonians which are non-interacting, second-quantized, and which commute with the number operator\footnote{Due to the presence of pair creation and annihilation terms, general free-fermionic Hamiltonians do not commute with the number operator. While the general case includes physically interesting models, such as the Kitaev chain \cite{Kitaev2001}, this thesis only considers the particle-number conserving case.}. As will be seen in \cref{prop-second-quantized-pseudoHerm}, second-quantized time-evolution and intertwining operators associated to the second-quantized Hamiltonian are determined through $\Gamma$. Explicitly, $d \Gamma$ and $\Gamma$ are
\begin{align}
d \Gamma(H) &= \sum_{i} \hat{a}(H \psi_i)^\dag\hat{a}(\psi_i), \\
\Gamma &= e^{d \Gamma},
\end{align}
where $\psi_i$ denotes an orthonormal basis of $\mathcal{H}_1$. While the definition of $d \Gamma$ provided here requires a basis, $d \Gamma$ is independent of which orthonormal basis is chosen.

The next goal of this section is to quantify the algebraic properties of $d \Gamma$ and $\Gamma$. I'll start with the description of $d \Gamma$, which turns out to be an involutive Lie algebra homomorphism, a kind of structure-preserving map. More explicitly, $d \Gamma$ satisfies three key properties:
\begin{align}
[d \Gamma(h_1), d \Gamma(h_2)]_- &= d\Gamma([h_1,h_2]_-) &\quad& \forall h_1, h_2 \in \mathcal{B}(\mathcal{H}_1) \\
d \Gamma(\alpha h_1 + \beta h_2) &= \alpha \,d \Gamma(h_1) + \beta \,d \Gamma(h_2) &\quad& \forall h_1, h_2 \in \mathcal{B}(\mathcal{H}_1), \alpha, \beta \in \mathbb{C} \\
d \Gamma(h^{}_1)^\dag &= d\Gamma(h_1^\dag) &\quad& \forall h_1 \in \mathcal{B}(\mathcal{H}_1).
\end{align}
\begin{theorem} \label{prop-second-quantized-pseudoHerm}
If $M = M^\dag \in {\normalfont \text{GL}}(\mathcal{B}(\mathcal{H}_1))$ is an intertwining operator for $H \in \mathcal{B}(\mathcal{H}_1)$, then $\Gamma(\log M)$ is an invertible, Hermitian intertwining operator for $d \Gamma(H)$. If $M$ is positive-definite, then $\Gamma(\log M)$ is a metric operator rendering $d \Gamma(H)$ quasi-Hermitian.
\end{theorem}
\begin{proof}

For every $X \in \mathcal{B}(\mathcal{H})$, define $\text{ad}_X:\mathcal{B}(\mathcal{H}) \to \mathcal{B}(\mathcal{H})$ by 
\begin{align}
\text{ad}_X(Y) := [X,Y]_- &\quad& \forall Y \in \mathcal{B}(\mathcal{H}).
\end{align}
Since $d \Gamma$ is an involutive Lie algebra homomorphism, the \textit{infinitesimal Baker-Campbell-Hausdorff formula} \cite[eq. (2.1)]{muger2019notes}, which is
\begin{align}
e^X Y e^{-X} = e^{\text{ad}_X} Y,
\end{align}
implies
\begin{align}
\Gamma(X) \, d \Gamma(H) \, \Gamma(-X) = d \Gamma(e^{X^{}} H e^{-X}).
\end{align}
That $\Gamma(\log M)$ is an intertwining operator follows from the previous formula with the choice $X = \log M$. When $M$ is positive-definite, since $d \Gamma$ is a positive map, $\Gamma(\log M)$ is also positive-definite.
\end{proof}

The exponential form of the intertwining operator of a second-quantized Hamiltonian may seem ad-hoc. An alternative definition is that the intertwining operator, $\eta:= \Gamma(\log M)$, is the unique solution to the equations
\begin{align}
\eta \hat{a}(\psi)^{\dag} &= \hat{a}(M \psi)^{\dag} \eta \label{reduced metric to metric} \\
\eta \ket{\emptyset} &= \ket{\emptyset}.
\end{align} 
If $M$ is invertible, an expression for the action of $\eta$ on the annihilation operators is 
\begin{align}
\eta \hat{a}(\psi) &= \hat{a}(M^{-1} \psi) \eta.
\end{align}

For every invertible first-quantized intertwining operator, $M = M^\dag \in \text{GL}(\mathcal{B}(\mathcal{H}_1))$, the \textit{number operator},
\begin{align}
\hat{n}:&= \sum_i \hat{a}(\psi_i)^{\dag} \hat{a}(\psi_i),
\end{align}
is pseudo-Hermitian with the intertwining operator $\Gamma(\log M)$.
Thus, given two second-quantized Hamiltonians related by a chemical potential, $d\Gamma(H') = d\Gamma(H) + \mu \hat{n}$, the spaces of their second-quantized intertwining operators are identical.

When $\mathcal{H}_1 = \mathbb{C}^n$, a neat formula for matrix-elements of $\eta$ is determined by the minors of the first-quantized metric, $M \in \mathfrak{M}_n(\mathbb{C})$. Given a subset, $S \subseteq \dbrac{n}$, the matrix elements of $\eta$ are
\begin{align}
\ket{S} &:= \prod_{i \in S} a^\dag_i \ket{\emptyset} \\
\braket{S|\eta|S'} &= \begin{cases}
\det M_{S S'} & \text{ if } \text{card} (S) = \text{card} (S')\\
0 & \text{otherwise},
\end{cases}
\end{align}
where $M_{S S'}$ denotes the matrix formed by taking all rows and columns with indices in $S$ and $S'$ respectively.

\subsection{Second-Quantized Hamiltonian Eigensystems}
With second-quantization of pseudo-Hermitian systems established in the previous section, we are now ready to address the eigenvalue problem associated to second-quantized Hamiltonians. Notably, the eigensystem of the second-quantized Hamiltonian $d \Gamma(H)$ is determined by the eigensystem of $H$. For simplicity, this section is devoted to the finite-dimensional case with $\text{Dom}(\hat{a}) = \mathbb{C}^n$.

A subset of the spectrum of $d \Gamma(H)$ is determined by the spectrum of $H$,
\begin{align}
\left\{\sum_{\lambda \in S} \lambda \,|\, S \in \mathbb{P}(\sigma(H)) \right\} \subseteq \sigma(d\Gamma(H)),
\end{align}
where equality is obtained if $H$ is diagonalizable. The eigenspaces of $d \Gamma(H)$ can be described after introducing some notation. Given a vector space, $V \subseteq \mathbb{C}^n$, let 
\begin{align}
\mathcal{H}_V := \text{span}
\left\{ \left(\prod\limits_{i = 1}^{\dim{V}} \hat{a}(v_i)^\dag \right) \ket{\emptyset} \,|\, v \in \times_{i=1}^{\dim V} V \right\},
\end{align}
where $\times$ denotes the Cartesian product, so $v$ refers to a tuple of vectors in the above expression.
Given a subset, $S \subseteq \sigma(H)$, of eigenvalues of $H$, a corresponding set of eigenvectors of $d \Gamma(H)$ is generated by corresponding eigenvectors of $H$,
\begin{align}
\mathcal{H}_{\oplus_{\lambda \in S}\ker(\lambda \mathbb{1} - H)} \subseteq \ker\left(\left(\sum\limits_{\lambda \in S} \lambda \right) \mathbb{1} - d \Gamma(H)\right).
\end{align}
Two special cases of eigenvalue/eigenvector pairs of the second-quantized Hamiltonian are
\begin{align}
\ket{\emptyset} &\in \ker(d \Gamma(H)) \\
\left(\prod_{i = 1}^n \hat{a}^\dag_i \right) \ket{\emptyset} &\in \ker \left(d \Gamma(H) - \sum_{\lambda \in \sigma(H)} \lambda \mathbb{1} \right).
\end{align} 

Assume $H$ is diagonalizable, so it can be expressed as
\begin{align}
H_{ij} = \sum_k U_{ik} \epsilon_k U^{-1}_{kj}.
\end{align}
In this case, a set of \textit{pseudo-Fermionic operators} \cite{Bagarello2022Book} is used to diagonalize the Hamiltonian. These operators are
\begin{align}
\hat{c}_k &= \sum_i U_{ik} \hat{a}^\dag_i,\\
\hat{d}_k &= \sum_j U^{-1}_{kj} \hat{a}_j,
\end{align}
which satisfy the relations 
\begin{align}
[\hat{c}_i,\hat{d}_j]_+ &= \delta_{ij} \mathbbm{1}, \\
[\hat{c}_i,\hat{c}_j]_+ &= 0 = [\hat{d}_i,\hat{d}_j], 
\end{align}
and which diagonalize the Hamiltonian via 
\begin{align}
d\Gamma(H) &= \sum_{k} \epsilon_k \hat{c}_k \hat{d}_k.
\end{align}
Note the vacuum is in the kernel of every $\hat{d}_i$.

Assuming $H$ is quasi-Hermitian, a metric for $d \Gamma(H)$ can be constructed from a set of eigenstates of $d \Gamma(H)^\dag$ \cite{BiOrthogonal}. In particular, one class of metrics for free fermions is
\begin{align}
\eta &= \sum_{S} \kket{S_\Xi} \bbra{S_\Xi} \\
\kket{S_\Xi} &:= \prod_{i \in S} \sum_{j} \Xi_{ij}^{1/2} \hat{d}^\dag_k \ket{\emptyset}, \label{free_fermion_metric}
\end{align}
where $S$ is a subset of indices, $S \subseteq \dbrac{n}$, and $\Xi$ is Hermitian, positive-definite, and commutes with $\text{diag}(\epsilon)$, its role is to select a possibly different set of eigenvectors of $H$. In this case, we have 
\begin{align}
M = (U \Xi U^\dag)^{-1}.
\end{align}


Motivated by \cite{Korff2008}, this metric can be equivalently expressed as the unique solution to the operator equations
\begin{align}
\eta \hat{c}_k &= \sum_l \Xi_{k l} \hat{d}^{\dag}_l \eta, \label{rule1}\\
\eta \ket{\emptyset} &= \ket{\emptyset}. \label{rule2}
\end{align}
The condition \cref{rule2} is imposed for Hermiticity: if $\eta$ satisfies \cref{rule1} as well as Hermiticity,
\begin{equation}
0 = \eta \hat{d}_l \ket{\emptyset} = \sum_{k} \Xi^{-1}_{lk} \hat{c}^{\dag}_k \eta \ket{\emptyset}.
\end{equation}
Multiplying by $\Xi U^{-1}$, this reduces to
\begin{align}
\hat{a}_i \eta \ket{\emptyset} = 0 &\quad& \forall i \in \dbrac{n}.
\end{align} 
Since the vacuum is unique, $\eta \ket{\emptyset}$ must be a multiple of the vacuum, $\eta \ket{\emptyset} = \eta_{00} \ket{\emptyset}$. Rescaling the metric preserves expectation values, so $\eta_{00}$ can be set to 1.

The metric of \cref{free_fermion_metric} is not the most general metric associated with $d \Gamma(H)$; if there are degeneracies in the spectrum of $d \Gamma(H)$ which are not present\footnote{One simple example where this happens is when the first-quantized Hamiltonian has a chiral symmetry, so that the eigenvalues of $H$ come in pairs $(\epsilon, -\epsilon)$.} in $H$, then there are metric operators associated with $d \Gamma(H)$ which are not of the form \cref{free_fermion_metric}. Furthermore, a general metric can have distinct actions in the different eigenspaces of the number operator.

The operators defined by
\begin{align}
\hat{n}_k &:= \hat{c}_k \hat{d}_k
\end{align}
are examples of quasi-Hermitian observables.

\subsection{Local Observable Algebras} \label{free fermion locality}
A naive attempt to construct local observable algebras for models of fermions would be to use the Jordan-Wigner transform to map the problem into a tensor product model, and then  
apply results that hold in that case, such as \cref{quasilocal_theorem}. There are two issues with this approach. One is simply aesthetic, it would be desirable to generate the local observable algebras directly from the canonical anticommutation relations, rather than from a Hilbert space representation. A more severe issue is that the Hilbert space dimension scales exponentially with lattice size. In the likely scenario where numerics are needed to solve resulting operator equations, deriving local observable algebras can become impossible even for relatively small system sizes.

This section characterizes local observable algebras associated to systems with $n$ quasi-Hermitian free fermions. 

The technique is to reduce the problem to one in the first-quantized setting. A relationship between first and second-quantized quasi-Hermitian free fermions is detailed in \cref{FermionsIntro}. The computations in this section will be performed with the metric operator defined in \cref{reduced metric to metric}.

A correspondence between first-quantized observables algebras and second-quantized observable algebras is desired. A first step towards this correspondence comes from the observation that if $o$ is quasi-Hermitian with the metric $M$ 
if and only if 
\begin{equation}
O = \sum_{ij} o_{i j} \hat{a}^\dag_i \hat{a}_j \label{proj_obs} 
\end{equation}
is quasi-Hermitian observable with respect to a metric $\eta$ which reduces to $M$ via \cref{reduced metric to metric}. 
The subclass of operators of the form \cref{proj_obs} which are local to a subsystem $A \subset \dbrac{n}$ are those satisfying $o_{ij} = 0$ if $i \in A^c$ or $j \in A^c$, where 
\begin{align}
A^c := \dbrac{n}\setminus A
\end{align}
denotes the complement of $A$. Note $o$ is local in the sense of \cref{finite-lattice-locality}. Extensively local observables satisfy the additional constraint that for all $i \in A$, there exists $j \in A$ such that either $o_{ij} \neq 0$ or $o_{ji} \neq 0$. Let us refer to such matrices $o$ as extensively local reduced observables. The existence of extensively local reduced observables turns out to be necessary for the existence of extensively local quasi-Hermitian observables for the free fermion metric of \cref{reduced metric to metric}, as will be shown shortly.

Before proving this result, I will elaborate on extensively local reduced observables. 
Since extensively local reduced observables are block matrices, let us define some notation relating to block decompositions of matrices. Let $M^{AB}$ denote the block of matrix elements $M_{ij}$ with $i \in A, j \in B$. In particular (rearranging columns and rows in $M$ as necessary),
\begin{equation}
M = \begin{pmatrix}
M^{AA} & M^{A A^c} \\ M^{A^c A} & M^{A^c A^c}
 \end{pmatrix},
\end{equation}
and $M^{\{i\} \dbrac{n}}$ denotes the $i^\text{th}$ row of $M$. When necessary, the matrix elements $M^{AB}$ will be considered as an operator mapping $\text{span} \{e_i:i\in B\}$ to $\text{span} \{e_i:i\in A\}$, where $\gls{ei}$ denotes the canonical basis of $\mathbb{C}^n$.

In addition, let $K(A) := \dim \ker M^{A^c A}$, and let $\{w^\mu | \mu \in \dbrac{K(A)}\}$ denote a basis of $\ker M^{A^c A}$. The most general local quasi-Hermitian reduced observable was determined in \cref{quasi-Hermitian-Lattice-Locality}; $K(A)$ is a measure of how many local observables are in subsystem $A$.

The next theorem relates the properties of observables containing an arbitrary number of products of the creation and annihilation operators to the properties of $\ker M^{A^c A}$.
\begin{theorem} \label{polynomial theorem}
Extensively local observables in subsystem $A$ which are quasi-Hermitian with respect to the metric of \cref{reduced metric to metric} exist if and only if 
\begin{equation} 
K(A) := \dim \ker M^{A^c A} > \dim \ker M^{S^c S} \label{Kernel equation} 
\end{equation} for all proper subsets $S \subsetneq A$. In addition, an operator, $O \in \mathfrak{A}_A(\eta^{1/2})$, is a local quasi-Hermitian observable if and only if it can be expressed of the form
\begin{equation}
O = \sum_{{\mtiny S_1, S_2}} O_{S_1 S_2} \left(\prod_{\mu \in S_1} \hat{a}(w^\mu)^\dag\right) \left( \prod_{\nu \in S_2} \hat{a}(M^{AA} w^\nu) \right)\label{GeneralBKQH},
\end{equation}
where $S_1, S_2 \in \mathbbm{P}(\dbrac{K(A)})$, $O_{S_1 S_2} = O^*_{S_2 S_1}$, $O_{S_1 S_2} = 0$ when $|S_1|+|S_2| \equiv 1 \mod 2$, and $\{w^\mu|\mu \in \dbrac{K(A)} \}$ is a basis of $\ker M^{A^c A}$. 
\end{theorem}
\begin{proof}
Let $O_A \in \mathfrak{A}_A \setminus \{0\}$. Using \cref{OperatorSpan}, for every $O_A$, there exists linearly independent sets of vectors $\{f^\mu| \mu \in \dbrac{\mathcal{F}}\}, \{g^\nu|\nu \in \dbrac{\mathcal{G}}\}\subset \mathbbm{C}^n$ satisfying $f^\mu,g^\nu \in \text{span}\{e_i:i\in A\}$ such that
\begin{align}
O_A &= \sum_{S_1 \in \mathbbm{P}(\dbrac{\mathcal{F}})}\sum_{S_2 \in \mathbbm{P}(\mathcal{G})} O_{S_1 S_2} \left(\prod_{\mu \in S_1} \hat{a}(f^\mu)^\dag\right) \left( \prod_{\nu \in S_2} \hat{a}(g^\nu) \right)
, \label{sum}
\end{align}
where $O_{S_1 S_2} \in \mathbbm{C}$. Without loss of generality, $f^\mu$ and $g^\nu$ can be chosen such that every such vector appears in the sum in \cref{sum} at least once. The quasi-Hermiticity condition applied to $O_A$ implies
\begin{align}
O^\dag &= \eta O \eta^{-1} \\
&= \sum_{S_1 \in \mathbbm{P}(\dbrac{\mathcal{F}})}\sum_{S_2 \in \mathbbm{P}(\dbrac{\mathcal{G})}} O_{S_1 S_2} \left( \prod_{\mu \in S_1} \hat{a}(M f^\mu)^\dag \right) \left( \prod_{\nu \in S_2} a(M^{-1} g^\nu) \right).
\end{align}
Requiring $O^\dag \in \mathfrak{A}_A$ results in the following vector identities 
\begin{equation}
f^\mu \in \ker M^{A^c A},\quad g^\mu \in \ker {M^{-1}}^{A^c A}.
\end{equation}
Note $2\times 2$ matrix inversion results in the following kernel identity,
\begin{equation}
\ker {M^{-1}}^{A^c A} = M^{A A}( \ker M^{A^c A}).
\end{equation}
Thus, expressing $f,g$ in terms of a basis $\{w^\mu|\mu\in \dbrac{K(A)}\}$ of $\ker M^{A^c A}$, the operator $O_A$ can be re-expressed in the form of \cref{GeneralBKQH}. The constraint $O_{S_1 S_2} = O^*_{S_2 S_1}$ follows from demanding quasi-Hermiticity.

Importantly, note that extensively local reduced observables exist if and only if $K(A) > K(S)$ for all proper subsets $S \subsetneq A$. Intuitively, the vectors in $\ker M^{S^c S}$ correspond to observables local to $S$, so vectors in $\ker M^{A^c A}$ which do not belong to any subset $\ker M^{S^c S}$ cannot correspond to an observable local to any subset $S$, thus, their corresponding observables are extensively local to $A$.
\end{proof}

Due to the potentially nonlocal string of $Z$ factors in the Jordan-Wigner transform, an observable dictated by the previous theorem in \cref{proj_obs} is not necessarily local in the sense given by a tensor product structure. Generalizing \cref{polynomial theorem} to address the notion of locality defined via a Jordan-Wigner transform is not so simple. However, consider the case of a \textit{connected} subsystem: so that for all $i, j \in A$, every integer $k \in \dbrac{n}$ satisfying $i<k<j$ is also in the subsystem $k \in A$. In this case, the inserted $Z$ factors in \cref{proj_obs} are local, so an observable which is local to a connected subsystem is additionally local in the tensor product model defined via the Jordan-Wigner transform with the map $p$. Note for every subsystem $A$ containing local quasi-Hermitian observables, there exists a permutation which maps it into a connected subsystem, so by \cref{JP1,JP2}, there exists a Jordan-Wigner transform such that there are local quasi-Hermitian observables in the corresponding tensor product model.

Notice that altering the diagonal entries of the metric bears no impact on the existence of observables local to a subsystem, since the diagonal entries never contribute to $M^{A^c A}$.

Simple examples of reduced metrics with an analytical understanding of locality are presented below.

Suppose the reduced metric block reduces to a set $S$, so 
\begin{align}
M_{ij} = 0 &\quad& \forall i \in S, \, \forall j \notin S. \label{M block}
\end{align}
Then, there exists an observable which is both extensively local and extensively local in the tensor product model defined by any Jordan-Wigner transform, $U_p$, in the subsystem $S$, 
\begin{equation}
\hat{n}_S = \sum_{i \in S} \hat{a}^\dag_i \hat{a}_i. \label{how many particles}
\end{equation}
This observable is not quasi-Hermitian if \cref{M block} does not hold in $S$. In particular, this implies that diagonal reduced metrics have observables in every subsystem. 


A special case of reduced metrics are those which reduce as a direct sum over $1\times 1$ and $2\times 2$ block matrices. Equivalently, there exists an \textit{associated involution}, $f:\dbrac{n}\rightarrow \dbrac{n}, f \circ f = 1$, such that $M_{i j} \neq 0 \, \Leftrightarrow \, i = f(j) \, \text{or} \, i = j$. The following corollary clarifies the relationship between the involution $f$ and the local observable algebras. 
\begin{corollary} \label{Involution theorem}
Given a reduced metric which decomposes into $1 \times 1$ and $2 \times 2$ blocks, extensively local observables exist in subsystems, $A \subset \dbrac{n}$, if and only if the subsystem's image under the reduced metric's associated involution, $f$, is itself $f(A) = A$. In addition, the most general local observable in this case is
\begin{align}
O = \sum_{{\mtiny S_1,S_2 \in \mathbbm{P}(\dbrac{A})}} O_{S_1 S_2} \left(\prod_{i \in S_1} \hat{a}(e_i)^\dag\right) \left( \prod_{j \in S_2} \hat{a}(M^{AA} e_j) \right)\label{General2x2BKQH},
\end{align}
where $O_{S_1 S_2} = O^*_{S_2 S_1}$. 
\end{corollary}
\begin{proof}
The construction of extensively local observables in subsystems closed under $f$ follows from $M^{A^c A} = 0$, so that $\ker M^{A^c A} = \text{span}\{e_i :i \in A\}$.

Given a site $i \in A$ whose dual satisfies $f(i) \notin A$, and a vector $v \in \ker M^{A^cA}$, the equation $M^{\{f(i)\} A} v = 0$ immediately implies $v_i = 0$, and the involution symmetry of $M$ implies $(M^{AA} v)_i = 0$. As a consequence, any observable of the form \cref{General2x2BKQH} is local to the subsystem $S = A\setminus\{i\}$, and, therefore, is not extensive.
\end{proof}


The special case of observables localized at a single site is quite simple to analyse.
\begin{corollary} \label{single site theorem}
Quasi-Hermitian observables, with respect to the metric of \cref{reduced metric to metric}, which are extensively local to a single site, $i$, exist if and only if the reduced metric is a block matrix, $M_{ij} = M_i \delta_{ij}$.
\end{corollary}
\begin{proof}
If $M$ block reduces, $\hat{n}_{\{i\}}$ is an observable. If $M$ does not block reduce, $\ker M^{A^c A} = \emptyset$, and there are no observables.
\end{proof}

Corollary~\eqref{Involution theorem} is quite strong when applied to the toy model introduced in \cref{nearestNeighbour}, that is a one-dimensional lattice model with a pair of non-Hermitian defect potentials at the center of the chain. To summarize earlier results, a one-parameter family of reduced metric decomposes into parity blocks,
\begin{equation}
M_{ij} \neq 0 \Leftrightarrow i = \bar{j},
\end{equation}
where $\bar{j} := n-j+1$. Denote this parameter as $\beta$.
Its matrix elements are given by the following recurrence relations
\begin{align}
\begin{pmatrix}
M_{mm} & M_{m \, m+1} \\
M_{m+1 \, m} & M_{m+1 \, m+1}
\end{pmatrix}
&= \begin{pmatrix}
1 & \displaystyle\frac{\beta - i  \gamma}{t_m}\cr
\displaystyle\frac{\beta + i  \gamma}{t^{*}_m} & 1
\end{pmatrix} \nonumber
\\
M_{ii} &= \displaystyle\frac{t_{n-i}}{t_i} M_{i+1\, i+1} \nonumber \\
M_{i \bar{i}} = M^*_{\bar{i} i} &= \displaystyle\frac{t_{n-i}}{t^*_i} M_{i+1\, n-i} \nonumber \\ M_{\bar{i} \bar{i}}   &= \displaystyle\frac{t_{n-i}}{t_i} M_{n-i\, n-i} \label{badass metric},
\end{align}
where $\gamma \in \mathbb{R}, t_i \in \mathbb{C}$ are physical parameters which appear in the Hamiltonian. Positivity of the metric requires $\gamma^2 + \beta^2 < |t_m|^2$. This operator block decomposes into parity sectors. 
Thus, for the metric of \cref{badass metric} with either $\gamma \neq 0$ or $\beta \neq 0$, \textit{extensively local observables exist in and only in Parity symmetric subsystems.} In the case of $\gamma = \beta = 0$, the metric is diagonal, so there are extensively local observables in \textit{every} subsystem.
In this case, the off-diagonal elements of the Hamiltonian are irrelevant in determining which subsystems contain local observables.  

\subsection{One-Dimensional Chain with Farthest Impurities} \label{ApplicationOfTheorems}

This section will apply the construction of local observable algebras associated a general model of free fermions, given in \cref{polynomial theorem}, to a $\mathcal{PT}$-symmetric tight-binding model studied in \cref{uniformSection}. This is a one-dimensional lattice model with uniform hopping-amplitudes, open boundary conditions, and a pair of non-Hermitian defect potentials at the edges of the chain. The important feature needed here is the matrix elements of the reduced metric 
\cite{farImpurityMetric,Ruzicka2015},
\begin{equation}
M_{ij} = \begin{cases}
\,\,\,\,\,1 &  i = j \\
-\mathfrak{i} \, \gamma \, \left(\Delta - \mathfrak{i} \gamma\right)^{j-i-1} &  i < j \\
\,\,\,\,\,\mathfrak{i} \,\gamma \, \left(\Delta + \mathfrak{i} \gamma \right)^{i-j-1} &  i > j
\end{cases}, \label{not positive2}
\end{equation}
where $\gamma, \Delta$ are dimensionless real parameters\footnote{In the notation of \cref{uniformSection}, I am setting $t = 1$.}. 
Positivity of this metric operator is examined in \cref{pseudoHermCriticalChain}. To simplify select equations, I will define
\begin{align}
z := \Delta + \mathfrak{i} \gamma,
\end{align}
and I will use the shorthand $|S| := \text{card}(S)$ for the cardinality of a set.

For this model, the existence of observables local to a subsystem is related to whether the subsystem is connected. Some related notation is defined in the following paragraph:

Consider the graph $G_A = (A, E_A)$ with vertices $A$ and edges $E_A = \{(i, i+1): i, i+1 \in A\}$. Let $\mathfrak{C}_A$ denote the set of connected components of $G_A$. A distance between components, $d_A:\mathfrak{C}_A \times \mathfrak{C}_A: \rightarrow \mathbbm{R}$, is defined as
\begin{align}
d_A(C_1, C_2) = \min \left\{d(i,j):i \in C_1, j \in C_2 \right\},
\end{align}
where $d$ is the geodesic distance in $G_{\dbrac{n}}$. Intuitively, $d_A$ measures the number of sites between $C_1$ and $C_2$. Next, denote the \textit{leftmost}, $C_L$, and \textit{rightmost}, $C_R$, or collectively \textit{edge} components of $G_A$ to be the connected components containing $\min A$ and $\max A$, respectively.
Lastly, define
\begin{align}
A_{<i} &= \{k<i:i \in A\}, \\ A_{>i} &= \{k>i:i\in A\}, \label{A<} \\
v_{<i} &= \sum_{j<i} v_j e_j, \\ v_{>i} &= \sum_{j>i} v_j e_j, \label{v<}
\end{align} 
where the sums are set to zero if they sum over an empty set.

When $z^* z = 1$, the extensively local observables are comparatively simple to construct. An additional special feature of the case $z^* z = 1$ is that the spectrum is exactly computable \cite{Willms2008}. The eigenvalues are displayed in \cref{closedFormEvalsTable}.

\begin{proposition} \label{UnitDiskLocality}
For quasi-Hermitian theories with respect to the metric of \cref{not positive} and $z^* z = 1$, extensively local observables in subsystem $A$ exist if and only if either $G_A$ contains no connected components with exactly one site or $A = \{1,n\} \cup B$, where $B$ contains no connected components with exactly one site.
\end{proposition}
\begin{proof}
Suppose $A$ contains a connected component with exactly one site, $i$. Suppose there exists $v \in K(A)$. Assuming $i \notin \{1,n\}$, the kernel equation, $\sum_j M^{A^c A}_{ij} v_j = 0$, for indices $i-1,i+1$ is
\begin{align}
\begin{pmatrix}
z^{-2} & -1 & {z^*}^{2}  \\
1 & 1 & 1
\end{pmatrix} \begin{pmatrix}
M^{\{i+1\} A_{<i}}\, v_{<i} \\
\mathfrak{i} \gamma\, v_i \\
M^{\{i+1\} A_{>i}}\, v_{>i}
\end{pmatrix} = 0.
\end{align}

For the second case of the proposition, suppose $i = 1 \in A$, but $2, n \notin A$, $n>3$ (the cases $n = 2,3$ are trivial and follow from \cref{single site theorem}).  The other case, $i = n \in A, 1, n-1 \notin A$, follows from $\mathcal{PT}-$symmetry. Then the 
kernel equations at sites $2, n$ are
\begin{equation}
\begin{pmatrix}
1 & 1 \\
z^{n-2} & -(z^{*})^{4-n}
\end{pmatrix}
\begin{pmatrix}
\mathfrak{i} \gamma \, v_1 \\
M^{\{2\} A_{>1}}\, v_{>1}
\end{pmatrix} = 0.
\end{equation}

In all cases mentioned above, since $z^* z = 1$, these equations imply $v_i = 0$. In addition, note $M^{\{i\} A\setminus\{i\}} v= 0$ since
\begin{align}
M^{\{i\}\, A\setminus\{i\}} v = z^{-1} M^{\{i+1\} A_{<i}}\, v_{<i} + z^{*} M^{\{i+1\} A_{>i}}\, v_{>i} = 0.
\end{align}

Thus, $M^{S^c S} v = 0$ for $S = A \setminus \{i\}$. Thus, by \cref{polynomial theorem}, if either $A$ has a single-site connected component between the endpoints of the lattice, or exactly one of $1,n$ is in $A$, no extensively local observables exist in $A$.

The converse follows from explicit construction of extensively local observables. If $C$ is a connected subset of $A$, then $\ker M^{C^c C} = ((1, z^*, \dots {z^*}^{|C|})^\dag)^\perp$, so $K(C) = |C|-1$. If $C = \{1, n\}$, then $\left(1, z^{n-3} \right)^\intercal \in \ker M^{C^c C}$. As a consequence, extensively local observables exist in every connected subset of $\dbrac{n}$, as well as the subset $\{1,n\}$. Taking suitable linear combinations of the above vectors demonstrates the existence of extensively local observables in the subsystem $A = \{1,n\} \cup B$, where $B$ is a union of connected components with at least two sites.
\end{proof}

The remainder of the section is devoted to the case $z^* z \neq 1$. The final result is summarized in proposition~\eqref{m=1Final}. Some examples of subsystems containing local observables are shown in \cref{ExampleRegions}. 

\begin{figure}[h!]
\centering
\includegraphics[width = 120mm]{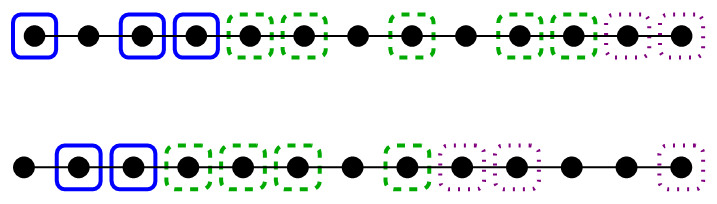}
\caption{An example of an $n = 13$ chain, depicted in black, with non-Hermitian impurities ($\Delta + \mathfrak{i} \gamma, \Delta - \mathfrak{i} \gamma$) at the edges of the chain with $|z| \neq 1$ and $\gamma \neq 0$, and hopping amplitudes $t_i = 1$. The top chain demonstrates three subsystems, shown in different colours and line styles, which contain extensively local quasi-Hermitian observables with respect to the metric of \protect\cref{not positive}. The bottom chain shows three subsystems which do not contain extensively local observables. Notably, the subsystem marked with a green, dashed line in the bottom chain does have local observables, but they are also local to the collection of its leftmost three sites. In addition, the solid blue and dashed green subsystems in the top chain only contain extensively local observables if $|z| \neq 1$.}
\label{ExampleRegions}
\end{figure}

To simplify the analysis, we start with the special case where the subsystem is connected.

\begin{lemma} \label{ConnectedRegions}
For a quasi-Hermitian model with the metric of \cref{not positive} for $|z| \neq 1$,
every subset $C \subset \dbrac{n}$ such that $G_C$ is connected with at least three sites contains extensively local observables. In addition, $\{1,2\}$ and $\{n-1,n\}$ contain extensively local observables. No other connected subgraph contains local observables. 
\end{lemma}

\begin{proof} 

By \cref{single site theorem}, if $|C| = 1$, there are no nontrivial local observables.

In the case of connected subsystems, $\dim \ker M^{C^c C}$ is easy to find, since the rows labelled by indices to the left of $C$ are all multiples of each other and, similarly, all rows labelled by indices to the right of $C$ are multiples of each other. Note the set of rows of $M^{C^c C}$ to the left or right of $C$ does not exist if either $1 \in C$ or $n \in C$, so in these cases, $K(C)$ increases by one. Thus,
\begin{align}
\ker M^{C^c C} &= \text{span} \left\{\begin{array}{l}
1_{C^c}(\{1\})(1, z^*, \dots, {z^*}^{|C|})^\dag,
1_{C^c}(\{n\})({z}^{|C|}, \dots, z, 1)^\dag
\end{array} \right\}^\perp,\\
K(C) &= |C|-2+1_A(\{1\})+1_A(\{n\}),
\end{align}
where $1_S:\mathbbm{P}(\dbrac{n})\rightarrow \mathbbm{P}(\dbrac{n})$ is the indicator function,
\begin{align}
1_S(T) = 1-\delta_{T\cap S \,\emptyset}.
\end{align}
$K(C)$ is nonzero if and only if $|C| \geq 3$, $C = \{1,2\}$, or $C = \{n-1,n\}$, proving that these subsystems are the only connected subsystems with local observables.

Note that removing any number of sites from $C$ necessarily reduces $K(C)$, so the subsystem $C$ also contains extensively local observables. 
\end{proof}

If a subsystem is a union of disjoint connected subsystems of the form above, there are observables which are extensively local to said subsystem. It only remains to check subsystems which have an isolated site or pair of sites.

\begin{lemma} \label{conds}
Observables which are extensively local to a subsystem, $A \subseteq \dbrac{n}$, and quasi-Hermitian with the metric of \cref{not positive} exist only if the following conditions on its connected components, $A = \cup \mathfrak{C}_A$, are met:

\begin{enumerate} 
\item If $A$ contains a connected component, $C \in \mathfrak{C}_A$ with $|C|\leq2$, and $C \notin \{\{1,2\},\, \{n-1,n\}\}$, then $A$ must contain at least one more connected component, so $A \neq C$. 
\item For every connected component with a single site, $C = \{i\}$, then $i-2, i+2 \in A$ when $i-2, i+2 \in \dbrac{n}$. Intuitively, this connected component is separated from the rest of the subsystem by at most one site from both the left and the right. \label{isolated}
\item The connected components with two sites, $|C| = 2$, satisfy $\min_{C^c \in \mathfrak{C}_{A-C}} d_A(C, C^c) \leq 2,$ unless $C = \{1,2\}$ or $C = \{n-1,n\}$. Intuitively, this connected component is separated from the rest of the subsystem by at most one site from either the left or the right. 
\item The edge components may have a single site only if that component is $C_L = \{1\}$ or $C_R = \{n\}$. Otherwise, $|C_{L,R}| > 1$.
\end{enumerate}

The set of all subsets $A \subseteq \dbrac{n}$ satisfying the above criteria will be referred to as $\mathcal{R}$.
\end{lemma}
\begin{proof}

\leavevmode
\begin{enumerate}
\item This follows from lemma~\ref{ConnectedRegions}.
\item Assume $\{i\}$ is a connected component of $A$. Assume $v \in \ker M^{A^c A}$. Then, using the notation of \cref{A<,v<}, 
\begin{align}
V^\intercal &:= \left(M^{\{i-1\} A_{<i}}\, v_{<i}, \mathfrak{i} \gamma\, v_i, M^{\{i-1\} A_{>i}}\, v_{>i} \right)\\
V &\in X := \text{span}\left\{\begin{array}{l}\arraycolsep=1.4pt\def\arraystretch{1.8} 1_{A^c}(\{i-2\})\left(z^{-1},-z^*,z^*\right)^\dag\\1_{A^c}(\{i-1\})\left(1, -1, 1\right)^\dag,\\
1_{A^c}(\{i+1\})\left(z^2, 1, (z^{*})^{-2} \right)^\dag\\
1_{A^c}(\{i+2\})\left(z^3,z,(z^*)^{-3} \right)^\dag\end{array}\right\}^\perp.
\end{align}

For a nontrivial solution $V$ to exist, $\dim X \leq 2$, which only happens if either $|z| = 1$ or $i-2,i+2\in A$ when $i-2,i+2\in \dbrac{n}$. 

If all vectors $v \in \ker M^{A^c A}$ satisfy $V = 0$, then  $\ker M^{A^c A} \subset \ker M^{S' S}$ for $S = A \setminus \{i\}$. Therefore, there is no extensively local observable in $A$ in such a case, and extensively local observables exist if and only if criteria~\eqref{isolated} is satisfied.
\item Suppose $i, i+1 \in A$, and assume $n>4$, since $n=4$ reduces to \cref{ConnectedRegions}. 
Consider the kernel conditions $(M^{A^cA}v)_j = 0$ for the choices of $j\in \{i-2, i-1, i+2, i+3\}\cap \dbrac{n}$:

\begin{align}
V^\intercal &:= \left(M^{\{i-2\} A_{<i}} v_{A_{<i}},
\mathfrak{i} \gamma v_i,
\mathfrak{i} \gamma v_{i+1},
M^{\{i-2\} A_{>i+1}} v_{A_{>i+1}}\right), \\
V \in X &:=\text{span} \left\{
\begin{array}{l}
1_{A^c}(\{i-2\})\left(1, -z^*, -{z^*}^2, 1\right)^\dag \\
1_{A^c}(\{i-1\})\left(z, -1, -z^{*}, {z^{*}}^{-1}\right)^\dag \\
1_{A^c}(\{i+2\})\left(z^4, z, 1, {z^{*}}^{-4}\right)^\dag \\
1_{A^c}(\{i+3\})\left(z^5, z^2, z, {z^{*}}^{-5}\right)^\dag
\end{array} \right\}^\perp. \label{|C|=2}
\end{align}
For $V$ to be nontrivial, $\dim X \leq 3$. This only happens when $|z| = 1$, or $i-2\in A$ when $i-2 \in \dbrac{n}$, or $i+3 \in A$ when $i+3 \in \dbrac{n}$. The same logic from the proof of criteria~\eqref{isolated} demonstrates that extensively local observables exist only if a nontrivial $V$ exists, proving this case.

\item Suppose $i$ is the leftmost site, $i-1 \in \dbrac{n}$, and $i+1 \in A^c$. Then, \begin{align}
\begin{pmatrix}
-1 & 1 \\
1 & {z^{*}}^{-2}
\end{pmatrix}
\begin{pmatrix}
\mathfrak{i} \gamma v_i\\
M^{\{i-1\} A_{>i}} v_{>i}
\end{pmatrix} = 0.
\end{align}
Thus, $v_i = M^{\{i-1\} A_{>i}} v_{>i} = 0$, so $M^{\{i\} A} v = 0$. Consequently, $\ker M^{A^c A} \subset \ker{M^{S^c S}}$ with $S = A \setminus \{i\}$, so by \cref{polynomial theorem}, there are no extensively local observables in $A$. The case where $i$ is the rightmost site follows from $\mathcal{PT}$ symmetry.
\end{enumerate}
\end{proof}
The remainder of this section is dedicated to showing that subsystems $A$ satisfying the enumerated criteria of lemma~\eqref{conds}, $A \in \mathcal{R}$, do contain extensively local observables.

\begin{proposition} \label{m=1Final}
Extensively local observables to a subsystem $A=\cup \mathfrak{C}_A \subseteq \dbrac{n}$, which are quasi-Hermitian with the metric of \cref{not positive}, exist if and only if the conditions of lemma~\eqref{conds} are met.
\end{proposition}
\begin{proof}

Define 
\begin{align}
\mathfrak{C}_k &= \{C \in \mathfrak{C}_A: |C| \leq k\}\\
\mathcal{R}_k &= \{A \in \mathcal{R}: \mathfrak{C}_A = \mathfrak{C}_k\}.
\end{align}
Intuitively, $\mathcal{R}_k$ is the set of all subsystems $A \in \mathcal{R}$ whose connected components contain at most $k$ sites. Note $\mathcal{R}_{n} = \mathcal{R}$. I will use an inductive argument to prove that for each $\mathcal{R}_k$, every element contains extensively local observables. 

Consider first the base case where $A \subset \mathcal{R}_2$, where all connected components have cardinality at most 2. Note $|\mathfrak{C}_2|-|\mathfrak{C}_1|$ denotes the number of connected components in $\mathfrak{C}_A$ with cardinality exactly two. A simple argument by counting the number of linearly independent rows in $M^{A^c A}$ results in the identity
\begin{align}
K(A) \geq |\mathfrak{C}_2|-|\mathfrak{C}_1| -1 + 1_A(\{1\}) + 1_A(\{n\}).
\end{align}
Thus, local observables exist for all $A \subset \mathcal{R}_2$.

The following proves by contradiction that the observables constructed above are extensively local. Suppose such an observable is not extensive, but is local to $S \subset A$ with cardinality at most $k$. This subset does not satisfy the criteria of lemma~\eqref{conds}, so the observable must be local to a subset $S$ with cardinality at most $k-1$. Repeating this argument inductively until $k = 1$ would imply the existence of a observable local to a single site, which contradicts corollary~\eqref{single site theorem}. 

To prove the inductive hypothesis, I demonstrate that for every $A \in \mathcal{R}_k$, with $k>2$, there exists a decomposition $A = \cup_i A_i$ such that either $A_i \in \mathcal{R}_{k-1}$, or $A_i$ is a union of connected components, $A_i = \cup C_{k,2} \subseteq \dbrac{n}$ with $|C_{k,2}| \geq 3$. Since $A_i$ either is assumed to have extensively local observables in the first case, or known to have extensively local observables by lemma~\eqref{ConnectedRegions} in the latter case, $A$ must have extensively local observables.

Suppose $A$ has $l$ connected components $C \in \mathfrak{C}_A$ with cardinality $|C| = k$. We will express $A$ as a union of the form $A = B_1 \cup B_2 \cup B_3$ such that $B_1, B_2 \in \mathcal{R}_k$, $B_3$ is connected with cardinality $|A_3|>2$, and $B_1$ and $B_2$ combined have $l-1$ connected components with cardinality $|C| = k$. An inductive argument on $l$, thus, constructs the decomposition $A_i$ from the previous paragraph.

Pick one connected component $C \in \mathfrak{C}_A$ with cardinality $|C| = k$. The construction of $B_1, B_2, B_3$ splits into four cases:

If $A = C$, the construction $B_1 = B_2 = \emptyset, B_3 = C$ is trivial.

If $\min_{C^c \in \mathfrak{C}_{A-C}} d_A(C, C^c)\geq 3$, then the sets $B_1 = A-C, B_2 = \emptyset$ must satisfy the axioms of lemma~\eqref{conds}, and $B_3 = C$ is of the desired form.

If there is a unique set $C_1$ such that $d(C, C_1) = 2$, set $B_3 = C$, $B_2 = \emptyset$. In the case where $\max C_1 < \min C$, set $B_1 = A_{<\min C + 2}$, else, set $B_1 = A_{>\max C - 2}$.

In the final case, there are two sets $C_1, C_2$ such that $d(C,C_1) = d(C,C_2) = 2$. Without loss of generality, assume $\max C_1 < \min C < \max C < \min C_2$. Set $B_3 = C$, $B_1 = A_{<\min C + 2}$, $B_2 = A_{>\max C - 2}$.
\end{proof}

\chapter{Conclusions} \label{conclusionChapter}

Three salient themes emerge from the previous chapters:
\begin{itemize}
\item Quasi-Hermitian frameworks generate new notions of locality in quantum theory. The modified inner product used in quasi-Hermitian quantum theory can be entangled, and local observable algebras have reduced dimension when compared to their Hermitian counterparts. This generalization is tame in the sense that no quasi-Hermitian strategy for a nonlocal game can outperform a Hermitian strategy.

\item Algebraic geometry yields insight into simple pseudo-Hermitian operators with parametric dependence. In particular, higher order exceptional points generically occur at cusp points of algebraic curves of exceptional points of lower order.

\item Representation theory and abstract algebraic considerations generate nontrivial examples of pseudo-Hermitian operators and their corresponding conserved quantities from simple examples.
\end{itemize}

Following these observations, various nontrivial examples of conserved quantities, eigensystems, and inner product structures associated to non-Hermitian operators were determined in \cref{Tactics,toyModelsChapter}. I hope researchers will find these examples illuminating and apply them, perhaps as toy models, to problems of mathematical or physical interest.

%


The following sections provide suggestions for future research which thesis informs or inspires.

\section{Time}
Nearly every result presented in this thesis pertains to kinematics associated to non-Hermitian quantum theory. In this section, I summarize three questions which emerge from contemplating the material of this thesis with time in mind.

\begin{itemize}
\item Exciting avenues for new fundamental physics come from features present in an $\eta$-inner product which are not present in traditional inner product structures. One interesting possibility is to include time-dependence in $\eta$ \cite{timeDependentInnerProd}. In particular, one problem which I believe is outstanding is the correspondence between time-evolution and Hamiltonian generators in the case where the metric is time-dependent. In other words, is there a generalization of Stone's theorem \cite{Stone1930,Stone1932} to the case with a time-dependent inner product structure?

\item The central structure discussed in \cref{LocalityChapter} is the local observable algebra. Given local observable algebras and a distinguished separating vector, a natural notion of local time evolution is given by Tomita and Takesaki's modular automorphisms \cite{Takesaki1970,TakesakiII}. This time evolution has been the subject of recent interest in quantum field theory \cite{Borchers2000}. What do modular automorphism groups associated to quasi-Hermitian observable algebras look like?

\item Non-Hermitian time evolution gives rise to phenomenon not found in traditional quantum theory. For example, the \textit{Leggett-Garg inequality} (LGI) quantifies correlations between measurements taken place at three different times \cite{Leggett1985}. Non-Hermitian time evolution provides algebraically maximal LGI violations \cite{Karthik2021,JoglekarLGI}. This contrasts with Bell's inequality violations \cite{Bell1964,Clauser1969,Cirelson1980,Hensen2015,Shalm2015,Giustina2015} and nonlocal games \cite{nonlocalGames,Ji2021}, which quantify correlations between different points in space. Is there a merger of these concepts, a spatio-temporal measure of non-classicality, which classifies what new physics is given by non-Hermitian Hamiltonians? A related notion is the \textit{Lieb-Robinson bound} \cite{Lieb1972}, which yields an emergent light-cone structure in quantum many-body problems.
\end{itemize}

\section{Entanglement}
Chapter~\ref{LocalityChapter} demonstrated that entangled quasi-Hermitian states cannot be used as a resource to outperform Hermitian strategies to nonlocal games, which are a generalization of Bells' inequalities. The more general question of quantifying entanglement in non-Hermitian quantum theory is of recent interest. For instance, entanglement entropy is analysed in a non-Hermitian setting in \cite{modak2021eigenstate}. To make a connection with a question from the previous section, a relative modular operator can be constructed which defines a measure for entanglement entropy in quantum field theory. Addressing the question of modular automorphism groups in quasi-Hermitian quantum theory would, thus, have physical consequences to the question of entanglement.

Perhaps more radical than analysing entanglement of states is introducing entanglement in the inner product structure of quasi-Hermitian theory. Following \cite{werner1989quantum}, in the tensor product model, I define an entangled metric operator as one which cannot be expressed as a sum of tensor products of local metric operators. A simple example of an entangled metric operator is one of the Werner states \cite{werner1989quantum} for a two-qubit system,
\begin{align}
\eta &= \frac{p}{3} P^{+} + \frac{2(1-p)}{3}P^{-},
\end{align}
where 
\begin{align}
P^{\pm} &= \frac{1}{2}\left(\mathbb{1} \pm \sum_{i,j = 1}^2 \ket{i} \bra{j} \otimes \ket{j} \bra{i} \right)
\end{align} 
and $p \in (0,1)$ is a tunable parameter. 
As a consequence of \cref{corollary:Schmidt-Rank}, the only local observables which are quasi-Hermitian with respect to this metric are the trivial multiples of the identity operator. Are there examples where the metric is entangled and there are nontrivial local observable algebras?

\section{Exceptional Points}
Pseudo-Hermitian operators exhibit unique symmetry breaking properties at exceptional points. Perturbative expressions for eigenvalues near exceptional points are given by \textit{Puiseux series} as opposed to the familiar Taylor series from Rayleigh-Schr{\"o}dinger perturbation theory. This enhanced sensitivity of eigenvalues at exceptional points has been observed experimentally in optical cavities and has applications to sensing technology \cite{Hodaei2017,Chen2017}. Thus, a systematic characterization of exceptional points is of both mathematical and physical interest. 

This thesis displays the insight that higher order exceptional points generically correspond to singular points of contours of exceptional points. This was demonstrated using a perturbative argument for the case of third order exceptional points. I imagine this readily generalizes to $n$-th order exceptional points, in which case the higher order exceptional points appear on a ridge of an algebraic variety of dimension $n-1$, as exemplified for instance in the non-Hermitian SSH chain studied in \cref{EP Surface Section}.

Of course, I emphasize the above insight is one among many characterizing exceptional points \cite{Motzkin1955,Moiseyev1980,Kato1995,Klaiman2008}.

Some questions generated by my curiosity regarding exceptional points include the following:
\begin{itemize}
\item Finite-dimensional matrices which are polynomial functions of their parameters have at most a finite number of exceptional and diabolical points \cite{Kato1995}. However, operators acting on infinite-dimensional Hilbert spaces can have an infinite set of exceptional points \cite{bender1998real}. This set can have limit points. When do limit points occur? What constraints must be satisfied for the exceptional points to be isolated? 

\item Can set-theoretic properties, such as cardinality, of exceptional points in an infinite-dimensional setting be classified? For instance, the model of \cite{bender1998real} has a countably infinite set of exceptional points. Can the set of exceptional points be dense? Can a continuous operator-valued map have an uncountably infinite set of exceptional points? If a measure is defined on a $\sigma$-algebra of the domain of this map, such the Borel $\sigma$-algebra, can the set of exceptional points have nonzero measure? The last question is particularly interesting if the measure of every countable set is zero, as is the case for the Lebesgue measure.
\end{itemize}
\section{Generalizations}

How can the results of this thesis be generalized? This section summarizes observations and questions which may yield new results.

\begin{itemize}
\item Sections~\ref{Constant Evalues} and \ref{General-Maximal-Breaking-Section} display examples where a subset of the spectrum of a Hamiltonian is independent of its non-Hermitian perturbation. This property has interesting physical consequences relating to the efficiency of transport \cite{ortega2019mathcal}. In some cases, but not all, this phenomena can be understood from the appearance of nodes in the eigenvectors via the Hellmann-Feynman theorem \cite{schrodinger1926quantisierung,Gttinger1932,Hellmann1933,Feynman1939}. I suspect a Su-Schrieffer-Heeger chain \cite{SSH}, with a pair of non-Hermitian perturbations at a suitable choice of sites $(m_1,m_2)$, also exhibits this behaviour. The tools from section~\ref{tridiagEsysSection} should be sufficient to analyse this problem and confirm my suspicion.

\item Chapter~\ref{Tactics} characterizes a large class of pseudo-Hermitian operators. In this chapter, I demonstrate that a small number of examples included in this class are the models studied in \cite{MyFirstPaper,Shi2022,Barnett2023} in examples~\ref{ex:Shi2022} and \ref{2011Example}. Are other examples from existing literature included in this class? What interesting new models can be analysed with \cref{representationTheoryTheorem} in hand? For example, using the representation of $2 \times 2$ matrices given in \cref{2x2-Matrix-Algebra-Representation}, choosing a unitary other than the identity or exchange matrices will yield examples which are distinct from examples~\ref{ex:Shi2022} and \ref{2011Example}. Perhaps more radically, we could consider representations of the algebra of $n \times n$ matrices, or even a $C^*$-algebra with no finite-dimensional representations. Furthermore, I imagine \cref{representationTheoryTheorem} can be generalized to reducible representations which are an infinite sum of smaller representations.

\item
An unusual phase transition which can occur for $\mathcal{PT}$-symmetric operators is one exhibiting what has been referred to as \textit{maximal symmetry breaking} \cite{MyFirstPaper}, where all of the eigenvalues gain a nonzero imaginary part as one tunes a parameter past an exceptional point. Representation theory allows for an understanding of why some $\mathcal{PT}$-symmetric operators exhibit maximal symmetry breaking, as demonstrated in \cref{Commutative-Section} and exemplified in \cref{nearestNeighbour}. 

I know of two examples which exhibit maximal symmetry breaking which I do not know how to express in the form given by \cref{homoHam}. The first example, which is discussed in \cite{Graefe2008} and \cref{su(2)tridiagExample}, is given by $S_x + i \gamma S_z$, where $S_x$ and $S_z$ are spin-$s$ representations of their corresponding generators in the Lie algebra $\mathfrak{su}(2)$. The second example is the $4 \times 4$ matrix pencil given by
\begin{align}
H(\gamma) = \begin{pmatrix}
i \gamma & 1 & 0 & 0\\
1 & i \gamma & 1 & 0 \\
0 & 1 & -i \gamma & 1 \\
0 & 0 & 1 & -i \gamma
\end{pmatrix},
\end{align}
which has second order exceptional points at $\gamma = \pm \sqrt{5}/4$. All four eigenvalues of $H(\gamma)$ have a nonzero imaginary part when $|\gamma| > \sqrt{5}/4$. Can either of these examples be expressed in the form given by \cref{homoHam}? Furthermore, with these examples in mind, it is reasonable to ask what the necessary and sufficient condition for maximal symmetry breaking is.

\item As demonstrated by \cref{non-commuting-representation-example}, given a set of matrices, $h_k$, which do not commute and a corresponding pseudo-Hermitian model generated by \cref{representationTheoryTheorem}, the $\mathcal{PT}$-symmetry breaking threshold can be greater than the threshold for the individual matrices $h_k$. Can a lower bound be placed on this threshold? I imagine this bound would involve the commutation properties of the $h_k$.
\end{itemize}
%
%
%
%
%

\section{Geometry}

One strategy I have utilized is to consider the set of exceptional points as an algebraic variety. I suspect we have only scratched the surface of connections between algebraic geometry and the perturbation theory of pseudo-Hermitian operators. Furthermore, I presume the more general setting of \textit{analytic geometry} and the study of \textit{analytic spaces} should apply to the study of holomorphic non-Hermitian operator-valued functions. A correspondence between algebraic and analytic geometry is given by the \textit{GAGA theorem} \cite{Serre1956}, which could potentially find applications to non-Hermitian perturbation problems.

\section{Bosons and Fermions}

Bosons and fermions are the generic types of particles \cite{Pauli1940}. Non-interacting, or free, models are analytically and numerically tractable. Section~\ref{FermionObservableAlgebraSection} examines quasi-Hermiticity in the context of free fermions. The construction of the inner product presented in this section begs for further examination and generalization. Some immediate questions include the following:

\begin{itemize}
\item Free models with pair creation and annihilation terms in the Hamiltonian, such as the \textit{Kitaev quantum wire} \cite{Kitaev2001}, have captivated the interest of physicists. In particular, the topological properties \cite{Chiu2016,Kawabata2019} of Kitaev's quantum wire means it can act as a superconductor, and the edge states of systems with Majorana fermions can possibly be used to design fault-tolerant qubits \cite{Nayak2008}. Non-Hermitian generalizations of the quantum wire have been explored \cite{Menke2017,KaustubhKitaev}. Using a Bogoljubov transformation\cite{Bogoljubov1958}, can the intertwining operators and local observable algebras studied in \cref{FermionObservableAlgebraSection} be generalized to models with pair creation and annihilation terms?

\item The number operator is a quasi-Hermitian observable associated to the metric operator utilized in \cref{FermionObservableAlgebraSection}. A result of \cite{QuasiHerm92} is that an irreducible set of observables uniquely defines a metric operator. If the Hamiltonian and the number operator are assumed to be simultaneously quasi-Hermitian, what constraint on the space of metrics is imposed?

\end{itemize}





\cleardoublepage 
\phantomsection  
\renewcommand*{\bibname}{References}

\addcontentsline{toc}{chapter}{\textbf{References}}
\printbibliography


\appendix

 
\chapter{Mathematical Preliminaries} \label{functionalAnalysisAppendix}
\addcontentsline{toc}{chapter}{Appendix A - Mathematical Preliminaries}

The goal of this appendix is to review mathematical constructs used in quantum theory and in this thesis. Topics which are essential but not reviewed here include calculus, linear algebra, complex numbers, and point-set topology.

Textbooks which introduce subsets of the mathematics discussed in this appendix in the context of quantum theory include \cite{vonNeumannBook,vonNeumannBookEnglish,hall2013quantum,moretti2013spectral,rejzner2016perturbative}. Textbooks which introduce functional analysis on its own merits include \cite{Istratescu1981,zaanen1953linear,riesz2012functional}. Textbooks devoted to the study of operator algebras include \cite{Abramovich2002,KadisonRingroseI,bratteli1987operator,Bratteli1997,TakesakiI,TakesakiII}. A classic book on perturbation theory of linear operators that also introduces various tools from functional analysis is \cite{Kato1995}, 

A graphical summary of the mathematical spaces studied in this appendix is given in \cref{graphOfSpaces}.

\section{Functional Analysis} \label{funcAnalSection}
\subsection{Normed Spaces}
\begin{defn} \label{normedSpace}
Let $X$ be a vector space over the field $\mathbb{F}$, where $\mathbb{F}$ is either $\mathbb{R}$ or $\mathbb{C}$. A map $||\cdot||:X \rightarrow \mathbb{R}$ is a \textit{norm} on $X$, and $(X, ||\cdot||)$ is a \textit{normed space}, if and only if for all $v, w \in X$ and $\lambda \in \mathbb{F}$,
\begin{enumerate}
\item $||v|| \geq 0$ \label{normSemidef}\\
\item $||v|| = 0 \Rightarrow v = 0$. \label{normdef}\\
\item $||\lambda v|| = |\lambda| \, ||v||$ \label{homogeneity}\\
\item $||v + w|| \leq ||v|| + ||w||$ \label{triangle} 
\end{enumerate}
Condition \eqref{triangle} is referred to as the \textit{triangle inequality}, the constraint \eqref{homogeneity} is the \textit{homogeneity} property. Together, conditions \eqref{normSemidef} and \eqref{normdef} are referred to as \textit{positive-definiteness}.
\end{defn}
\begin{defn}
Given a normed space, $(X, ||\cdot||)$, an \textit{open ball} is a set of form 
\begin{align}
\mathscr{B}_{||\cdot||}(a;r) := \{x \in X \, | \, ||x - a|| < r\}, \label{openBall}
\end{align} 
with $r \in \mathbb{R}$ and $a \in X$.
The \textit{norm topology} is the collection of subsets of $X$ which can be expressed as a union of open balls. Topological properties in normed spaces, such as closure, convergence, and continuity, are defined via the norm topology.
\end{defn}
\begin{defn} \label{defn:Banach}
A \textit{sequence} is a function whose domain is the natural numbers. The values of a sequence, $x$, for the input $n$ is typically denoted with a subscript, such as $x_n$. A \textit{Cauchy sequence}, $x$, is a sequence whose codomain is a normed space such that for every $\epsilon > 0$, there exists an $N_\epsilon$ such that for all $n, m \geq N_\epsilon$, $||x_n -x_m||< \epsilon$.
A normed space is called \textit{complete}, or a \textit{Banach space}, if every Cauchy sequence has a limit in the Banach space.
\end{defn}

\begin{definition}
Two norms, $||\cdot||_i:X \to \mathbb{R}$ with $i \in \{1, 2\}$, are called \textit{equivalent} if there exists $a, b > 0$ such that for every $x \in X$,
\begin{align}
a ||x||_1 < ||x||_2 < b ||x||_1.
\end{align}
Two norms are equivalent if and only if their norm topologies are the same \cite[Thm. 1.7.1]{Narici2010}.
\end{definition}

\begin{ex} \label{pnorm}
Given the field of either the real or complex numbers, $\mathbb{F} \in \{\mathbb{R}, \mathbb{C} \}$, and an integer $n \in \mathbb{Z}_+$, consider the coordinate space $\mathbb{F}^n$. An example of a norm on $\mathbb{F}^n$ is
one of the $p$\textit{-norms}, with $1 \leq p \leq +\infty$, which are defined by
\begin{align}
||z||_p^p &:= \sum_{i = 1}^n |z_i|^p &\quad& \forall p \in [1, \infty[ \\
||z||_\infty &:= \sup_{i \in \{1, \dots, n\}} |z_i|.
\end{align} 
Every norm on $\mathbb{F}^n$ is equivalent \cite[Thm. 1.7.1]{Narici2010}. Furthermore, every normed space on $\mathbb{F}^n$ is a Banach space. The \textit{Euclidean topology} is the induced norm topology.
\demo
\end{ex}

\begin{defn}
An \textit{isometry} is a norm preserving map between two normed spaces. Explicitly, if $(X, ||\cdot||_X)$ and $(Y, ||\cdot||_Y)$ are normed spaces, then $\iota:X \to Y$ is an isometry if $||\cdot||_Y \circ \iota = ||\cdot||_X$.
\end{defn}

\subsection{Spaces with Hermitian Forms}

\begin{definition} \label{hermitianForm}
Given a vector space, $V$, over $\mathbb{C}$, a map $Q:V \times V \to \mathbb{C}$ is a \textit{Hermitian sesquilinear form} if and only if $Q$ satisfies the following properties
\begin{itemize}
\item \textit{Conjugate Symmetry}:
\begin{align}
Q(x,y) =  Q(y,x)^*.
\end{align}

\item \textit{Linearity}: 
\begin{align}
Q(z, \alpha x + \beta y) &= \alpha Q(z,x) + \beta Q(z,y) &\enspace& \forall \alpha, \beta \in \mathbb{C}, \, x,y,z \in V. 
\end{align}
\end{itemize}
\end{definition}

I follow the definition of an indefinite inner product given in \cite[\S 2.1]{Gohberg2005}.
\begin{definition} \label{indefiniteInnerProduct}
Consider a Hermitian sesquilinear form, $[\cdot| \cdot]: V \times V \rightarrow \mathbb{C}$, over a complex vector space, $V$, as defined in \cref{hermitianForm}. This space is an \textit{indefinite inner product space} if $[\cdot| \cdot]$ is \textit{non-degenerate}, which means for every $x \in V \setminus \{0\}$, there exists a $y \in V \setminus \{0\}$ such that $[x|y] \neq 0.$ The form $[\cdot| \cdot]$ referred to as the \textit{indefinite inner product}.
\end{definition}

\begin{definition} \label{innerProduct}
An \textit{inner product space} is a complex vector space, $V$, equipped with a map, $\braket{\cdot| \cdot}: V \times V \rightarrow \mathbb{C}$, referred to as the \textit{inner product}, which is a Hermitian sesquilinear form which is also \textit{positive-definite}, so that 
\begin{align}
x \in V \setminus \{0\}\, \Rightarrow \, \braket{x|x} > 0.
\end{align}
Some treatments refer to inner product spaces as \textit{unitary spaces} \cite{silberstein1962symmetrisable,Kato1995}.
\end{definition}

\begin{proposition}
Every inner product space, $V$, is a normed space in the \textit{induced norm}, the map $||\cdot||:V \to \mathbb{R}$ defined by
\begin{align}
||\psi|| := \sqrt{\braket{\psi|\psi}}.
\end{align}
\end{proposition}

\begin{defn} \label{HilbertSpaceDefn}
A \textit{Hilbert space} is an inner product space which is complete in the induced norm \cite[chp. II \S 1]{vonNeumannBook,vonNeumannBookEnglish}. Topological properties in Hilbert space, such as closure, convergence, and continuity are defined with respect to the induced norm topology.
\end{defn}



\begin{defn}
Two vectors in an inner product space are \textit{orthogonal} if their inner product equals zero.
\end{defn}

\begin{defn} \label{OrthoComplement-Defn}
Given a subset, $S$, of an inner product space, $V$, 
the \textit{orthogonal complement} of $S$ is the linear subspace of elements in $V$ which are orthogonal to every element of $S$, or more explicitly,
\begin{align}
S^\perp := \{v \in V \,|\, \braket{v|s} = 0 \, \forall s \in S\}.
\end{align}
\end{defn}

\begin{defn}
Given an inner product space, $V$, a subset, $S \subset V$ is called \textit{orthonormal} if its elements are mutually orthogonal and all have norm one. More explicitly, a set is orthonormal if for every $\psi, \phi \in S$, $\braket{\psi|\phi} = \delta^\phi_{\psi}$ holds. Every orthonormal set is linearly independent. An \textit{orthonormal basis} is an orthonormal subset of a Hilbert space, $S \subsetneq \mathcal{H}$, such that $S^\perp = \{0\}$. An equivalent definition of an orthonormal basis is an orthonormal subset satisfying $\text{cl}(\text{span}(S)) = \mathcal{H}$
\end{defn}
\begin{ex} \label{canonicalBasis}
The $n$-dimensional \textit{complex coordinate space}, $\mathbb{C}^n$, which is the set of ordered $n$-tuples of complex numbers, is a Hilbert space with the inner product 
\begin{align}
\braket{z|w} := \sum_{i = 1}^n z_i^*  w_i.
\end{align} 
One orthonormal basis of $\mathbb{C}^n$ is the \textit{canonical basis}, which is the set of tuples, $\{\gls{ei}\,|\,i \in \{1, \dots, n\} \}$, whose components are given by
\begin{align}
(e_i)_j = \delta_{ij},
\end{align}
where $\delta_{ij}$ is the Kronecker delta.
\demo
\end{ex}

\begin{ex}
The set of $n \times n$ matrices with elements in $\mathbb{C}$, $\mathfrak{M}_n(\mathbb{C})$ is a Hilbert space. The inner product is the \textit{Frobenius} or \textit{Hilbert-Schmidt} inner product, $\braket{\cdot|\cdot}_{F}: \mathfrak{M}_n(\mathbb{C}) \times \mathfrak{M}_n(\mathbb{C}) \to \mathbb{C}$, defined by
\begin{align}
\braket{M|N}_{F} := \text{Tr} (M^\dag N), \label{hilbertSchmidtHilbertAlg}
\end{align}
where $\dag$ denotes complex conjugate transposition.
\demo
\end{ex}

\begin{defn}
A topological space is \textit{separable} if and only if it has a countable dense subset. In particular, a Hilbert space is separable if and only if it has a countable orthonormal basis.
\end{defn}

\begin{theorem}[Pythagoras] \label{Pythagoras}
Given a sequence of vectors in a Hilbert space, $\psi_i \in \mathcal{H}$ for $i \in S$ with $S$ countable, suppose these vectors are mutually orthogonal, so $\braket{\psi_i|\psi_j} = \delta_{ij} \braket{\psi_i|\psi_i}$. If $\sum_{i \in S} ||\psi_i||^2 <+\infty$ converges, then $\sum_{i \in S} \psi_i \in \mathcal{H}$ converges, and 
\begin{align}
\bigg| \bigg| \sum_{i \in S} \psi_i \bigg| \bigg|^2 = \sum_{i \in S} ||\psi_i||^2.
\end{align}
\end{theorem}

\subsection{Operators}
\begin{definition} \label{operatorDefn}
In this thesis, an \textit{operator}, $A$, is either a linear or an antilinear function whose domain and codomain are vector spaces over the field $\mathbb{C}$. A function, $A$, is \textit{linear} if
\begin{align}
A(\alpha x + \beta y) = \alpha A(x) + \beta A(y) &\enspace& \forall x,y \in \text{Dom}(A), \, \alpha,\beta \in \mathbb{C},
\end{align}
or \textit{antilinear} if 
\begin{align}
A(\alpha x + \beta y) = \alpha^* A(x) + \beta^* A(y) &\enspace& \forall x,y \in \text{Dom}(A), \, \alpha,\beta \in \mathbb{C}. \label{antilinearDefn}
\end{align}
Antilinear functions are also referred to as \textit{conjugate linear} in some references.
An operator whose codomain is $\mathbb{C}$, $A:\text{Dom}(A) \rightarrow \mathbb{C}$, is referred to as a \textit{functional}.
An operator, $A$, whose domain, $\text{Dom}(A)$, is a dense subset of a topological vector space\footnote{A \textit{topological vector space} over a topological field is a vector space and a topological space where vector addition and scalar multiplication are continuous functions in the product topologies. A \textit{topological field} is a topological space and a field such that the field operations are continuous. A normed space in the norm topology is an example of a topological vector space. Normed spaces are vector spaces over the topological field $\mathbb{C}$ with the Euclidean topology.}, $X$, is called \textit{densely defined}. Every operator in this thesis is densely defined. 

Given $\psi \in \text{Dom}(A)$, I use the shorthand $A \psi$ for $A(\psi)$.
\end{definition}


\begin{defn} \label{boundedOperators}
An operator, $A:X \to Y$, whose domain, $X$, and codomain, $Y$, are normed spaces is \textit{bounded} if there exists $M > 0$ such that for every $x \in X$, $||A x||_Y < M ||x||_X$. The set of all bounded linear operators from $X$ to $Y$ is denoted $\mathcal{B}(X,Y)$, and the set of all bounded linear operators from $X$ to $X$ is denoted $\mathcal{B}(X)$.

Densely defined bounded operators on a Banach space admit unique extensions whose domain is the entire Banach space. 
\end{defn}

\begin{theorem}
An operator is continuous in the norm topology if and only if it is bounded.
\end{theorem}

\begin{defn}
Given normed spaces, $X, Y$, the \textit{operator norm} of a bounded operator is 
\begin{align}
||A|| &:= \sup \left\{\frac{||A x||_Y}{||x||_X} \, \vert \, x \in X \setminus \{0\} \right\}.
\end{align}
The operator norm is a norm on $\mathcal{B}(X,Y)$. $\mathcal{B}(X,Y)$ with the operator norm is a Banach space whenever $X$ and $Y$ are Banach spaces.
\end{defn}

\begin{ex}
Consider the set of linear maps $A:\mathbb{C}^m \to \mathbb{C}^n$. A bijective map from linear maps to matrices, $A_{ij} \in \mathbb{C}^m \otimes \mathbb{C}^n$, is realized via 
\begin{align}
A e_i = \sum_{j = 1}^n A_{ji} \tilde{e}_j
\end{align}
where \gls{ei} and $\tilde{e}_j$ are the canonical bases of $\mathbb{C}^m$ and $\mathbb{C}^n$ respectively, given in \cref{canonicalBasis}. The map $A$ is bounded regardless of what norm is specified on $\mathbb{C}^n$ and $\mathbb{C}^m$. Considered as a map from $(\mathbb{C}^m, ||\cdot||_p)$ to $(\mathbb{C}^n,||\cdot||_p)$with $p \in \{1, 2, \infty\}$ the operator norm of $A \in \mathcal{B}(\mathbb{C}^d)$ is 
\begin{align}
||A||_1 &= \max_{j \in \{1, \dots, n\}} \sum_{i \in \{1, \dots, m\} }|A_{ij}|\\
||A||_2 &= \sqrt{\max \sigma(A^\dag A)} \label{spectralNorm}\\
||A||_\infty &= \max_{i \in \{1, \dots, m\}} \sum_{j \in \{1, \dots, n\} }|A_{ij}| = ||A^T||_1,
\end{align}
where the definition of $||A||_2$ relies on the notion of its spectrum, which is defined in \cref{spectrumDefn}.
The $||\cdot||_2$ operator norm is referred to as the \textit{spectral norm}.
\demo
\end{ex}

\begin{defn}
A bounded operator, $A$, is a \textit{contraction} if $||A|| \leq 1$.
\end{defn}

\begin{theorem}[Riesz Representation \cite{riesz1907,frechet1907}] \label{Riesz}
Let $\mathcal{H}$ be a Hilbert space, and $\varphi \in \mathcal{B}(\mathcal{H}, \mathbb{C})$ be a continuous linear functional. There exists a unique $f_\varphi \in \mathcal{H}$ such that for all $\psi \in \mathcal{H}$,
\begin{align}
\varphi(\psi) = \braket{f_{\varphi}|\psi}.
\end{align}
Additionally, 
\begin{align}
||f_\varphi|| = ||\varphi||.
\end{align} 
\end{theorem}
%

\begin{definition}[Adjoint] \label{adjoint}
Suppose $A:\text{Dom}(A) \to \mathcal{H}_2$ is a densely defined operator, $\text{Dom}(A) \subseteq \mathcal{H}_1$, and $\mathcal{H}_1$ and $\mathcal{H}_2$ are Hilbert spaces with the inner products $\braket{\cdot|\cdot}_1$ and $\braket{\cdot|\cdot}_2$ respectively. Let $\text{Dom}(A^\dag)$ be the space of all $\phi \in \mathcal{H}_2$ such that the linear functional 
\begin{align}
\text{Dom}(A) \ni \psi \mapsto \begin{cases}
\braket{\phi |A \psi}_2 & \text{ if } A \text{ is linear } \\
\braket{\phi |A \psi}_2^* & \text{ if } A \text{ is antilinear }
\end{cases}
\end{align} 
is bounded. Define the \text{adjoint} of $A$ to be the operator $A^\dag: \text{Dom}(A^\dag) \to \mathcal{H}_1$ such that $A^\dag \phi \in \mathcal{H}_1$ is the unique vector given by the Riesz Representation \cref{Riesz} satisfying 
\begin{align}
\braket{A^\dag \phi|\psi}_1 = \begin{cases}
\braket{\phi|A \psi}_2 & \text{ if } A \text{ is linear }\\
\braket{\phi|A \psi}_2^* & \text{ if } A \text{ is antilinear }
\end{cases}. \label{adjointDefn}
\end{align} 
An operator is called \textit{self-adjoint} or \hypertarget{link:Hermitian}{\textit{Hermitian}} if $A = A^\dag$.
\end{definition}

\begin{ex}
If $A \in \mathcal{B}(\mathbb{C}^d)$, then its adjoint is the complex conjugate transpose,
\begin{align}
(A^\dag)_{ij} = A_{ji}^*. \label{conjugateTranspose}
\end{align}
\demo
\end{ex}

\begin{defn}
Given a Hilbert space, $\mathcal{H}$, a \textit{projector}, $P$, is a linear idempotent satisfying
\begin{align}
P P = P.
\end{align}
An \textit{orthogonal projector} is a self-adjoint projector.
\end{defn}

\begin{theorem} \label{orthoProjection}
For every closed linear subspace of a Hilbert space, $\mathcal{K}$, there exists a unique orthogonal projection whose range equals $\mathcal{K}$. I refer to this projector as $\mathscr{P}_\mathcal{K}$.
\end{theorem}





\begin{defn} \label{diagonalizable}
A linear operator, $D$, is \textit{diagonal} in the basis $x_i$ if and only if the basis elements are eigenvectors of $D$. A linear operator is \hypertarget{link:Diagonable}{\textit{diagonalizable}} in the basis $x_i$ if and only if it is similar to a diagonal operator. An operator is diagonalizable in a finite-dimensional space if and only if it is diagonalizable in every linearly independent basis. An operator on a finite-dimensional space is \textit{defective} if and only if it is not diagonalizable.
\end{defn}
\begin{proposition}[Spectral Theorem for Compact Normal Operators]
Every compact normal operator is unitarily diagonalizable.
\end{proposition}

\begin{defn}\label{invariantSubspaceOperatorDefn}
Consider a set of operators, $\mathcal{S}$, on a vector space. A subspace $V$ is an \textit{invariant subspace} if the image of $V$ under every element $s \in \mathcal{S}$ is contained in $V$, or more explicitly, $s(V) \subseteq V$. If $\mathcal{S}$ consists of a single operator, then $V$ is an invariant subspace of that operator.
\end{defn}

\begin{ex}
If the domain and codomain of an operator are equal, then this domain is an invariant subspace of the operator. The eigenspaces of an operator are always invariant subspaces, as well as the set containing the zero vector.
\end{ex}

\begin{lemma} \label{Invariant-Subspace-Of-Adjoint}
If $V$ is an invariant subspace of a bounded operator on a Hilbert space, $A:\mathcal{H} \to \mathcal{H}$, then $V^\perp$ is an invariant subspace of this operator's adjoint, $A^\dag$.
\end{lemma}

\subsection{Invertibility}

This section discusses the inverses of operators. Firstly, I remind the reader of what invertibility means for functions.
\begin{defn}
A function, $f:X \to Y$, is 
\begin{itemize}
\item \textit{surjective}, or \textit{onto}, if the range of $f$ is $Y$. \\
\item \textit{injective}, or \textit{one-to-one}, if $f(x_1) = f(x_2) \, \Rightarrow x_1 = x_2$ \\
\item \textit{bijective} if $f$ is injective and surjective.
\end{itemize}
\end{defn}
A function is bijective if and only if it has an inverse function, $f^{-1}:Y \to X$, such that $f \circ f^{-1} = \text{id}_Y$ and $f^{-1} \circ f = \text{id}_X$, where $\text{id}_S$ denotes the identity map on a set $S$.

The next goal of this section is to analyse invertibility of linear maps. Injectivity is discussed in the next lemma.
\begin{lemma}
A linear map is injective if and only if its kernel consists of only the zero vector. 
\end{lemma}
More can be said for bounded linear maps between Banach spaces. A first result, which can be generalized to topological vector spaces, is given below
\begin{theorem}[Banach bounded inverse theorem] \label{BoundedInverseTheorem}
A bijective bounded linear operator whose domain and codomain are Banach spaces has a bounded inverse {\normalfont \cite[Thm. 14.5.1]{Narici2010}}.
\end{theorem}
One useful class of operators that helps clarify when an inverse exists is the class of \textit{bounded below} operators. 
\begin{defn}
An operator $A$ between normed spaces is \textit{bounded below} if 
\begin{align}
\text{inf} \left\{\frac{||A \psi||}{||\psi||} \, | \, \psi \in \text{Dom}(A) \setminus \{0\} \right\} > 0. \label{boundedBelow}
\end{align}
\end{defn}
An equivalent characterization of bounded below operators is as follows.
\begin{theorem} \label{thm-bounded-below-inverse}
A bounded linear operator between Banach spaces is bounded below if and only if this operator is injective and has a closed range. {\normalfont \cite[Thm. 2.5]{Abramovich2002}}
\end{theorem}
Since a surjective map between Banach spaces necessarily has a closed range, a bounded bijective linear operator is necessarily bounded below. The next result regards bounded linear operators on Hilbert spaces.
\begin{theorem}
A bounded linear operator on a Hilbert space, $A \in \mathcal{B}(\mathcal{H})$, is bijective if and only if $A$ and $A^\dag$ are bounded below.
\end{theorem}
\begin{proof}
This is a straightforward corollary of \cref{thm-bounded-below-inverse} and the identity 
\begin{align}
\mathcal{R}(A)^\perp = \ker(A^\dag).
\end{align}
\end{proof}

\section{Associative Algebras}

In this section, we will introduce relevant mathematical structures from functional analysis. More detailed introductions to the theory of $C^*$-algebras can be found in \cite[chp. 2]{bratteli1987operator}, \cite[chp. 2.1]{rejzner2016perturbative}, and \cite{moretti2013spectral,KadisonRingroseI,TakesakiI,TakesakiII}. 


\begin{defn}
An \textit{associative algebra}, $\mathfrak{A}$, over $\mathbb{C}$ is a complex vector space equipped with an associative bilinear operation, $\cdot:\mathfrak{A} \times \mathfrak{A} \to \mathfrak{A}$, referred to as the \textit{product}, which distributes over addition. Given $x, y \in \mathfrak{A}$, I use the shorthand $x y$ to denote $x \cdot y$. A \textit{normed algebra} is an associative algebra and a normed space, where the norm must satisfy the \textit{product inequality},
\begin{align}
\forall x, y \in \mathfrak{A}, ||xy|| \leq ||x|| \, ||y||. \label{productInequality}
\end{align}
A \textit{Banach algebra} is a normed algebra and a Banach space in the same norm.
\end{defn}
The reason why the product inequality is imposed is so that the product is continuous in the product norm topology.

\begin{ex}
Given a normed space, $X$, $\mathcal{B}(X)$ is a normed algebra, where the product is function composition, and the norm is the operator norm.
\demo
\end{ex}

\begin{defn} \label{c*Defn}
An \textit{involutive complex algebra}, also referred to as a ${}^*$-\textit{algebra}, $\mathfrak{A}$, is an associative algebra over $\mathbb{C}$ equipped with a map, $*: \mathfrak{A} \rightarrow \mathfrak{A}$, which satisfies the following properties for every $a,b \in \mathfrak{A}$, $\alpha,\beta \in \mathbb{C}$:
\begin{itemize}
\item $*$ is an involution, so $(a^*)^* = a$. \\
\item $*$ is antilinear, so $(\alpha a + \beta b)^* = \overline{\alpha} a^* + \overline{\beta} b^* $, where $\overline{\alpha}$ denotes the complex conjugate of $\alpha$.\\
\item $*$ is an \textit{algebra antihomomorphism}, which means $(a b)^*= b^* a^*$.
\end{itemize}
A $C^*$-\textit{algebra} is a ${}^*$-algebra and a Banach algebra for which the $C^*$ identity holds, namely 
\begin{align}
||a^* a|| = ||a^*|| \, ||a||.
\end{align}
\end{defn}

\begin{ex}
The prototypical example of a $C^*$-algebra is the set of bounded linear operators, $\mathcal{B}(\mathcal{H})$, on a Hilbert space, $\mathcal{H}$, where the involution is the adjoint, and the norm is the operator norm. Understanding this example is crucial to understanding the general case $C^*$-algebra, due to the existence of representations on closed subalgebras of $\mathcal{B}(\mathcal{H})$ given by the Gelfand-Naimark theorem~\ref{GelfandNaimark}.
\demo
\end{ex}

\begin{defn} \label{generated-Algebra-defn}
Given a set of subsets of a $C^*$-algebra, $S_i \in \mathfrak{A}$, the $C^*$-algebra \textit{generated by} the collection $S_i$ is, intuitively, the smallest $C^*$-algebra containing the union of $S_i$. Explicitly, the generated algebra is the intersection of every $C^*$-subalgebra of $\mathfrak{A}$ which contains every $S_i$. This algebra is denoted by $\vee_i S_i$. The case of a $C^*$-algebra generated by two subsets, $S_{1}, S_2$ is denoted by $S_1 \vee S_2$.
\end{defn}

\begin{defn}
Given a subset, $S$, of a ring, $R$, the \textit{commutant}, denoted by $S'$, of $S$ is the set of all elements in $R$ which commute with respect to multiplication with each element of $S$,
\begin{align}
S' := \{r \in R \,|\, \forall s \in S, r s = s r\}.
\end{align}
\end{defn}

\subsection{Homomorphisms}

\begin{defn}
Given two associative algebras, $\mathfrak{A}$ and $\mathfrak{B}$, a linear map, $\phi:\mathfrak{A} \rightarrow \mathfrak{B}$, satisfying
\begin{align}
\phi(xy) &= \phi(x) \phi(y)
\end{align}
is an \textit{algebra homomorphism}. If $\mathfrak{A}$ and $\mathfrak{B}$ are $*$-algebras, a ${}^*$-\textit{homomorphism} is an algebra homomorphism which satisfies
\begin{align}
\phi(x^*) &= \phi(x)^*.
\end{align} 
A bijective ${}^*$-homomorphism is a ${}^*$-\textit{isomorphism}, and a ${}^*$-isomorphism whose codomain is equal to its domain is a ${}^*$-\textit{automorphism}.
\end{defn}

Homomorphisms can be interpreted as structure preserving maps between algebras. Often, one can utilize a homomorphism to map the solution to a problem in one algebra to a problem in another algebra.


\begin{theorem}[Gelfand-Naimark  \cite{gelfand1943imbedding}] \label{GelfandNaimark}
Every $C^*$-algebra is isometrically $*$-isomorphic to a subalgebra of bounded linear operators on a Hilbert space.
\end{theorem}

\subsection{Unital Algebras} \label{Unital-Algebras-Section}
The importance of unital algebras stems from their ability to define inverses and, consequently, spectra of algebra elements.

\begin{defn}
A \textit{unital} associative algebra is an associative algebra, $\mathfrak{A}$ with a two-sided identity with respect to the multiplication, $\gls{id}$, satisfying $\gls{id} a = a = a \gls{id}$ for all $a \in \mathfrak{A}$.
\end{defn}

The identity of a unital associative algebra is unique. The identity of a ${}^*$-algebra satisfies 
\begin{align}
\gls{id} = \gls{id}^*.
\end{align}
In a $C^*$-algebra,
\begin{align}
||\gls{id}|| = 1.
\end{align}

It is often safe to assume that an algebra is unital, since for every algebra which is not unital, we can define a unital algebra, $\tilde{\mathfrak{A}}$, such that $\mathfrak{A}$ is isometrically $*$-isomorphic to a subalgebra of $\tilde{\mathfrak{A}}$.

\begin{defn}
Let $\mathfrak{A}$ be an algebra which is not unital. Consider the set of pairs,
\begin{align}
\tilde{\mathfrak{A}} := \{ (\alpha, a) \,| \, \alpha \in \mathbb{C}, a \in \mathfrak{A} \}.
\end{align}
Given $\alpha, \beta \in \mathbb{C}$, $a,b \in \mathfrak{A}$ the following rules define addition and multiplication in $\tilde{\mathfrak{A}}$:
\begin{align}
\beta (\alpha, a) &:= (\alpha \beta, \beta a) \\
(\alpha, a)(\beta, b) &:= (\alpha \beta, \beta a + \alpha b + ab) \\
(\alpha, a) + (\beta, b) &:= (\alpha + \beta, a + b)
\end{align}
The set $\tilde{\mathfrak{A}}$ with addition and multiplication defined as above is an algebra, referred to as the \textit{unitization} of $\mathfrak{A}$. $\tilde{\mathfrak{A}}$ is a unital algebra with the identity $\gls{id} = (1, 0)$. 
\end{defn}

\begin{proposition}
Let $\mathfrak{A}$ be an associative algebra without an identity. 
If $\mathfrak{A}$ is a ${}^*$-algebra, then $\tilde{\mathfrak{A}}$ with the involution
\begin{align}
(\alpha, a)^* := (\alpha^*, a^*)
\end{align}
is a ${}^*$-algebra. If $\mathfrak{A}$ is a normed algebra, then $\tilde{\mathfrak{A}}$ is a normed algebra with the norm
\begin{align}
||(\alpha, a)|| &:= \sup\{||\alpha b + a b|| \,|\, b \in \mathfrak{A} \text{ and } ||b|| = 1\}.
\end{align}
If $\mathfrak{A}$ is a $C^*$-algebra, then so is $\tilde{\mathfrak{A}}$ \cite[prop. 2.1.5]{bratteli1987operator}.
\end{proposition}

\begin{defn}
A map, $\phi$, whose domain and codomains are unital associative algebras is called \textit{unital} if $\phi(\gls{id}) = \gls{id}$.
\end{defn}

\begin{ex} \label{leftMultHomo}
Let $\mathfrak{A}$ be a unital normed algebra. Then, the map $\pi_l:\mathfrak{A} \to \mathcal{B}(\mathfrak{A})$ defined by 
\begin{align}
\pi_l(\xi) \eta = \xi \eta
\end{align}
is a unital algebra homomorphism and an isometry.
\demo
\end{ex}

To conclude this section, I discuss the notion of invertibility for unital associative algebras.
\begin{defn} \label{inverseDefn}
Given a unital associative algebra, $\mathfrak{A}$, an element, $a \in \mathfrak{A}$, is \textit{invertible} if there exists an element, $b \in \mathfrak{A}$, such that $a b = b a = \mathbb{1}$, and $b$ is called an \textit{inverse} of $a$. If an element is invertible, its inverse is unique, and is denoted by $a^{-1}$. The set of all invertible elements in $\mathfrak{A}$ is a group, referred to as the \textit{general linear group} of $\mathfrak{A}$, denoted by $\text{GL}(\mathfrak{A})$. The general linear group of a unital associative algebra always contains the identity, since $\mathbb{1}^{-1} = \mathbb{1}$.
\end{defn}

\begin{defn} \label{similarDefn}
Two elements in a unital associative algebra, $a, b \in \mathfrak{A}$, are \textit{similar} if there exists $s \in \text{GL}(\mathfrak{A})$, referred to as a \textit{similarity transform}, such that $s a s^{-1} = b$.
\end{defn}

\begin{lemma}
Suppose $\phi$ is a $*$-homomorphism whose domain and co-domain are unital $C^*$-algebras. Then $\phi$ is unital if and only if $\phi(\mathbb{1})$ is invertible.
\end{lemma}
\begin{proof}
By \cref{inverseDefn}, the identity in the co-domain of $\phi$ is invertible. Thus, we only need to demonstrate that if $\phi(\mathbb{1})$ is invertible, then $\phi(\mathbb{1}) = \mathbb{1}$. This follows from the following brief computation,
\begin{align}
\mathbb{1} &= \phi(\mathbb{1})^{-1} \phi(\mathbb{1}) \\
&= \phi(\mathbb{1})^{-1} \phi(\mathbb{1} \mathbb{1}) \\
&= \phi(\mathbb{1})^{-1} \phi(\mathbb{1}) \phi(\mathbb{1}) \\
&= \phi(\mathbb{1}).
\end{align}
\end{proof}

\subsection{Spectral Theory} \label{Section:Spectral-Theory}


\begin{defn}[Spectrum] \label{spectrumDefn}
For every element of a unital associative algebra, $a \in \mathfrak{A}$, we define its \textit{spectrum}, $\sigma_{\mathfrak{A}}(a)$, as 
\begin{align}
\sigma_{\mathfrak{A}}(a) := \{\lambda \in \mathbb{C} \,|\, \lambda \gls{id} - a \notin \text{GL}(\mathfrak{A})\}.
\end{align}
If $\mathfrak{A}$ is not unital, then we define the spectrum of $a \in \mathfrak{A}$ through unitization,
\begin{align}
\sigma_{\mathfrak{A}}(a) &:= \sigma_{\tilde{\mathfrak{A}}}(a).
\end{align}
If $A$ is an operator on a Banach space, $X$, the \textit{spectrum} of $A$ is the set
\begin{align}
\sigma(A) = \left\{\lambda \in \mathbb{C} | \lambda \mathbb{1} - A \notin \text{GL}(\mathcal{B}(X))\right\}.
\end{align}
The \textit{spectral radius}, $\text{spr}$, of an element of a Banach algebra is the supremum of the norm of the elements in the spectrum of this element,
\begin{align}
\text{spr}(a) := \sup\{|\lambda|\,|\,\lambda \in \sigma(a)\}.
\end{align} 
\end{defn}

The following theorem can be found in \cite[prop. 2.8]{TakesakiI}.
\begin{theorem}[Spectral Mapping] \label{spectralCalculus}
For every element of a unital Banach algebra, $a \in \mathfrak{A}$, and for every polynomial, $p$, the following identity holds:\footnote{The spectral mapping theorem can be generalized to apply to functions which are holomorphic on an open neighbourhood of the spectrum of $a$, although we have no need to do so in this thesis.}
\begin{align}
\sigma_{\mathfrak{A}}(p(a)) &= p(\sigma_{\mathfrak{A}}(a)).
\end{align}
If $a$ is invertible,
\begin{align}
\sigma_{\mathfrak{A}}(a^{-1}) &= (\sigma_{\mathfrak{A}}(a))^{-1}.
\end{align}
If $a$ is an element of a $C^*$-algebra, then
\begin{align}
\sigma(a^*) &= (\sigma(a))^*.
\end{align}
\end{theorem}

\begin{theorem}[Spectral Radius] \label{thm:Spectral-Radius}
For every element, $a$, of a Banach algebra, 
\begin{enumerate}
\item ${\normalfont \text{spr}}(a) \leq ||a||$, \\
\item ${\normalfont \text{spr}}(a) = \lim_{n \to \infty} ||a^n||^{1/n}$.
\end{enumerate}
\end{theorem}
\begin{corollary}
If $a$ is an element of a $C^*$-algebra, then $||a|| = \sqrt{{\normalfont \text{spr}}(a^* a)}$.
\end{corollary}
\begin{proof}
One proof of this corollary utilizes the fact that the involution in a $C^*$-algebra is an isometry \cite[Thm. 16.1]{Doran2018}.
\end{proof}

\begin{theorem}
If $a \in \mathfrak{A}$ and $a \in \mathfrak{B}$, where $\mathfrak{A}, \mathfrak{B}$ are $C^*$-algebras, then $\sigma_{\mathfrak{A}}(a) = \sigma_{\mathfrak{B}}(a)$. Thus, in the case where $a$ is an element of a $C^*$-algebra, $\mathfrak{A}$, the shorthand $\sigma(a) := \sigma_{\mathfrak{A}}(a)$ can be safely used.
\end{theorem}

\begin{ex} \label{matrixAlgebraExample}
For every $C^*$-algebra, $\mathfrak{A}$, and $n \in \mathbb{N}$, we can construct a $C^*$-algebra denoted by $\gls{MnA}$. This algebra is the set of $n \times n$ matrices whose matrix elements are elements of $\mathfrak{A}$, where the product of $M, M' \in \mathfrak{A}$ is given by the matrix product,
\begin{align}
(M M')_{ik} = \sum_{j = 1}^n M_{ij} M'_{jk}.
\end{align} 
Given a matrix with entries $M_{ij}$, the involution in $\gls{MnA}$ is defined as $(M^*)_{ij} = (M_{ji})^*$. The norm is defined via a $*$-homomorphism from $\gls{MnA}$ to $\oplus_{i = 1}^n \mathcal{B}(\mathcal{H})$, where $\mathcal{H}$ is a Hilbert space given by the Gelfand-Naimark theorem~\eqref{GelfandNaimark}, see \cite[p.2-3]{paulsen2002completely} for details. When $\mathfrak{A} = \mathbb{C}$, the norm is the \textit{spectral norm} of \cref{spectralNorm} and the involution is complex conjugate transposition, as in \cref{conjugateTranspose}.

If $\mathfrak{A}$ is a unital $C^*$-algebra, then so is $\gls{MnA}$. In this case, we will use the shorthand 
\begin{align}
\gls{GLnA} := \text{GL}(\mathfrak{M}_n(\mathfrak{A})). \label{GnAdefn}
\end{align}
\demo
\end{ex}

\begin{theorem} \label{homomorphism-Contraction}
If $\phi$ is a ${}^*$-homomorphism whose domain and codomain are $C^*$-algebras, then $\phi$ is a contraction. If $\phi$ is a unital ${}^*$-homomorphism, then $\sigma(\phi(a)) \subseteq \sigma(a)$ for all $a \in \text{Dom}(\phi)$. If $\phi$ is a unital ${}^*$-isomorphism, then $\phi$ is an isometry and $\sigma(\phi(a)) = \sigma(a)$ {\normalfont \cite[Thm. 4.1.8]{KadisonRingroseI}}. 
\end{theorem}

\begin{defn} \label{pointSpectrum}
Given an operator on a Banach space, $A:\text{Dom}(A) \to X$, with $\text{Dom}(A) \subseteq X$, $\lambda \in \mathbb{C}$ is referred to as a \textit{eigenvalue} of $A$ if $\{0\} \subsetneq \ker(\lambda \mathbb{1} - A)$. A vector in $\ker(\lambda \mathbb{1} - A)$ is referred to as an \textit{eigenvector} associated to $\lambda$. The set of all eigenvalues of $A$ is referred to as the \textit{point spectrum} of $A$, and is denoted by $\sigma_p(A)$. Denoting the range of $\lambda \mathbb{1} - A$ by $\mathcal{R}(\lambda \mathbb{1} - A)$, we can further partition the point spectrum into a union of disjoint subsets,
\begin{align}
\sigma_{p,1}(A) &= \{\lambda \in \sigma_p(A)\,|\, \text{cl}(\mathcal{R}(\lambda \mathbb{1} - A)) \neq X\}\\
\sigma_{p,2}(A) &= \{\lambda \in \sigma_p(A)\,|\, \text{cl}(\mathcal{R}(\lambda \mathbb{1} - A)) = X\}.
\end{align}
\end{defn}

\begin{proposition}
Given an operator, $A$, on a finite-dimensional Banach space, $\sigma(A) = \sigma_{p,1}(A)$.
\end{proposition}

\begin{defn} \label{continuousResidual}
Given an operator on a Banach space, $A:\text{Dom}(A) \to X$, with $\text{Dom}(A) \subseteq X$, the \textit{continuous spectrum}, denoted by $\sigma_c(A)$, is the set of all $\lambda \in \sigma(A)$ such that $\lambda \mathbb{1} - A$ is injective, but not surjective, and has dense range. The set of all elements of the spectrum of $A$ which are not included in the point or continuous spectrum is referred to as the \textit{residual spectrum}, and is denoted by $\sigma_r(A)$.
\end{defn}

\subsection{Positivity} \label{section:Positive}

I initiate this section's discussion of positivity in the abstract setting of a $C^*$-algebra. Then, I relate the notions of positivity to the specific $C^*$-algebra of bounded operators on a Hilbert space.

\begin{defn} \label{positivityDef}
An element, $a$, of a $C^*$-algebra, $\mathfrak{A}$, is \textit{normal} if and only if $a a^* = a^* a$, \textit{Hermitian} if and only if $a = a^*$, and \textit{positive} if and only if $a = b^* b$ for some $b \in \mathfrak{A}$. All Hermitian elements are normal, and all positive elements are Hermitian. I use the abbreviation $a \geq 0$ to mean that $a$ is positive and the set of positive elements of $\mathfrak{A}$ is denoted as $\mathfrak{A}^+$.
\end{defn}

\begin{theorem} \label{thm-Characterize-Positive}
The following definitions for a positive element, $a \in \mathfrak{A}^+$, of a $C^*$-algebra, $\mathfrak{A}$, are equivalent {\normalfont \cite[chp. 2.2.2]{bratteli1987operator}}: 
\begin{enumerate}
\item There exists $b \in \mathfrak{A}$ such that $a = b^* b$ (\cref{positivityDef}). 

\item There exists a unique positive element, $b \in \mathfrak{A}^+$, which is the square root of $a$, $a = b^2$.

\item $a$ is normal and $\sigma(a) \subseteq [0, \infty)$.
\end{enumerate}
\end{theorem} 

\begin{defn}
An element of a unital $C^*$-algebra is called \textit{strictly-positive} if and only if it is positive and invertible.
\end{defn}

\begin{lemma}
The inverse of a strictly-positive element is also strictly-positive. 
\end{lemma}

\begin{defn}
An operator, $\eta$, on a Hilbert space, $\mathcal{H}$, with $\text{Dom}(\eta) \subseteq \mathcal{H}$ is called \textit{positive-definite} if and only if for every $\psi \in \text{Dom}(\eta) \setminus \{0\}$,
\begin{align}
\braket{\psi| \eta \psi} > 0
\end{align}
\end{defn}

\begin{theorem} \label{thm-Characterize-Positive-Hilbert}
A bounded operator on a Hilbert space, $a \in \mathcal{B}(\mathcal{H})$ is:
\begin{enumerate}
\item \textbf{positive} if and only if for every $\xi \in \mathcal{H}$, $\braket{\xi| a \xi} \geq 0$ {\normalfont \cite[\S 104]{riesz2012functional}}. \\
\item \textbf{positive-definite} if and only if it is positive and injective. \\
\item \textbf{strictly-positive} if and only if there exists a positive real number, $c > 0$, such that $\braket{\xi| a \xi} \geq c$ for all $\xi \in \mathcal{H}$. \label{strictly-positive-characterize}
\end{enumerate}
\end{theorem}
\begin{proof}
While the proof is straightforward, I mention that \cref{strictly-positive-characterize} follows from the Banach bounded inverse \cref{BoundedInverseTheorem}.
\end{proof}

\subsection{Complete Positivity}

Completely positive maps are used in physics to define quantum channels.
One textbook devoted to complete positivity is \cite{paulsen2002completely}.

\begin{defn}
A map, $Q: \mathfrak{A} \rightarrow \mathfrak{B}$, where $\mathfrak{A}, \mathfrak{B}$ are both $C^*$-algebras\footnote{Completely positive maps are defined in the same way for maps whose domain is an \textit{operator system}, which is a $*$-closed vector subspace of a $C^*$-algebra.}, is called \textit{positive} if $a \geq 0$ implies $Q(a) \geq 0$.
$Q$ is \textit{completely positive} if for all $n \in \gls{Z+}$, the map $Q_n: \gls{MnA} \rightarrow \mathfrak{M}_n(\mathfrak{B})$ defined by $(Q_n(M))_{ij} = (Q(M_{ij}))$ is positive.
\end{defn}

\begin{ex}
$^*$-homomorphisms are completely positive.
\demo
\end{ex}

\begin{ex} \label{cpexample}
Given a unital $C^*$-algebra, $\mathfrak{A}$, and an element $b \in \mathfrak{A}$, the map $Q: \mathfrak{A} \to \mathfrak{A}$ defined by $Q(a) = b^* a b$ is completely positive. 
\demo
\end{ex}

When $b$ from \cref{cpexample} is a unitary on a Hilbert space, corresponding to time evolution, a physicist would identify the operation $Q(a)$ as the effect of evolving $a$ in the Heisenberg picture.

\begin{thm}[Choi's theorem \cite{Choi1975}] \label{Choi-Theorem}
$\Phi:\mathfrak{M}_n(\mathbb{C}) \to \mathfrak{M}_m(\mathbb{C})$ is completely positive if and only if \begin{align}
\sum_{i = 1}^n \sum_{j = 1}^n E_{ij}\otimes \Phi(E_{ij}) \geq 0,
\end{align}
where $E_{ij} \in \mathfrak{M}_n(\mathbb{C})$ is the matrix with entries $(E_{ij})_{kl} := \delta_{ik} \delta_{jl}$.
\end{thm}

\begin{defn} \label{stateDef}
A \textit{normalized} linear functional on a $C^*$-algebra, $\varphi:\mathfrak{A}\rightarrow \mathbb{C}$, satisfies $||\varphi|| = 1$. A positive normalized linear functional is referred to as a \textit{state}.
\end{defn}
\begin{lemma}
A positive linear functional on a unital $C^*$-algebra is a state if and only if the functional is unital.
\end{lemma}

The following theorem characterizes states on the algebra of finite-dimensional matrices as \textit{density matrices}, which are positive trace-class operators with trace equal to one.
\begin{theorem}
$\varphi$ is a state on $\mathfrak{M}_n(\mathbb{C})$ if and only if there exists $\rho \in \mathfrak{M}_n(\mathbb{C})^+$ satisfying $\text{Tr}(\rho) = 1$ such that 
\begin{align}
\varphi(a) = {\normalfont \text{Tr} {}{}} (\rho a). \label{finiteDimensionalState}
\end{align} 
\end{theorem}
\begin{proof}
One proof is given in \cite{finiteC*}. I present an alternative here.

Given a linear functional of the form \cref{finiteDimensionalState}, the assumption $\text{Tr}(\rho) = 1$ implies $\varphi$ is normalized. By the spectral theorem for positive matrices, there exists a decomposition $\rho = \sum_k p_k \ket{k}\bra{k}$, where $p_k \geq 0$ and $\ket{k} \in \mathbb{C}^n$ is an orthonormal basis, so the proof that $\varphi$ is positive follows from $\text{Tr}(\rho M) = \sum_k p_k \braket{k| M | k}$. 

For the other direction of the proof, note $\mathfrak{M}_n(\mathbb{C})$ is a Hilbert space with the Frobenius inner product, which is defined in \cref{hilbertSchmidtHilbertAlg}. Thus, by the Riesz representation \cref{Riesz}, every linear functional $\varphi: \mathfrak{M}_n(\mathbb{C}) \to \mathbb{C}$ can be expressed as 
\begin{align}
\varphi(M) = \braket{\rho|M}_{HS}
\end{align}
for some $\rho \in \mathfrak{M}_n(\mathbb{C})$. For $\varphi$ to be normalized, we must have $\text{Tr}(\rho) = 1$. Assuming $\varphi$ is positive implies $\rho \geq 0$, as proven by contradiction: If $\rho$ is not positive, then there exists $\ket{\xi} \in \mathbb{C}^n$ such that 
\begin{align}
\varphi(\ket{\xi} \bra{\xi}) &= \text{Tr} (\rho \ket{\xi} \bra{\xi}) \\
&= \braket{\xi| \rho |\xi} < 0,
\end{align}
which is a contradiction since $\ket{\xi} \bra{\xi}$ is positive.
\end{proof}

\section{Measure Theory}
\begin{definition} \label{sigmaAlgAppendix}
A $\sigma$-\textit{algebra} over the set $X$ is a collection of subsets, $\Sigma \subseteq \mathbb{P}(X)$, satisfying
\begin{enumerate}
\item $X \in \Sigma$ 
\item $E \in \Sigma \, \Rightarrow \, X \setminus E \in \Sigma$ 
\item If $\{E_i\}_{i \in \gls{N}} \subseteq \Sigma$, then $\cup_{i \in \gls{N}} E_i \in \Sigma$.
\end{enumerate}
\end{definition}

\begin{definition} \label{Borel}
The $\sigma$-algebra \textit{generated by} a collection of subsets, $S \in \mathbb{P}(X)$, is the intersection of every $\sigma$-algebra over $X$ which contains every element of $S$. Given a topological space, $X$, the \textit{Borel} $\sigma$-\textit{algebra}, denoted by $\mathfrak{B}_X$, is the $\sigma$-algebra generated by the topology on $X$. 
\end{definition}

\begin{ex}
The Borel $\sigma$-algebra generated by a discrete topological space is the discrete topology.
\demo
\end{ex}

\begin{definition}
Given a $\sigma$-algebra, $\Sigma$, a \textit{positive measure} is a map, $\mu:\Sigma \to [0, +\infty]$ which satisfies $\mu(\emptyset) = 0$ and 
\begin{align}
\mu(\sum_{i \in \gls{N}} E_i) = \sum_{i \in \gls{N}} \mu(E_i) &\quad& \text{if } \{E_i\}_{i \in \gls{N}} \subseteq \Sigma \text{ and } E_i \cap E_j = \emptyset \text{ when } i \neq j.
\end{align}
\end{definition}

\begin{ex}
The \textit{Lebesgue measure} is a measure over the $\sigma$-algebra of subsets of $\mathbb{R}^n$ satisfying the \textit{Carathéodory criterion}, which contains the Borel $\sigma$-algebra generated by the Euclidean topology on $\mathbb{R}^n$. The Lebesgue measure was introduced in \cite{Lebesgue1902} and Carathéodory's criterion was introduced in \cite{caratheodory1914lineare,Edgar2019}.
\demo
\end{ex}

\begin{definition} \label{operator-Valued-Measures}
Let $\Sigma$ be a $\sigma$-algebra and let $\mathcal{H}$ be a Hilbert space. An operator valued measure, $E$, is a map, $E: \Sigma \to \mathcal{B}(\mathcal{H})$ that is weakly countably additive. More precisely, if $B_i \in \Sigma$ is a countable collection of disjoint sets with union $\cup_i B_i = B$, then
\begin{align}
\braket{E(B) x, y} = \sum_i \braket{E(B_i) x, y}. \label{weakAdditive}
\end{align}
for all $x, y \in \mathcal{H}$. Furthermore,
\begin{itemize}
%
%
%

\item $E$ is \textit{self-adjoint} if $E(B)$ is self-adjoint for all $B$.

\item $E$ is \textit{spectral} if $E(B_1 \cap B_2) = E(B_1) E(B_2)$ for all $B_1, B_2 \in \Sigma$.

\item $E$ is \textit{projection-valued} if it is self-adjoint, spectral, and $E(X) = \gls{id}$.
\end{itemize}
\end{definition}

\begin{ex}
Let $\mathfrak{X}$ be an orthonormal basis of a Hilbert space. The induced norm topology on $\mathfrak{X}$ is the discrete topology on $\mathfrak{X}$. The map $E: \mathbb{P}(B) \to \mathcal{B}(\mathcal{H})$ defined by $E(S) := \mathscr{P}_{\text{cl}(\text{span}(S)}$ is a projective valued measure.
\demo
\end{ex}

\begin{ex}
Given two Hilbert spaces, $\mathcal{H}$ and $\mathcal{K}$, a unital $*$-homomorphism, $\pi: \mathcal{B}(\mathcal{H}) \to \mathcal{B}(\mathcal{K})$, and a projective valued measure $E: \mathfrak{B}_X \to \mathcal{B}(H)$, the map $\pi \circ E$ is a projective valued measure.
\demo
\end{ex}

\textit{Integration} is a linear map on (measurable) functions over measure spaces. Integration can be performed with respect to measures either or with respect to projective-valued measures. No integrals with respect to operator-valued measures are explicitly performed in this thesis; thus, instead of explicitly defining integration, I refer the reader to a reference \cite[\S 8.3.2]{moretti2013spectral}.
\begin{theorem}[Spectral Theorem] \label{SpectralTheorem}
For every self-adjoint linear operator on a Hilbert space, $A$, there exists a unique projection-valued measure, $P_A:\mathfrak{B}_{\sigma(A)} \to \mathcal{B}(\mathcal{H})$ such that
\begin{align}
A = \int_{\sigma(A)} \lambda \, d P_A(\lambda).
\end{align} 
Furthermore, the spectrum of $A$ is a subset of the reals, 
\begin{align}
\sigma(A) \subseteq \mathbb{R}.
\end{align}
{\normalfont \cite[Thm. 10.4]{hall2013quantum}}
\end{theorem}

\section{Chebyshev Polynomials} \label{Chebyshev Appendix}
This section tabulates the formulas satisfied by the Chebyshev polynomials that are applied in this thesis. Many more properties of Chebyshev polynomials can be located in resources such as \cite{abramowitz1972handbook,olver2010nist,mason2002chebyshev}. The Chebyshev polynomials are sequences of orthogonal 
polynomials. All four polynomials, $P_n(x)$, are solutions to the linear recurrence relation
\begin{align}
P_n(x) = 2 x P_{n-1}(x) - P_{n-2}(x),
\end{align}
where $n \in \mathbb{Z}$ and $x \in \mathbb{C}$.
These polynomials merely differ in prescriptions for the initial conditions of the recurrence relation. The four kinds of Chebyshev polynomials and their root sets, $P_n^{-1}(\{0\})$, are displayed in \cref{ChebyshevTable}.

\begin{table*}[htp!]
\centering
\begingroup
\setlength{\tabcolsep}{8pt} 
\renewcommand{\arraystretch}{2} 
\begin{tabular}{|m{3cm}|l|l|l|l|}
\hline
Kind, $P_n$
& First, $T_n$
& Second, $U_n$
& Third, $V_n$
& Fourth, $W_n$ \\
\hhline{|=|=|=|=|=|}
$P_n(\cos \theta)$
& $\cos(n \theta)$
& $\dfrac{\sin((n+1) \theta)}{\sin \theta}$ 
& $\dfrac{\cos((n+1/2) \theta)}{\cos(\theta/2)}$
& $\dfrac{\sin((n+1/2) \theta)}{\sin(\theta/2)}$ \\
\hline
$P_n\left(\dfrac{z + z^{-1}}{2}\right)$
& $\dfrac{z^n + z^{-n}}{2}$
& $\dfrac{z^{n+1}-z^{-n-1}}{z - z^{-1}}$
& $\dfrac{z^{n+1/2}+z^{-n-1/2}}{z^{1/2} + z^{-1/2}}$
& $\dfrac{z^{n+1/2}-z^{-n-1/2}}{z^{1/2} - z^{-1/2}}$ \\
\hline
\parbox{3cm}{$\arccos(P_n^{-1}(\{0\}))$ \\ $k \in \{1, \dots, n\}$}
& $\dfrac{(2k-1) \pi}{n} $ 
& $\dfrac{k \pi}{n+1}$ 
& $\dfrac{(2 k-1) \pi}{2 n+1} $ 
& $\dfrac{2 k \pi}{2 n+1} $ \\
\hline
$P_n(1)$
& $1$ 
& $n+1$ 
& $1$ 
& $2n+1$ \\
\hline
$P_n(-1)$
& $(-1)^n$ 
& $(-1)^n (n+1)$
& $(-1)^n (2n+1)$
& $(-1)^n$ \\
\hline
\end{tabular}
\centering
\caption{Definitions and roots of the Chebyshev polynomials, see \cite[chp. 2]{mason2002chebyshev} for more details. A root for a Chebyshev polynomial is derived from an entry in the second row by applying the cosine.}
\label{ChebyshevTable}
\endgroup
\end{table*}

Further linear relations between the Chebyshev polynomials include
\begin{align}
2T_n(x) &= U_n(x) - U_{n-2}(x)                &\quad& \text{\cite[eq. (1.7)]{mason2002chebyshev}}\\
V_n(x) &= U_n(x) - U_{n-1}(x)             &\quad& \text{\cite[eq. (1.17)]{mason2002chebyshev}} \label{ChebyshevV} \\
W_n(x) &= U_n(x) + U_{n-1}(x)             &\quad& \text{\cite[eq. (1.18)]{mason2002chebyshev}} \label{ChebyshevW} \\
U_{n \pm 1}(x) &= x U_n(x) \pm T_{n+1}(x) &\quad&
\end{align}
Consequently, properties of the any kind of Chebyshev polynomial can often be determined from corresponding properties of Chebyshev polynomials of the second kind. 

The derivative is
\begin{align}
U_n'(x) &= \begin{cases}
\dfrac{-nx U_n(x)+(n+1) U_{n+1}(x)}{1-x^2}& \text{if } x \neq \pm 1 \\
(\pm 1)^{n+1}\left( \dfrac{n(n+1)(n+2)}{3}\right)& \text{if } x = \pm 1
\end{cases} 
\end{align}

An interesting composition identity of the Chebyshev polynomials is\footnote{This formula appeared as a textbook problem before the publication of both references mentioned here, see \cite[\S 1.5]{mason2002chebyshev}.}
\cite[Lemma 3]{Zhang2004}\cite[Thm. 5]{Brandi2020}
\begin{align}
U_{mk-1}(x) = U_{k-1} \left(T_m (x) \right) U_{m-1} \left(x \right),
\end{align}

The resultant between two Chebyshev polynomials of the second kind is \cite{Jacobs2011,Louboutin2013}
\begin{align}
&\text{Res}_x\left(U_n(x), U_m(x)\right)\nonumber\\ &= \begin{cases}
0 & \text{if } \text{gcd}(m+1,n+1) > 1 \\
(-1)^{mn/2} * 2^{mn} & \text{if } \text{gcd}(m,n) = 1
\end{cases}.
\end{align}

Some nifty product formulae include
\begin{align}
U_{n-m-1}(x) &= U_{n-1}(x) U_{m}(x) - U_{m-1}(x) U_n(x)\\
U_{2m}(x) &= U_{m}^2(x) - U_{m-1}^2(x).
\end{align} 

A series expansion is
\begin{align}
U_n(x) &= \sum^{\floor{n/2}}_{i = 0} (-1)^m \frac{(n-m)!}{m! (n-2m)!} (2x)^{n-2m}.
\end{align}

\section{Proof of Theorem~\ref{unbrokenRealSpectrum}} \label{UnbrokenTheoremProof}

As a reminder, one of the claims proven in \cref{pseudoEquivThm} is that a matrix is similar to a matrix with real elements if and only if it has an involutive antilinear symmetry.

This section references the Jordan normal form of an operator and the notion of unbroken antilinear symmetry. In the interest of remaining self-contained, I review these concepts before their usage.
\begin{theorem}[Jordan Normal Form]
Given a matrix $A \in \mathfrak{M}_n(\mathbb{C})$, there exists a set of projectors, $P_\lambda$, and a set of nilpotents\footnote{An operator is \textit{nilpotent} if there exists a finite power of this operator which equals zero.}, $D_\lambda$, such that {\normalfont \cite[chp. I \S 5]{Kato1995}}
\begin{align}
A = \sum_{\lambda \in \sigma(A)} \lambda P_\lambda + D_\lambda.
\end{align}
The projectors and nilpotents satisfy the identities
\begin{align}
P_\lambda D_{\lambda'} &= D_\lambda P_{\lambda'} = \delta^\lambda_{\lambda'} P_\lambda \\
D_\lambda D_{\lambda'} &= \delta^\lambda_{\lambda'} D^2_\lambda \\
P_\lambda P_{\lambda'} &= \delta^\lambda_{\lambda'} P_\lambda \\
\sum_{\lambda \in \sigma(A)} P_\lambda &= \mathbb{1}.
\end{align}
\end{theorem}

\begin{defn}
Given an antilinear operator on a Hilbert space, $\Theta$, a linear operator, $A: \text{Dom}(A) \to \mathcal{H}$, is called $\Theta$-\textit{unbroken} if and only if $\Theta$ is an antilinear symmetry of $A$ and the eigenspaces of $A$ are invariant subspaces of $\Theta$,
\begin{align}
\Theta \ker (\lambda \mathbb{1} - A) \subseteq \ker (\lambda \mathbb{1} - A)  &\quad& \forall \lambda \in \sigma_p(A).
\end{align}
\end{defn}

\begin{thm}
The spectrum of a matrix, $A \in \mathfrak{M}_n(\mathbb{C})$, is a subset of the reals, $\sigma(A) \subset \mathbb{R}$, if and only if there exists a involutive antilinear symmetry, $\Theta$, so that $\Theta^2 = \gls{id}$, such that $A$ is $\Theta$-unbroken. If $A$ is $\Theta$-unbroken for one invertible antilinear symmetry, then $A$ is $\Theta$-unbroken for all invertible antilinear symmetries.
\end{thm}

\begin{proof}
First, I show that matrices with real eigenvalues have an unbroken antilinear symmetry satisfying $\Theta^2 = \gls{id}$, and that every antilinear symmetry of such a matrix is unbroken. Every matrix is similar to its Jordan normal form \cite[\S 7.3]{hoffmann1971linear}. In the case where a matrix has a real spectrum, its Jordan normal form is real, so by \cref{pseudoEquivThm}, a matrix with real eigenvalues has an antilinear symmetry satisfying $\Theta^2 = \gls{id}$.

For my next trick, I prove that a matrix with real eigenvalues, $A$, is $\Theta$-unbroken for every invertible antilinear $\Theta$ satisfying $\Theta A= A \Theta$. This is because for every eigenvalue $\lambda \in \sigma_p(A)$, since $\lambda$ is real by supposition,
\begin{align}
\{0\} &= (\lambda \mathbb{1} - A) \ker(\lambda \mathbb{1} - A) \\
&= \Theta (\lambda \mathbb{1} - A) \ker(\lambda \mathbb{1} - A) \\
&= (\lambda \mathbb{1} - A) \Theta \ker(\lambda \mathbb{1} - A).
\end{align}
Thus, $\ker (\lambda \mathbb{1} - A)$ is an invariant subspace of $\Theta$ for all $\lambda \in \sigma_p(A)$, so $\Theta$ is unbroken by definition. 

The second direction of the proof in the case where all eigenvalues of $A$ are simple is given in \cite{bender1999pt}. The general case follows from the definition of unbroken symmetry and \cref{PT-Unbroken-Lemma-2} of \cref{PT-Unbroken-Lemma}.
\end{proof}

\pagebreak

\begin{figure}
\includegraphics[height=.6\paperheight]{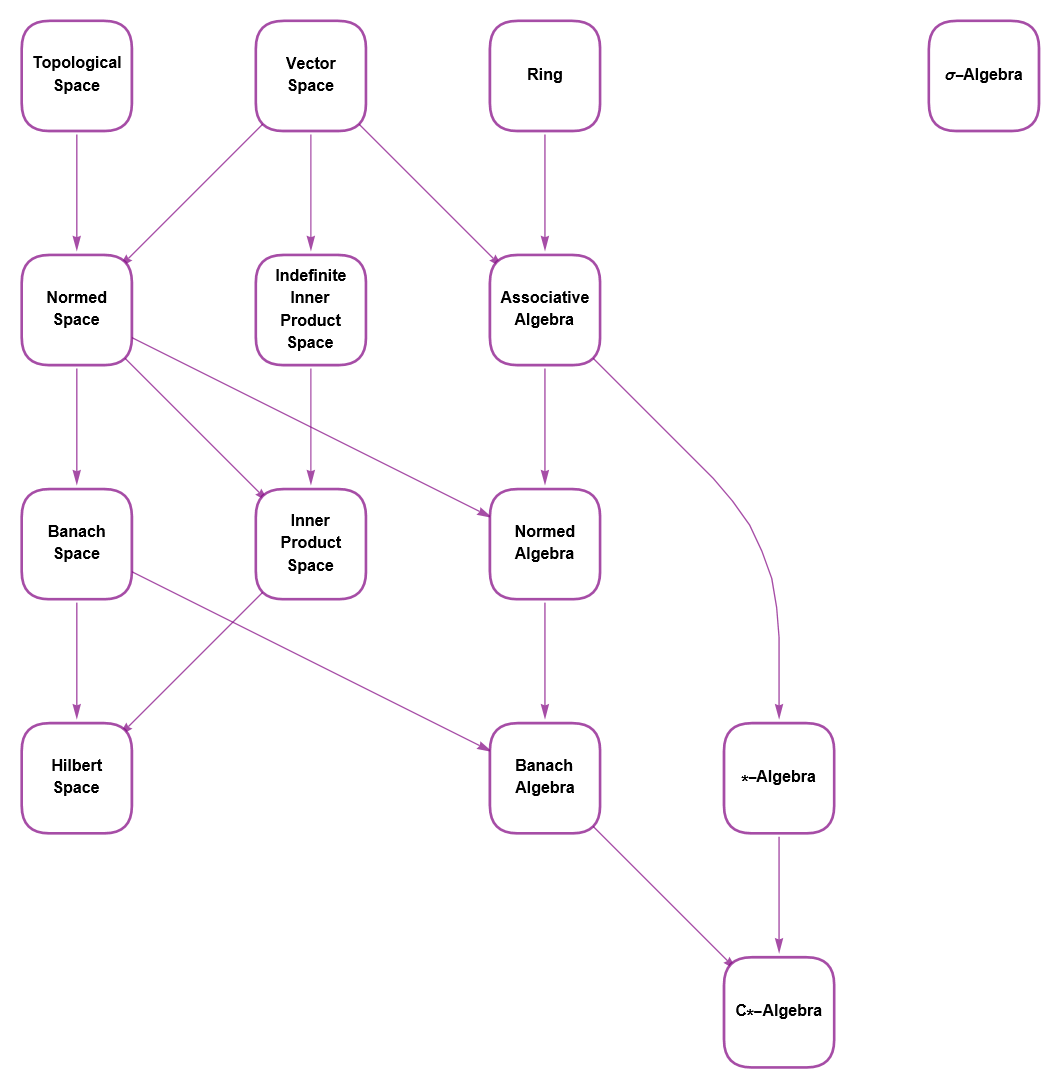}
\caption{This graph summarizes the mathematical spaces studied in this appendix and their relationships. An arrow pointing from one space to another indicates that all elements of the latter space are included in the former.}
\label{graphOfSpaces}
\end{figure}

\pagebreak



%


\end{document}